%% file: sigma24-099.tex
\pdfoutput=1
\RequirePackage{ifpdf}
\ifpdf 
\documentclass[pdftex]{sigma}
\else
\documentclass{sigma}
\fi

\usepackage[dvipsnames]{xcolor}
\usepackage{mathtools,bbm,mathdots,tensor}
\usepackage[]{dsfont}
\usepackage{empheq}
\usepackage{longtable}

\usepackage{framed}
\usepackage{manfnt}
\usepackage{simpler-wick}
\usepackage[new]{old-arrows}
\usepackage{stmaryrd}
\usepackage{epic}
\usepackage{eepic}
\usepackage[matrix,arrow]{xy}
\usepackage[calc,useregional]{datetime2}
\usepackage{mathrsfs}
\usepackage{pgfmath}
\usepackage{ytableau}
\usepackage{thmtools}

\usepackage{slashed,extarrows,bigints,cancel}
\usepackage{float,multirow,makecell}
\usepackage{bm}

\usepackage{booktabs}

\usepackage{subcaption}
\usepackage[labelsep=period,labelfont=bf,font=small]{caption}

\usepackage{pgfplots}
\usepackage{tikz}
\usetikzlibrary{cd, decorations.markings, decorations.pathmorphing, arrows,calc,knots,hobby,
	decorations.pathreplacing,
	shapes.geometric,
	calc, arrows, arrows.meta,patterns,intersections, pgfplots.fillbetween}

\usepackage{pgfplots}
\tikzset{snake it/.style={decorate, decoration=snake}}
\tikzset{->-/.style={decoration={
			markings,
			mark=at position #1 with {\arrow{>}}},postaction={decorate}}}
\tikzset{-<-/.style={decoration={
			markings,
			mark=at position #1 with {\arrow{<}}},postaction={decorate}}}

\tikzset{cross/.style={path picture={
			\draw[black]
			(path picture bounding box.south east) -- (path picture bounding box.north west) (path picture bounding box.south west) -- (path picture bounding box.north east);
}}}

\usepackage{tikz-cd}

\newcommand{\mscr}[1]{\mathscr{#1}}

\newcommand{\mfk}[1]{\mathfrak{#1}}
\newcommand{\mbs}[1]{\boldsymbol{#1}}
\newcommand{\mbb}[1]{\mathbb{#1}}

\newcommand{\mbf}[1]{\mathbf{#1}}
\newcommand{\msf}[1]{\mathsf{#1}}
\newcommand{\mcal}[1]{\mathcal{#1}}

\newcommand{\sd}{\msf{d}}
\newcommand{\si}{\sigma}

\newcommand{\pa}[1]{\partial_{#1}}

\newcommand{\ot}{\leftarrow}

\newcommand*\bbar[1]{%
	{\hbox{%
			\vbox{%
				\hrule height 0.5pt %
				\kern0.2ex%
				\hbox{%
					\kern-0.1em%
					\ensuremath{#1}%
					\kern-0.1em%
				}%
			}%
		}%
	}
}

\DeclareMathOperator{\SU}{\tenofo{SU}}

\DeclareMathOperator{\U}{\tenofo{U}}

\newcommand{\tenofo}[1]{\text{\normalfont #1}}
\newcommand{\wt}[1]{\widetilde{#1}}

\newcommand{\y}[1]{\msf{y}_{#1}}

\newcommand{\higgs}{\mathcal{M}_H}

\newcommand{\thdQ}{\mscr{Q}}
\newcommand{\twdQ}{\msf{Q}}

\newcommand{\A}{\tenofo{A}}
\newcommand{\B}{\tenofo{B}}
\newcommand{\C}{\tenofo{C}}

\theoremstyle{plain}
\newtheorem{thr}{Theorem}[section]
\newtheorem{conj}{Conjecture}[section]
\newtheorem{lem}{Lemma}[section]
\newtheorem{cor}{Corollary}[section]
\newtheorem{prop}{Proposition}[section]

\theoremstyle{definition}
\newtheorem{dfn}{Definition}[section]
\newtheorem{ex}{Example}[section]
\newtheorem{rmk}{Remark}[section]

\newenvironment{eqaligned*}
{%
	\begin{equation*}
		\begin{aligned}
		}
		{%
		\end{aligned}
	\end{equation*}
	\ignorespacesafterend}

\newenvironment{eqgathered}
{%
	\begin{equation}
		\begin{gathered}
		}
		{%
		\end{gathered}
	\end{equation}
	\ignorespacesafterend}

\newenvironment{eqgathered*}
{%
	\begin{equation*}
		\begin{gathered}
		}
		{%
		\end{gathered}
	\end{equation*}
	\ignorespacesafterend}

\definecolor{mygrey}{rgb}{.45,.4,.7}
\definecolor{mypurple}{RGB}{127, 73, 235}
\definecolor{mycyan}{rgb}{.1,.75,.95}
\definecolor{myred}{rgb}{.9,.1,.15}
\definecolor{myblue}{rgb}{.05,.15,.85}
\definecolor{mygreen}{rgb}{.05, .6, .1}

\newcommand{\rf}[1]{(\ref{#1})}

\newcommand{\tr}{\operatorname{tr}}
\newcommand{\dd}{\text{d}}
\newcommand{\id}{\mathds{1}}
\newcommand{\wtd}{\widetilde}
\newcommand{\what}{\widehat}
\newcommand{\ov}[1]{\overline{#1}}

\renewcommand{\Im}{\text{Im}}
\newcommand{\Hom}{\text{Hom}}
\newcommand{\End}{\text{End}}

\newcommand{\lan}{\left\langle}
\newcommand{\ran}{\right\rangle}

\newcommand{\bra}[1]{\left\langle #1 \right|}
\newcommand{\ket}[1]{\left|#1\right\rangle}

\newcommand{\mcr}{\mathscr}
\newcommand{\mfr}{\mathfrak}

\newcommand{\cF}{{\mathcal F}}

\newcommand{\cH}{{\mathcal H}}

\newcommand{\cJ}{{\mathcal J}}
\newcommand{\cK}{{\mathcal K}}
\newcommand{\cL}{{\mathcal L}}
\newcommand{\cM}{{\mathcal M}}
\newcommand{\cN}{{\mathcal N}}

\newcommand{\cQ}{{\mathcal Q}}
\newcommand{\cR}{{\mathcal R}}

\newcommand{\bC}{\mathbb{C}}

\newcommand{\bE}{\mathbb{E}}

\newcommand{\bL}{\mathbb{L}}
\newcommand{\bM}{\mathbb{M}}

\newcommand{\bP}{\mathbb{P}}
\newcommand{\bR}{\mathbb{R}}
\newcommand{\bV}{\mathbb{V}}
\newcommand{\bW}{\mathbb{W}}
\newcommand{\bZ}{\mathbb{Z}}

\newcommand{\sfA}{\mathsf{A}}

\newcommand{\sfD}{\mathsf{D}}
\newcommand{\sfF}{\mathsf{F}}
\newcommand{\sfG}{\mathsf{G}}

\newcommand{\sfQ}{\mathsf{Q}}

\newcommand{\scrN}{\mathscr{N}}

\newcommand{\ii}{\text{i}}

\DeclareMathAlphabet{\mathpzc}{OT1}{pzc}{m}{it}

\newcommand{\GL}{\text{GL}}

\newcommand{\gl}{\mathfrak{gl}}
\renewcommand{\sl}{\mathfrak{sl}}
\newcommand{\su}{\mathfrak{su}}

\newcommand{\al}{\alpha}
\newcommand{\be}{\beta}
\newcommand{\ga}{\gamma}

\newcommand{\de}{\delta}
\newcommand{\De}{\Delta}

\newcommand{\ep}{\epsilon}
\newcommand{\vep}{\varepsilon}
\newcommand{\ze}{\zeta}
\newcommand{\tht}{\theta}
\newcommand{\vth}{\vartheta}

\newcommand{\ka}{\kappa}
\newcommand{\la}{\lambda}
\newcommand{\La}{\Lambda}

\newcommand{\Om}{\Omega}
\newcommand{\Si}{\Sigma}

\newcommand{\Up}{\Upsilon}

\newcommand{\beq}{\begin{equation}}
	\newcommand{\eeq}{\end{equation}}
\newcommand{\bea}{\begin{eqnarray}}
	\newcommand{\eea}{\end{eqnarray}}
\newcommand{\nn}{\nonumber}

\usepackage{environ}
\NewEnviron{draftnote}{
	\color{gray}
	\noindent\colorbox{gray!70}{\textcolor{white}{Draft note}}
	\hrule
	\BODY
	\hrule
}

\newcommand{\efunction}{\bm{f}}
\newcommand{\estab}{\mbf{Stab}}
\newcommand{\ermatrix}{\mbf{R}}

\newcommand{\tfunction}{\msf{f}}
\newcommand{\tstab}{\msf{Stab}}
\newcommand{\trmatrix}{\msf{R}}

\newcommand{\rfunction}{\mathrm{f}}
\newcommand{\rstab}{\mathrm{Stab}}
\newcommand{\rrmatrix}{\mathrm{R}}

\input{Figures_and_Tables.tex}

\numberwithin{equation}{section}

\newtheorem*{Theorem*}{Theorem}

 { \theoremstyle{definition}

 }

\begin{document}

\allowdisplaybreaks

\newcommand{\arXivNumber}{2308.12333}

\renewcommand{\PaperNumber}{099}

\FirstPageHeading

\ShortArticleName{Elliptic Stable Envelopes and Dynamical $R$-Matrices}

\ArticleName{Elliptic Stable Envelopes for Certain Non-Symplectic\\ Varieties and Dynamical $\boldsymbol{R}$-Matrices for Superspin\\ Chains from the Bethe/Gauge Correspondence}

\Author{Nafiz ISHTIAQUE~$^{\rm a}$, Seyed Faroogh MOOSAVIAN~$^{\rm b}$ and Yehao ZHOU~$^{\rm c}$}

\AuthorNameForHeading{N.~Ishtiaque, S.F.~Moosavian and Y.~Zhou}

\Address{$^{\rm a)}$~Institut des Hautes \'Etudes Scientifiques, 35 Rte de Chartres, 91440 Bures-sur-Yvette, France}
\EmailD{\href{ishtiaque@ihes.fr}{ishtiaque@ihes.fr}}

\Address{$^{\rm b)}$~Department of Physics, McGill University, Ernest Rutherford Physics Building,\\
\hphantom{$^{\rm b)}$}~3600 Rue University, Montr\'eal, QC H3A 2T8, Canada}
\EmailD{\href{mailto:sfmoosavian@gmail.com}{sfmoosavian@gmail.com}}

\Address{$^{\rm c)}$~Kavli Institute for the Physics and Mathematics of the Universe (WPI), University of Tokyo,\\
\hphantom{$^{\rm c)}$}~Kashiwa, Chiba 277-0882, Japan}
\EmailD{\href{yehao.zhou@ipmu.jp}{yehao.zhou@ipmu.jp}}

\ArticleDates{Received February 07, 2024, in final form October 11, 2024; Published online October 31, 2024}

\Abstract{We generalize Aganagic--Okounkov's theory of elliptic stable envelopes, and its physical realization in Dedushenko--Nekrasov's and Bullimore--Zhang's works, to certain varieties without holomorphic symplectic structure or polarization. These classes of varieties include, in particular, classical Higgs branches of 3d $\mathcal N=2$ quiver gauge theories. The Bethe/gauge correspondence relates such a gauge theory to an anisotropic/elliptic superspin chain, and the stable envelopes compute the $R$-matrix that solves the dynamical Yang--Baxter equation (dYBE) for this spin chain. As an illustrative example, we solve the dYBE for the elliptic $\mathfrak{sl}(1|1)$ spin chain with fundamental representations using the corresponding 3d $\mathcal N=2$ SQCD whose classical Higgs branch is the Lascoux resolution of a determinantal variety. Certain Janus partition functions of this theory on $I \times \mathbb E$ for an interval $I$ and an elliptic curve $\mathbb E$ compute the elliptic stable envelopes, and in turn the geometric elliptic $R$-matrix, of the anisotropic $\mathfrak{sl}(1|1)$ spin chain. Furthermore, we consider the 2d and 1d reductions of elliptic stable envelopes and the $R$-matrix. The reduction to 2d gives the K-theoretic stable envelopes and the trigonometric $R$-matrix, and a further reduction to 1d produces the cohomological stable envelopes and the rational $R$-matrix. The latter recovers Rim\'anyi--Rozansky's results that appeared recently in the mathematical literature.}


\Classification{81R12; 81T60; 55N34}

\newpage

\tableofcontents

\section{Introduction}
\label{sec:introduction}

\subsection{Background and motivation}
In this paper, we study some mathematical and physical aspects of one-dimensional quantum integrable spin chains with anisotropy based on Lie superalgebras $\mfr g$ of type A. Given a Lie algebra/superalgebra $\mfr g$, there is a family of integrable spin chains with varying spectral curves and spectrum-generating algebras:
\begin{table}[H]\renewcommand{\arraystretch}{1.2}
\setlength\tabcolsep{4pt}
\centering
\begin{tabular}{|ccc|}
\hline
Spin chain & Spectral curve & Spectrum generating algebra \\
\hline
XXX/Rational & $\bC$ & Yangian, $\msf Y_\hbar(\mfr g)$ \\
XXZ/Trigonometric & $\bC^\times$ & Quantum affine, $\mcr U_\hbar(\hat{\mfr g})$ \\
XYZ/Elliptic & $\bE_\tau$ & Elliptic dynamical quantum group, $E_{\tau, \hbar}(\mfr g)$ \\
\hline
\end{tabular}
\caption{One-dimensional integrable quantum spin chains based on the Lie algebra $\mfr g$. Here $\hbar$ is the usual quantization parameter, $\tau$ denotes the complex structure of the elliptic curve $\bE_\tau$, and~$\hat {\mfr g}$ is the affine Lie algebra of $\mfr g$.}
\label{table:spinchains}
\end{table}
Given one of these choices, to complete the description of a specific spin chain, we need to choose the number of sites $L$ in the chain, along with the $\mfr g$-module $V_i$ and the spectral curve valued inhomogeneity $v_i$ at the $i$th site for $i=1,\dots, L$. As an example, the XYZ $\sl_2$ spin chain with $V_i = \bC^2$ being the fundamental representation for all $i$ can be defined by the following Hamiltonian \cite{baxter1982exactly}:
\beq
	H_{\text{XYZ-}\frac{1}{2}} = -\frac{1}{2} \sum_{i=1}^{L} \sum_{j=1}^3 J_j \si_{(i)}^j \si_{(i+1)}^j , \label{HXYZ1/2}
\eeq
where $\si^1$, $\si^2$, $\si^3$ are the standard Pauli matrices, $\si_{(i)}$ acts on $V_i$, and $J_1$, $J_2$, $J_3$ are three mutually distinct coupling constants. We consider periodic spin chains by identifying the $(L+1)$st site with the first one. Starting from the Hamiltonian, one can construct a large commutative algebra of operators containing the Hamiltonian, called the Bethe subalgebra of the spectrum-generating algebra. These operators provide the large number of conserved charges for the spin chain, which is the source of its integrability. The wave function and energy spectrum of such spin chains are obtained using its equivalence to the 2d statistical eight-vertex model and the Bethe ansatz~\cite{baxter1982exactly,Faddeev199605,1993qism.book.....K,Sklyanin1982, Sklyanin:1978pj, FaddeevSklyaninTakhtajan1980,sutherland1967, FaddeevTakhtajan1979, Yang196711, Yang196712}.

It became apparent from the pioneering works \cite{Baxter197203,Baxter197805,McGuire196405,SemenovTyanShanskii198310,Yang196711,Yang196712} that the more general and fundamental approach to defining integrable spin chains is utilizing the \emph{$R$-matrix} and the \emph{Yang--Baxter equation $($YBE$)$}. For any two $\mfr g$-modules $U$ and~$V$ with inhomogeneities $u$ and~$v$, the $R$-matrix is an operator $R(u-v) \in \End(U \otimes V)$ depending meromorphically on the difference~${u-v}$.\footnote{The difference is taken in the appropriate abelian group, $\bC$, $\bC^\times$ or $\bE_\tau$.} It is sometimes convenient to use a diagrammatic notation for the $R$-matrix, as in Figure~\ref{fig:R}.

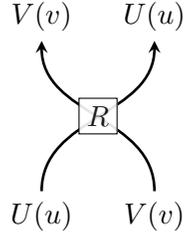
\begin{figure}[h]
\centering
\begin{tikzpicture}[x={(.75cm,0cm)}, y={(0cm,1cm)}]
 \draw[-stealth, line width=1, name path=line1] (0,0) .. controls (0,1) and (2,1) .. (2,2);
 \draw[-stealth, line width=1, name path=line2] (2,0) .. controls (2,1) and (0,1) .. (0,2);
 \node[below=0] at (0,0) {$U(u)$};
 \node[below=0] at (2,0) {$V(v)$};
 \node[above=0] at (0,2) {$V(v)$};
 \node[above=0] at (2,2) {$U(u)$};

 \fill [name intersections={of=line1 and line2, by={intersection}}] (intersection) node[rectangle, draw, fill=white, fill opacity=.9, inner sep=3pt] {$R$};
\end{tikzpicture}
\caption{Diagrammatic representation of the $R$-matrix $R(u-v)\colon U\otimes V \to U \otimes V$. The notation~$U(u)$ means that at the site with the module $U$ the inhomogeneity is $u$. For rational spin chains, given a~$\mfr g$-module $U$, $U(u)$ becomes the evaluation module for the Yangian $\msf Y_\hbar(\mfr g)$ with the evaluation parameter~$u$.} \label{fig:R}
\end{figure}

The defining property of the $R$-matrix is that it satisfies the YBE. For $\mfr g$-modules~$U$,~$V$ and~$W$ with inhomogeneities $u$, $v$ and $w$, the YBE can be written as
\begin{align}
 R_{UV}(u-v)R_{UW}(u-w)R_{VW}(v-w)=R_{VW}(v-w)R_{UW}(u-w)R_{UV}(u-v).
\label{YBE}
\end{align}
It is an equality of operators in $U \otimes V \otimes W$. Here $R_{UV}$ is the $R$-matrix in $\End(U \otimes V)$ and its action is extended to $U \otimes V \otimes W$ by identity, similarly for $R_{VW}$ and $R_{UW}$. Figure~\ref{fig:YBE} shows the diagrammatic YBE.
\begin{figure}[h]
\centering
\begin{tikzpicture}[x={(.8cm,0cm)}, y={(0cm,1cm)}, baseline=39pt]
 \draw[-stealth, line width=1, name path=line1] (0,0) .. controls (0,1.5) and (1.5,1.5) .. (3,3);
 \draw[-stealth, line width=1, name path=line2] (1.5,0) .. controls (1.5,.75) and (2.75,.75) .. (2.75,1.5) .. controls (2.75,2.25) and (1.5, 2.25) .. (1.5, 3);
 \draw[-stealth, line width=1, name path=line3] (3,0) .. controls (1.5,1.5) and (0,1.5) .. (0,3);
 \node[below=0] at (0,0) {$U(u)$};
 \node[below=0] at (1.5,0) {$V(v)$};
 \node[below=0] at (3,0) {$W(w)$};
 \node[above=0] at (0,3) {$W(w)$};
 \node[above=0] at (1.5,3) {$V(v)$};
 \node[above=0] at (3,3) {$U(u)$};

 \fill [name intersections={of=line1 and line3, by={intersection}}] (intersection) node[rectangle, draw, fill=white, fill opacity=.9, inner sep=2pt] {$R_{UW}$};
 \fill [name intersections={of=line1 and line2, by={intersection}}] (intersection) node[rectangle, draw, fill=white, fill opacity=.9, inner sep=2pt] {$R_{UV}$};
 \fill [name intersections={of=line2 and line3, by={intersection}}] (intersection) node[rectangle, draw, fill=white, fill opacity=.9, inner sep=2pt] {$R_{VW}$};
\end{tikzpicture}
$=$
\begin{tikzpicture}[x={(.8cm,0cm)}, y={(0cm,1cm)}, baseline=39pt]
 \draw[-stealth, line width=1, name path=line1] (0,0) .. controls (1.5,1.5) and (3,1.5) .. (3,3);
 \draw[-stealth, line width=1, name path=line2] (1.5,0) .. controls (1.5,.75) and (.25,.75) .. (.25,1.5) .. controls (.25,2.25) and (1.5, 2.25) .. (1.5, 3);
 \draw[-stealth, line width=1, name path=line3] (3,0) .. controls (3,1.5) and (1.5,1.5) .. (0,3);
 \node[below=0] at (0,0) {$U(u)$};
 \node[below=0] at (1.5,0) {$V(v)$};
 \node[below=0] at (3,0) {$W(w)$};
 \node[above=0] at (0,3) {$W(w)$};
 \node[above=0] at (1.5,3) {$V(v)$};
 \node[above=0] at (3,3) {$U(u)$};

 \fill [name intersections={of=line1 and line3, by={intersection}}] (intersection) node[rectangle, draw, fill=white, fill opacity=.9, inner sep=2pt] {$R_{UW}$};
 \fill [name intersections={of=line1 and line2, by={intersection}}] (intersection) node[rectangle, draw, fill=white, fill opacity=.9, inner sep=2pt] {$R_{UV}$};
 \fill [name intersections={of=line2 and line3, by={intersection}}] (intersection) node[rectangle, draw, fill=white, fill opacity=.9, inner sep=2pt] {$R_{VW}$};
\end{tikzpicture}
\caption{Diagrammatic Yang--Baxter equation. A line joining two $R$-matrices is a contraction of indices for the vector space labeling the line. In particular, the order of matrix multiplication in the YBE \rf{YBE} is read off from the diagrams by reading from bottom to top.} \label{fig:YBE}
\end{figure}
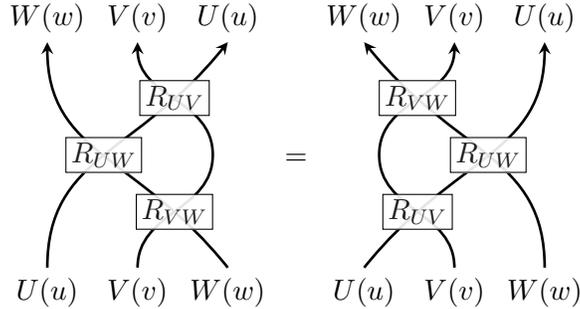

The $R$-matrix is used to construct the commuting charges of the spin chain as follows. Given the spin chain Hilbert space $\cH = \bigotimes_{i=1}^L V_i$, we add an auxiliary site $U$ to it with inhomogeneity~$u$.\footnote{Inhomogeneity of the auxiliary site is also referred to as the spectral parameter in literature. We use the terms inhomogeneity of a site and its spectral parameter interchangeably.} Then we define the transfer matrix
\[
	t(u) := \tr_U[R_{UV_1}(u-v_1) R_{UV_2}(u-v_2) \cdots R_{UV_L}(u-v_L)] \in \End(\cH) ,
\]
which can also be depicted in a diagram as
\[
t(u) =
\begin{tikzpicture}[x={(1.5cm,0cm)}, y={(0cm,1cm)}, baseline=11pt]
\foreach \i in {1,...,4}
	\path[-stealth, line width = 1, name path = v\i] (-3+\i,-1) -- (-3+\i,2);

\draw[-stealth, line width = .75, name path = u2] (-2.5,.9) .. controls (-1.5,1.25) and (.5,1.2) .. (1.5,.9);

\draw[-stealth, line width = .75, name path = u1] (1.5,.9) .. controls (2,.7) and (2,.3) .. (1.5,.1) .. controls (.5,-.2) and (-1.5,-.25) .. (-2.5,.1) .. controls (-3,.3) and (-3,.7) .. (-2.5,.9);

\foreach \i in {1,2,4}
	\fill [name intersections={of=u2 and v\i, by={intersection}}] (intersection) node[rectangle, fill=white, minimum size=5pt] {};
	
\foreach \i in {1,2,4}
	\draw[-stealth, line width = 1, name path = v\i] (-3+\i,-1) -- (-3+\i,2);

\foreach \i in {1,2}
	\fill [name intersections={of=u1 and v\i, by={intersection}}] (intersection) node[rectangle, draw, fill=white, fill opacity=.9, inner sep=2pt] {$R_{UV_{\i}}$};

\fill [name intersections={of=u1 and v4, by={intersection}}] (intersection) node[rectangle, draw, fill=white, fill opacity=.9, inner sep=2pt] {$R_{UV_{L}}$};

\foreach \i in {1,2}
	\fill [name intersections={of=v3 and u\i, by={intersection}}] (intersection) node[rectangle, fill=white, inner sep=7pt] {$\bm\cdots$};
	
\node at (0,.5) {$\cdots$};
\node[below=0] at (0,-1) {$\cdots$};

\node[below=0] at (-2,-1) {$V_1(v_1)$};
\node[below=0] at (-1,-1) {$V_2(v_2)$};
\node[below=0] at (1,-1) {$V_L(v_L)$};

\node at (2,1) {$U(u)$};
\end{tikzpicture} .
\]
As a consequence of the YBE, these transfer matrices commute with each other for arbitrary inhomogeneities
$[t(u), t(v)] = 0$.
Therefore, if we expand the transfer matrix as $t(u) = \sum_{n=0}^\infty u^n h_n$, then the coefficients $h_n \in \End(\cH)$ form a large commutative algebra of operators:
$[h_m, h_n] = 0$.
This becomes the Bethe subalgebra of the spin chain. Any element of this subalgebra can be taken to be the Hamiltonian of the quantum mechanical system.

The $R$-matrix, as a solution of the YBE, has been computed in all cases from Table~\ref{table:spinchains} for bosonic Lie \cite{Bazhanov198709,Bazhanov198710,BazhanovStroganov198207, BelavinDrinfeld198207,BelavinDrinfeld198307,BelavinDrinfeld199808,KulishReshetikhinSklyanin198109,KulishSklyanin198207,Reshetikhin198710,Sklyanin198210,
Stolin199112} and affine Lie algebra \cite{Jimbo198612,Reshetikhin198710}. For Lie superalgebras, the first examples of solutions to YBE appeared in \cite{KulishSklyanin198207}. Certain rational solutions were also obtained in \cite{Kulish198611}. The classical trigonometric solutions to YBE are constructed in \cite{LeitesSerganova198401} while the explicit trigonometric quantum $R$-matrices first appeared in \cite{BazhanovShadrikov198712}. %
Furthermore, the rational solution for $\mfr g = \sl(1|1)$ is constructed using the notion of stable envelopes in \cite{RimanyiRozansky202105}. In this paper, we propose a general construction of elliptic $R$-matrices for Lie superalgebras of type~A and we explicitly compute the elliptic $R$-matrix for $\mfr g = \sl(1|1)$. Reduction of the elliptic $R$-matrix automatically gives the trigonometric $R$-matrix. Further reduction of the trigonometric $R$-matrix gives the rational $R$-matrix. The reduction of our elliptic $\sl(1|1)$ $R$-matrix down to the rational case matches with the previously computed rational $\sl(1|1)$ $R$-matrix from \cite{RimanyiRozansky202105}.

The tool that we use to construct the elliptic $R$-matrices is a conjectural relationship between quantum integrable systems and supersymmetric gauge theories, called the \emph{Bethe--Gauge correspondence}, proposed by Nekrasov and Shatashvili \cite{NekrasovShatashvili200901c,NekrasovShatashvili200901b,NekrasovShatashvili201405}. Part of their conjecture is an action of the spectrum-generating algebra of a quantum integrable system, Yangian for example, on the cohomology (or generalized cohomologies) of the Higgs branch of the corresponding gauge theory. The correspondence and its string-theory origin, when $\mfk{g}$ is a bosonic Lie algebra, have been extensively studied \cite{BullimoreKimLukowski201708,ChungYoshida201605,CostelloYagi201810,FodaManabe201907,GalakhovLiYamazaki202206,KannoSugiyamaYoshida201806,KimuraZhu202012,
Nekrasov2013,Nekrasov2014,NekrasovShatashvili200901c,NekrasovShatashvili200901b,NekrasovShatashvili201405,NekrasovWitten201002,OkudaYoshida201501,
Orlando201310,OrlandoReffert201111}, generalized to higher dimensions \cite{DingZhang202304,DingZhang202303,JeongLeeNekrasov202103,JeongNekrasov201806,LeeNekrasov202009,NekrasovPestunShatashvili201312, NekrasovShatashvili200901c,NekrasovWitten201002}, and the mathematical formulation has been established \cite{AganagicOkounkov201704,AganagicOkounkov201604,BykovZinnJustin201904,MaulikOkounkov201211,Okounkov201512,Okounkov201710}. For $\mfk{g}$ being a Lie superalgebra, the correspondence and its string-theory origin were studied recently \cite{BoujakhroutSaidiAhlLaamaraDrissi202304,IshtiaqueMoosavianRaghavendranYagi202110,Nekrasov201811,OrlandoReffert201005} but despite progress \cite{ChernyakLeurentVolin202004,RimanyiRozansky202105}, a full mathematical treatment was lacking. Providing such a mathematical formulation for $\mfk{g}$ being a~Lie superalgebra of type A has been one of the central motivations of this work.

Notably, for a 3d gauge theory with $\mcal N=4$ supersymmetry, the Higgs branch is a complex symplectic variety. A particular example is the 3d $\mcal N=4$ quiver gauge theory, whose Higgs branch is the \emph{Nakajima quiver variety} \cite{nakajima1998quiver}, and in this case, the works of Nakajima \cite{nakajima2001quiver} and Varagnolo \cite{varagnolo2000quiver} establish the quantum affine algebra $\mscr U_\hbar(\hat{\mathfrak g}_Q)$ (resp.\ Yangian $\mathsf Y_{\hbar}(\mathfrak g_Q)$) actions on equivariant cohomology (resp.\ equivariant K-theory) of the Nakajima quiver variety $\mcal M_Q$, where~$\mathfrak g_Q$ is the Kac--Moody algebra associated to the quiver $Q$. Their construction relies on the generators and relations of the Yangian $\mathsf Y_{\hbar}(\mathfrak g_Q)$ or quantum affine algebra $\mscr U_\hbar(\hat{\mathfrak g}_Q)$.

 From the quantum integrability point of view, the spectrum-generating algebras are better explained in the framework of $R$-matrices and Yang--Baxter equations \rf{YBE}, instead of generators and relations. In the context of the Bethe--Gauge correspondence, the vector spaces on which the $R$-matrices act are generalized cohomologies of the Higgs branches of the corresponding gauge theories.

Providing the mathematical formulation of the Bethe--Gauge correspondence has been one of the main motivations of Maulik and Okounkov in their fundamental work \cite{MaulikOkounkov201211}. To construct $R$-matrices acting on the cohomologies, Maulik and Okounkov \cite{MaulikOkounkov201211} came up with a map, called the \emph{stable envelope}, from the equivariant cohomology of the torus fixed point set to the equivariant cohomology of the whole Higgs branch
\begin{align}
 \mathrm{Stab}\colon\ H_{\mathsf T}\bigl(\mathsf X^{\mathsf A}\bigr)\to H_{\mathsf T}(\mathsf X). \label{HXAtoHX}
\end{align}
Here $\mathsf X$ is assumed to be a smooth complex symplectic variety, with a torus $\mathsf T$ action such that the subtorus $\mathsf A\subset\mathsf T$ preserves the symplectic structure. Their construction depends on a choice of chamber $\mathfrak C$ which is a certain subset of the real Lie algebra $\mathrm{Lie}(\mathsf A)_{\mathbb R}$, and the $R$-matrix is defined as
\begin{align}
 R_{\mathfrak C_2\shortleftarrow\mathfrak C_1}:=\mathrm{Stab}_{\mathfrak C_2}^{-1}\circ\mathrm{Stab}_{\mathfrak C_1}\colon\ H_{\mathsf T}\bigl(\mathsf X^{\mathsf A}\bigr)\to H_{\mathsf T}\bigl(\mathsf X^{\mathsf A}\bigr). \label{R=StabInvStab}
\end{align}
Then it follows from the definition that $R$-matrices satisfy the braiding relation
\begin{align}
 R_{\mathfrak C_1\shortleftarrow\mathfrak C_{n}}\circ R_{\mathfrak C_n\shortleftarrow\mathfrak C_{n-1}}\circ\cdots\circ R_{\mathfrak C_2\shortleftarrow\mathfrak C_1}=1, \label{RR...R}
\end{align}
which implies the Yang--Baxter equation. The choice of chamber $\mathfrak C$ induces a partial order~$\preceq$ on the set of connected components of $\mathsf X^{\mathsf A}$ such that $\mathrm{Stab}_{\mathfrak C}$ is triangular with respect to $\preceq$ and~${R_{-\mathfrak C\shortleftarrow\mathfrak C}=\mathrm{Stab}_{-\mathfrak C}^{-1}\circ\mathrm{Stab}_{\mathfrak C}}$ gives a Gauss decomposition of $R_{-\mathfrak C\shortleftarrow\mathfrak C}$. A similar kind of Gauss decomposition of $R$-matrices has been discussed by Khoroshkin and Tolstoy in the algebraic setting \cite{KhoroshkinTolstoy199102,Razumov202210}. Since the paper of Maulik and Okounkov came out, the construction of stable envelopes has been extended to K-theory \cite{Okounkov201512} and elliptic cohomology \cite{AganagicOkounkov201604}, giving rise to trigonometric and elliptic solutions to the Yang--Baxter equations respectively.

From the physics point of view, a strategy for the construction of cohomological stable envelopes was suggested by Nekrasov \cite{Nekrasov2013,Nekrasov2014}, and was successfully implemented for gauge theories associated with $\mfk{sl}(2)$ spin chains in \cite{BullimoreKimLukowski201708}. The latter also provides a recipe for the computation of the $R$-matrix by an A-model localization computation. The systematic gauge-theoretic construction of elliptic stable envelopes was developed in \cite{BullimoreZhang202109,DedushenkoNekrasov202109}. In these works the elliptic stable envelope was identified as certain Janus partition functions of 3d $\cN=4$ theories on $I \times \bE_\tau$ where~$I$ is some interval and $\bE_\tau$ is the elliptic curve. The Janus interface interpolates between certain boundary conditions corresponding to the vacua of massless and massive 3d theories. These boundary conditions correspond to equivariant elliptic cohomology classes of the Higgs branch of the 3d theory and the (flavor) torus fixed points of the Higgs branch respectively~\cite{Dedushenko202211}. The elliptic stable envelope is then seen as a map between these cohomologies (the elliptic version of \rf{HXAtoHX}). Dedushenko and Nekrasov \cite{DedushenkoNekrasov202109} construct Janus interfaces for more general~3d $\cN=2$ theories, though they provide concrete boundary conditions corresponding to vacua of 3d $\cN=4$ theories and they compute the stable envelopes for complex symplectic varieties corresponding to 3d $\cN=4$ Higgs branches. In this paper we build on the techniques of \cite{DedushenkoNekrasov202109} and generalize the boundary conditions to vacua of 3d $\cN=2$ theories and use them to compute the elliptic stable envelopes of non-symplectic varieties corresponding to 3d $\cN=2$ Higgs branches. Some related topics such as vertex functions for Higgs branches of 3d $\mcal{N}=4$ (and $\mcal{N}=2$) theories and the role of elliptic stable envelopes in the transformation of vertex functions under mirror symmetry/symplectic duality have been discussed in \cite{CrewZhangZhao202306,DedushenkoNekrasiv202306}.

A key ingredient in the mathematical construction of elliptic stable envelope \cite{AganagicOkounkov201604} is a \emph{polarization}, i.e., a decomposition of the tangent bundle in the K-theory
\begin{align}\label{eqn: polarization}
 T_{\mathsf X}=T^{1/2}_{\mathsf X}+\hbar^{-1}\bigl(T^{1/2}_{\mathsf X}\bigr)^\vee\in K_{\mathsf T}(\mathsf X).
\end{align}
This kind of decomposition shows up naturally on the Higgs branches of 3d $\mcal N=4$ gauge theories, for example, hypertoric varieties and Nakajima quiver varieties. Computations of stable envelopes have been done for many examples of these types of varieties \cite{Dinkins202011,Dinkins20207,FelderRimanyiTarasov201702,RimanyiTarasovVarchenko201212,RimanyiTarasovVarchenko201411,RimanyiTarasovVarchenko201705,RimanyiWeber202007}. However, for more general 3d $\mcal N=2$ theories, their Higgs branches do not have complex symplectic structure or polarization, even worse they can be singular. There have been efforts to construct stable envelopes for varieties without complex symplectic structure or polarization~\cite{RimanyiRozansky202105}, the authors constructed cohomological stable envelopes using the \emph{super weight functions}, which generalized their previous works on weight functions and stable envelopes on cotangent bundles of flag varieties \cite{RimanyiTarasovVarchenko201212,TarasovVarchenko199311}. They found the $R$-matrix of $\mathsf Y_{\hbar}(\mathfrak{sl}(1|1))$ as the geometric $R$-matrix of the total space of the bundle $\mscr O(-1)^{\oplus 2}\to \mathbb P^1$. The super Lie algebra~$\mathfrak{sl}(1|1)$ justifies the terminology ``super" weight functions.

Another approach to generalizing the construction of stable envelopes was proposed by Okounkov in \cite{okounkov2021inductive}, with a focus on the elliptic version. An observation in \cite{okounkov2021inductive} is that polarization is not crucial for the existence of elliptic stable envelopes. A technical alternative to polarization is a notion called an \emph{attractive line bundle} $\mscr S_{\mathsf X}$ on the $\mathsf T$-equivariant elliptic cohomology~$\mathrm{Ell}_{\mathsf T}(\mathsf X)$ (see Appendix \ref{subsec: Equivariant Elliptic Cohomology} for a review on equivariant elliptic cohomology). Roughly speaking, $\mscr S_{\mathsf X}$ provides a ``square root'' of the elliptic Thom line bundle~$\Theta(T_{\mathsf X})$, in parallel to a polarization being ``half'' of the tangent bundle. For a precise definition, see Section~\ref{subsec: Attractive Line Bundle}. It is shown in \cite{okounkov2021inductive} that if~$\mathsf X$ admits a polarization \eqref{eqn: polarization}, then the elliptic Thom line bundle \smash{$\Theta\bigl(T^{1/2}_{\mathsf X}\bigr)$} is an attractive line bundle, which makes the construction of elliptic stable envelopes in \cite{AganagicOkounkov201604} a special case of the more general one in \cite{okounkov2021inductive}.

Part of the objectives in this paper is to use the idea developed in \cite{okounkov2021inductive} to extend the construction of elliptic stable envelopes in \cite{AganagicOkounkov201604} to the classical Higgs branches of a certain class of~3d~${\mcal N=2}$ gauge theories which do not have $\mcal N=4$ supersymmetry. In mathematical terminology, we construct the elliptic stable envelopes for certain smooth varieties which do not necessarily admit complex symplectic structure or polarization. The replacement for the polarization in our construction is a certain structure called the \emph{partial polarization}, which means a~decomposition of the tangent bundle in the K-theory
$ T_{\mathsf X}=\mathrm{Pol}_{\mathsf X}+\hbar^{-1}\mathrm{Pol}_{\mathsf X}^\vee+\mscr E\in K_{\mathsf T}(\mathsf X)$,
where $\mscr E\in K_{\mathsf T}(\mathsf X)$ is a K-theory class that is required to have the following properties:
\begin{itemize}\itemsep=0pt
 \item $\mscr E=\mscr E^\vee$ in $K_{\mathsf T}(\mathsf X)$, and that
 \item the elliptic Thom bundle $\Theta(\mscr E)$ admits a square root.
\end{itemize}
Then it turns out that \smash{$\Theta(\mathrm{Pol}_{\mathsf X})\otimes \sqrt{\Theta(\mscr E)}$} is an attractive line bundle, thus the machinery developed in \cite{okounkov2021inductive} is applicable to our situation.

We will give a sufficient condition for the existence of partial polarization in Proposition~\ref{prop: criterion for partial polarization}. A class of examples that satisfy the criterion in Proposition~\ref{prop: criterion for partial polarization} are quiver varieties. They are similar to Nakajima quiver varieties, namely our quiver variety is built from the cotangent bundle to the representation of a quiver by imposing some relations and stability conditions and then taking the quotient by the gauge group; but the difference is that we allow some of the nodes to be free from imposing complex moment map equations, these nodes will be called \emph{odd}, and the ordinary nodes where complex moment map equations are imposed will be called \emph{even}. We use the following notation for the odd/even nodes
\begin{equation*}
 \begin{tikzpicture}[x={(1.5cm,0cm)}, y={(0cm,1.5cm)}]
 \node[draw, circle, label=above:even] at (2,0) (n3) {};
 \node[draw, circle, cross, label=above:odd] at (0,0) (n1) {};
 \end{tikzpicture}.
\end{equation*}
The following is an example of a type A quiver (framed and doubled)
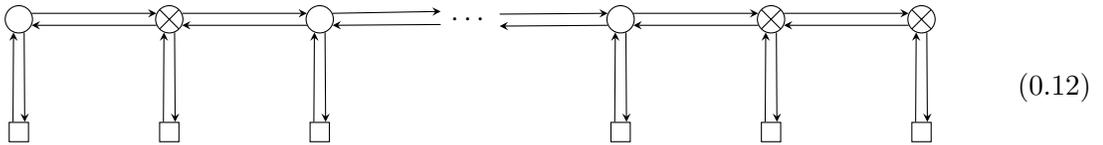
\beq
\begin{tikzpicture}[x={(2cm,0cm)}, y={(0cm,2cm)}, baseline=-1cm]
\node[draw, circle] at (0,0) (n1) {};
\node[draw, circle,cross] at (1,0) (n2) {};
\node[draw, circle] at (2,0) (n3) {};
\node at (3,0) (ndot) {$\cdots$};
\node[draw, circle] at (4,0) (n4) {};
\node[draw, circle,cross] at (5,0) (n5) {};
\node[draw, circle,cross] at (6,0) (n6) {};

\draw[-stealth] (n1.25) to (n2.155);
\draw[-stealth] (n2.25) to (n3.155);
\draw[-stealth] (n3.25) to (ndot.165);
\draw[-stealth] (ndot.10) to (n4.155);
\draw[-stealth] (n4.25) to (n5.155);
\draw[-stealth] (n5.25) to (n6.155);

\draw[-stealth] (n6.205) to (n5.335);
\draw[-stealth] (n5.205) to (n4.335);
\draw[-stealth] (n4.205) to (ndot.347);
\draw[-stealth] (ndot.190) to (n3.335);
\draw[-stealth] (n3.205) to (n2.335);
\draw[-stealth] (n2.205) to (n1.335);

\node[draw, rectangle] at (0,-.75) (f1) {};
\node[draw, rectangle] at (1,-.75) (f2) {};
\node[draw, rectangle] at (2,-.75) (f3) {};
\node[draw, rectangle] at (4,-.75) (f4) {};
\node[draw, rectangle] at (5,-.75) (f5) {};
\node[draw, rectangle] at (6,-.75) (f6) {};

\draw[-stealth] (n1.292) to (f1.60);
\draw[-stealth] (f1.120) to (n1.248);
\draw[-stealth] (n2.292) to (f2.60);
\draw[-stealth] (f2.120) to (n2.248);
\draw[-stealth] (n3.292) to (f3.60);
\draw[-stealth] (f3.120) to (n3.248);
\draw[-stealth] (n4.292) to (f4.60);
\draw[-stealth] (f4.120) to (n4.248);
\draw[-stealth] (n5.292) to (f5.60);
\draw[-stealth] (f5.120) to (n5.248);
\draw[-stealth] (n6.292) to (f6.60);
\draw[-stealth] (f6.120) to (n6.248);
\end{tikzpicture}.
\label{an example of type A quiver}
\eeq
Our notation might remind the reader of Kac--Dynkin diagrams for Lie superalgebras \cite{frappat1996dictionary}, for example, the Kac--Dynkin diagram corresponding to \eqref{an example of type A quiver} is the following
\beq
\begin{tikzpicture}[x={(2cm,0cm)}, y={(0cm,2cm)}]
\node[draw, circle] at (0,0) (n1) {};
\node[draw, circle,cross] at (1,0) (n2) {};
\node[draw, circle] at (2,0) (n3) {};
\node at (3,0) (ndot) {$\cdots$};
\node[draw, circle] at (4,0) (n4) {};
\node[draw, circle,cross] at (5,0) (n5) {};
\node[draw, circle,cross] at (6,0) (n6) {};

\draw (n1) to (n2);
\draw (n2) to (n3);
\draw (n3) to (ndot);
\draw (ndot) to (n4);
\draw (n4) to (n5);
\draw (n5) to (n6);

\end{tikzpicture}.
\label{an example of Kac--Dynkin diagram}
\eeq
Indeed our choice of notation is on purpose, due to the following conjecture we make.
\begin{conj}\label{conj: algebra action}
Let $Q$ be a finite or affine type A quiver with nodes decorated as above, and let $\mathfrak g_Q$ be the Lie superalgebra associated with the corresponding Kac--Dynkin diagram. For a~generic stability parameter $\zeta$, and a pair of dimension vectors $\mathbf v,\mathbf w\in \mathbb N^{Q_0}$, denote by $\mcal M^{\zeta}(\mathbf v,\mathbf w)$ the quiver variety of the given data. There exists an action of elliptic dynamical quantum super group $E_{\hbar,\tau}(\mathfrak g_Q)$ on the extended equivariant elliptic cohomology scheme
\[
 \mathsf E_{\mathsf T}\bigl(\mcal M^{\zeta}(\mathbf w)\bigr)=\bigsqcup_{\mathbf v} \mathsf E_{\mathsf T}\bigl(\mcal M^{\zeta}(\mathbf v,\mathbf w)\bigr).
\]
Similarly, there exists an action of quantum affine super algebra $\mscr U_{\hbar}(\hat{\mathfrak g}_Q)$ on the equivariant K-theory
\[
 K_{\mathsf T}\bigl(\mcal M^{\zeta}(\mathbf w)\bigr)=\bigoplus_{\mathbf v} K_{\mathsf T}\bigl(\mcal M^{\zeta}(\mathbf v,\mathbf w)\bigr),
\]
and an action of super Yangian $\mathsf Y_{\hbar}(\mathfrak g_Q)$ on the equivariant cohomology
\[
 H_{\mathsf T}\bigl(\mcal M^{\zeta}(\mathbf w)\bigr)=\bigoplus_{\mathbf v} H_{\mathsf T}\bigl(\mcal M^{\zeta}(\mathbf v,\mathbf w)\bigr).
\]
Moreover, all the actions factor through the corresponding Maulik--Okounkov quantum groups constructed via the stable envelopes for $\mcal M^{\zeta}(\mathbf w)$.
\end{conj}

In anisotropic spin chains, as in \rf{HXYZ1/2}, when all coupling constants are distinct, the $R$-matrix depends non-trivially on them. This dependence is manifested as a dependence of the $R$-matrix on an additional spectral curve valued parameter called the Dynamical parameter. Thus the elliptic $R$-matrix, and the corresponding YBE which involves nontrivial shifts of the spectral parameters, are also called the dynamical $R$-matrix and the dynamical Yang--Baxter equation (dYBE) \cite{Felder199412,Felder199407}. Our Conjecture \ref{conj: algebra action} is supported by the evidence of the weight structure occurring in the shifts in the dYBE. In fact, we will show (see \eqref{eqn: DYBE for quiver}) that the $R$-matrices constructed from the elliptic stable envelopes for quiver varieties satisfy the following dynamical Yang--Baxter equation
\begin{equation}\label{eqn: dYBE_intro}
 \mathbf R_{12}(\mbs{z}) \mathbf R_{13}\bigl(\mbs{z}-\hbar\mbs{\mu}^{(2)}\bigr) \mathbf R_{23}(\mbs{z})= \mathbf R_{23}\bigl(\mbs{z}-\hbar\mbs{\mu}^{(1)}\bigr) \mathbf R_{13}(\mbs{z}) \mathbf R_{12}\bigl(\mbs{z}-\hbar\mbs{\mu}^{(3)}\bigr),
\end{equation}
where $z$ is the dynamical parameter and we have suppressed the spectral parameter dependence of $R$-matrices. Here the dynamical shifts $\mbs{\mu}$ are weights in a certain highest-weight module of~$\widehat{\mathfrak{sl}}(m|n)$ or ${\mathfrak{sl}}(m|n)$.

\subsection{Summary of the results}

Let us now briefly summarize the main contributions of this work.

{\bf Elliptic stable envelopes for partially polarized varieties.}
\begin{itemize}\itemsep=0pt
 \item We define the notion of partial polarization on a smooth variety $\mathsf X$ with a torus $\mathsf T$ action (see Definition~\ref{dfn: partial polarization}), and give a sufficient condition for the existence of partial polarization (see Proposition~\ref{prop: criterion for partial polarization}). It turns out that a large class of GIT quotients admits partial polarization (see Example~\ref{ex: pol for quotient of Hamiltonian reduction}), this includes all abelian gauge group cases, and all the quiver varieties considered in Section~\ref{subsubsec: Example: The Doubled Quivers}, and the total spaces of vector bundles on partial flag varieties (see Example~\ref{ex: bundles over flag variety}).
 \item We show that a partial polarization induces an attractive line bundle on the equivariant elliptic cohomology $\mathrm{Ell}_{\mathsf T}(\mathsf X)$, a notion defined by Okounkov in \cite{okounkov2021inductive}. Then the existence and uniqueness of elliptic stable envelope for partially polarized varieties follow from the general result in \cite{okounkov2021inductive}, see Theorem~\ref{thm: stable envelope}. In particular, elliptic stable envelopes exist for all cases considered in Examples \ref{ex: pol for quotient of Hamiltonian reduction}, \ref{ex: partial polarization for the quivers}, \ref{ex: bundles over flag variety} and \ref{ex: bundles over flag variety_type A}. The last example was also considered in the recent work of Rim\'anyi and Rozansky \cite[Section~2.3.3]{RimanyiRozansky202105}, and our result implies the existence of stable envelopes in this example.

 \item We compute the elliptic stable envelope for the quiver with a single node which is odd, i.e., the total space of direct sum of $L$ copies of tautological bundles on $\mathrm{Gr}(N,L)$, see equation~\eqref{eqn: stable envelope for M(N,L)}.

 \item The triangle lemma and the duality for the elliptic stable envelopes on partially polarized varieties are obtained verbatim to that in \cite{AganagicOkounkov201604}, see Section~\ref{subsec: Triangle Lemma and Duality}.

 \item We extend the abelianization procedure described by Aganagic and Okounkov \cite{AganagicOkounkov201604} to the partially polarized case in Section~\ref{subsec: Abelianization}. This is potentially helpful in the computation of stable envelopes for more general quivers.

 \item In Section~\ref{subsec: R-Matrices}, the $R$-matrices are defined in the same way as \cite{AganagicOkounkov201604,MaulikOkounkov201211}, and dynamical Yang--Baxter equations \eqref{eqn: DYBE} follows from the braiding relations. The $R$-matrices and dynamical Yang--Baxter equations for the quiver varieties are discussed in Section~\ref{subsec: R-Matrices for Quiver Varieties}. For general quivers of finite or affine type A, we observe that the dynamical shift $\mbs{\mu}$ is closely related to a highest-weight module of the corresponding Lie superalgebra.

 \item In Section~\ref{sec:k-theory and cohomology limit}, we show that the analog of the limit procedure in \cite{AganagicOkounkov201604} also works for the partially polarized varieties and produces K-theoretic and cohomological stable envelopes, see Corollary~\ref{cor:k-theoretic stable envelope from the elliptic one}. We present the explicit K-theoretic and cohomological stable envelopes for the quiver with a single node which is odd, see \eqref{eqn: K-theoretic stable envelope for M(N,L)} and \eqref{eqn: cohomological stable envelope for M(N,L)}. The cohomological stable envelope \eqref{eqn: cohomological stable envelope for M(N,L)} recovers the recent result of Rim\'anyi and Rozansky \cite[Theorem~7.3]{RimanyiRozansky202105}.
\end{itemize}

{\bf A couple of distinctions between 3d $\boldsymbol{\cN=2}$ and 3d $\boldsymbol{\cN=4}$ cases.}
For the purpose of illustration, we solve the dYBE for the Lie superalgebra $\mfk{g}=\mfk{sl}(1|1)$ from the corresponding gauge theory, which is a 3d $\mcal{N}=2$ SQCD. While rather similar in computation, there are a few important conceptual differences from the computation of $\sl(2)$ elliptic stable envelopes from 3d~$\cN=4$ SQCDs.

\begin{itemize}\itemsep=0pt
 \item \textbf{Higgs branch: Non-hyperk\"ahler and quantum correction.} Higgs branches of 3d $\cN=2$ theories without $\cN=4$ supersymmetry are K\"ahler manifolds, unlike 3d $\cN=4$ Higgs branches which are hyperk\"ahler. For example, in the case of 3d $\cN=2$ SQCD, the classical Higgs branch $\higgs(N,L)$ is given by the direct sum of $L$ copies of the tautological bundle over $\tenofo{Gr}(N,L)$, the so-called resolved determinantal variety \cite{Lascoux1978}. Also, unlike in case of the 3d $\cN=4$ Higgs branches, 3d $\cN=2$ Higgs branches receive quantum corrections \cite{Aharony:1997bx,deBoer:1997kr}. In this paper, unless explicitly mentioned otherwise, all mentions of 3d $\cN=2$ Higgs branches will refer to the \emph{classical} Higgs branches.

 \item \textbf{Non-isolated fixed points.} Another departure from the 3d $\cN=4$ SQCD case is that there can still be continuous Higgs branch moduli in 3d $\cN=2$ SQCDs in the presence of generic flavor twisted masses (even at the quantum level \cite{Aharony:1997bx,deBoer:1997kr}). As a result, we can not have a general correspondence between fully massive vacua and Bethe eigenstates of integrable superspin chains, as is more familiar in the bosonic Bethe/Gauge correspondence~\cite{NekrasovShatashvili200901c,NekrasovShatashvili200901b}. More specifically, in the context of interval partition functions computing stable envelopes or $R$-matrices for bosonic spin chains, boundary conditions are chosen representing these fully massive vacua \cite{BullimoreKimLukowski201708,DedushenkoNekrasov202109}. On the Higgs branch, these fully massive vacua correspond to isolated fixed points of the (flavor) torus action. In contrast, the fixed point set of the 3d $\cN=2$ Higgs branches decomposes into connected components of nonzero dimensions \rf{MHA}. So we introduce brane-type boundary conditions preserving~${\cN=(0,2)}$ supersymmetry mimicking these connected components (see Section~\ref{sec:thimble}), rather than isolated points.
\end{itemize}

{\bf The solution to dYBE for $\boldsymbol{\mfk{sl}(1|1)}$ from 3d $\boldsymbol{\mcal{N}=2}$ SQCD.}
Let us now present the formulas.
\begin{itemize}\itemsep=0pt
\item \textbf{Elliptic stable envelopes and the elliptic $\boldsymbol{R}$-matrix for $\boldsymbol{\sl(1|1)}$.}
We compute the stable envelopes as Janus partition functions in 3d $\cN=2$ SQCD on $I \times \bE_\tau$. The theory has~$\U(N)$ gauge group and $\SU(L) \times \U(1)_\hbar$ flavor symmetry. There is also a topological~$\U(1)_\text{top}$ symmetry acting only on monopoles. A Janus interface is created by varying~$\U(1)^{L-1}$ twisted masses $m_1, \dots, m_L$ (satisfying $m_1 + \dots + m_L=0$) along $I$. The masses interpolate between nonzero generic values satisfying $m_1 < \dots < m_L$ and zero. A~choice of such ordering creates a chamber $\mfr C$ in $\bR^{L-1}$. The formula for the stable envelope is (see Section~\ref{sec:explicit construction of elliptic stable envelope} for details)
\begin{align}
 \estab _{\mfr C}(p) :={}& (-1)^{\sum_a\#(i<p(a))}\nonumber\\
 &\times\text{Sym}_{S_N} \Biggl[ \Biggl(\prod_a \efunction_{\mfr C, p(a)}(s_a, x, \hbar, z)\Biggr) \Biggl( \prod_{p(a) > p(b)} \frac{1}{\vth\left(s_a s_b^{-1}\right)} \Biggr) \Biggr],\label{eq:elliptic stable envelope for sl(1|1),introduction}
\end{align}
where $\vartheta$ is the elliptic theta function (defined in \eqref{eq:elliptic theta function in the aganagic-okounkov convention}), $\tenofo{Sym}_{S_N}$ denotes the symmetrization with respect to the gauge holonomies $s_a$, and
\[
 \efunction_{\mfr C, n}(s, \bm x, \hbar, z) := \Biggl(\prod_{i < n} \vth(s x_i ) \Biggr)\frac{\vth\bigl(s x_n z \hbar^{n - L}\bigr)}{\vth\bigl( z \hbar^{n - L} \bigr)} \Biggl( \prod_{i > n} \vth(s x_i \hbar ) \Biggr).
\]
The elliptic equivariant parameters $s$, $x$, $\hbar$ and $z$ correspond to background holonomies for the $\U(N)$, $\SU(L)$, $\U(1)_\hbar$, and $\U(1)_\text{top}$ connections respectively.

The corresponding $R$-matrix, which satisfies dYBE, is given by (see Section~\ref{sec:elliptic r-matrix for sl(1|1)} for details)
\begin{equation}\label{eq:elliptic gl(1|1) rational r-matrix, introduction}
	\def\arraystretch{1.5}
	\setlength{\arrayrulewidth}{.6pt}
 \setlength{\arraycolsep}{6pt}
	\ermatrix_{-\mfk{C}\ot\mfk{C}}(u,z)=
	\left(
	\begin{array}{c|cc|c}
		1 & 0 & 0 & 0
		\\
		\hline
		0 & \frac{\vth\left(u\right) \vth(z) \vth(z\hbar^{-2})}{\vth\left(u^{-1} \hbar\right) \vth(z \hbar^{-1})^2} & \frac{\vth(\hbar) \vth\left(u z \hbar^{-1}\right)}{\vth\left(u^{-1} \hbar\right) \vth\left(z \hbar^{-1}\right)} & 0
 \\[4pt]
 0 & \frac{\vth(\hbar) \vth\left(u^{-1} z \hbar^{-1}\right)}{\vth\left(u^{-1} \hbar\right) \vth\left(z \hbar^{-1}\right)} & \frac{\vth\left(u\right)}{\vth\left(u^{-1} \hbar\right)} & 0
		\\[4pt]
		\hline
		0 & 0 & 0 & \frac{\vth\left(u \hbar\right)}{\vth\left(u^{-1} \hbar\right)}
	\end{array}\right),
\end{equation}
where $-\mfk{C}$ denotes the opposite chamber, and the diagonal blocks from top to bottom correspond to zero, one, and two-magnon sectors, respectively. To the best of our knowledge, the elliptic stable envelopes and the elliptic $R$-matrix for $\mfk{sl}(1|1)$ have not been constructed in the literature before. Furthermore, this is also the first time that a new solution to dYBE has been constructed directly from the gauge theory using the Bethe/Gauge correspondence that we could find.

\item \textbf{K-theoretic stable envelopes and the trigonometric $\boldsymbol{R}$-matrix for $\boldsymbol{\sl(1|1)}$.}
The 3d~$\to$ 2d reduction of elliptic stable envelopes gives K-theoretic stable envelopes and the trigonometric $R$-matrix. The precise way to get the K-theoretic stable envelope is described in Section~\ref{sec:k-theory and cohomology limit}. First, we define the slope parameter $\msf{s}$ associated to the K\"ahler parameter~$z$~by
 \[
 \msf{s}:=\lim_{\substack{\ln z\to \infty \\ \ln q\to\infty}}\tenofo{Re}\left(-\frac{\ln z}{\ln q}\right)\in\mbb{R}\backslash\mbb{Z},
 \]
where $q:={\rm e}^{2\pi\mfk{i}\tau}$. Then, the K-theoretic stable envelope, as the limit of \eqref{eq:elliptic stable envelope for sl(1|1),introduction}, is given by (see Example~\ref{sec:k-theory and cohomology limit for resolved determinantal variety} and Section~\ref{sec:the K-theory limit} for details)
\begin{align}
 \tstab_{\mfk{C},\msf{s}}(p)={}&(-1)^{\sum_a\#(i<p(a))}\hbar^{\frac{1}{2}\sum_a\#(i>p(a))}\nonumber
 \\
 &\times\! \tenofo{Sym}_{S_N}\!\!\left[\!\left(\prod_{a=1}^N \tfunction_{\mfk{C},\msf{s},p(a)}(s_a,\mbs{x},\hbar)\!\right)\!
 \cdot\!\Biggl(\prod_{\substack{a>b}}\frac{1}{\bigl(s_as_b^{-1}\bigr)^{\frac{1}{2}}\!-\!\bigl(s_as_b^{-1}\bigr)^{-\frac{1}{2}}}\!\Biggr)\!\right]\!,\label{eq:k-theoretic stable envelope for sl(1|1), introduction}
\end{align}
where
\[
 \tfunction_{\mfk{C},\msf{s},n}(s,\mbs{x},\hbar)\equiv (sx_{n})^{\lfloor{\msf{s}}\rfloor}\prod_{i<n}\bigl(1-(sx_i)^{-1}\bigr)\prod_{i>n}\bigl(1-(sx_i\hbar)^{-1}\bigr).
\]
The corresponding trigonometric $R$-matrix, as the reduction of \eqref{eq:elliptic gl(1|1) rational r-matrix, introduction}, is ${\msf{G}(x):= x^{\frac{1}{2}}-x^{-\frac{1}{2}}}$
\begin{equation}\label{eq:gl(1|1) trigonometric r-matrix, introduction}
	\def\arraystretch{1.5}
	\setlength{\arrayrulewidth}{.8pt}
 \setlength{\arraycolsep}{6pt}
	\trmatrix_{-\mfk{C}\ot\mfk{C}}(u)=
	\left(
	\begin{array}{c|cc|c}
		1 & 0 & 0 & 0
		\\
		\hline
		0 & \frac{\msf{G}(u)}{\msf{G}\bigl(u^{-1}\hbar\bigr)} & u^{\lfloor \msf{s}\rfloor+\frac{1}{2}}\frac{\msf{G}(\hbar)}{\msf{G}\bigl(u^{-1}\hbar\bigr)} & 0
 \\[6pt]
 0 & u^{-\lfloor \msf{s}\rfloor-\frac{1}{2}}\frac{\msf{G}(\hbar)}{\msf{G}\bigl(u^{-1}\hbar\bigr)} & \frac{\msf{G}(u)}{\msf{G}\bigl(u^{-1}\hbar\bigr)} & 0
		\\[4pt]
		\hline
		0 & 0 & 0 & \frac{\msf{G}(u\hbar)}{\msf{G}\bigl(u^{-1}\hbar\bigr)}
	\end{array}\right).
\end{equation}
Note that all holonomies are still multiplicative but are now $\mbb{C}^\times$-valued.

\item \textbf{Cohomological stable envelopes and the rational $\boldsymbol{R}$-matrix for $\boldsymbol{\sl(1|1)}$.} A~further~2d~$\to$ 1d reduction of \eqref{eq:k-theoretic stable envelope for sl(1|1), introduction} gives cohomological stable envelopes (see Section~\ref{sec:k-theory and cohomology limit} for the precise mathematical procedure)
\[
 \rstab_{\mfk{C}}(p)=(-1)^{\sum_a\#(i<p(a))}\tenofo{Sym}_{S_N}\left[\left(\prod_{a=1}^N\rfunction_{\mfk{C},m}(s,\mbs{x},\hbar)\right)\cdot\Biggl(\prod_{\substack{a>b}}\frac{1}{s_a-s_b}\Biggr)\right],
\]
where
\smash{$\rfunction_{\mfk{C},n}(s,\mbs{x},\hbar):= \prod_{i<n}(s+x_i)\prod_{i>n}(s+x_i+\hbar)$},
while the reduction of \eqref{eq:gl(1|1) trigonometric r-matrix, introduction} produces the rational $R$-matrix
\[
	\def\arraystretch{1.5}
	\setlength{\arrayrulewidth}{.6pt}
 \setlength{\arraycolsep}{6pt}
	\rrmatrix_{-\mfk{C}\ot\mfk{C}}(u)=
	\left(
	\begin{array}{c|cc|c}
		1 & 0 & 0 & 0
		\\
		\hline
		0 & \frac{u}{\hbar-u} & \frac{\hbar}{\hbar-u} & 0
 \\[4pt]
 0 & \frac{\hbar}{\hbar-u} & \frac{u}{\hbar-u} & 0
		\\[4pt]
		\hline
		0 & 0 & 0 & \frac{\hbar+u}{\hbar-u}
	\end{array}\right).
\]
All holonomies are now additive and $\mbb{C}$-valued. These reproduce the results of Rim\'anyi and Rozansky \cite{RimanyiRozansky202105} (see Example~\ref{sec:k-theory and cohomology limit for resolved determinantal variety} and Section~\ref{sec:the cohomology limit} for details).
\end{itemize}

{\bf Dynamical $\boldsymbol{R}$-matrices for the fundamental representation of $\boldsymbol{\mfk{sl}(n|1)}$.}
We write down explicit elliptic dynamical $R$-matrices for the fundamental representation of $\mathfrak{sl}(n|1)$. Details are in Example~\ref{exa:elliptic dynamical $R$-matrix for sl(n|1)}. Let us label the basis for $\mathbb C^{n+1}$ by $v_{\alpha}$ where $\alpha\in \{0,1,\dots,n\}$. Then there exists a dynamical $R$-matrix for the tensor product $\mathbb C^{n+1}\otimes \mathbb C^{n+1}$ which reads as follows (see Example~\ref{exa:elliptic dynamical $R$-matrix for sl(n|1)} for details)
\begin{gather}\label{eqn: $R$-matrix for sl(n|1)}
\mathbf R(u,z_1,\dots,z_n)(v_\alpha\otimes v_{\beta})\nonumber\\
\qquad=\begin{cases}
v_\alpha\otimes v_{\beta},& \alpha=\beta<n, \\
D(u)v_\alpha\otimes v_{\beta},& \alpha=\beta=n,\\
\displaystyle C(u)v_\alpha\otimes v_{\beta}+B\left(u,\hbar^{-\delta_{\beta,n}}\prod_{i=\alpha+1}^\beta z_i\right)v_\beta\otimes v_{\alpha}, & \alpha<\beta,\\
\displaystyle A\left(u,\hbar^{-\delta_{\alpha,n}}\prod_{i=\beta+1}^\alpha z_i\right)v_\alpha\otimes v_{\beta}+B\left(u,\hbar^{\delta_{\alpha,n}}\prod_{i=\beta+1}^\alpha z_i^{-1}\right)v_\beta\otimes v_{\alpha}, & \beta<\alpha.
\end{cases}\hspace{-10mm}
\end{gather}
Here we use the shorthand notations for some special functions
\begin{gather*}
A(u,z)=\frac{\vartheta(z\hbar)\vartheta\bigl(z\hbar^{-1}\bigr)\vartheta(u)}{\vartheta(z)^2\vartheta\bigl(u\hbar^{-1}\bigr)},\qquad B(u,z)=\frac{\vartheta(\hbar)\vartheta(u z)}{\vartheta(z)\vartheta\bigl(u^{-1}\hbar\bigr)},\qquad
C(u)=\frac{\vartheta(u)}{\vartheta\bigl(u\hbar^{-1}\bigr)},\\
D(u)=\frac{\vartheta(u\hbar)}{\vartheta\bigl(u^{-1}\hbar\bigr)}.
\end{gather*}
Then
$\mathbf R$ satisfies the dynamical Yang--Baxter equation \eqref{eqn: dYBE_intro} for $\mbs{\mu}=(\mu_1,\dots,\mu_n)$ given by weights for the fundamental representation of $\mathfrak{sl}(n|1)$ (see \eqref{eqn: weights for sl(n|1)} for details on the weights). This suggests that we should actually treat $\mathbb C^{n+1}$ as $\mathbb C^{n|1}$, with basis elements $v_0,v_1,\dots,v_{n-1}$ being even and~$v_n$ being odd.

The rational limit of \eqref{eqn: $R$-matrix for sl(n|1)} is
\[
\mathrm{R}(u)
=\mathrm{P}\left(\frac{u}{u-\hbar}\Pi-\frac{\hbar}{u-\hbar}\mathbf 1\right).
\]
Here $\mathrm{P}$ is the usual swapping-tensor operator, i.e., $\mathrm{P}(v_{\alpha}\otimes v_{\beta})=v_{\beta}\otimes v_{\alpha}$, and $\Pi$ is the super swapping-tensor operator, i.e., $\Pi(v_{\alpha}\otimes v_{\beta})=(-1)^{|v_{\alpha}|\cdot |v_{\beta}|}v_{\beta}\otimes v_{\alpha}$, where $|v|$ is the parity of the vector $v$ which equals to $0$ if $v$ is even and $1$ if $v$ is odd. Note that we can rewrite
\[
(u-\hbar)\Pi \mathrm{P}\mathrm R(u)=u\mathbf 1-\hbar\Pi,
\]
 the right hand side is the more familiar rational $R$-matrix for the fundamental representation of~$\mathfrak{sl}(n|1)$ in the literature \cite{RagoucySatta200706}.

\subsection{Future directions}\label{sec:discussion and future directions}

Let us mention some directions related to the study of this paper that we find interesting:

{\bf Stable envelopes from 4d Chern--Simons theory.} A unifying theme for $(1+1)d$ integrable systems has emerged in recent years. 4d Chern--Simons theory describes quantum integrable spin-chain models and a large class of classical 2d integrable field theories~\cite{Costello201303,Costello201308,CostelloWittenYamazaki201709,CostelloWittenYamazaki201802,CostelloYamazaki201908}. In the particular case of integrable spin-chain models, the $R$-matrix can be constructed by a~relatively straightforward Feynman diagram calculation. However, the mathematical formulation of the BGC shows that more basic objects to study are stable envelopes through which $R$-matrix can be constructed easily. It would be very interesting to provide a 4d Chern--Simons perspective of the mathematical works \cite{AganagicOkounkov201604,MaulikOkounkov201211,Okounkov201512,Okounkov201710} and specially the construction of various stable envelopes. The string-theory realization of the Bethe--Gauge correspondence in which 4d Chern--Simons theory plays a crucial role should be a very useful guide in providing such a~perspective (see \cite{CostelloYagi201810,Dedushenko202107,IshtiaqueMoosavianRaghavendranYagi202110} for some related works in this direction).

{\bf 3d $\boldsymbol{\mcal{N}=2}$ mirror symmetry for stable envelopes.} It is believed that the stable envelopes for two mirror-dual pairs $\bigl(X,X^\vee\bigr)$ must match after certain identification of parameters. This conjecture has been proven in many special cases related to 3d $\mcal{N}=4$ mirror symmetry: cotangent bundle of Grassmannian \cite{RimanyiSmirnovVarchenkoZhou201902}, complete flag variety of type-$A_n$ \cite{RimanyiSmirnovVarchenkoZhou201906} and arbitrary type~\cite{RimanyiWeber201904}, hypertoric varieties \cite{RuanWenZhou202011,SmirnovZhou202006}, and many examples of Bow varieties \cite{RimanyiShou202012}. We expect that a similar conclusion holds for the case of 3d $\mcal{N}=2$ mirror symmetry. Therefore, the first step in proving this conjecture would be to construct the elliptic stable envelope for the Higgs branch of the mirror-dual of the gauge theory we considered here. It is expected that the physical setup of \cite{DedushenkoNekrasov202109} which we employed here in the presence of a Janus interface in the space of FI parameters (instead of a mass-Janus interface) produces the elliptic stable envelopes for the mirror-dual of Higgs branches considered in this paper. This would provide a solid ground for progress towards the mathematical formulation of 3d $\mcal{N}=2$ mirror symmetry (or symplectic duality).\looseness=-1

\subsection{Organization of the paper}

The mathematical and physical contents of the paper are mainly organized in Part \ref{part:math} and Part~\ref{part:physics} respectively. In Section~\ref{sec:mathematical construction of higgs branches}, we characterize the varieties we are interested in, including the classical Higgs branches of 3d $\cN=2$ theories, and in Section~\ref{sec:elliptic stable envelopes for partially-polarized varieties}, we lay down the mathematical foundation of elliptic stable envelopes for these spaces. In Section~\ref{sec:the setup}, we give motivations, in the context of the Bethe/Gauge correspondence, for computing the stable envelopes as certain mass Janus interfaces. In particular, we argue that certain mass Janus partition functions for~3d~$\cN=2$ SQCDs on $I \times \bE_\tau$ should give us the $\sl(1|1)$ elliptic stable envelopes. To compute the Janus partition functions for the 3d theories on intervals we need specific half-BPS boundary conditions corresponding to the vacua of the theory -- we describe these in Section~\ref{sec:3d N=2 SQCD and half-bps boundary conditions}. In Section~\ref{sec:stable envelopes and the $R$-matrix from gauge theory}, we do the explicit computation of the Janus partition functions and find the formulas for the elliptic $\sl(1|1)$ stable envelopes and the $R$-matrix satisfying the dYBE. In the last main section, Section~\ref{sec:2d and 1d avatars of elliptic stable envelopes}, we take certain limits of the elliptic stable envelopes and recover the K-theoretic and the cohomological stable envelopes and their respective $R$-matrices for the $\mfk{sl}(1|1)$ superspin chains.

We include a brief review of equivariant elliptic cohomology in Appendix~\ref{subsec: Equivariant Elliptic Cohomology}. Our conventions regarding supersymmetry are summarized in Appendix~\ref{sec:details of three-dimensional computations}. Lastly, in Appendix~\ref{sec:geometric r-matrix from the a-model localization}, we provide an alternative computation of the rational $R$-matrix for $\sl(1|1)$ spin chains based on the observations in \cite{BullimoreKimLukowski201708,Nekrasov2013,Nekrasov2014} using the A-model computation.

{\bf Note added.} After the appearance of this work, a related work appeared in the context of stable envelopes and 3d $\mcal{N}=2$ gauge theories in \cite{Tamagni202311}. This work does not have any overlap with the present paper.

\subsection{Glossary of notations}
Some often-used symbols:

{\bf Latin}
\begin{longtable}{@{}rl@{}}
$\msf A$ & Maximal torus of $F_{\msf A}$. \\
$\mfr a$ & Lie algebra of $\msf A$. \\
\multirow{2}{*}{$\mfr C$} & A chamber for the $\msf A$-twisted masses, \\[-2pt]
& for $F_{\msf A} = \SU(L)$ it corresponds to a choice of ordering for the $m_i$s. \\
$\bE_\tau$ & An elliptic curve with complex moduli $\tau$, part of the 3d space-time. \\
$\mathrm{Ell}_{\mathsf T}(\mathsf X)$ & $\mathsf T$-equivariant elliptic cohomology of $\mathsf X$.\\
$\mathsf{E}_{\mathsf T}(\mathsf X)$ & Extended $\mathsf T$-equivariant elliptic cohomology of $\mathsf X$.\\
$\mscr E_{\mathsf T}$ & $\mathsf T$-equivariant elliptic cohomology of a point. \\
$\mscr E_{\mathsf K}$ & $\mathsf K$-equivariant elliptic cohomology of a point. \\
$F$ & \begin{tabular}{ll} Total flavor symmetry, \\ $F=F_{\msf A} \times \bC^\times_\hbar$ in Part~\ref{part:math} and $F=F_{\msf A} \times \U(1)_\hbar$ in Part~\ref{part:physics}. \end{tabular} \\
$F_{\msf A}$ & \begin{tabular}{ll} Part of the flavor symmetry, the framing in quivers, \\ complex in Part~\ref{part:math} and real in Part~\ref{part:physics}. In SQCDs $F_{\msf A} = \SU(L)$. \end{tabular} \\
$\cF_p$ & \begin{tabular}{ll} Connected component of the $\msf A$-fixed point set in $\higgs$ labeled by $p$, \\ $\higgs^{\msf A} = \bigsqcup_p \cF_p$. In SQCDs $p$ is a choice of $N$ integers out of $\{1,\dots, L\}$. \end{tabular} \\
$\what\cF_p$ & \begin{tabular}{ll} Lift of $\cF_p$ to the space $\mbf M$ of chirals, \\ projection: $\pi\colon\mbf M^s \twoheadrightarrow \higgs$, $\pi\bigl(\what\cF_p\bigr) = \cF_p$. \end{tabular} \\
$\mfr f$ & Lie algebra of $F$. \\
$\mfr f_{\msf A}$ & Lie algebra of $F_{\msf A}$. \\
$G$ & \begin{tabular}{ll} Gauge group, complex in Part \ref{part:math} and real in Part \ref{part:physics}, \\ in SQCDs $G = \SU(N)$. \end{tabular} \\
$\mfr g$ & Lie algebra of $G$. \\
$\msf H$ & Maximal torus of $G$ \\
$\mfr h$ & Lie algebra of $\msf H$. \\
$\hbar$ & Equivariant elliptic parameter/fugacity for $\U(1)_\hbar$. \\
$I$ & A finite interval $[y_-, y_+]$, part of the 3d space-time. \\
$\mathsf K$ & K\"ahler torus.\\
$\mbf M$ & Space of chiral multiplets. $\mbf M = T^*\Hom\bigl(\bC^L, \bC^N\bigr)$ for SQCDs. \\
$\mbf M^s$ & Stable locus of $\mbf M$. \\
$\mbf M_0$ & Space of massless chirals in the vacuum $p$. \\
$\mbf M_{\mfr C; \pm}(p)$ & Space of $+/-$-vely massive chirals in the vacuum $p$ and chamber $\mfr C$. \\
$\higgs$ & Classical Higgs branch. \\
$m_i$, $m_i^\bC$ & $\msf A$-twisted (in 3d) and $\msf A_\bC$-twisted (in 2d) masses. \\
$m_i^I$ & Triplet of $\sfA$-twisted masses in 1d, $I=1,2,3$. \\
$\mathrm{Pol}_{\mathsf X}$ & Partial polarization of $\mathsf X$. \\
$\mathrm{Pol}^{\mathrm{op}}_{\mathsf X}$ & Opposite partial polarization of $\mathsf X$. \\
$Q$ & Quiver/Kac--Dynkin diagram (e.g., \eqref{an example of Kac--Dynkin diagram}). \\
$\ov Q$ & Framed and doubled quiver (e.g., \eqref{an example of type A quiver}). \\
$\wtd Q$ & Tripled/gauge theory quiver (e.g., \eqref{quiver31}). \\
$\cQ$, $\wtd \cQ$ & Chiral multiplets in SQCDs valued in $\Hom\bigl(\bC^L, \bC^N\bigr)$ and $\Hom\bigl(\bC^N, \bC^L\bigr)$. \\
$\mcr Q$, $\msf Q$, $\mbf Q$ & Supercharges in 3d, 2d, and 1d. \\
$\cR$, $\wtd \cR$ & Space of the $Q$s and $\wtd Q$s in SQCD. \\
$\cR_0(p)$, $\wtd\cR_(p)$ & Spaces of massless $Q$s and $\wtd Q$s in the vacuum $p$. \\
$\cR_{\mfr C;\pm}(p)$, $\wtd \cR_{\mfr C;\pm}(p)$ & $+/-$-ve mass subspaces of $\cR$, $\wtd\cR$ in the vacuum $p$ and chamber $\mfr C$. \\
$\ermatrix$, $\trmatrix^{\msf{s}}$, $\rrmatrix$ & Elliptic, K-theoretic, and rational $R$-matrices. \\
$\bR_\hbar$, $\bR_\text{top}$ & Lie algebras of $\U(1)_\hbar$, and $\U(1)_\text{top}$. \\
$\msf{s}$ & Slope in the K-theoretic stable envelope. \\
$s_a$ & Equivariant elliptic parameter/fugacity for $\msf H$. In SQCDs $a=1,\dots, N$. \\
$S^1_\A$, $S^1_\B$ & \begin{tabular}{ll} Non-contractible cycles of $\bE_\tau = S^1_\A \times S^1_\B$, \\ we go from 3d $\rightarrow$ 2d $\rightarrow$ 1d by reducing $S^1_\B$ first and then $S^1_\A$. \end{tabular} \\
$\estab$, $\tstab^{\msf{s}}$, $\rstab$ & Elliptic, K-theoretic, and cohomological stable envelopes. \\
$\msf T$ & Maximal torus of $F$. \\
$\mfr t$ & Lie algebra of $\msf T$. \\
$T_{\mathsf X}$ & Tangent bundle of $\mathsf X$.\\
$\mscr U$ & Universal line bundle on elliptic cohomology scheme.\\
$\U(1)_\hbar$, $\U(1)_\text{top}$ & Parts of the flavor symmetry. $\U(1)_\text{top}$ acts only on monopoles in 3d. \\
$\msf X$ & Smooth quasi-projective complex variety with $\msf T$ action. \\
$x_i$ & Equivariant elliptic parameter/fugacity for $\msf A$. In SQCDs $i=1,\dots,L$. \\
$y$ & Real coordinate along $I$. \\
$z$ & Equivariant elliptic parameter/fugacity for $\U(1)_\text{top}$.
\end{longtable}

{\bf Greek}
\begin{longtable}{@{}rl@{}}
$\ze$, $\ze^\bC$ & Real (in 3d) and complex in 2d FI parameters. \\
$\ze^I$ & Triplet of real FI parameters in 1d, $I=1,2,3$. \\
$\mu$ & Complex moment map for the gauge group $G$ action.\\
$\mu_{\mathrm{ev}}$ & Complex moment map for the even part $G_{\mathrm{ev}}$ of the gauge group action.\\
$\mu_{\mfr g}$ & Real moment map for the gauge group $G$ action. \\
$\si$ & Real adjoint scalar in 3d $\cN=2$ vector multiplet. \\
$\si^\bC$ & Complex adjoint scalar in 2d $\cN=(2,2)$ vector multiplet.
\end{longtable}

\part{Stable envelopes: beyond symplectic varieties} \label{part:math}

\section[Classical Higgs branches of 3d $\mcal N=2$ gauge theories]{Classical Higgs branches of 3d $\boldsymbol{\mcal N=2}$ gauge theories}
\label{sec:mathematical construction of higgs branches}

The essential data extracted from a 3d $\mcal{N}=2$ gauge theory is a complex algebraic group $G$, a~complex $G$-representation $\mathbf M$, with a $G$-invariant algebraic function $\mcal{W}\colon\mathbf M\to \mathbb C$, and a character~${\zeta\colon G\to \mathbb C^{\times}}$. The Higgs branch of the 3d $\mcal{N}=2$ gauge theory associated to $(G,\mathbf M,\mcal W,\zeta)$ is then defined as the GIT quotient
\begin{align}
 \higgs(G,\mathbf M,\mcal W,\zeta):=\mathrm{Crit}(\mcal W)^{\zeta-ss}/G, \label{MHgitDef}
\end{align}
where $\mathrm{Crit}(\mcal W)^{\zeta-ss}$ is the $\zeta$-semistable locus of the critical locus $\mathrm{Crit}(\mcal W)$ which is the solution to the equation $\mathsf d\mcal W=0$.

{\bf Assumption.} We assume that the semistable locus $\mathrm{Crit}(\mcal W)^{\zeta-ss}$ is smooth and the action of $G$ on it is free so that we get a smooth Higgs branch $\higgs(G,\mathbf M,\mcal W,\zeta)$. As we will see later, this assumption is satisfied in many examples that we consider. This setting is not the most general one for generic 3d $\mcal{N}=2$ gauge theories, but this is enough for our purpose.

\subsection{Example: GIT quotients}\label{subsubsec: Example: Quotient of Hamiltonian Reduction}
A typical example of $\higgs(G,\mathbf M,\mcal W,\zeta)$ is as follows. Let $G=G_{\mathrm{ev}}\times G_{\mathrm{odd}}$ that acts on a~representation $\mcal R$ and then we take $\mathbf M:=\mcal R\oplus \mcal R^{\vee}\oplus \mathfrak g_{\mathrm{ev}}$. We choose a complex moment map~${\mu\colon \mcal R\oplus \mcal R^{\vee}\!\!\to \mathfrak{g}}$ for the $G$ action, and define $\mu_{\mathrm{ev}}\colon\mcal R\oplus \mcal R^{\vee}\!\!\to \mathfrak{g}_{\mathrm{ev}}^{\vee}$ to be the composition $\mathrm{pr}_{\mathrm{ev}}\circ \mu$, where~${\mathrm{pr}_{\mathrm{ev}}\colon \mathfrak{g}^\vee\!\!\to \mathfrak{g}^\vee_{\mathrm{ev}}}$ is the projection to the even part. Note that $\mu_{\mathrm{ev}}$ is the moment map for the action of $G_{\mathrm{ev}}$ on~$\mathbf M$. We take $\mcal W=\langle X,\mu_{\mathrm{ev}}\rangle$ where $X$ is the coordinate on $\mathfrak g_{\mathrm{ev}}$ and $\langle\cdot,\cdot\rangle$ is the pairing between~$\mathfrak g_{\mathrm{ev}}$ and~$\mathfrak g_{\mathrm{ev}}^{\vee}$. Then we choose a character $\zeta\colon G\to \mathbb C^{\times}$. In this case, the Higgs branch is then isomorphic~to
$ \higgs(G,\mathbf M,\mcal W,\zeta)\cong\mu_{\mathrm{ev}}^{-1}(0)^{\zeta-ss}/\!\!/G$,
where $\mu_{\mathrm{ev}}^{-1}(0)^{\zeta-ss}$ is the $\zeta$-semistable locus of the~$\mu_{\mathrm{ev}}^{-1}(0)$. In other words, $\higgs(G,\mathbf M,\mcal W,\zeta)$ is a~further quotient of a~Hamiltonian reduction by a group action
\begin{align}\label{eqn: quotient of Hamiltonian reduction}
 \higgs(G,\mathbf M,\mcal W,\zeta)\cong\bigl(\bigl(\mcal R\oplus \mcal R^{\vee} \bigr)^{\zeta-ss}{/\!\!/\!\!/} G_{\mathrm{ev}}\bigr)/\!\!/G_{\mathrm{odd}}.
\end{align}
Notice that we have a $G$-equivariant closed embedding
$\mu^{-1}(0)\hookrightarrow \mu^{-1}\bigl(\mathfrak g_{\mathrm{ev}}^{\perp}\bigr)=\mu_{\mathrm{ev}}^{-1}(0)$,
which leads to a closed embedding
\[
\mu^{-1}(0)^{\zeta-ss}/\!\!/G\hookrightarrow \mu_{\mathrm{ev}}^{-1}(0)^{\zeta-ss}/\!\!/G.
\]
In other words, if we consider the moment map $\mu$ for the whole gauge group $G$, and perform the Hamiltonian reduction for the $G$-action, then the resulting variety is a smooth and closed (but could be empty) subvariety of $\higgs(G,\mathbf M,\mcal W,\zeta)$.

\begin{ex}\label{ex: abelian gauge theory}
When $G$ is abelian, $G$ acts trivially on its Lie algebra $\mathfrak g$, and we have a commutative diagram
\begin{equation*}
\begin{tikzcd}
\mu^{-1}(0)^{\zeta-ss}/\!\!/G \arrow[r,hook] \arrow[d] & \mu^{-1}_{\mathrm{ev}}(0)^{\zeta-ss}/\!\!/G \arrow[r,hook] \arrow[d] & \bigl(\mcal R\oplus \mcal R^{\vee}\bigr)^{\zeta-ss}/\!\!/G \arrow[d] \\
\{0\} \arrow[r,hook] & \mathfrak g_{\mathrm{ev}}^{\perp} \arrow[r,hook] & \mathfrak g^\vee
\end{tikzcd}
\end{equation*}
such that every square is Cartesian. The whole diagram is equivariant under the action of the flavour torus $\msf{A}=(\mathbb C^{\times})^{\operatorname{rk}\mcal R}/G$. The central fiber $\mu^{-1}(0)^{\zeta-ss}/\!\!/G$ is a hypertoric variety and the GIT quotient \smash{$\bigl(\mcal R\oplus \mcal R^{\vee}\bigr)^{\zeta-ss}/\!\!/G$} is known as the Lawrence toric variety \cite{hausel2002toric}. For a generic choice of $\zeta$, we have \smash{$\bigl(\mcal R\oplus \mcal R^{\vee}\bigr)^{\zeta-ss}=\bigl(\mcal R\oplus \mcal R^{\vee}\bigr)^{\zeta-s}$} and the latter is smooth over the base $\mathfrak g^\vee$, therefore \smash{$\bigl(\mcal R\oplus \mcal R^{\vee}\bigr)^{\zeta-ss}/\!\!/G$} is flat over the base $\mathfrak g^\vee$ and $\mu^{-1}_{\mathrm{ev}}(0)^{\zeta-ss}/\!\!/G$ is flat over $\mathfrak g_{\mathrm{ev}}^{\perp}$. In the case when the charge matrix $A\colon\mathbb Z^{\operatorname{rk}\mcal R}\to \mathrm{Char}(G)$ is surjective and unimodular,\footnote{The action of $G$ on $\mcal R$ gives a homomorphism $G\to (\mathbb C^{\times})^{\operatorname{rk}\mcal R}$, the charge matrix $A$ is the induced map on characters. $A$ is called unimodular if and only if every $\operatorname{rk}G\times \operatorname{rk}G$ submatrix has determinant $\in \{0,\pm 1\}$.} the action of~$G$ on~\smash{$\bigl(\mcal R\oplus \mcal R^{\vee}\bigr)^{\zeta-ss}$} is free,\footnote{In fact, $G$ acts on \smash{$\bigl(\mcal R\oplus \mcal R^{\vee}\bigr)^{\zeta-ss}$} freely $\Longleftrightarrow$ $G$ acts on $\mu^{-1}(0)^{\zeta-ss}$ freely $\Longleftrightarrow$ charge matrix is surjective and unimodular \cite{hausel2002toric}.} so \smash{$\bigl(\mcal R\oplus \mcal R^{\vee}\bigr)^{\zeta-ss}/\!\!/G$} is smooth over $\mathfrak g^\vee$, thus $\mu^{-1}_{\mathrm{ev}}(0)^{\zeta-ss}/\!\!/G$ is a~smooth deformation of the hypertoric variety $\mu^{-1}(0)^{\zeta-ss}/\!\!/G$ over the base $\mathfrak g_{\mathrm{ev}}^{\perp}$.
\end{ex}

In the abelian gauge group example, we have the implication
\[
\mu^{-1}_{\mathrm{ev}}(0)^{\zeta-ss}\neq \varnothing \Rightarrow \mu^{-1}(0)^{\zeta-ss}\neq \varnothing.
\]
 In fact, if $\mu^{-1}_{\mathrm{ev}}(0)^{\zeta-ss}$ is nonempty, then $\mu^{-1}_{\mathrm{ev}}(0)^{\zeta-ss}/\!\!/G$ is projective over the affine scheme $\mu^{-1}_{\mathrm{ev}}(0)/\!\!/G$. The projection is equivariant under the $\mathbb C^{\times}_{r}$ which acts on $\mcal R\oplus \mcal R^{\vee}$ with weight~$r$ and acts on $\mathfrak g_{\mathrm{ev}}^{\perp}$ and $\mathfrak g^\vee$ with weight $r^{2}$. $\mathbb C^{\times}_{r}$ contracts $\mu^{-1}_{\mathrm{ev}}(0)/\!\!/G$ to a unique point, which is the image of $0\in \mcal R\oplus\mcal R^\vee$, therefore the image of $\mu^{-1}_{\mathrm{ev}}(0)^{\zeta-ss}/\!\!/G$ in $\mu^{-1}_{\mathrm{ev}}(0)/\!\!/G$, which is a $\mathbb C^{\times}_{\hbar}$-invariant closed subset, must contain the image of $0$. In particular, $\mu^{-1}(0)^{\zeta-ss}/\!\!/G$ is nonempty.

 To summarize, in the case when the gauge group is abelian and under the unimodular assumption, we always get a smooth deformation of the corresponding hypertoric variety over the base $\mathfrak g_{\mathrm{ev}}^{\perp}$, and the deformation is equivariant under the action of $\msf{A}\times \mathbb C^{\times}_{r}$, where $\msf{A}$ is the flavour torus $(\mathbb C^{\times})^{\operatorname{rk}\mcal R}/G$. The smooth deformation gives rise to an isomorphism between cohomology groups $H^*\bigl(\mu^{-1}(0)^{\zeta-ss}/\!\!/G\bigr)\cong H^*\bigl(\mu^{-1}_{\mathrm{ev}}(0)^{\zeta-ss}/\!\!/G\bigr)$, and we will see that the stable envelopes for~${\mu^{-1}(0)^{\zeta-ss}/\!\!/G}$ and $\mu^{-1}_{\mathrm{ev}}(0)^{\zeta-ss}/\!\!/G$ are the same (see Example~\ref{ex: abelian stable envelope}).

 However, when $G$ is nonabelian, the cohomologies of $\mu^{-1}(0)^{\zeta-ss}/\!\!/G$ and $\mu^{-1}_{\mathrm{ev}}(0)^{\zeta-ss}/\!\!/G$ are not isomorphic in general; more drastically, it is possible that $\mu^{-1}_{\mathrm{ev}}(0)^{\zeta-ss}$ is nonempty but $\mu^{-1}(0)^{\zeta-ss}$ is empty. See Example~\ref{ex: quiver with one f node 2} below.

\begin{rmk}[tautological generation]
We say that the K-theory $K\bigl(\mu_{\mathrm{ev}}^{-1}(0)^{\zeta-ss}/\!\!/G\bigr)$ is generated by tautological classes if the natural map $K_G(\mathrm{pt})\to K\bigl(\mu_{\mathrm{ev}}^{-1}(0)^{\zeta-ss}/\!\!/G\bigr)$ is surjective. This property is also known as Kirwan surjectivity. For hyperk\"ahler quotients (i.e., $G_{\mathrm{odd}}=1$), it is~known that both hypertoric varieties (i.e., $G$ is torus) and the Nakajima quiver varieties have~the Kirwan surjectivity property, see \cite{harada2004properties,mcgerty2018kirwan}. Let $\zeta_b$ be the restriction of $\zeta$ to $G_{\mathrm{ev}}$, then we have a natural open embedding \smash{$\bigl(\mcal R\oplus\mcal R^\vee\bigr)^{\zeta-ss}\subset \bigl(\mcal R\oplus\mcal R^\vee\bigr)^{\zeta_b-ss}$}, which induces open embedding:
\begin{align*}
 \bigl(\mcal R\oplus \mcal R^{\vee} \bigr)^{\zeta-ss}{/\!\!/\!\!/} G_{\mathrm{ev}}\hookrightarrow \bigl(\mcal R\oplus \mcal R^{\vee}\bigr) {/\!\!/\!\!/}_{\zeta_b} G_{\mathrm{ev}},
\end{align*}
where the right-hand-side is the usual hyperk\"ahler reduction. Assume that $K_{G_{\mathrm{odd}}}\bigl(\bigl(\mcal R\oplus \mcal R^{\vee}\bigr) \allowbreak\smash{{/\!\!/\!\!/}_{\zeta_b} G_{\mathrm{ev}}\bigr)}$ is generated by $K_{G_{\mathrm{ev}}\times G_{\mathrm{odd}}}(\mathrm{pt})=K_{G}(\mathrm{pt})$, then \smash{$K_{G_{\mathrm{odd}}}\bigl(\bigl(\mcal R\oplus \mcal R^{\vee}\bigr)^{\zeta-ss}{/\!\!/\!\!/} G_{\mathrm{ev}}\bigr)$} is also generated by $K_{G}(\mathrm{pt})$, whence \eqref{eqn: quotient of Hamiltonian reduction} implies that $K\bigl(\mu_{\mathrm{ev}}^{-1}(0)^{\zeta-ss}/\!\!/G\bigr)$ is generated by $K_{G}(\mathrm{pt})$. In other words,
\begin{itemize}\itemsep=0pt
 \item $G_{\mathrm{odd}}$-equivariant Kirwan surjectivity of $\bigl(\mcal R\oplus \mcal R^{\vee}\bigr) {/\!\!/\!\!/}_{\zeta_b} G_{\mathrm{ev}}$ implies the Kirwan surjectivity of $\mu_{\mathrm{ev}}^{-1}(0)^{\zeta-ss}/\!\!/G$.
\end{itemize}
Since for the hypertoric varieties and the Nakajima quiver varieties, the Kirwan surjectivity holds for flavour group equivariant K-theory, we conclude that: if either $G$ is a torus or $(\mcal R,G)$ is constructed from a quiver representation (see Section~\ref{subsubsec: Example: The Doubled Quivers}), then $\mu_{\mathrm{ev}}^{-1}(0)^{\zeta-ss}/\!\!/G$ has the Kirwan surjectivity.
\end{rmk}

\subsection{Example: quiver varieties}\label{subsubsec: Example: The Doubled Quivers}
A large class of GIT quotients comes from 3d $\mcal{N}=2$ quiver gauge theories, and we recall the construction of the Higgs branches of such theories here.

 Let $Q=(Q_0,Q_1)$ be a quiver, i.e., $Q_0$ and $Q_1$ are finite sets (known as the set of nodes and edges respectively) together with two maps $h,t\colon Q_1\to Q_0$ sending an edge to its head and tail respectively. We add one more structure to the quiver $Q$, namely, we separate $Q_0$ into two parts~${Q_0=Q_0^{\mathrm{ev}}\sqcup Q_0^{\mathrm{odd}}}$, which will be called even and odd respectively. Our notation for the even and odd nodes is as follows
\begin{equation*}
 \begin{tikzpicture}[x={(1.5cm,0cm)}, y={(0cm,1.5cm)}]
 \node[draw, circle, label=above:even] at (2,0) (n3) {};
 \node[draw, circle, cross, label=above:odd] at (0,0) (n1) {};
 \end{tikzpicture}.
\end{equation*}
Let $\mathbf w,\mathbf v\in \mathbb N^{Q_0}$ be $Q_0$-tuples of natural numbers, then the gauge group and the chiral multiplets are built from the above data as follows.

 The gauge group $G=G_{\mathrm{ev}}\times G_{\mathrm{odd}}$ is such that
\begin{align*}
 G_{\mathrm{ev}}=\prod_{i\in Q_0^{\mathrm{ev}}}\mathrm{GL}(\mathbf v_i),\qquad G_{\mathrm{odd}}=\prod_{i\in Q_0^{\mathrm{odd}}}\mathrm{GL}(\mathbf v_i).
\end{align*}
And the space of chiral multiplets $\mathbf M=\mcal R\oplus \mcal R^{\vee}\oplus \mathfrak g_{\mathrm{ev}}$ is such that
\begin{align*}
 \mcal R=\bigoplus_{i\in Q_0}\mathrm{Hom}(V_i,W_i)\oplus \bigoplus_{a\in Q_1}\mathrm{Hom}(V_{t(a)},V_{h(a)}),
\end{align*}
where $W_i$ is the $\mathbf w_i$-dimensional complex vector space and $V_i$ is the $\mathbf v_i$-dimensional complex vector space. The action of $G$ on $\mcal R$ is the obvious one, namely $\mathrm{GL}(\mathbf v_i)$ acts on $V_i$ as the fundamental representation. We give notations to the field contents of $\mcal R$ by letting $\alpha_i\in\mathrm{Hom}(V_i,W_i)$ for~${i\in Q_0}$ and $x_{a}\in \mathrm{Hom}(V_{t(a)},V_{h(a)})$ for~${a\in Q_1}$; for the dual representation $\mcal R^\vee$, we use the notation $\widetilde{\alpha}_i\in \mathrm{Hom}(W_i,V_i)$ and $\widetilde{x}_{a}\in \mathrm{Hom}(V_{h(a)},V_{t(a)})$ for $a\in Q_1$.

 The holomorphic symplectic form on $\mcal R\oplus \mcal R^{\vee}$ is
\[
 \omega=\sum_{a\in Q_1}\sd x_{a}\wedge \sd \widetilde{x}_{a}+\sum_{i\in Q_0}\sd \alpha_{i}\wedge \sd \widetilde{\alpha}_{i},
\]
so there is a moment map $\mu\colon\mcal R\oplus \mcal R^{\vee}\to \mathfrak g^\vee$ such that
\[
 \mu\bigl(x_a,\widetilde{x}_a,\alpha_i,\widetilde{\alpha}_i\bigr)=\sum_{a\in Q_1}\bigl[x_a,\widetilde{x}_a\bigr]+\sum_{i\in Q_0} \widetilde{\alpha}_{i}\alpha_{i}.
\]
We define $\mu_{\mathrm{ev}}\colon \mcal R\oplus \mcal R^{\vee}\to \mathfrak g^\vee_{\mathrm{ev}}$ to be the composition of $\mu$ with the projection $\mathfrak g^\vee\to \mathfrak g^\vee_{\mathrm{ev}}$.

 We introduce the flavour symmetry group $F=G_W\times \mathbb C^{\times}_{\hbar}$, such that
\smash{$G_W=\prod_{i\in Q_0}\mathrm{GL}(\mathbf w_i)$},
$\mathrm{GL}(\mathbf w_i)$ acts on $W_i$ by fundamental representation, and $\mathbb C^{\times}_{\hbar}$ acts on $\mathbf M$ by scaling $\mcal R^{\vee}$ with weight~$\hbar^{-1}$ and scaling $\mathfrak g_{\mathrm{ev}}$ with weight $\hbar$ and fixing $\mcal R$. Note that the action of $F$ on $\mathbf M$ commutes with the action of $G$. We also fix a maximal torus $\mathsf T$ of $F$, namely $\mathsf T=T_W\times \mathbb C^{\times}_{\hbar}$ such that~${T_W=\prod_{i\in Q_0}(\mathbb C^{\times })^{\mathbf w_i}}$ is a maximal torus of $G_W$.

 Finally, we choose a character $\zeta\colon G\to \mathbb C^{\times}$, $\zeta$ can be written as $\zeta(g)=\prod_{i\in Q_0}\det(g_i)^{\zeta_i}$. Then the Higgs branch of the 3d $\mcal{N}=2$ quiver gauge theory associated to the quiver $Q$ with gauge rank $\mathbf v$ and flavour rank $\mathbf w$ is isomorphic to the GIT quotient
\begin{align}\label{HiggsSyQ}
 \mcal M^{\zeta}(\mathbf v,\mathbf w)=\mu_{\mathrm{ev}}^{-1}(0)^{\zeta-ss}/\!\!/G,
\end{align}
and we write
\begin{align*}
 \mcal M^{\zeta}(\mathbf w):=\underset{\mathbf v\in \mathbb N^{Q_0}}{\bigsqcup} \mcal M^{\zeta}(\mathbf v,\mathbf w),
\end{align*}
for a fixed $\mathbf w$. In other words, we impose complex moment map equations only for even nodes.

 If $\zeta$ equals to $\zeta_+:=(1,\dots,1)$ or $\zeta_-:=(-1,\dots,-1)$, then according to the King's criterion for the stability \cite{king1994moduli}, a quiver representation $\bigl(V,x_a,\widetilde{x}_a,\alpha_i,\widetilde{\alpha}_i\bigr)\in \mcal R\oplus \mcal R^{\vee}$ is $\zeta$-semistable if and only if the following condition is satisfied respectively:
\begin{itemize}\itemsep=0pt
 \item[($\zeta_+$)] If $S_i\subset V_i$ are subspaces such that $S$ is preserved under the maps $\bigl(x_a,\widetilde{x}_a\bigr)$, and that $S_i\supset \Im(\widetilde{\alpha}_i)$ for all $i\in Q_0$, then $S=V$.
 \item[($\zeta_-$)] If $T_i\subset V_i$ are subspaces such that $T$ is preserved under the maps $\bigl(x_a,\widetilde{x}_a\bigr)$, and that $T_i\subset \mathrm{Ker}(\alpha_i)$ for all $i\in Q_0$, then $T=0$.
\end{itemize}
In the following discussion, we always assume that
\begin{itemize}
 \item $\zeta$ is generic, i.e., $\bigl(\mcal R\oplus \mcal R^{\vee}\bigr)^{\zeta-s}=\bigl(\mcal R\oplus \mcal R^{\vee}\bigr)^{\zeta-ss}$.
\end{itemize}
For example, $\zeta_\pm$ are generic.

It is well known that if $\zeta$ is generic, then \smash{$\mu\colon\bigl(\mcal R\oplus \mcal R^{\vee}\bigr)^{\zeta-ss}\to \mathfrak g^\vee$} is smooth, and the action of $G$ on \smash{$\bigl(\mcal R\oplus \mcal R^{\vee}\bigr)^{\zeta-ss}$} is free. Since the projection $\mathrm{pr}_{\mathrm{ev}}\colon\mathfrak g^\vee\to \mathfrak g^\vee_{\mathrm{ev}}$ is smooth, the composition~${\mu_{\mathrm{ev}}=\mathrm{pr}_{\mathrm{ev}}\circ\mu}$ is smooth. Therefore, we have the following lemma.

\begin{lem}\label{lem:G/Higgs}
Assume that $\zeta$ is generic then, $\mcal M^{\zeta}(\mathbf v,\mathbf w)$ is a smooth variety and the quotient map $\mu_{\mathrm{ev}}^{-1}(0)^{\zeta-ss}\to \mcal M^{\zeta}(\mathbf v,\mathbf w)$ is a principal $G$-bundle.
\end{lem}

\begin{ex}\label{ex: quiver with no f node}
If there is no odd node, i.e., $Q^{\mathrm{odd}}_0$ is empty, then $\mu_{\mathrm{ev}}\!=\!\mu$ and in this case~$\mcal M^{\zeta}(\mathbf v,\mathbf w)$ is a Nakajima quiver variety \cite{Nakajima199411,nakajima1998quiver}.
\end{ex}

\begin{ex}\label{ex: quiver with one f node 1}
Let $Q$ be an $A_{n}$ quiver with $Q^{\mathrm{ev}}_0=\{1,\dots,n-1\}$ and $Q^{\mathrm{odd}}_0=\{n\}$, take framing dimension $\mathbf w=(r,0,\dots,0)$ and gauge dimension $\mathbf v=(v_1,\dots,v_{n})$. Below is the doubled quiver~$\overline{Q}$,
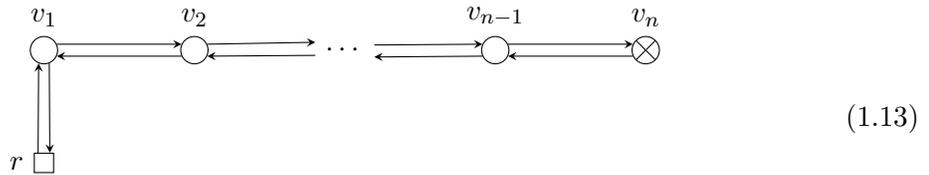
\[
\begin{tikzpicture}[x={(2cm,0cm)}, y={(0cm,2cm)}, baseline=-1cm]
\node[draw, circle, label=above:$v_1$] at (0,0) (n1) {};
\node[draw, circle, label=above:$v_2$] at (1,0) (n2) {};
\node at (2,0) (ndot) {$\cdots$};
\node[draw, circle, label=above:$v_{n-1}$] at (3,0) (n3) {};
\node[draw, circle,cross, label=above:$v_n$] at (4,0) (n4) {};

\draw[-stealth] (n1.25) to (n2.155);
\draw[-stealth] (n2.25) to (ndot.165);
\draw[-stealth] (ndot.10) to (n3.155);
\draw[-stealth] (n3.25) to (n4.155);

\draw[-stealth] (n4.205) to (n3.335);
\draw[-stealth] (n3.205) to (ndot.347);
\draw[-stealth] (ndot.190) to (n2.335);
\draw[-stealth] (n2.205) to (n1.335);

\node[draw, rectangle, label=left:$r$] at (0,-.75) (f1) {};

\draw[-stealth] (n1.292) to (f1.60);
\draw[-stealth] (f1.120) to (n1.248);
\end{tikzpicture}
\]
$\mcal M^{\zeta_+}(\mathbf v,\mathbf w)$ is nonempty if and only if $r\ge v_1\ge\cdots\ge v_n$, and if nonempty then $\mcal M^{\zeta_+}(\mathbf v,\mathbf w)$ is isomorphic to $\mathrm{GL}_r{\times}^P\mathfrak{m}$, where $P$ is the parabolic subgroup of $\mathrm{GL}_r$ that stabilizes a fixed flag~${F_\bullet=F_1\subset F_2\subset\cdots\subset F_n\subset F_{n+1}=\mathbb C^r}$ such that $\dim F_{n+1}/F_i=v_i$, and $\mathfrak m$ is the Lie algebra of the subgroup $M\subset P$ that acts on $F_{i+1}/F_{i}$ as identity for all $i$ except for $i=n$. $\mcal M^{\zeta_+}(\mathbf v,\mathbf w)$ is a vector bundle over the flag variety $\mathrm{Fl}_{\mathbf v}=\mathrm{GL}_r/P$. Note that $M$ contains the unipotent radical of $P$, therefore $\mcal M^{\zeta_+}(\mathbf v,\mathbf w)$ contains $T^*\mathrm{Fl}_{\mathbf v}$ as a closed subvariety.
\end{ex}

\begin{ex}\label{ex: quiver with one f node 2}
Let $Q$ be an $A_{n}$ quiver with $Q^{\mathrm{ev}}_0=\{2,\dots,n\}$ and $Q^{\mathrm{odd}}_0=\{1\}$, take framing dimension $\mathbf w=(r,0,\dots,0)$ and gauge dimension $\mathbf v=(v_1,\dots,v_{n})$. Below is the doubled quiver~$\overline{Q}$,
\[
\begin{tikzpicture}[x={(2cm,0cm)}, y={(0cm,2cm)}, baseline=-1cm]
\node[draw, circle, cross, label=above:$v_1$] at (0,0) (n1) {};
\node[draw, circle, label=above:$v_2$] at (1,0) (n2) {};
\node at (2,0) (ndot) {$\cdots$};
\node[draw, circle, label=above:$v_{n-1}$] at (3,0) (n3) {};
\node[draw, circle, label=above:$v_n$] at (4,0) (n4) {};

\draw[-stealth] (n1.25) to (n2.155);
\draw[-stealth] (n2.25) to (ndot.165);
\draw[-stealth] (ndot.10) to (n3.155);
\draw[-stealth] (n3.25) to (n4.155);

\draw[-stealth] (n4.205) to (n3.335);
\draw[-stealth] (n3.205) to (ndot.347);
\draw[-stealth] (ndot.190) to (n2.335);
\draw[-stealth] (n2.205) to (n1.335);

\node[draw, rectangle, label=left:$r$] at (0,-.75) (f1) {};

\draw[-stealth] (n1.292) to (f1.60);
\draw[-stealth] (f1.120) to (n1.248);
\end{tikzpicture}
\]
$\mcal M^{\zeta_+}(\mathbf v,\mathbf w)$ is nonempty if and only if $v_n\le r$ and $v_{i+1}\le v_i\le v_{i+1}+r$. $\mcal M^{\zeta_+}(\mathbf v,\mathbf w)$ can be described as follows. Let $d_i=v_i-v_{i+1}$ for $1\le i\le n-1$ and $d_n=v_n$, then $\mcal M^{\zeta_+}(\mathbf v,\mathbf w)$ is a~vector bundle over the convolution of affine Grassmannian orbits
\[
\widetilde{\mathrm{Gr}}_{\mathrm{GL}_r}^{\vec\lambda}=\mathrm{Gr}_{\mathrm{GL}_r}^{\omega_{d_n}}\widetilde{\times}\mathrm{Gr}_{\mathrm{GL}_r}^{\omega_{d_{n-1}}}\widetilde{\times}\cdots \widetilde{\times}\mathrm{Gr}_{\mathrm{GL}_r}^{\omega_{d_1}},
\]
where $\vec\lambda=(\omega_{d_n},\dots,\omega_{d_1})$ and $\omega_i$ is the $i$-th fundamental coweight of $\mathrm{GL}_r$. The convolution of the affine Grassmannian gives a map \smash{$m\colon \widetilde{\mathrm{Gr}}_{\mathrm{GL}_r}^{\vec\lambda}\to \overline{\mathrm{Gr}}_{\mathrm{GL}_r}^{v_1\omega_1}$}, and the latter represents sub-lattices~$\mathcal L$ in $\mathbb C[z]^{\oplus r}$ of codimension $v_1$, so it has a $v_1$-dimensional vector bundle $\mathcal V$ which is the universal quotient $\mathbb C[z]^{\oplus r}/\mathcal L_{\mathrm{univ}}$. $\mcal M^{\zeta_+}(\mathbf v,\mathbf w)$ is the total space of the vector bundle $\mathrm{Hom}(\mathcal V,\mathcal O^{\oplus r})$ on the~\smash{$\widetilde{\mathrm{Gr}}_{\mathrm{GL}_r}^{\vec\lambda}$}.
\end{ex}

In the last example, if we replace the odd node with a even one, then $\mcal M^{\zeta_+}(\mathbf v,\mathbf w)$ is nonempty if and only if $r\ge v_1\ge v_2\ge\cdots\ge v_n$. In particular, if $n\ge 2$ then there exists $\mathbf v$ such that~$\mu^{-1}(0)^{\zeta_+-ss}$ is empty but $\mu^{-1}_{\mathrm{ev}}(0)^{\zeta_+-ss}$ is nonempty, for instance $\mathbf v=(nr,(n-1)r,\dots,r)$.

 Let $\msf{A}\subset T_W$ be a subtorus, such that $W$ decomposes as eigenspaces
\[
 W=\bigoplus_{\lambda\in \mathrm{Char}(\msf{A})}W^\lambda,
\]
and we write $\mathbf w=\sum_{\lambda}\mathbf w^\lambda$ for the dimension vector, then it is well known that the $\msf{A}$-fixed points set $\mcal M^{\zeta}(\mathbf w)^{\msf{A}}$ is isomorphic to the product
\[
 \mcal M^{\zeta}(\mathbf w)^{\msf{A}}\cong \prod_{\lambda\in \mathrm{Cochar}(\msf{A})}\mcal M\bigl(\mathbf w^\lambda\bigr).
\]

\section{Elliptic stable envelopes for partially-polarized varieties}
\label{sec:elliptic stable envelopes for partially-polarized varieties}
Aganagic--Okounkov defined the elliptic stable envelope and showed their existence and uniqueness for hypertoric varieties and Nakajima quiver varieties \cite{AganagicOkounkov201604}, later Okounkov extended the result on the existence and uniqueness of elliptic stable envelope to algebraic stacks with a~polarization~\cite{okounkov2020nonabelian} and to varieties with an attractive line bundle on its equivariant elliptic cohomology~\cite{okounkov2021inductive}.

 A polarization, namely a decomposition of tangent bundle \smash{$T_X=T^{1/2}_X+\bigl(T^{1/2}_X\bigr)^\vee$} in the K-theory class of a variety or stack $X$, is manifested for the Higgs branches of 3d $\mcal N=4$ gauge theories, since the moduli space is holomorphic symplectic by construction.

 However, in the 3d $\mcal N=2$ setting, we do not have a polarization in general. In fact, the Higgs branch can be odd-dimensional, for example, the resolved conifold in the case of SQED with two flavours (see equation \rf{MH12}). In this section, we show that the condition of having a polarization can be weakened to having a partial polarization (defined below). We basically follow the line of arguments in \cite{okounkov2021inductive}, the main point is that we can always construct an attractive line bundle from a partial polarization.

 Throughout this section, we denote by $\mathsf X$ a smooth quasi-projective complex variety with a~torus~$\mathsf T$ action, we fix a nontrivial group homomorphism $\mathsf T\to \mathbb C^{\times }_{\hbar}$ and a subtorus $\mathsf A\subset {\ker\bigl(\mathsf T\to \mathbb C^{\times }_{\hbar}\bigr)}$.

\begin{dfn}\label{dfn: partial polarization}
A \emph{partial polarization} on $\mathsf X$ is the following data:
\begin{itemize}\itemsep=0pt
 \item a decomposition of the tangent bundle
 \[
 T_{\mathsf X}=\mathrm{Pol}^+_{\mathsf X}+\hbar^{-1}\bigl(\mathrm{Pol}^+_{\mathsf X}\bigr)^\vee+\mathrm{Pol}^-_{\mathsf X}+\hbar(\mathrm{Pol}^-_{\mathsf X})^\vee+\mscr E \in K_{\mathsf T}(\mathsf X),
\]
for some $\mathsf T$-equivariant K-theory classes $\mathrm{Pol}_{\mathsf X}:=\mathrm{Pol}^+_{\mathsf X}+\mathrm{Pol}^-_{\mathsf X}$ and $\mscr E$,
\end{itemize}
such that
\begin{itemize}\itemsep=0pt
 \item[(1)] $\mscr E=\mscr E^\vee$ in $K_{\mathsf T}(\mathsf X)$, and that
 \item[(2)] the elliptic Thom line bundle $\Theta(\mscr E)$ admits a square root.\footnote{For a review of elliptic cohomology and Theta bundles, see Appendix \ref{subsec: Equivariant Elliptic Cohomology}.}
\end{itemize}
We call such $\mathsf X$ a partially-polarized variety, with partial polarization $\mathrm{Pol}_{\mathsf X}$. And we define the opposite partial polarization to be
$
\mathrm{Pol}^{\mathrm{op}}_{\mathsf X}=\mathrm{Pol}^{\mathrm{op},+}_{\mathsf X}+\mathrm{Pol}^{\mathrm{op},-}_{\mathsf X}=\hbar^{-1}\bigl(\mathrm{Pol}^+_{\mathsf X}\bigr)^\vee+\hbar(\mathrm{Pol}^-_{\mathsf X})^\vee$.
\end{dfn}

\begin{rmk}
In particular, the existence of partial polarization implies that $T_{\mathsf X}$ equals to $T^\vee_{\mathsf X}$ in $K_{\mathsf A}(\mathsf X)$.
\end{rmk}
We give a criterion when a variety $\mathsf X$ has a partial polarization in the following.
\begin{prop}\label{prop: criterion for partial polarization}
Let $\mathsf{P}^{\pm}_{\mathsf X}\in K_{\mathsf T}(\mathsf X)$, and moreover assume that there exist
\begin{itemize}\itemsep=0pt
 \item[$(1)$] a finite set of pairs $(G_i,n_i)_{i\in I}$, where $n_i\in\mathbb Z_{\neq 0}$ and $G_i$ is a reductive group whose derived subgroup $[G_i,G_i]$ is simply connected and every simple constituent is of type $\mathrm{A}$, $\mathrm{C}$, $\mathrm{D}$, $\mathrm{E}_6$, or $\mathrm{G}_2$,
 \item[$(2)$] for every $i\in I$, a $\mathsf T$-equivariant principal $G_i$ bundle $\mathcal P_i$ on $\mathsf X$,
\end{itemize}
such that the tangent bundle $T_{\mathsf X}$ decomposes in the $\mathsf T$-equivariant K-theory into
\[
 T_{\mathsf X}=\mathsf{P}^+_{\mathsf X}+\hbar^{-1}\bigl(\mathsf{P}^+_{\mathsf X}\bigr)^\vee+\mathsf{P}^-_{\mathsf X}+\hbar(\mathsf{P}^-_{\mathsf X})^\vee+\sum_{i\in I}n_i\mathrm{adj}(\mcal P_i) \in K_{\mathsf T}(\mathsf X),
\]
where $\mathrm{adj}(\mcal P_i):=\mcal P_i\times^{G_i}\mathfrak g_i$ is the adjoint bundle associated to $\mcal P_i$. Then $\mathsf{P}_{\mathsf X}=\mathsf{P}^+_{\mathsf X}+\mathsf{P}^-_{\mathsf X}$ is a~partial polarization on $\mathsf X$.
\end{prop}

The above proposition is the direct consequence of the following lemma.

\begin{lem}\label{lem: square root}
Let $G$ be a reductive group whose derived subgroup is simply connected and every simple constituent is of type $\mathrm{A}$, $\mathrm{C}$, $\mathrm{D}$, $\mathrm{E}_6$, or $\mathrm{G}_2$. Let $\mcal P$ be a $\mathsf T$-equivariant principal $G$-bundle on~$\mathsf X$, then the elliptic Thom line bundle $\Theta(\mathrm{adj}(\mcal P))$ has a square root in $\mathrm{Pic}(\mathrm{Ell}_{\mathsf T}(\mathsf X))$.
\end{lem}

\begin{proof}
It suffices to prove the existence of a square root in the universal case, i.e., $\Theta(\mathfrak g)$ on~${\mbb{E}\otimes_{\mathbb Z}\mathbb X_*/W}$ has a square root, where $\mathbb X_*$ is the cocharacter lattice of $G$ and $W$ is the Weyl group and for any $G$-representation $V$, $\Theta(V)$ is the pullback of the $\Theta$-bundle along the induced map~${\mbb{E}\otimes_{\mathbb Z}\mathbb X_*/W\to S^{\operatorname{rk} V}\mbb{E}}$ between the moduli of semistable bundles on $\mbb{E}$.

 We first assume that $G$ is simple, then the cocharacter lattice $\mathbb X_*$ agrees with the coroot lattice $Q^\vee$. According to the result of Looijenga \cite{looijenga1976root}, $\mbb{E}\otimes_{\mathbb Z}Q^{\vee}/W$ is isomorphic to the weighted projective space $\mathbb P(1,g_1,\dots,g_r)$ such that
\begin{align*}
 \theta^\vee=\sum_{i=1}^r g_i\alpha_i^\vee,
\end{align*}
is the decomposition of the dual of the highest root in terms of simple coroots. It is well known that the Picard group of $\mathbb P(1,g_1,\dots,g_r)$ is freely generated by $\mscr O(m)$, where $m=\mathrm{lcm}(g_1,\dots,g_r)$ \cite{beltrametti1985itronduction}, and Looijenga also shows that $\Theta(\mathfrak g)$ is isomorphic to $\mscr O\bigl(2h^\vee\bigr)$ where $h^\vee=1+\sum_{i=1}^r g_i$ is the dual Coxeter number \cite{looijenga1976root}, therefore $\Theta(\mathfrak g)$ has a square root if and only if $h^\vee$ is divisible by $m$, which is true if and only if $G$ is of type $\mathrm{A}$, $\mathrm{C}$, $\mathrm{D}$, $\mathrm{E}_6$ or $\mathrm{G}_2$.

 In the general case, $G$ is the central quotient of $\widetilde G:=H\times G'$ by a finite abelian group $\Gamma$, where $H$ is a torus and $G'$ is the derived subgroup of $G$. Then $\mbb{E}\otimes_{\mathbb Z}\mathbb X_*/W$ is the quotient of~${\mbb{E}\otimes_{\mathbb Z}\widetilde{\mathbb X}_*/W}$ by $H^1(\mbb{E},\Gamma)$, where $\widetilde{\mathbb X}_*$ is the cocharacter lattice of $\widetilde G$, so there is a fiber sequence
\begin{gather}\label{eqn: fiber seq of Pic 1}
H^1(\mbb{E},\Gamma)^\vee\longrightarrow \mathrm{Pic}(\mbb{E}\otimes_{\mathbb Z}\mathbb X_*/W)\longrightarrow \mathrm{Pic}(\mbb{E}\otimes_{\mathbb Z}\widetilde{\mathbb X}_*/W)^{H^1(\mbb{E},\Gamma)},
\end{gather}
where $H^1(\mbb{E},\Gamma)^\vee$ is the Pontryagin dual of $H^1(\mbb{E},\Gamma)$. Note that we can identify $\mbb{E}\otimes_{\mathbb Z}\widetilde{\mathbb X}_*/W=\mscr E_H\times \bigl(\mbb{E}\otimes_{\mathbb Z}Q^{\vee}/W\bigr)$, where $Q^{\vee}$ is the coroot lattice of $G'$. We also note that there is an \'etale locally trivial fibration $\mbb{E}\otimes_{\mathbb Z}\mathbb X_*/W\to \mscr E_{H/\Gamma}$ with fibers isomorphic to $\mbb{E}\otimes_{\mathbb Z}Q^{\vee}/W$, such that the trivial fibration $\mscr E_H\times \bigl(\mbb{E}\otimes_{\mathbb Z}Q^{\vee}/W\bigr)$ is the pullback along the quotient map $\mscr E_{H}\to \mscr E_{H/\Gamma}$. The fiber sequence \eqref{eqn: fiber seq of Pic 1} is compatible with its counterpart coming from the quotient map between the bases $\mscr E_H\to \mscr E_{H/\Gamma}$. In other words, we have the commutative diagram of fiber sequences
\begin{equation*}
\begin{tikzcd}
H^1(\mbb{E},\Gamma)^\vee \arrow[r]\arrow[d,equal] & \mathrm{Pic}(\mscr E_{H/\Gamma}) \arrow[r]\arrow[d] & \mathrm{Pic}(\mscr E_H)^{H^1(\mbb{E},\Gamma)} \arrow[d]\\
H^1(\mbb{E},\Gamma)^\vee \arrow[r] & \mathrm{Pic}(\mbb{E}\otimes_{\mathbb Z}\mathbb X_*/W) \arrow[r] & \mathrm{Pic}\bigl(\mbb{E}\otimes_{\mathbb Z}\widetilde{\mathbb X}_*/W\bigr)^{H^1(\mbb{E},\Gamma)}.
\end{tikzcd}
\end{equation*}
Now the image of $\Theta(\mathfrak g)$ in \smash{$\mathrm{Pic}\bigl(\mbb{E}\otimes_{\mathbb Z}\widetilde{\mathbb X}_*/W\bigr)^{H^1(\mbb{E},\Gamma)}$} is the $\mscr O_{\mscr E_H}\boxtimes \Theta(\mathfrak g')$, which has a square root~\smash{$\mscr O_{\mscr E_H}\boxtimes \Theta(\mathfrak g')^{\otimes\frac{1}{2}}$} by the previous step. The square root is invariant under the action $H^1(\mbb{E},\Gamma)$ since $H^1(\mbb{E},\Gamma)$ acts trivially on $\mathrm{Pic}\bigl(\mbb{E}\otimes_{\mathbb Z}Q^{\vee}/W\bigr)$. We lift such square root to $\mscr L\in\mathrm{Pic}(\mbb{E}\otimes_{\mathbb Z}\mathbb X_*/W)$ so that $\Theta(\mathfrak g)\cong \mscr L^{\otimes 2}\otimes \mscr K$ for some $\mscr K\in H^1(\mbb{E},\Gamma)^\vee$. Since $H^1(\mbb{E},\Gamma)^\vee\in \mathrm{Pic}^0(\mscr E_{H/\Gamma})$, so we can find \smash{$\mscr K_1\in \mathrm{Pic}^0(\mscr E_{H/\Gamma})$} such that \smash{$\mscr K\cong \mscr K_1^{\otimes 2}$}, hence we have \smash{$\Theta(\mathfrak g)\cong (\mscr L\otimes \mscr K_1)^{\otimes 2}$}.
\end{proof}

\begin{rmk}
When $G=\mathrm{GL}_n$, there is an alternative proof of the above lemma. In this case~${\mathrm{adj}(\mcal P)\cong\mathrm{End}(\mcal V)}$ where $\mcal V$ is a $\mathsf T$-equivariant vector bundle. Notice that $\Theta\bigl(\mcal V\mcal V^\vee\bigr)$ is isomorphic to pullback $\Delta^*\Theta\bigl(\mcal V_1\mcal V_2^\vee\bigr)$ along the diagonal morphism $\mathrm{Ell}_{\mathsf T}(\mathsf X)\hookrightarrow \mathrm{Ell}_{\mathsf T}(\mathsf X)\times \mathrm{Ell}_{\mathsf T}(\mathsf X)$, where~${\mcal V_i=\mathrm{pr}_i^*\mcal V}$ are pullback of $\mcal V$ along projections to the first and the second component. Moreover, $\mathrm{(12)}^*\Theta\bigl(\mcal V_1\mcal V_2^\vee\bigr)\cong \Theta\bigl(\mcal V_1^\vee\mcal V_2\bigr)$ which is isomorphic to $\Theta\bigl(\mcal V_1\mcal V_2^\vee\bigr)$, where (12) is the permutation of factors, thus we can apply \cite[Lemma 6.1]{nekrasov2014membranes} to conclude that $\Theta\bigl(\mcal V\mcal V^\vee\bigr)$ admits a~square root.
\end{rmk}

\begin{ex}\label{ex: pol for quotient of Hamiltonian reduction}
Suppose that we are in the situation of Section~\ref{subsubsec: Example: Quotient of Hamiltonian Reduction}, namely, $G=G_{\mathrm{ev}}\times G_{\mathrm{odd}}$ acts on a representation $\mcal R$ and $\mu_{\mathrm{ev}}\colon\mcal R\oplus\mcal R^\vee\to \mathfrak g_{\mathrm{ev}}$ is the moment map for $G_{\mathrm{ev}}$ action. Assume moreover that $G_{\mathrm{odd}}$ is the reductive group whose derived subgroup is simply connected and every simple constituent is of type $\mathrm{A}$, $\mathrm{C}$, $\mathrm{D}$, $\mathrm{E}_6$, or $\mathrm{G}_2$. Then $\mathsf X:=\mu^{-1}_{\mathrm{ev}}(0)^{\zeta-ss}/\!\!/G$ is a partially-polarized variety with
$\mathrm{Pol}_{\mathsf X}=\mathrm{Pol}^+_{\mathsf X}=\mscr R-\mathrm{adj}(\mcal P_b)$,
where $\mscr R$ is the bundle associated to the representation $\mcal R$ and $\mcal P_b\times\mcal P_f$ is the principal $G_{\mathrm{ev}}\times G_{\mathrm{odd}}=G$ bundle $\mu^{-1}_{\mathrm{ev}}(0)^{\zeta-ss}\to \mathsf X$. We have the decomposition
$T_{\mathsf X}=\mathrm{Pol}_{\mathsf X}+\hbar^{-1}\mathrm{Pol}_{\mathsf X}^\vee-\mathrm{adj}(\mcal P_f)\in K_{\mathsf T}(\mathsf X)$.
\end{ex}

\begin{ex}\label{ex: partial polarization for the quivers}
In the quiver case, adjoint bundles can be written in terms of tautological bundles. Let $Q$ be a quiver with decomposition of nodes $Q_0=Q^{\mathrm{ev}}_0\sqcup Q^{\mathrm{odd}}_0$ into even and odd parts. In this case, $\mathsf T=\mathsf A\times \mathbb C^{\times}_{\hbar}$, where $\mathsf A=T_W$ is the maximal torus of the flavour group. Our choice of $\mathbb C^{\times}_{\hbar}$-weight convention is as follows. We split the set of arrows into two parts~${Q_1=Q_1^+\sqcup Q_1^-}$, such that
\begin{itemize}\itemsep=0pt
 \item[$(\star)$] All arrows connected to an even node are from the same group, i.e., either all of them are in $Q_1^+$ or all of them are in $Q_1^-$. Denote by $Q_0^{\mathrm{ev},\pm}$ the subset of $Q^{\mathrm{ev}}_0$ consisting of even nodes such that arrows connected to them are all in $Q_1^{\pm}$.
\end{itemize}
Then we have $Q^{\mathrm{ev}}_0=Q_0^{\mathrm{ev},+}\sqcup Q_0^{\mathrm{ev},-}$ by our assumption.
We define the action of $\mathbb C^{\times}_{\hbar}$ on the doubled quiver $\overline{Q}$ by assigning the $\mathbb C^{\times}_{\hbar}$-weights as follows:
\begin{center}\renewcommand{\arraystretch}{1.2}
\begin{tabular}{|c | c | c | c | c | c|}
 \hline
 & $x_a$ & $\widetilde{x}_a,a\in Q^{\pm}_1$ & $\alpha_i$ & $\widetilde{\alpha}_i,i\in Q^{\mathrm{odd}}_0$ & $\widetilde{\alpha}_i,i\in Q^{\mathrm{ev},\pm}_0$ \\ [0.5ex]
 \hline
 $\mathbb C^{\times}_{\hbar}$-weight & $0$ & $\mp 1$ & $0$ & $-1$ & $\mp 1$ \\ [1ex]
 \hline
\end{tabular}
\end{center}
Let $\zeta$ be generic stability parameter, then $\mcal M^{\zeta}(\mathbf v,\mathbf w)$ is smooth. Denote by $\mcal V_i$ the tautological bundle on $\mcal M^{\zeta}(\mathbf v,\mathbf w)$ corresponding to the $i$-th gauge node, and denote by $\mcal W_i$ the tautological bundle on $\mcal M^{\zeta}(\mathbf v,\mathbf w)$ corresponding to the $i$-th framing, then $\mathrm{Pol}_{\mcal M}=\mathrm{Pol}^+_{\mcal M}+\mathrm{Pol}^-_{\mcal M}$ such that
\begin{gather}
\mathrm{Pol}^+_{\mcal M}:=\sum_{i\in Q_0\setminus Q^{\mathrm{ev},-}_0}\mcal W_i\mcal V_i^\vee+\sum_{a\in Q^+_1}\mcal V_{h(a)}\mcal V_{t(a)}^\vee-\sum_{j\in Q^{\mathrm{ev},+}_0}\mcal V_j\mcal V_j^\vee,\nonumber\\
\mathrm{Pol}^-_{\mcal M}:=\sum_{i\in Q^{\mathrm{ev},-}_0}\mcal W_i\mcal V_i^\vee+\sum_{a\in Q^-_1}\mcal V_{h(a)}\mcal V_{t(a)}^\vee-\sum_{j\in Q^{\mathrm{ev},-}_0}\mcal V_j\mcal V_j^\vee\label{eqn: pol in the quiver case}
\end{gather}
is a partial polarization on $\mcal M(\mathbf v,\mathbf w)$. In fact,
\[
 T_{\mcal M^{\zeta}(\mathbf v,\mathbf w)}=\mathrm{Pol}^+_{\mcal M}+\hbar^{-1}\bigl(\mathrm{Pol}^+_{\mcal M}\bigr)^\vee+\mathrm{Pol}^-_{\mcal M}+\hbar(\mathrm{Pol}^-_{\mcal M})^\vee-\sum_{i\in Q^{\mathrm{odd}}_0}\mcal V_i\mcal V_i^\vee,
\]
in $K_{\mathsf T}\bigl(\mcal M^\zeta(\mathbf v,\mathbf w)\bigr)$.
\end{ex}

\begin{ex}\label{ex: bundles over flag variety}
Let $G$ be a semisimple and simply connected algebraic group, $P$ be a parabolic subgroup, and $U\subset P$ be the unipotent radical. Let $M$ be a connected normal subgroup of $P$ such that $M$ contains $U$, and $\mathfrak m$ be its Lie algebra. Denote by $K=P/M$, which is a quotient of the Levi $L\cong P/U$. Take $\mathsf X=G{\times}^P\mathfrak{m}$, and $\mathsf T=\mathsf A\times \mathbb C^{\times}_{\hbar}$ acts on $\mathsf X$ such that $\mathsf A$ is maximal torus of $G$ and $\mathbb C^{\times}_{\hbar}$ only acts on the fiber $\mathfrak{m}$ with weight $\hbar^{-1}$. Then tangent bundle of $\mathsf X$ decomposes into
\[
T_{\mathsf X}=T_{G/P}+\hbar^{-1}T^\vee_{G/P}+\mscr L-\mscr K \in K_{\mathsf T}(\mathsf X),
\]
where $T_{G/P}$ is the pullback of the tangent bundle of $G/P$ via the canonical projection $G{\times}^P\mathfrak{m}\to G/P$, and $\mscr L$ (resp.\ $\mscr K$) is the associated adjoint bundle of the principal $L$-bundle (resp.\ $K$-bundle) induced from the principal $P$-bundle $G{\times}\mathfrak{m}\to \mathsf X$. Assume that every simple constituent of $L$ is of type $\mathrm{A}$, $\mathrm{C}$, $\mathrm{D}$, $\mathrm{E}_6$, or $\mathrm{G}_2$, then the same holds for $K$. Therefore, $T_{G/P}$ is a partial polarization of $\mathsf X$ under this assumption.
\end{ex}

\begin{ex}\label{ex: bundles over flag variety_type A}
A special case of the previous example is when $G=\mathrm{SL}_n$ and $P$ is the subgroup fixing a flag $F_{\bullet}=(0=F_0\subset \mathcal F_1\subset\cdots\subset F_{k-1}\subset F_k=\mathbb C^n)$, we choose the Lie algebra $\mathfrak m_{s}$ according to a marking number $s\in \{\pm 1\}^k$ such that
\[
 \mathfrak m_{s}(F_i)\subset
 \begin{cases}
 F_{i-1}, & s_i=+1,\\
 F_{i}, & s_i=-1,
 \end{cases}
\]
then $\mathsf X=G{\times}^P\mathfrak{m}_s$ has a partial polarization for any choice of marking $s$. These varieties show up in the work of Rim\'anyi and Rozansky \cite[Section~2.3.3]{RimanyiRozansky202105} as the moduli space of certain quiver-like diagram construction.
\end{ex}

\subsection{Chambers and attracting sets}\label{subsec: Chambers and Attracting Sets}
Let $\mathrm{Cochar}(\mathsf A)$ be the cocharacter lattice of $\mathsf A$, and we denote
$\mathfrak a_{\mathbb R}:=\mathrm{Cochar}(\mathsf A)\otimes_{\mathbb Z}\mathbb R\subset \mathfrak a=\mathrm{Lie}(\mathsf A)$.
We define the \emph{roots} of the pair $(\mathsf X,\mathsf A)$ to be the set of weights $\{\alpha\}$ appearing in the normal bundle to $\mathsf X^{\mathsf A}$. A \emph{chamber} is defined to be a connected component of the complement of hyperplanes cut out by roots, i.e.,
\begin{align}
 \mathfrak a_{\mathbb R}\big\backslash \bigcup_{\alpha\in \mathrm{roots}}\alpha^{\perp}=\bigsqcup_i \mathfrak C_i. \label{chamberDefM}
\end{align}
For every $\sigma\in \mathfrak C_i\cap \mathrm{Cochar}(\mathsf A)$, we have $\mathsf X^{\sigma}=\mathsf X^{\mathsf A}$.

 Let $\mathfrak C$ be a chamber, then we say that a root $\alpha$ is attracting (resp.\ repelling) if $\alpha$ is positive (resp.\ negative) on $\mathfrak C$. For a connected component $\mathcal F\subset\mathsf X^{\mathsf A}$, we define the attracting part~\smash{$N^+_{\mathsf X/\mathcal F}$} of normal bundle to be the span of attracting root space in \smash{$N_{\mathsf X/\mathcal F}$}; similarly, the repelling part~\smash{$N^-_{\mathsf X/\mathcal F}$} is defined to be the span of repelling root space in \smash{$N_{\mathsf X/\mathcal F}$}. Then we define attracting submanifold~$\mathsf{Attr}_{\mathfrak C}(\mathcal F)$ to be the subset
\begin{align}
 \{x\in \mathsf X \mid \lim_{t\to 0}\sigma(t)\cdot x\in \mathcal F\}, \label{attrDef}
\end{align}
for some $\sigma\in \mathfrak C\cap \mathrm{Cochar}(\mathsf A)$. $\mathsf{Attr}_{\mathfrak C}(\mathcal F)$ does not depend on the choice of $\sigma$, in fact, it is the exponential of attracting part of the normal bundle \smash{$N^+_{\mathsf X/\mathcal F}$}. Define the union ${\mathsf{Attr}_{\mathfrak C}:=\coprod_{\mathcal F}\mathsf{Attr}_{\mathfrak C}(\mathcal F)}$, then $\mathsf{Attr}_{\mathfrak C}$ admits an immersion
$\mathsf{Attr}_{\mathfrak C}\hookrightarrow \mathsf X\times \mathsf X^{\mathsf A}$, $ x\mapsto (x,\lim_{t\to 0}\sigma(t)\cdot x)$.
In general, $\mathsf{Attr}_{\mathfrak C}$ is not closed in $\mathsf X\times \mathsf X^{\mathsf A}$, since $\lim_{t\to \infty}\sigma(t)\cdot x$, if exists, is not in $\mathsf{Attr}_{\mathfrak C}$. We define~\smash{$\mathsf{Attr}^{f}_{\mathfrak C}$} to be the set of pairs $(x,y)$ that belongs to a chain of closures of attracting $\mathsf A$-orbits. \smash{$\mathsf{Attr}^{f}_{\mathfrak C}$} is closed in~${\mathsf X\times \mathsf X^{\mathsf A}}$.

A partial order $\preceq$ on the set of connected components of $\mathsf X^{\mathsf A}$ is generated by letting
\[
 \mathcal F_j\subset \overline{\mathsf{Attr}_{\mathfrak C}(\mathcal F_i)}\ \Rightarrow \ \mathcal F_j\preceq \mathcal F_i.
\]
We define closed subvarieties $\mathsf{Attr}^{<}_{\mathfrak C}\subset \mathsf{Attr}^{\le}_{\mathfrak C}\subset \mathsf X\times \mathsf X^{\mathsf A}$ by
\[
 \mathsf{Attr}^{<}_{\mathfrak C}:=\bigcup_{\mathcal F_j\prec \mathcal F_i} \mathsf{Attr}_{\mathfrak C}(\mathcal F_j)\times \mathcal F_i,
 \qquad
 \mathsf{Attr}^{\le}_{\mathfrak C}:=\bigcup_{\mathcal F_j\preceq \mathcal F_i} \mathsf{Attr}_{\mathfrak C}(\mathcal F_j)\times \mathcal F_i.
\]
Note that $\mathsf{Attr}^{f}_{\mathfrak C}$ is a closed subvariety of \smash{$\mathsf{Attr}^{\le}_{\mathfrak C}$} and
\smash{$\mathsf{Attr}^{f}_{\mathfrak C}\cap \bigl(\mathsf{Attr}^{\le}_{\mathfrak C}\setminus\mathsf{Attr}^{<}_{\mathfrak C}\bigr)=\mathsf{Attr}_{\mathfrak C}$}.

\begin{ex}
In the case when $\mathsf X$ is constructed from quiver representation, there is a nice description of attracting sets, due to Andrei Negu\c{t} \cite{neguct2023quantum}. The setting in loc.\ cit.\ is for a~Nakajima quiver variety (i.e., $G_{\mathrm{odd}}=1$), nevertheless the arguments are applicable to the situation when $G_{\mathrm{odd}}\neq 1$. We record the result here. Consider $\mathsf X=\mcal M^{\zeta_+}(\mathbf v,\mathbf w)$ for a quiver $Q$ with dimension vectors $(\mathbf v,\mathbf w)$ and the stability condition $\zeta_+$, let $\mathsf A=\mathbb C^{\times}_a$ that acts on framing vector space as \smash{$W=a W^{(1)}\oplus W^{(2)}$}, so that the framing dimension vector decomposes accordingly~${mathbf w=\mathbf w^{(1)}+\mathbf w^{(2)}}$, then
\[
\mathsf X^{\mathsf A}=\bigsqcup_{\mathbf v^{(1)}+\mathbf v^{(2)}=\mathbf v}\mcal M^{\zeta}\bigl(\mathbf v^{(1)},\mathbf w^{(1)}\bigr)\times\mcal M^{\zeta}\bigl(\mathbf v^{(2)},\mathbf w^{(2)}\bigr).
\]
 We choose the chamber $\mathfrak C$ such that $a>0$, then the same argument in \cite[Lemma 3.10]{neguct2023quantum} can be applied, and $\mathsf{Attr}_{\mathfrak C}$ consists of quiver representations $(V,W)$, $\bigl(V^{(1)},W^{(1)}\bigr)$, $\bigl(V^{(2)},W^{(2)}\bigr)$ in $\mathsf X\times \mathsf X^{\mathsf A}$ that fits into exact sequences
\smash{$0\to V^{(1)}\hookrightarrow V\twoheadrightarrow V^{(2)}\to 0$},
that commutes with $\bigl(x_a,\widetilde{x}_a\bigr)$ and also commutes with $(\alpha_i,\widetilde{\alpha}_i)$ through the splitting exact sequence
\smash{$0\to W^{(1)}\hookrightarrow W\twoheadrightarrow W^{(2)}\to 0$}.
Moreover, the same argument in \cite[Lemma 3.11]{neguct2023quantum} can be applied (since the moment map equation is not essential in the proof), and \smash{$\mathsf{Attr}^f_{\mathfrak C}$} consists of quiver representations $(V,W)$, $\bigl(V^{(1)},W^{(1)}\bigr)$, $\bigl(V^{(2)},W^{(2)}\bigr)$ in $\mathsf X\times \mathsf X^{\mathsf A}$ that fits into the complex
\smash{$V^{(1)}\overset{f}{\longrightarrow} V\overset{g}{\longrightarrow} V^{(2)}$},
such that
\begin{itemize}\itemsep=0pt
 \item the composition $g\circ f=0$, and
 \item $f$ and $g$ commutes with $(x_a,\widetilde{x}_a)$ and also commutes with $(\alpha_i,\widetilde{\alpha}_i)$ through the splitting exact sequence
\smash{$ 0\to W^{(1)}\hookrightarrow W\twoheadrightarrow W^{(2)}\to 0$}.
\item there exist filtrations of quiver representations
\begin{align*}
 &V^{(1)}=E^0\twoheadrightarrow E^1\twoheadrightarrow\cdots \twoheadrightarrow E^{k-1}\twoheadrightarrow E^k=\mathrm{Im}(f),\\
 &\mathrm{Ker}(g)=F^k\twoheadrightarrow F^{k-1}\twoheadrightarrow\cdots \twoheadrightarrow F^{1}\twoheadrightarrow F^0= V^{(2)},
\end{align*}
such that the kernels of $E^l\twoheadrightarrow E^{l+1}$ and $F^{l+1}\twoheadrightarrow F^l$ are isomorphic quiver representations.
\end{itemize}
In particular,
\[
 \mcal M^{\zeta}\bigl(\mathbf v^{(1)},\mathbf w^{(1)}\bigr)\times\mcal M^{\zeta}\bigl(\mathbf v^{(2)},\mathbf w^{(2)}\bigr)\preceq \mcal M^{\zeta}\bigl(\mathbf v^{(1)}+\mathbf u,\mathbf w^{(1)}\bigr)\times\mcal M^{\zeta}\bigl(\mathbf v^{(2)}-\mathbf u,\mathbf w^{(2)}\bigr),
\]
for all $\mathbf u\in \mathbb N^{Q_0}$.
\end{ex}

The restriction of partial polarization $\mathrm{Pol}_{\mathsf X}$ to $\mathsf X^{\mathsf A}$ decomposes according to the chamber $\mathfrak C$ as
\[
 \smash{\mathrm{Pol}_{\mathsf X}|_{\mathsf X^{\mathsf A}}=\mathrm{Pol}_{\mathsf X}|_{\mathsf X^{\mathsf A},>0}+\mathrm{Pol}_{\mathsf X}|_{\mathsf X^{\mathsf A},\mathrm{fixed}}+\mathrm{Pol}_{\mathsf X}|_{\mathsf X^{\mathsf A},<0}.}
\]

\begin{dfn}
We define the index bundle
\smash{$\mathsf{ind}=\mathrm{Pol}^+_{\mathsf X}|_{\mathsf X^{\mathsf A},>0}-\mathrm{Pol}^-_{\mathsf X}|_{\mathsf X^{\mathsf A},>0}\in K_{\mathsf T}\bigl(\mathsf X^{\mathsf A}\bigr)$},
the alternating sum of attracting part of partial polarization.
\end{dfn}

\begin{lem}
\smash{$\mathrm{Pol}_{\mathsf X}|_{\mathsf X^{\mathsf A},\mathrm{fixed}}$} is a partial polarization on $\mathsf X^{\mathsf A}$.
\end{lem}

\begin{proof}
Let $T_{\mathsf X}=\mathrm{Pol}^+_{\mathsf X}+\hbar^{-1}\bigl(\mathrm{Pol}^+_{\mathsf X}\bigr)^\vee+\mathrm{Pol}^-_{\mathsf X}+\hbar(\mathrm{Pol}^-_{\mathsf X})^\vee+\mscr E$, then there exists a decomposition of $\mscr E$ into eigenspaces of $\mathsf A$-actions
\begin{align*}
 \mscr E|_{\mathsf X^{\mathsf A}}=\sum_{\lambda\in \mathrm{Cochar}(\mathsf A)} \mscr E^{\lambda}.
\end{align*}
By the duality $\mscr E\cong \mscr E^\vee$, we have $\mscr E^{-\lambda}\cong \bigl(\mscr E^{\lambda}\bigr)^\vee$, therefore we have
\begin{align*}
 \Theta\bigl(\mscr E^0\bigr)=\Theta(\mscr E|_{\mathsf X^{\mathsf A}})\otimes \bigotimes_{\lambda(\mathfrak C)>0}\Theta\bigl(\mscr E^{\lambda}\bigr)^{\otimes 2}.
\end{align*}
In particular, we see that $\Theta\bigl(\mscr E^0\bigr)$ admits a square root. Moreover, we have the decomposition of the tangent bundle of $\mathsf X^{\mathsf A}$ in $K_{\mathsf T}\bigl(\mathsf X^{\mathsf A}\bigr)$
\begin{align*}
 T_{\mathsf X^{\mathsf A}}=\mathrm{Pol}^+_{\mathsf X}|_{\mathsf X^{\mathsf A},\mathrm{fixed}}+\hbar^{-1}\left(\mathrm{Pol}^+_{\mathsf X}|_{\mathsf X^{\mathsf A},\mathrm{fixed}}\right)^\vee+\mathrm{Pol}^-_{\mathsf X}|_{\mathsf X^{\mathsf A},\mathrm{fixed}}+\hbar\left(\mathrm{Pol}^-_{\mathsf X}|_{\mathsf X^{\mathsf A},\mathrm{fixed}}\right)^\vee+\mscr E^0,
\end{align*}
which implies that $\mathrm{Pol}_{\mathsf X}|_{\mathsf X^{\mathsf A},\mathrm{fixed}}$ is a partial polarization on $\mathsf X^{\mathsf A}$.
\end{proof}

We denote by $\mathrm{Pol}_{\mathsf X^{\mathsf A}}$ the restriction of partial polarization $\mathrm{Pol}_{\mathsf X}|_{\mathsf X^{\mathsf A},\mathrm{fixed}}$ on $\mathsf X^{\mathsf A}$.

\subsection{Attractive line bundle}\label{subsec: Attractive Line Bundle}
This is the most crucial part of our construction of elliptic stable envelopes for partially-polarized varieties. An important ingredient that we need from \cite{okounkov2021inductive} is the following.
\begin{dfn}
A line bundle $\mscr L$ on $\mathrm{Ell}_{\mathsf T}(\mathsf X)$ is called attractive for a given chamber $\mathfrak C$ if
\smash{$ \deg_{\mathsf A}\mscr L=\deg_{\mathsf A}\Theta(N^-_{\mathsf X/\mathsf X^{\mathsf A}})$},
where $\deg_{\mathsf A}\mscr L$ is the degree of the restriction of $\mscr L$ to the fiber along the projection $\mathrm{Ell}_{\mathsf T}\bigl(\mathsf X^{\mathsf A}\bigr)\to \mathrm{Ell}_{\mathsf T/\mathsf A}\bigl(\mathsf X^{\mathsf A}\bigr)$.\footnote{The degree of a line bundle $\mscr L$ on an abelian variety $\mcal A$ is defined as the image of $\mscr L$ in the N\'eron--Severi group~${\mathrm{NS}(\mcal A)=\mathrm{Pic}(\mcal A)/\mathrm{Pic}^0(\mcal A)}$, see Appendix \ref{subsec: Equivariant Elliptic Cohomology} for details.} The $\deg_{\mathsf A}$ takes value in $H^0\bigl(\mathsf X^{\mathsf A},S^2\mathrm{Char}(\mathsf A)\bigr)$.
\end{dfn}

\begin{dfn}
From now on, we fix a square root for $\Theta(\mscr E)$ and define $\mscr S_{\mathsf X}:=\Theta(\mathrm{Pol}_{\mathsf X})\otimes \smash{\Theta(\mscr E)^{\otimes \frac{1}{2}}}$.
\end{dfn}

\begin{prop}\label{prop: attractive line bundle}
$\mscr S_{\mathsf X}$ is an attractive line bundle for every chamber $\mathfrak C$ of $\mathrm{Cochar}(\mathsf A)$.
\end{prop}

\begin{proof}
Since $S^2\mathrm{Char}(\mathsf A)$ is a torsion-free abelian group, it suffices to show that $\deg_{\mathsf A}\bigl(\mscr S_{\mathsf X}^{\otimes 2}\bigr)=\smash{\deg_{\mathsf A}\bigl(\Theta\bigl(N^-_{\mathsf X/\mathsf X^{\mathsf A}}+N^{-\vee}_{\mathsf X/\mathsf X^{\mathsf A}}\bigr)\bigr)}$. Since \smash{$N^{-\vee}_{\mathsf X/\mathsf X^{\mathsf A}}$} is $\mathsf A$-equivariantly isomorphic to \smash{$N^{+}_{\mathsf X/\mathsf X^{\mathsf A}}$}, we have
\[
\deg_{\mathsf A}\bigl(\Theta\bigl(N^-_{\mathsf X/\mathsf X^{\mathsf A}}+N^{-\vee}_{\mathsf X/\mathsf X^{\mathsf A}}\bigr)\bigr)=\deg_{\mathsf A}\bigl(\Theta(N_{\mathsf X/\mathsf X^{\mathsf A}})\bigr)=\deg_{\mathsf A}(\Theta(T_{\mathsf X})).
\]
 On the other hand, \smash{$\mscr S_{\mathsf X}^{\otimes 2}\cong \Theta\bigl(\mathrm{Pol}_{\mathsf X}+\mathrm{Pol}^\vee_{\mathsf X}+\mscr E\bigr)$}, and its $\mathsf A$-degree equals to that of
\[
\Theta\bigl(\mathrm{Pol}^+_{\mathsf X}+\hbar^{-1}\bigl(\mathrm{Pol}^+_{\mathsf X}\bigr)^\vee+\mathrm{Pol}^-_{\mathsf X}+\hbar(\mathrm{Pol}^-_{\mathsf X})^\vee+\mscr E\bigr)=\Theta(T_{\mathsf X}),
\]
this proves the lemma.
\end{proof}

\begin{rmk}
For a line bundle $\mscr L$ on $\mathrm{Ell}_{\mathsf T}(\mathsf X)$, define $\mscr L^{\triangledown}:=\mscr L^{\vee}\otimes\Theta(T_{\mathsf X})$. Then there is an isomorphism \smash{$\mscr S_{\mathsf X}^{\triangledown}\cong\Theta\bigl(\mathrm{Pol}^{\mathrm{op}}_{\mathsf X}\bigr)\otimes \Theta(\mscr E)^{\otimes \frac{1}{2}}$}, therefore $\mscr S_{\mathsf X}^{\triangledown}$ is the attractive line bundle associated to the opposite partial polarization.
\end{rmk}

The attractive line bundle associated to a polarization naturally restricts to $\mathsf A$-fixed points. In fact we can define \smash{$\mscr S_{\mathsf X,\mathsf A}:=i^*\mscr S_{\mathsf X}\otimes\Theta(-N^-_{\mathsf X/\mathsf X^{\mathsf A}})$} where $i\colon \mathrm{Ell}_{\mathsf T}\bigl(\mathsf X^{\mathsf A}\bigr)\to \mathrm{Ell}_{\mathsf T}(\mathsf X)$ is the map induced by the inclusion $\mathsf X^{\mathsf A}\hookrightarrow \mathsf X$. On the other hand, we also have the line bundle $\mscr S_{\mathsf X^{\mathsf A}}$ defined using the restriction of the partial polarization $\mathrm{Pol}_{\mathsf X^{\mathsf A}}\in \mathrm{Pic}\bigl(\mathrm{Ell}_{\mathsf T}\bigl(\mathsf X^{\mathsf A}\bigr)\bigr)$, i.e.,
\smash{$
\mscr S_{\mathsf X^{\mathsf A}}=\Theta(\mathrm{Pol}_{\mathsf X^{\mathsf A}})\otimes \Theta\bigl(\mscr E^0\bigr)^{\otimes \frac{1}{2}}$}.
$\mscr S_{\mathsf X,\mathsf A}$ is not isomorphic to $\mscr S_{\mathsf X^{\mathsf A}}$ in general, and they are related as follows
\begin{align}\label{eqn: dynamical shifts in attractive line bundle}
 \mscr S_{\mathsf X,\mathsf A}\otimes \mscr U\cong \mscr S_{\mathsf X^{\mathsf A}}\otimes \Theta(\hbar)^{-\operatorname{rk} \mathsf{ind}}\otimes \tau(-\hbar\det\mathsf{ind})^*\mscr U.
\end{align}
$\mscr U$ is the universal line bundle and $\tau(-\hbar\det\mathsf{ind})$ is the translation \cite{AganagicOkounkov201604}, whose definitions are recalled below.

\subsection{Universal line bundle, K\"ahler torus and resonant locus}
Let $\mcal A$ be an abelian variety and $\mcal A^\vee$ its dual abelian variety, then there is a universal line bundle~$\mscr U_{\text{Poincar\'e}}$ on $\mcal A^\vee\times \mcal A$. The sections of $\mscr U_{\text{Poincar\'e}}$ on the universal cover of $\mcal A^\vee_s\times \mcal A_z$ has the same factor of automorphy as the function
\smash{$\frac{\vartheta(sz)}{\vartheta(s)\vartheta(z)}$}.

\textbf{Assumption.} We assume that $\mathrm{Pic}(\mathsf X)$ is finitely generated as an abelian group, and we fix a~set of generators $\{\mscr L^{\circ}_i\}_{i=1}^r$, which induces a group homomorphism
$\mathsf K=(\mathbb C^{\times})^r\twoheadrightarrow \mathrm{Pic}(\mathsf X)\otimes_{\mathbb Z}\mathbb C^{\times}$.
$\mathsf K$ is called the K\"ahler torus.

 We choose an equivariant lift $\mscr L_i\in \mathrm{Pic}_{\mathsf T}(\mathsf X)$ for each $\mscr L^{\circ}_i$, then the elliptic Chern class gives a map $c_i\colon\mathrm{Ell}_{\mathsf T}(\mathsf X)\to \mathbb E$. We consider $c_i\times 1\colon \mathrm{Ell}_{\mathsf T}(\mathsf X)\times \mscr E_{z_i}\to \mathbb E\times \mscr E_{z_i}$ and define
\[
 \mscr U(\mscr L_i,z_i):=(c_i\times 1)^*\mscr U_{\text{Poincar\'e}},
\]
where we identify $\mathbb E\cong \mscr E_{z_i}^\vee$. Moreover, we define the extended equivariant elliptic cohomology~${\mathsf E_{\mathsf T}(\mathsf X):=\mathrm{Ell}_{\mathsf T}(\mathsf X)\times \mscr E_{\mathsf{K}}}$
which is a scheme over
$
 \mscr B_{\mathsf T,\mathsf X}:=\mscr E_{\mathsf T}\times \mscr E_{\mathsf{K}}$.
Using the identification~${\mscr E_{\mathsf K}\cong \prod_{i=1}^r\mscr E_{z_i}}$, the line bundle $\mscr U(\mscr L_i,z_i)$ is naturally defined on the extended equivariant elliptic cohomology $\mathsf E_{\mathsf T}(\mathsf X)$, and we set
$ \mscr U:=\bigotimes_{i=1}^r \mscr U(\mscr L_i,z_i)$.

 Let $\mathcal F_j\preceq \mathcal F_i$ be a pair of connected components of $\mathsf X^{\mathsf A}$ and consider the commutative diagram
\[
\begin{tikzcd}
\mathsf E_{\mathsf T}(\mathcal F_j\times \mathcal F_i)\arrow[r,"\phi"]\arrow[dr,"\psi"] \arrow[d,"p"]& \mathsf E_{\mathsf T/\mathsf A}(\mathcal F_j\times \mathcal F_i) \arrow[d,"p"]\\
\mscr B_{\mathsf T,\mathsf X}\arrow [r,"\phi"] & \mscr B_{\mathsf T/\mathsf A,\mathsf X},
\end{tikzcd}
\]
where $\mscr B_{\mathsf T/\mathsf A,\mathsf X}:=\mscr E_{\mathsf T/\mathsf A}\times \mscr E_{\mathsf{K}}$.
\begin{dfn}
Let $\mscr S_{\mathsf X,\mathcal F}$ be the restriction of $\mscr S_{\mathsf X,\mathsf A}$ to a connected component $\mathcal F$ of $\mathsf X^{\mathsf A}$. The resonant locus $\mathbf\Delta$ is defined as the union of
\[
 p\bigl(\mathrm{supp}R\phi_*\bigl(\mscr S_{\mathsf X,\mathcal F_j}\otimes\mscr U\boxtimes (\mscr S_{\mathsf X,\mathcal F_i}\otimes\mscr U)^\vee\bigr)\bigr)\subset \mscr B_{\mathsf T/\mathsf A,\mathsf X},
\]
over all pairs $\mathcal F_j\preceq \mathcal F_i$ of connected components of $\mathsf X^{\mathsf A}$. We also use the same notation $\mathbf \Delta$ for its preimages.
\end{dfn}
Verbatim to that of \cite[Proposition~2.6]{okounkov2021inductive}, we have the following.
\begin{lem}
The complement of $\mathbf \Delta$ in $\mathsf E_{\mathsf T}(\mathsf X)$ is open and dense.
\end{lem}
For a homomorphism between abelian varieties $g\colon \mscr E_{\mathsf T}\to \mscr E_{\mathsf{K}}$, we associate an automorphism~${\tau(g)\colon\mscr B_{\mathsf T,\mathsf X}\cong \mscr B_{\mathsf T,\mathsf X}}$ by
$
 (t,z)\mapsto (t,z+g (t))$,
where $t$ is the coordinate on $\mscr E_{\mathsf T}$. We also use the same notation $\tau(g)$ for its pullback to $\mathsf E_{\mathsf T}(\mathsf X)$. Such transformation does not affect $\mathrm{Ell}_{\mathsf T}(\mathsf X)$, only the K\"ahler parameters are shifted. An example is as follows: for a pair
\begin{align*}
 \mu\in \mathrm{Char}(\mathsf T)=\mathrm{Hom}(\mscr E_{\mathsf T},\mathbb E),\qquad \lambda\in \mathrm{Cochar}(\mathsf K)=\mathrm{Hom}(\mathbb E,\mscr E_{\mathsf{K}}),
\end{align*}
$\lambda\mu$ is an abelian variety homomorphism, so we have a shift map $\tau(\lambda\mu)$.

 The line bundle $\tau(g)^*\mscr U\otimes\mscr U^{-1}$ depends linearly on $g$, i.e.,
\begin{align*}
 \frac{\tau(g_1)^*\mscr U}{\mscr U}\otimes\frac{\tau(g_2)^*\mscr U}{\mscr U}\cong \frac{\tau(g_1\cdot g_2)^*\mscr U}{\mscr U}.
\end{align*}
In the case when $\mu\in \mathrm{Char}(\mathsf T)$ and $\lambda\in \mathrm{Cochar}(\mathsf K)$, there is a meromorphic section
\smash{$\frac{\vartheta(\lambda\cdot\mu)}{\vartheta(\lambda)\vartheta(\mu)}$},
of~the line bundle $\tau(\lambda\mu)^*\mscr U\otimes\mscr U^{-1}$, here $\lambda\cdot\mu$ is the coordinate multiplication \cite[Lemma~2.4]{AganagicOkounkov201604}. We will write $\mscr U(\lambda,\mu)$ for the line bundle $\tau(\lambda\mu)^*\mscr U\otimes\mscr U^{-1}$.

\subsection{Elliptic stable envelope} \label{sec:ellStab}
Now we have all the ingredients, and the main result of this section: the existence and uniqueness of the elliptic stable envelope for partially-polarized varieties, directly follows from the general construction in \cite{okounkov2021inductive}.
\begin{thr}\label{thm: stable envelope}
If $\mathsf X$ is partially polarized, then there exists a unique section
\[
 \mathbf{Stab}_{\mathfrak C,\mscr S_{\mathsf X}}\in \Gamma\bigl(\mathsf E_{\mathsf T}\bigl(\mathsf X\times \mathsf X^{\mathsf A}\bigr)\backslash \mathbf \Delta,\mscr S_{\mathsf X}\otimes\mscr U\boxtimes (\mscr S_{\mathsf X,\mathsf A}\otimes\mscr U)^{\triangledown}\bigr),
\]
such that it is supported on \smash{$\mathsf{Attr}^f_{\mathfrak C}$} and that its restriction to the complement of $\mathsf{Attr}^{<}_{\mathfrak C}$ is given by~$(-1)^{\operatorname{rk} \mathsf{ind}}[\mathsf{Attr}_{\mathfrak C}]$.\footnote{For the definition of supports, see Appendix \ref{subsec: Equivariant Elliptic Cohomology}.}
\end{thr}

\begin{proof}
This is a special case of \cite[Theorem~1]{okounkov2021inductive} applied to the attractive line bundle $\mscr S_{\mathsf X}$.
\end{proof}

Equivalently, $\mathbf{Stab}_{\mathfrak C,\mscr S_{\mathsf X}}$ gives rise to a map between sheaves $\mscr S_{\mathsf X,\mathsf A}\otimes\mscr U\to \mscr S_{\mathsf X}\otimes\mscr U$ over the non-resonant locus $\mscr B_{\mathsf T,\mathsf X}\backslash \mathbf \Delta$, which is uniquely characterized by
\begin{itemize}\itemsep=0pt
 \item[(1)] The support of $\mathbf{Stab}_{\mathfrak C,\mscr S_{\mathsf X}}$ is triangular with respect to $\mathfrak C$, i.e., \smash{$\mathbf{Stab}_{\mathfrak C,\mscr S_{\mathsf X}}([\mathcal F_i])|_{\mathcal F_j}=0$} for a~pair of connected components of $\mathsf X^{\mathsf A}$ such that $\mathcal F_j\npreceq \mathcal F_i$.
 \item[(2)] The diagonal \smash{$\mathbf{Stab}_{\mathfrak C,\mscr S_{\mathsf X}}([\mathcal F_i])|_{\mathcal F_i}=(-1)^{\operatorname{rk} \mathsf{ind}}\vartheta(N^-_{\mathsf X/\mathcal F_i})$}, where \smash{$\vartheta(N^-_{\mathsf X/\mathcal F_i})$} is the canonical section of the elliptic Thom line bundle \smash{$\Theta(N^-_{\mathsf X/\mathcal F_i})$}.
\end{itemize}
We will use the above point of view on the stable envelope in the following discussions.

 According to \eqref{eqn: dynamical shifts in attractive line bundle}, the elliptic stable envelope is a map between sheaves
\[
 \mathbf{Stab}_{\mathfrak C,\mscr S_{\mathsf X}}\colon\ \mscr S_{\mathsf X^{\mathsf A}}\otimes \Theta(\hbar)^{-\operatorname{rk} \mathsf{ind}}\otimes \tau(-\hbar\det\mathsf{ind})^*\mscr U\to \mscr S_{\mathsf X}\otimes\mscr U.
\]
This turns out to be more practical since it is usually easier to compute $\mathrm{Pol}_{\mathsf X^{\mathsf A}}$ than $\mscr S_{\mathsf X,\mathsf A}$.

{\bf Change of partial polarization.} If we change the partial polarization by
\[
 \mathrm{Pol}'_{\mathsf X}=\mathrm{Pol}_{\mathsf X}-(\mscr F_++\mscr F_-)+\hbar^{-1}\mscr F^\vee_+ +\hbar \mscr F^\vee_-,
\]
for some $\mscr F_+,\mscr F_-\in K_{\mathsf T}(\mathsf X)$, then the associated attractive line bundle changes by
\[
 \mscr S'_{\mathsf X}\otimes\mscr U\cong \mscr S_{\mathsf X}\otimes \tau(\hbar\det(\mscr F_+-\mscr F_-))^*\mscr U\otimes\Theta(\hbar)^{\operatorname{rk}\mscr F},
\]
and the elliptic stable envelope changes accordingly
\[
 \mathbf{Stab}_{\mathfrak C,\mscr S'_{\mathsf X}}(z)=\mathbf{Stab}_{\mathfrak C,\mscr S_{\mathsf X}}(z+\hbar\det(\mscr F_+-\mscr F_-))\cdot (-1)^{\operatorname{rk}\mscr F_{\mathrm{moving}}} ,
\]
where $z$ is the K\"ahler parameter (coordinate on $\mscr E_{\mathsf{K}}$), and $\mscr F_{\mathrm{moving}}$ is the $\mathsf A$-moving part of $\mscr F|_{\mathsf X^{\mathsf A}}$. In particular, if we choose the opposite partial polarization \smash{$\mathrm{Pol}^{\mathrm{op}}_{\mathsf X}=\hbar^{-1}\bigl(\mathrm{Pol}^+_{\mathsf X}\bigr)^\vee+\hbar(\mathrm{Pol}^-_{\mathsf X})^\vee$}, then
\[
 \mathbf{Stab}_{\mathfrak C,\mscr S^{\triangledown}_{\mathsf X}}(z)=\mathbf{Stab}_{\mathfrak C,\mscr S_{\mathsf X}}\bigl(z+\hbar\det\bigl(\mathrm{Pol}^+_{\mathsf X}-\mathrm{Pol}^-_{\mathsf X}\bigr)\bigr)\cdot(-1)^{\operatorname{rk} \mathrm{Pol}_{\mathsf X,\mathrm{moving}}}.
\]

\begin{ex}[abelian Gauge groups]\label{ex: abelian stable envelope}
Let us consider the Higgs branch of an Abelian $\mcal N=2$ gauge theory, i.e., Example~\ref{ex: abelian gauge theory}. In this case $\mathsf X:=\mu^{-1}_{\mathrm{ev}}(0)^{\zeta-ss}/\!\!/G$ is a $\mathsf T$-equivariant smooth deformation of $\mathsf X_0:=\mu^{-1}(0)^{\zeta-ss}/\!\!/G$ over the base $\mathfrak g_{\mathrm{ev}}^{\perp}$, where $\mathsf T=\mathsf A\times \mathbb C^{\times}_{\hbar}=\bigl(\bigl(\mathbb C^{\times}\bigr)^{\operatorname{rk}\mcal R}/G\bigr)\times \mathbb C^{\times}_{\hbar}$. $\mathbb C^{\times}_{\hbar}$ acts on $\mcal R\oplus \mcal R^\vee$ as $\mcal R\oplus \hbar^{-1}\mcal R^\vee$ and acts on $\mathfrak g$ as $\hbar^{-1}\mathfrak g$. Then every $\mathsf A$-fixed point of the hypertoric variety $\mathsf X_0$ extends to an \'etale cover of the base $\mathfrak g_{\mathrm{ev}}^{\perp}$, which must be trivial cover since every point in this cover is repelled by the $\mathbb C^{\times}_{\hbar}$ action from the fixed point. Let $\mathcal F_p$ be the component of $\mathsf X^{\mathsf A}$ that contains fixed point $p\in \mathsf X^{\mathsf A}_0$, then for any chamber $\mathfrak C$ of $\mathrm{Lie}(\mathsf A)_{\mathbb R}$, $\mathcal F_p\preceq \mathcal F_q$ in $\mathsf X$ if and only if $p\preceq q$ in $\mathsf X_0$. In other words, the chamber $\mathfrak C$ induces exactly the same partial order on~$\mathsf X^{\mathsf A}$ and $\mathsf X^{\mathsf A}_0$. Moreover, the inclusion $\mathsf X_0\hookrightarrow \mathsf X$ induces isomorphism on equivariant elliptic cohomology~${\mathrm{Ell}_{\mathsf T}(\mathsf X_0)\cong \mathrm{Ell}_{\mathsf T}(\mathsf X)}$. It is straightforward to see that \smash{$\mscr S_{\mathsf X}\cong\Theta(\mscr R)\cong \mscr S_{\mathsf X_0}$}, i.e., $\mathsf X$ and~$\mathsf X_0$ have the same attractive line bundle. The normal bundles are isomorphic \smash{$N^{\pm}_{\mathsf X/\mathcal F_p}|_p=N^{\pm}_{\mathsf X_0/p}|_p$}, so the index bundles are also isomorphic. To summarize, $\mathsf X$ and $\mathsf X_0$ have the same elliptic stable envelope
$\mathbf{Stab}_{\mathfrak C,\mscr S_{\mathsf X}}([\mathcal F_p])= \mathbf{Stab}_{\mathfrak C,\mscr S_{\mathsf X_0}}([p])$.
The explicit formula of the elliptic stable envelope for hypertoric varieties can be found in \cite{AganagicOkounkov201604}, see also \cite{Smirnov201804}.
\end{ex}

\begin{ex}[resolved determinantal varieties]\label{ex: stab for M(N,L)}
Consider the quiver $Q$ with one odd node, i.e., $Q_0=Q^{\mathrm{odd}}_0=\{1\}$ and $\mathbf v=N,\mathbf w=L$, and $\zeta=\zeta_-$, then $\mcal M^{\zeta}(N,L)$ is nonempty if and only if $N\le L$, and if nonempty then it is isomorphic to the total space of the vector bundle
$
\mathrm{Tot}\bigl(\mcal V^{\oplus L}\to\mathrm{Gr}(N,L)\bigr)$,
 where $\mcal V$ is the tautological bundle of rank $N$ on the Grassmannian. It is a resolution of the determinantal variety which parameterizes $L\times L$ matrices of rank $\le N$.

 In this case $\mathsf A=T_W=(\mathbb C^{\times })^L$ and $\mathsf T=\mathsf A\times \mathbb C^{\times }_{\hbar}$. $\mathsf A$ acts on both the base and the fiber of the bundle $\mathrm{Tot}\bigl(\mcal V^{\oplus L}\to\mathrm{Gr}(N,L)\bigr)$, and $\mathbb C^{\times }_{\hbar}$ scales the fiber with weight $\hbar^{-1}$. Let the equivariant parameters of the gauge group $\mathrm{GL}_N$ be $\mbs{s}=\{s_a\}_{a=1}^N$ and the equivariant parameters of $\mathsf A$ be $\mbs{x}=\{x_i\}_{i=1}^L$. Connected components of $\mathsf A$-fixed points are labelled by order-preserving embedding $p\colon\{1,\dots,N\}\hookrightarrow\{1,\dots,L\}$. Denote by $\{m_1,\dots,m_L\}$ the coordinates on $\mathrm{Lie}(\mathsf A)$, then we choose the chamber $\mathfrak C$ such that $m_1<\cdots<m_L$. We choose the partial polarization as in \eqref{eqn: pol in the quiver case}. Then we claim that the stable envelope is the following:\footnote{\label{foot:normStab}For the gauge-theoretic derivation of \eqref{eqn: stable envelope for M(N,L)}, see Section~\ref{sec:explicit construction of elliptic stable envelope}. Note that the gauge-theoretic computation uses the normalization of stable envelope such that \smash{$\mathbf{Stab}_{\mathfrak C,\mscr S_{\mathsf X}}([\mathcal F])|_{\mathcal F}=\vartheta(N^-_{\mathsf X/{\mathcal F}})$}.}
\begin{align}\label{eqn: stable envelope for M(N,L)}
 \mathbf{Stab}_{\mathfrak C,\mscr S_{\mathsf X}}([\mathcal F_p])=\tenofo{Sym}_{S_N}\left[\left(\prod_{a=1}^N\efunction_{p(a)}(s_a,\mbs{x},\hbar,z)\right)\cdot\left(\prod_{\substack{a>b}}\frac{1}{\vartheta\bigl(s_as_b^{-1}\bigr)}\right)\right],
\end{align}
where $\mathrm{Sym}_{S_N}$ means summation over permutations $\{s_a\}\mapsto \{s_{\sigma(a)}\}_{\sigma\in S_N}$, and $\efunction_m(s,\mbs{x},\hbar,z)$ is the following function:
\[
\efunction_m(s,\mbs{x},\hbar,z):=\frac{\vartheta\bigl(sx_{m}\hbar^{m-L}z\bigr)}{\vartheta\bigl(\hbar^{m-L}z\bigr)}\prod_{i<m}\vartheta(sx_i)\prod_{j>m}\vartheta(sx_j\hbar),
\]
where $z$ is the K\"ahler parameter corresponding to the ample line bundle $(\det \mcal V)^{-1}$.

 Let us show that \eqref{eqn: stable envelope for M(N,L)} is indeed the stable envelope. The first thing to check is that \eqref{eqn: stable envelope for M(N,L)} is a~map between line bundles $\mscr S_{\mathsf X^{\mathsf A}}\otimes \Theta(\hbar)^{-\operatorname{rk} \mathsf{ind}}\otimes \tau(-\hbar\det\mathsf{ind})^*\mscr U\to \mscr S_{\mathsf X}\otimes\mscr U$. Note that~$\mscr S_{\mathsf X^{\mathsf A}}$ is a~trivial line bundle, and $\mathsf{ind}$ has equivariant weights $x_i/x_{p(a)}$ such that $i>p(a)$. Every summand in the symmetrization \eqref{eqn: stable envelope for M(N,L)} can be rewritten as
\begin{gather*}
\Biggl(\prod_{i,a}\vartheta(s_ax_i)\Biggr)\Biggl(\prod_a\frac{\vartheta\bigl(s_ax_{p(a)}\hbar^{p(a)-L}z\bigr)}{\vartheta(s_ax_{p(a)})\vartheta\bigl(\hbar^{p(a)-L}z\bigr)}\Biggr)
\Biggl(\prod_{i>p(a)}\frac{\vartheta(s_ax_i\hbar)}{\vartheta(s_ax_i)\vartheta(\hbar)}\Biggr)\\
\qquad\times\Biggl(\prod_{\substack{a>b}}\frac{1}{\vartheta\bigl(s_as_b^{-1}\bigr)}\Biggr)\vartheta(\hbar)^{\#(i>p(a))},
\end{gather*}
which is a meromorphic section of the line bundle
\begin{align*}
 \Theta(\mathrm{Pol}_{\mathsf X})\otimes\mscr U\otimes(\tau(-\hbar\det\mathsf{ind})^*\mscr U|_{\mathcal F_p})^{-1}\otimes \Theta\bigl(\mcal V\mcal V^\vee\bigr)^{\otimes\frac{-1}{2}}\otimes \Theta(\hbar)^{\operatorname{rk} \mathsf{ind}},
\end{align*}
which is isomorphic to $\mscr S_{\mathsf X}\otimes\mscr U\otimes (\tau(-\hbar\det\mathsf{ind})^*\mscr U)^{-1}\otimes \Theta(\hbar)^{\operatorname{rk} \mathsf{ind}}$, this shows that \eqref{eqn: stable envelope for M(N,L)} lives in the correct line bundle.

 Next, we verify that \eqref{eqn: stable envelope for M(N,L)} is triangular with respect to $\mathfrak C$. In fact, if $\mathcal F_{p'}\npreceq \mathcal F_p$ is another connected component of fixed points, in other words, $p'\colon \{1,\dots,N\}\hookrightarrow\{1,\dots,L\}$ is an order-preserving embedding and there exists $1\le a\le N$ such that $p'(a)<p(a)$, then for every $\sigma\in S_N$ there exists $b$ such that $p'(\sigma(b))<p(b)$, so every summand in the symmetrization \eqref{eqn: stable envelope for M(N,L)} vanishes.

Finally, the restriction of \eqref{eqn: stable envelope for M(N,L)} to diagonal is
\begin{align*}
(-1)^{\#(i>p(a))}\prod_{a=1}^N\prod_{\substack{i<p(a)\\i\notin \mathrm{Im}(p)}}\vartheta\left(\frac{x_i}{x_{p(a)}}\right)\prod_{i>p(a)}\vartheta\left(\frac{x_{p(a)}}{x_i\hbar} \right),
\end{align*}
which equals to \smash{$(-1)^{\operatorname{rk} \mathsf{ind}}\vartheta(N^-_{\mathsf X/\mathcal F_p})$}. This verifies that \eqref{eqn: stable envelope for M(N,L)} is indeed the stable envelope for $\mcal M^{\zeta}(N,L)$ with respect to the chamber $\mfk C$.
\end{ex}

\subsection{Triangle lemma and duality}\label{subsec: Triangle Lemma and Duality}
Let $\mathfrak C'\subset\mathfrak C$ be a face, and let $\mathsf A'\subset \mathsf A$ be the subtorus associated to the span of $\mathfrak C'$ in $\mathrm{Lie}(\mathsf A)$.
\begin{prop}[triangle lemma]
We have
\smash{$\mathbf{Stab}_{\mathfrak C,\mscr S_{\mathsf X}}=\mathbf{Stab}_{\mathfrak C',\mscr S_{\mathsf X}}\circ \mathbf{Stab}_{\mathfrak C/\mathfrak C',\mscr S_{\mathsf X,\mathsf A'}}$}.
Moreover,
\begin{align}\label{eqn: stab on a wall}
 \mathbf{Stab}_{\mathfrak C/\mathfrak C',\mscr S_{\mathsf X,\mathsf A'}}(z)=\mathbf{Stab}_{\mathfrak C/\mathfrak C',\mscr S_{\mathsf X^{\mathsf A'}}}(z -\hbar \det\mathsf{ind}_{\mathsf X^{\mathsf A'}}),
\end{align}
where \smash{$\mathsf{ind}_{\mathsf X^{\mathsf A'}}=\mathrm{Pol}^+_{\mathsf X}|_{\mathsf X^{\mathsf A'},>0}-\mathrm{Pol}^-_{\mathsf X}|_{\mathsf X^{\mathsf A'},>0}$} is the alternating sum of the attracting part of partial polarization with respect to the chamber $\mathfrak C'$.
\end{prop}
\begin{proof}
This is verbatim to that of \cite[Proposition~3.3]{AganagicOkounkov201604}, replacing polarization with partial polarization.
\end{proof}

Assume that $\mathsf X^{\mathsf T}$ is proper, then we have the equivariant pushforward $f_{\circledast}\colon\Theta(T_{\mathsf X})\!\to\mscr O_{\mscr E_{\mathsf T},\mathrm{localized}}$ for the map $f\colon\mathsf X\to \mathrm{pt}$, and accordingly we define
\smash{$\mathbf{Stab}^*_{\mathfrak C,\mscr S_{\mathsf X}}(z):=\mathbf{Stab}_{-\mathfrak C,\mscr S^{\triangledown}_{\mathsf X}}(-z)^\vee$},
where \smash{$z\in \mscr E_{\mathrm{Pic}_{\mathsf T(\mathsf X)}}$} is the K\"ahler parameter and the dual means transpose with respect to the duality
\[
\mscr S_{\mathsf X}\otimes \mscr S^{\triangledown}_{\mathsf X}\overset{\mathrm{product}}{\longrightarrow}\Theta(T_{\mathsf X})\overset{f_{\circledast}}{\rightarrow}\mscr O_{\mscr E_{\mathsf T},\mathrm{localized}}.
\]

\begin{prop}[duality]
We have
\begin{align}\label{eqn: duality}
\mathbf{Stab}^*_{\mathfrak C,\mscr S_{\mathsf X}}\circ \mathbf{Stab}_{\mathfrak C,\mscr S_{\mathsf X}}=1.
\end{align}
\end{prop}
\begin{proof}
This is verbatim to that of \cite[Proposition~3.4]{AganagicOkounkov201604}, replacing polarization by partial polarization.
\end{proof}

Combine the duality \eqref{eqn: duality} with the formula of the stable envelope for opposite partial polarization, and we obtain the following corollary.

\begin{cor}[inverse stable envelope for the opposite chamber]
We have
\[
\mathbf{Stab}^{-1}_{-\mathfrak C,\mscr S_{\mathsf X}}(z)=(-1)^{\operatorname{rk} \mathrm{Pol}_{\mathsf X,\mathrm{moving}}} \mathbf{Stab}_{\mathfrak C,\mscr S_{\mathsf X}}\bigl(-z+\hbar\det\bigl(\mathrm{Pol}^+_{\mathsf X}-\mathrm{Pol}^-_{\mathsf X}\bigr)\bigr)^\vee.
\]
where $z$ is the K\"ahler parameter $($coordinate on $\mscr E_{\mathsf{K}})$.
\end{cor}

From now on, unless specified, we will drop the attractive line bundle in the notation and write $\mathbf{Stab}_{\mathfrak C}$ for $\mathbf{Stab}_{\mathfrak C,\mscr S_{\mathsf X}}$ and $\mathbf{Stab}_{\mathfrak C/\mathfrak C'}$ for \smash{$\mathbf{Stab}_{\mathfrak C/\mathfrak C',\mscr S_{\mathsf X^{\mathsf A'}}}$}, etc.

\subsection{Abelianization}\label{subsec: Abelianization}
A general strategy to compute elliptic stable envelope for GIT quotients by nonabelian gauge group is abelianization. This is described in detail for the symplectic reduction in \cite[Section~4.3]{AganagicOkounkov201604}, and in this subsection, we adapt their technique to the more general setting in Section~\ref{subsubsec: Example: Quotient of Hamiltonian Reduction}. In the following, we point out the necessary modification to \cite[Section~4.3]{AganagicOkounkov201604}.

 Let $S_b$ (resp.\ $S_f$) be a maximal torus of $G_{\mathrm{ev}}$ (resp.\ $G_{\mathrm{odd}}$) and let $B_b\supset S_b$ (resp.\ $B_f\supset S_f$) be a Borel subgroup with Lie algebra $\mathfrak b_b$ (resp.\ $\mathfrak b_f$). Denote by $S:=S_b\times S_f$ and $B:=B_b\times B_f$, then we have a diagram similar to that of \cite[equation~(59)]{AganagicOkounkov201604}
\begin{equation}\label{cd: abelianization}
\begin{tikzcd}
\mathsf{Fl}\arrow[r,"\mathsf j_+"]\arrow[d,"\pi"] & \mu^{-1}\bigl(\mathfrak b_b^{\perp}\bigr)^{\zeta-ss}/\!\!/ S \arrow[r,"\mathsf j_-"] &\mathsf X_S\\
\mathsf X.
\end{tikzcd}
\end{equation}
\smash{$\mathsf{Fl}=\mu^{-1}\bigl(\mathfrak g_{\mathrm{ev}}^{\perp}\bigr)^{\zeta-ss}/\!\!/ B\cong (\mu\times \mu_{\mathbb R})^{-1}\bigl(\mathfrak g_{\mathrm{ev}}^{\perp}\times \{\zeta\}\bigr)/U\cap S$}, where $\mu_{\mathbb R}\colon\mcal R\oplus\mcal R^\vee\to \mathrm{Lie}(U)^\vee$ is the real moment map with respect to the action of a maximal compact subgroup $U\subset G$ and a $U$-invariant Hermitian metric on $\mcal R\oplus\mcal R^\vee$, and $\mathsf X_S=\mu^{-1}\bigl(\mathfrak s_b^{\perp}\bigr)/\!\!/_{\zeta}S$ is the GIT quotient of $\mu^{-1}\bigl(\mathfrak s_b^{\perp}\bigr)$ by the torus $S$. Note that $\mathsf j_+$ is a non-holomorphic embedding and $\mathsf j_-$ is a holomorphic immersion. $\pi$ is a $G/B$-bundle, so every statement in \cite[Section~4.3.2]{AganagicOkounkov201604} still holds for $\pi$.

 We have decompositions
\[
 \pi^*\mscr G_{\mathrm{ev}}=\mscr N_b+\mscr N_b^\vee+\mathrm{Lie}(S_b)\otimes\mscr O_{\mathsf{Fl}},\qquad \pi^*\mscr G_{\mathrm{odd}}=\mscr N_f+\mscr N_f^\vee+\mathrm{Lie}(S_f)\otimes\mscr O_{\mathsf{Fl}},
\]
where $\mscr N_b$ (resp.\ $\mscr N_b$) is the bundle on $\mathsf X_S$ associated to the adjoint action of $S$ on $\mathfrak n_b=[\mathfrak b_b,\mathfrak b_b]$ (resp.\ $\mathfrak n_f=[\mathfrak b_f,\mathfrak b_f]$). The relative tangent bundle of $\pi$ is $T_{\pi}=\mscr N^\vee=\mscr N_b^\vee+\mscr N_f^\vee$, and the normal bundle of $\mathsf j_+$ is \smash{$N_{\mathsf j_+}=\mscr N+\hbar^{-1}\mscr N_b$}, and the normal bundle of $\mathsf j_-$ is \smash{$N_{\mathsf j_-}=\hbar^{-1}\mscr N_b^\vee$}.

 For simplicity, let us choose the partial polarizations
\[
 \mathrm{Pol}_{\mathsf X}=\mathrm{Pol}^+_{\mathsf X}=\mscr R-\mscr G_{\mathrm{ev}},
 \qquad
 \mathrm{Pol}_{\mathsf X_S}=\mathrm{Pol}^+_{\mathsf X_S}=\mscr R+\hbar^{-1}\mscr N_b^\vee-\mscr N_b-\mathrm{Lie}(S_b)\otimes\mscr O_{\mathsf{X}_S}.
\]
Note that the partial polarization on $\mathsf X_S$ differs from the choice in Example~\ref{ex: pol for quotient of Hamiltonian reduction} by $\hbar^{-1}\mscr N_b^\vee-\mscr N_b$, which affects the stable envelope by
\[
\mathbf{Stab}_{\mfk C}(z)\mapsto \mathbf{Stab}_{\mfk C}(z+\hbar\det\mscr N_b)\cdot(-1)^{\operatorname{rk}\mscr N_{b,\mathrm{moving}}},
\]
 thus \eqref{eqn: newton polytope} still holds for such modified choice of partial polarization.

 Now we have the modified version of \cite[equation~(63)]{AganagicOkounkov201604}
\[
 \mathsf j^*\bigl(\mathrm{Pol}_{\mathsf X_S}-\hbar^{-1}\mscr N_b^\vee+\mscr N_f^{\vee}\bigr)=\pi^*\mathrm{Pol}_{\mathsf X}+T_\pi,
\]
where $\mathsf j=\mathsf j_-\circ\mathsf j_+$. Also notice that \smash{$\Theta\bigl(\mscr N_f^{\vee}\bigr)\cong\Theta(\mscr N_f)\cong \pi^*\Theta(\mscr G_{\mathrm{odd}})^{\otimes\frac{1}{2}}$}. Therefore, we have the modified version of \cite[equation~(64)]{AganagicOkounkov201604}
\begin{equation}\label{eqn: chain of maps}
\begin{tikzcd}
\mscr S_{\mathsf X} & \pi^*\mscr S_X\otimes\Theta(T_\pi) \arrow[l,"\pi_{\circledast}" '] & \mathsf j^*_-\mscr S_{\mathsf X_S}\otimes \Theta\bigl(-\hbar^{-1}\mscr N_b^\vee\bigr)\arrow[l,"\mathsf j^*_+" '] \arrow[r,"\mathsf j_{-\circledast}"] & \mscr S_{\mathsf X_S}.
\end{tikzcd}
\end{equation}
The map $\mathsf j^*_+$ in \eqref{eqn: chain of maps} is surjective since any tautological class extends to $\mathsf X_S$.

 Next, we choose a connected component $\mathcal F\subset\mathsf X^{\mathsf A}$, and let $\mathcal F'$ be the unique connected component of $\pi^{-1}(\mathcal F)^{\mathsf A}$ such that normal weights of $\mathcal F'$ in $\pi^{-1}(\mathcal F)$ are repelling with respect to the chamber $\mathfrak C$, i.e., $\left(T_\pi|_{\mathcal F'}\right)_{>0}=0$. This implies that $\bigl(\mscr N^\vee|_{\mathcal F'}\bigr)_{>0}=0$ by the same argument with~\cite[Section~4.3.8]{AganagicOkounkov201604}. Consider the following analog of \eqref{cd: abelianization}
\[
\begin{tikzcd}
\mathcal F'\arrow[r,"\mathsf j'_+"]\arrow[d,"\pi'"] & \mathcal F_S\cap\mu^{-1}\bigl(\mathfrak b_b^{\perp}\bigr)^{\zeta-ss}/\!\!/ S \arrow[r,"\mathsf j'_-"] &\mathcal \mathcal F_S\\
F,
\end{tikzcd}
\]
where $\mathcal F_S$ is the connected component of $\mathsf X^{\mathsf A}_S$ that contains $\mathcal F'$. Then by the same argument with~\cite[Section~4.3.9]{AganagicOkounkov201604}, $\mathsf j'^{*}_+$ is surjective, and
$\pi^*(\mathsf{ind}\mathcal F)=\mathsf j'^{*}(\mathsf{ind}\mathcal F_S)$,
where $\mathsf j'=\mathsf j'_-\circ \mathsf j'_+$.

 Now we have all ingredients to state the abelianization result. We claim that the following analog of \cite[equation~(72)]{AganagicOkounkov201604} is well defined and the diagram is commutative
\begin{equation}\label{eqn: abelianization}
\begin{tikzcd}
\mscr S_{\mathcal F}\otimes \mscr U' \arrow[rr,"\mathsf j'_{-\circledast}\left(\mathsf j'^{*}_+\right)^{-1}(\pi'_{\circledast})^{-1}"]\arrow[d,"\mathbf{Stab}"] & & \mscr S_{\mathcal F_S}\otimes \mscr U' \arrow[d,"\mathbf{Stab}_S"]\\
\mscr S_{\mathsf X}\otimes\mscr U\otimes \Theta(\hbar)^{\operatorname{rk} \mathsf{ind}} & & \mscr S_{\mathsf X_S}\otimes\mscr U\otimes \Theta(\hbar)^{\operatorname{rk} \mathsf{ind}} \arrow[ll,"\pi_{\circledast}\mathsf j_+^* \left(\mathsf j_{-\circledast}\right)^{-1}"],
\end{tikzcd}
\end{equation}
where $\mscr U'=\tau(-\hbar\det \mathsf{ind})^*\mscr U$. In fact, the argument in \cite[Section~4.3.11]{AganagicOkounkov201604} works verbatim, so the inverse of the top horizontal arrow is well defined; the normal weights to $\mathfrak b_b^{\perp}\subset\mathfrak g^\vee$ are non-attracting then it follows that
\smash{$
\mathsf{Attr}^f\bigl(\mu^{-1}\bigl(\mathfrak b_b^{\perp}\bigr)^{\mathsf A}\bigr)\subset \mu^{-1}\bigl(\mathfrak b_b^{\perp}\bigr)$},
 thus $\mathbf{Stab}_S\circ \mathsf j'_{-\circledast}$ factors through~$\mathsf j_{-\circledast}$, so the inverse of the bottom horizontal arrow is also well defined. To show that the diagram commutes, we only need to show that $\mathbf{Stab}$ agrees with the composition of the rest of the arrows in the diagram when restricted to the diagonal. The point is to compare the difference between~\smash{$\pi^*N^-_{\mathsf X/\mathcal F}$} and \smash{$N^-_{\mathsf X_S/\mathcal F'}$}, and from our previous discussion we learn that the new repelling directions are $(T_\pi|_{\mathcal F'})_{\mathrm{moving}}$ and $(N_{\mathsf j_-}|_{\mathcal F'})_{\mathrm{moving}}$. $(T_\pi|_{\mathcal F'})_{\mathrm{moving}}$ is cancelled out by~$\pi_{\circledast}$, and $(N_{\mathsf j_-}|_{\mathcal F'})_{\mathrm{moving}}$ is cancelled with \smash{$(\mathsf j_{-\circledast})^{-1}\circ \mathbf{Stab}_S\circ \mathsf j'_{-\circledast}$}. This concludes the proof of abelianization \eqref{eqn: abelianization}.

\subsection[$R$-matrices and dynamical Yang--Baxter equations]{$\boldsymbol{R}$-matrices and dynamical Yang--Baxter equations}\label{subsec: R-Matrices}

\begin{dfn}
Let $\mathfrak C_1$, $\mathfrak C_2$ be two chambers in $\mathrm{Lie}(\mathsf A)$, then we define the $R$-matrix
\begin{align}\label{eqn: dfn of R-matirx}
 \mathbf R_{\mathfrak C_2\shortleftarrow\mathfrak C_1}:=\mathbf{Stab}_{\mathfrak C_2}^{-1}\circ \mathbf{Stab}_{\mathfrak C_1}.
\end{align}
This is a map from $\mscr S_{\mathsf X^{\mathsf A}}\otimes \tau(-\hbar\det\mathsf{ind}_1)^*\mscr U$ to $\mscr S_{\mathsf X^{\mathsf A}}\otimes \tau(-\hbar\det\mathsf{ind}_2)^*\mscr U$, where $\mathsf{ind}_1$ and $\mathsf{ind}_2$ are index bundles for the chambers $\mathfrak C_1$ and $\mathfrak C_2$ respectively.
\end{dfn}

Obviously it follows from definition that
\begin{align}\label{eqn: composing R-matrices}
 \mathbf R_{\mathfrak C_3\shortleftarrow\mathfrak C_2}\mathbf R_{\mathfrak C_2\shortleftarrow\mathfrak C_1}=\mathbf R_{\mathfrak C_3\shortleftarrow\mathfrak C_1}.
\end{align}
Moreover, if $\mathfrak C_1$ and $\mathfrak C_2$ share a face $\mathfrak C'$, then \eqref{eqn: stab on a wall} implies that
\begin{align}\label{eqn: $R$-matrix on a wall}
 \mathbf R_{\mathfrak C_2\shortleftarrow\mathfrak C_1}(z)=\mathbf R_{\mathfrak C_2/\mathfrak C'\shortleftarrow\mathfrak C_1/\mathfrak C'}(z-\hbar\det \mathsf{ind}_{\mathsf X^{\mathsf A'},\mathfrak C'}),
\end{align}
where $\mathsf A'$ is the subtorus corresponding to the chamber $\mathfrak C'$ and \smash{$\mathsf{ind}_{\mathsf X^{\mathsf A'},\mathfrak C'}$} is the index bundle on~\smash{$\mathsf X^{\mathsf A'}$} for the chamber $\mathfrak C'$. If $\mathfrak C'$ is of codimension one in both $\mathfrak C_1$ and $\mathfrak C_2$, then we say $\mathfrak C'$ is a wall and define
$
 \mathbf R_{\mathrm{wall}}:=\mathbf R_{\mathfrak C_2/\mathfrak C'\shortleftarrow\mathfrak C_1/\mathfrak C'}$.
The dynamical Yang--Baxter equation is a direct consequence of \eqref{eqn: composing R-matrices} and \eqref{eqn: $R$-matrix on a wall} for $\mathsf A=(\mathbb C^{\times})^3$. To see this, let $m_1$, $m_2$, $m_3$ be coordinates on~$\mathrm{Lie}(\mathsf A)$, then the two paths connecting the chamber $\mathfrak C_{123}=\{m_1<m_2<m_3\}$ to its opposite~${\mathfrak C_{321}=\{m_3<m_2<m_1\}}$, namely $\mathfrak C_{123}\to \mathfrak C_{132}\to \mathfrak C_{312}\to \mathfrak C_{321}$ and $\mathfrak C_{123}\to \mathfrak C_{213}\to \mathfrak C_{231}\to \mathfrak C_{321}$, gives rise to equation between two sequence of composing $R$-matrices
\begin{gather}
\mathbf R_{12}(z-\hbar\det \mathsf{ind}_{3(12)}) \mathbf R_{13}(z-\hbar\det \mathsf{ind}_{(13)2}) \mathbf R_{23}(z-\hbar\det \mathsf{ind}_{1(23)})\nonumber\\
\qquad=\mathbf R_{23}(z-\hbar\det \mathsf{ind}_{(23)1}) \mathbf R_{13}(z-\hbar\det \mathsf{ind}_{2(13)}) \mathbf R_{12}(z-\hbar\det \mathsf{ind}_{(12)3}),\label{eqn: DYBE}
\end{gather}
where $\mathbf R_{ij}$ is the $R$-matrix for the $ij$-wall, and the subscript of $\mathsf{ind}$ means the wall and chamber for which the index bundle is taken, for example, $3(12)$ means the chamber $\{m_3<m_1=m_2\}$.

 The particular case $\mathfrak C_3=\mathfrak C_1$ in \eqref{eqn: composing R-matrices} implies the unitarity
$ \mathbf R_{21}(-u,\hbar,z)\mathbf R(u,\hbar,z)=1$,
where~$\mathbf R(u,\hbar,z)$ is the wall $R$-matrix and subscript $21$ means swapping the two tensor components.

\subsection[$R$-matrices for quiver varieties]{$\boldsymbol{R}$-matrices for quiver varieties}\label{subsec: R-Matrices for Quiver Varieties}
Let us write down the dynamical Yang--Baxter equation for the quiver variety $\mcal M^{\zeta_-}(\mathbf w)$ for a~quiver $Q$ with decomposition of nodes $Q_0=Q^{\mathrm{ev}}_0\sqcup Q^{\mathrm{odd}}_0$ into even and odd parts and decomposition of arrows $Q_1=Q^+_1\sqcup Q^-_1$ such that the condition $(\star)$ in Example~\ref{ex: partial polarization for the quivers} is satisfied. We take~$\mathsf A=(\mathbb C^{\times})^3\subset T_W$ such that the framing vector space decomposes into eigenspaces as
\smash{$ W=a_1W^{(1)}+a_2W^{(2)}+a_3W^{(3)}$}.
Choosing the partial polarization \eqref{eqn: pol in the quiver case}, we can write down the index bundle for the chamber corresponding to the weight decomposition $W=aW'+bW''$ such that $b/a$ is attracting
\[
\mathsf{ind}=\sum_{i,j\in Q_0} \mathsf D_{ij}\mathrm{Hom}\bigl(\mcal V'_i,\mcal W''_j\bigr)+ (\mathsf Q_{ij}-\mathsf P_{ij})\mathrm{Hom}\bigl(\mcal V'_i,\mcal V''_j\bigr),
\]
where $\mathsf Q_{ij}$ is the \emph{weighted} adjacency matrix of $Q$ defined by
\[\mathsf Q_{ij}=\#\bigl(a\in Q_1^+: h(a)=j,\,t(a)=i\bigr)-\#(a\in Q_1^-: h(a)=j,t(a)=i),\]
and $\mathsf D_{ij}$ is the diagonal matrix such that
\[
\mathsf D_{ii}=
\begin{cases}
\hphantom{-}1, & i\notin Q^{\mathrm{ev},-}_0,\\
-1, & i\in Q^{\mathrm{ev},-}_0,
\end{cases}
\]
and $\mathsf P_{ij}$ is the diagonal matrix such that
\[
\mathsf P_{ii}=
\begin{cases}
\hphantom{-}1, & i\in Q^{\mathrm{ev},+}_0,\\
-1, & i\in Q^{\mathrm{ev},-}_0,\\
\hphantom{-}0, & i\in Q^{\mathrm{odd}}_0.
\end{cases}
\]
Define the quiver $R$-matrix
$\mathbf R^Q(\mbs{z}):=\mathbf R_{\mathrm{wall}}\bigl(\mbs{z}+\hbar \bigl(\mathsf P-\mathsf Q^{\mathrm{t}}\bigr)\mathbf v\bigr)$,
where $\mbs{z}:=\{z_i\}_{i\in Q_0}$ are the equivariant parameters corresponding to the ample line bundles $(\det\mcal V_{i})^{-1}$. Let $\mbs{\mu}:=\mathsf D\mathbf w-\mathsf C\mathbf v$, where
\begin{align}\label{eqn: Cartan matrix}
 \mathsf C=2\mathsf P-\mathsf Q-\mathsf Q^{\mathrm{t}},
\end{align}
then \eqref{eqn: DYBE} is specialized to
\begin{gather}
 \mathbf R^Q_{12}(\mbs{z}) \mathbf R^Q_{13}\bigl(\mbs{z}-\hbar\mbs{\mu}^{(2)}\bigr) \mathbf R^Q_{23}(\mbs{z})= \mathbf R^Q_{23}\bigl(\mbs{z}-\hbar\mbs{\mu}^{(1)}\bigr) \mathbf R^Q_{13}(\mbs{z}) \mathbf R^Q_{12}\bigl(\mbs{z}-\hbar\mbs{\mu}^{(3)}\bigr).\label{eqn: DYBE for quiver}
\end{gather}
This kind of equation is called the dynamical Yang--Baxter equations and was first studied by Felder \cite{Felder199412,Felder199407}.

\begin{ex}[$\mathfrak{sl}(1|1)$]
Consider the resolved determinantal variety case (see Example~\ref{ex: stab for M(N,L)}), namely a quiver $Q$ with one odd node, i.e., $Q_0=Q^{\mathrm{odd}}_0=\{1\}$. We choose stability condition~${\zeta=\zeta_-}$ and partial polarization \eqref{eqn: pol in the quiver case}. We choose the framing dimensions $\mathbf w^{(1)}=\mathbf w^{(2)}=\mathbf w^{(3)}=1$, then \smash{$\mcal M^{\zeta}\bigl(\mathbf w^{(1)}\bigr)$} is disjoint union of two points, and
\begin{align*}
 \mcal M^{\zeta}\bigl(\mathbf w^{(1)}+\mathbf w^{(2)}\bigr)=\mathrm{pt}\sqcup \mathrm{Tot}\bigl(\mscr O(-1)^{\oplus 2}\to \mathbb P^1\bigr)\sqcup \mathbb C^4.
\end{align*}
Let us label the basis for the elliptic cohomology $\mathsf E_{\mathsf T}\bigl(\mcal M^{\zeta}\bigl(\mathbf w^{(1)}\bigr)\bigr)$, which is a free module of rank two over $\mscr E_{x}\times \mscr E_{\hbar}\times \mscr E_{z}$, by $v_0=\bigl[\mcal M^{\zeta}(0,1)\bigr]$ and $v_1=\bigl[\mcal M^{\zeta}(1,1)\bigr]$,
then we can write down the quiver $R$-matrix in the basis $\{v_0\otimes v_0,v_1\otimes v_0, v_0\otimes v_1, v_1\otimes v_1\}$
\begin{equation}\label{eqn: quiver $R$-matrix for sl(1|1)}
	\def\arraystretch{1.5}
	\setlength{\arrayrulewidth}{.8pt}
 \setlength{\arraycolsep}{6pt}
	\mathbf R^Q(u,z)=
	\left(
	\begin{array}{c|cc|c}
		1 & 0 & 0 & 0
		\\
		\hline
		0 & \frac{\vartheta(z)\vartheta(z\hbar^{-2})\vartheta(u)}{\vartheta\bigl(z\hbar^{-1}\bigr)^2\vartheta\bigl(u\hbar^{-1}\bigr)} & \frac{\vartheta(\hbar)\vartheta(u z\hbar^{-1})}{\vartheta\bigl(z\hbar^{-1}\bigr)\vartheta\bigl(u^{-1}\hbar\bigr)} & 0
 \\[6pt]
 0 & \frac{\vartheta(\hbar)\vartheta(zu^{-1}\hbar^{-1})}{\vartheta\bigl(z\hbar^{-1}\bigr)\vartheta\bigl(u^{-1}\hbar\bigr)} & \frac{\vartheta(u)}{\vartheta\bigl(u\hbar^{-1}\bigr)} & 0
		\\[4pt]
		\hline
		0 & 0 & 0 & \frac{\vartheta(u\hbar)}{\vartheta\bigl(u^{-1}\hbar\bigr)}
	\end{array}\right),
\end{equation}
where $u=x_2/x_1$ is the ratio between equivariant weights in the decomposition $W=x_1W'+x_2W''$. In this case $\mathsf C=0$, which is the Cartan matrix for $\mathfrak{sl}(1|1)$, and $\mbs{\mu}(v_0)=\mbs{\mu}(v_1)=1$, so the dynamical Yang--Baxter equation \eqref{eqn: DYBE for quiver} reads
\begin{gather}\label{eqn: DYBE for sl(1|1)}
 \mathbf R^Q_{12}(u,z) \mathbf R^Q_{13}(uw,z/\hbar) \mathbf R^Q_{23}(w,z)= \mathbf R^Q_{23}(w,z/\hbar) \mathbf R^Q_{13}(uw,z) \mathbf R^Q_{12}(u,z/\hbar),
\end{gather}
which is the dynamical Yang--Baxter equation for $\mathfrak{sl}(1|1)$ in the fundamental representation.
\end{ex}

\begin{ex}[$\mathfrak{sl}(n|1)$]\label{exa:elliptic dynamical $R$-matrix for sl(n|1)}
Next, we consider a generalization to the above example, namely the quiver in Example~\ref{ex: quiver with one f node 1}. Let $Q$ be an $A_{n}$ quiver with $Q^{\mathrm{ev}}_0=\{1,\dots,n-1\}$ and $Q^{\mathrm{odd}}_0=\{n\}$, and the orientation is such that $h(a)<t(a)$ for all arrows $a\in Q_1$. We take $Q_1=Q_1^+$. We choose stability condition $\zeta=\zeta_-$ and the partial polarization \eqref{eqn: pol in the quiver case}. We choose framing dimensions~${\mathbf w^{(1)}=\mathbf w^{(2)}=\mathbf w^{(3)}=(1,0,\dots,0)}$, then $\mcal M^{\zeta}\bigl(\mathbf w^{(1)}\bigr)$ is the disjoint union of $n+1$ points which is labeled by integers $\alpha\in \{0,1,\dots,n\}$ so that the gauge dimension vector of the corresponding quiver variety is
\begin{align*}
 \mathbf v_{\alpha}=(\overbrace{\underbrace{1,\dots,1}_{\alpha},0,\dots,0}^{n}).
\end{align*}
Using the above notation, we have
\begin{align}\label{eqn: sl(n|1) moduli}
\mcal M^{\zeta}\bigl(\mathbf v_\alpha+\mathbf v_\beta,\mathbf w^{(1)}+\mathbf w^{(2)}\bigr)=
\begin{cases}
\mathrm{pt},& \alpha=\beta<n,\\
T^*\mathbb P^1,& \alpha<\beta<n \text{ or } \beta<\alpha<n,\\
\mathrm{Tot}\bigl(\mscr O(-1)^{\oplus 2}\to \mathbb P^1\bigr),& \alpha<\beta=n \text{ or } \beta<\alpha=n,\\
\mathbb C^4,& \alpha=\beta=n.
\end{cases}
\end{align}
Let us label the basis for the elliptic cohomology $\mathsf E_{\mathsf T}\bigl(\mcal M^{\zeta}\bigl(\mathbf w^{(1)}\bigr)\bigr)$, which is a free module of rank~${n+1}$ over $\mscr E_{x}\times \mscr E_{\hbar}\times \mscr E_{\mbs{z}}$, by $v_\alpha=\bigl[\mcal M^{\zeta}\bigl(\mathbf v_\alpha,\mathbf w^{(1)}\bigr)\bigr]$. Here the K\"ahler parameter ${\mbs z=(z_i)_{i=1}^n}$ corresponds to the ample line bundles $\bigl(\det\mcal V_i^\vee\bigr)_{i=1}^n$, where $\mcal V_i$ is the tautological bundle corresponding to the $i$-th node.

 For the quiver $Q$, we have the Cartan matrix
\[
\mathsf C_{i,j}=2\delta_{i,j}(1-\delta_{i,n})-\delta_{i+1,j}-\delta_{i-1,j}, \qquad 1\le i,j\le n,
\]
which is the symmetric Cartan matrix of $\mathfrak{sl}(n|1)$ \cite{frappat1996dictionary}. The weights of $v_{\alpha}$ are
\begin{align}\label{eqn: weights for sl(n|1)}
\mbs{\mu}(v_\alpha)=\mathbf w^{(1)}-\mathsf C\mathbf v_\alpha=
\begin{cases}
(1,0,\dots,0),& \alpha=0,\\
(\underbrace{0,\dots,0}_{\alpha-1},-1,1,0,\dots,0),& 0<\alpha<n,\\
(0,\dots,0,1),& \alpha=n,
\end{cases}
\end{align}
which is aligned with the weights for the fundamental representation of $\mathfrak{sl}(n|1)$. Moreover, the K\"ahler parameter shifts between the wall $R$-matrix and the quiver $R$-matrix is
\[
\bigl(\mathsf P-\mathsf Q^{\mathrm{t}}\bigr)\mathbf v_\alpha=
\begin{cases}
(\underbrace{0,\dots,0}_{\alpha-1},1,0,\dots,0),& 0<\alpha<n,\\
(0,\dots,0),& \alpha=0\text{ or }\alpha=n.
\end{cases}
\]
Using the known results on the $R$-matrices for all four varieties in \eqref{eqn: sl(n|1) moduli}, see \eqref{eqn: quiver $R$-matrix for sl(1|1)} and \cite[equation~(85)]{AganagicOkounkov201604}, we can write down the quiver $R$-matrix in this case. We introduce shorthand notations for some special functions
\begin{gather*}
\begin{split}
& A(u,z)=\frac{\vartheta(z\hbar)\vartheta\bigl(z\hbar^{-1}\bigr)\vartheta(u)}{\vartheta(z)^2\vartheta\bigl(u\hbar^{-1}\bigr)},\qquad B(u,z)=\frac{\vartheta(\hbar)\vartheta(u z)}{\vartheta(z)\vartheta\bigl(u^{-1}\hbar\bigr)},\qquad
C(u)=\frac{\vartheta(u)}{\vartheta\bigl(u\hbar^{-1}\bigr)},\\
& D(u)=\frac{\vartheta(u\hbar)}{\vartheta\bigl(u^{-1}\hbar\bigr)},
\end{split}
\end{gather*}
where we have suppressed the $\hbar$-dependence. Then the quiver $R$-matrix in this case is the following:
\begin{gather*}
\mathbf R^Q(u,\mbs{z})(v_\alpha\otimes v_{\beta})\\
\qquad=\begin{cases}
v_\alpha\otimes v_{\beta},& \alpha=\beta<n, \\
D(u)v_\alpha\otimes v_{\beta},& \alpha=\beta=n,\\
C(u)v_\alpha\otimes v_{\beta}+B\bigl(u,\hbar^{-\delta_{\beta,n}}\prod_{i=\alpha+1}^\beta z_i\bigr)v_\beta\otimes v_{\alpha}, & \alpha<\beta,\\
A\bigl(u,\hbar^{-\delta_{\alpha,n}}\prod_{i=\beta+1}^\alpha z_i\bigr)v_\alpha\otimes v_{\beta}+B\bigl(u,\hbar^{\delta_{\alpha,n}}\prod_{i=\beta+1}^\alpha z_i^{-1}\bigr)v_\beta\otimes v_{\alpha}, & \beta<\alpha.
\end{cases}
\end{gather*}
$\mathbf R^Q$ satisfies the dynamical Yang--Baxter equation \eqref{eqn: DYBE for quiver} for $\mbs{\mu}$ given by weights for the fundamental representation of $\mathfrak{sl}(n|1)$ \eqref{eqn: weights for sl(n|1)}.
\end{ex}

For quivers of finite or affine type A, we make the following observation. In addition to the condition $(\star)$ in Example~\ref{ex: partial polarization for the quivers}, we also require that
\begin{itemize}\itemsep=0pt
 \item[$(\star\star)$] if the valency at an odd node is two then the arrows connected to it are from different groups, i.e., exactly one of the arrows is from $Q^+_1$ and the other is from $Q^-_1$.
\end{itemize}
Here is an example of an allowed quiver (after adding framing and doubling)
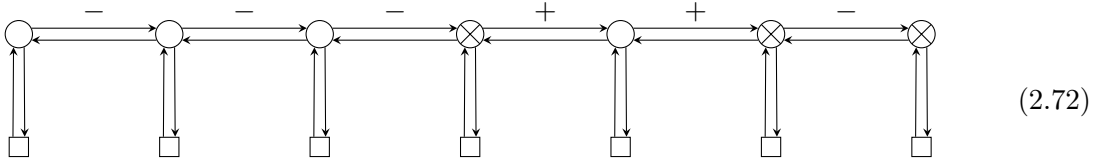
\[
\begin{tikzpicture}[x={(2cm,0cm)}, y={(0cm,2cm)}, baseline=-1cm]
\node[draw, circle] at (0,0) (n1) {};
\node at (0.5,0.15) (a1) {$-$};
\node[draw, circle] at (1,0) (n2) {};
\node at (1.5,0.15) (a2) {$-$};
\node[draw, circle] at (2,0) (n3) {};
\node at (2.5,0.15) (a3) {$-$};
\node[draw, circle, cross] at (3,0) (ndot) {};
\node at (3.5,0.15) (adot) {$+$};
\node[draw, circle] at (4,0) (n4) {};
\node at (4.5,0.15) (a4) {$+$};
\node[draw, circle,cross] at (5,0) (n5) {};
\node at (5.5,0.15) (a5) {$-$};
\node[draw, circle,cross] at (6,0) (n6) {};

\draw[-stealth] (n1.25) to (n2.155);
\draw[-stealth] (n2.25) to (n3.155);
\draw[-stealth] (n3.25) to (ndot.155);
\draw[-stealth] (ndot.25) to (n4.155);
\draw[-stealth] (n4.25) to (n5.155);
\draw[-stealth] (n5.25) to (n6.155);

\draw[-stealth] (n6.205) to (n5.335);
\draw[-stealth] (n5.205) to (n4.335);
\draw[-stealth] (n4.205) to (ndot.335);
\draw[-stealth] (ndot.205) to (n3.335);
\draw[-stealth] (n3.205) to (n2.335);
\draw[-stealth] (n2.205) to (n1.335);

\node[draw, rectangle] at (0,-.75) (f1) {};
\node[draw, rectangle] at (1,-.75) (f2) {};
\node[draw, rectangle] at (2,-.75) (f3) {};
\node[draw, rectangle] at (3,-.75) (fdot) {};
\node[draw, rectangle] at (4,-.75) (f4) {};
\node[draw, rectangle] at (5,-.75) (f5) {};
\node[draw, rectangle] at (6,-.75) (f6) {};

\draw[-stealth] (n1.292) to (f1.60);
\draw[-stealth] (f1.120) to (n1.248);
\draw[-stealth] (n2.292) to (f2.60);
\draw[-stealth] (f2.120) to (n2.248);
\draw[-stealth] (n3.292) to (f3.60);
\draw[-stealth] (f3.120) to (n3.248);
\draw[-stealth] (ndot.292) to (fdot.60);
\draw[-stealth] (fdot.120) to (ndot.248);
\draw[-stealth] (n4.292) to (f4.60);
\draw[-stealth] (f4.120) to (n4.248);
\draw[-stealth] (n5.292) to (f5.60);
\draw[-stealth] (f5.120) to (n5.248);
\draw[-stealth] (n6.292) to (f6.60);
\draw[-stealth] (f6.120) to (n6.248);
\end{tikzpicture}.
\label{an admissible type A quiver}
\]
For simplicity, we consider an \smash{$\widehat{\mathrm{A}}_{m+n-1}$} quiver $Q$ and label the nodes by $Q_0=\mathbb Z/(m+n)\mathbb Z$. We choose the orientation $Q_1=\{i-1\to i\colon i\in \mathbb Z/(m+n)\mathbb Z\}$, and choose a splitting $Q_1=Q_1^+\sqcup Q_1^-$ such that $|Q_1^+|=m$ and $|Q_1^-|=n$. Define
\[
 c_i=\begin{cases}
 +1,& \text{if }(i-1\to i)\in Q_1^+,
 \\
 -1,& \text{if }(i-1\to i)\in Q_1^-.
 \end{cases}
\]
Then the condition $(\star)$ is equivalent to $c_i=c_{i+1}$ for all $i\in Q^{\mathrm{ev}}_0$, and the condition $(\star\star)$ is equivalent to $c_i=-c_{i+1}$ for all $i\in Q^{\mathrm{odd}}_0$. The Cartan matrix \eqref{eqn: Cartan matrix} reads
\[
\mathsf C_{i,j}=(c_{i}+c_{i+1})\delta_{i,j}-c_i\delta_{i,j+1}-c_j\delta_{i+1,j}.
\]
$\mathsf C_{i,j}$ is exactly the symmetric Cartan matrix of the Kac--Dynkin diagram for $\widehat{\mathfrak{sl}}(m|n)$ associated to the quiver $Q$. Moreover, the quiver $R$-matrix $\mathbf R^Q$ satisfies the dynamical Yang--Baxter equation~\eqref{eqn: DYBE for quiver} for $\mbs{\mu}$ given by weights in a certain highest-weight module of $\widehat{\mathfrak{sl}}(m|n)$.

\subsection{K-theory and cohomology limit}\label{sec:k-theory and cohomology limit}
In the $q\to 0$ limit, the elliptic curve $\mbb{E}=\mathbb C^{\times}/q^{\mathbb Z}$ degenerates to a nodal compactification of $\mathbb C^{\times}$, and the elliptic cohomology $\mathrm{Ell}_{\mathsf T}(\mathsf X)$ degenerates to the K-theory $K_{\mathsf T}(\mathsf X)\otimes\mathbb C$, such that sections of line bundles on $\mathrm{Ell}_{\mathsf T}(\mathsf X)$ becomes sections of line bundles on $K_{\mathsf T}(\mathsf X)\otimes\mathbb C$. Strictly speaking, because of half-periodicity property of $\vartheta$-function: $\vartheta\bigl({\rm e}^{2\pi {\rm i}}x\bigr)=-\vartheta(x)$, their $q\to 0$ limit is defined on the double cover of $\mathbb C^{\times}$.

 We expect that the elliptic stable envelope degenerates to the K-theoretic stable envelope in the $q\to 0$ limit, such that the K\"ahler parameter $z$ degenerates to the slope parameter in the K-theory side. Suppose that
\begin{equation}\label{eq:the k-theory slope parameter}
 \lim_{\substack{\ln z\to\infty \\ \ln q\to\infty}} \tenofo{Re}\left(-\frac{\ln z}{\ln q}\right)=\msf{s}\in\mathrm{Pic}_{\mathsf T}(\mathsf X)\otimes_{\mathbb Z}\mathbb R,
\end{equation}
such that the slope $\msf{s}$ is \emph{generic}, and then define
\begin{equation}\label{eq:k-theoretic stable envelope as the limit of elliptic one}
 \msf{Stab}^{\mathsf{s}}_{\mfk{C}}:= \lim_{q\to 0} \bigl[(\det \mathrm{Pol}_{\mathsf X})^{-\frac{1}{2}}\circ\mathbf{Stab}_{\mfk{C}}\circ(\det \mathrm{Pol}_{\mathsf X^{\mathsf A}})^{\frac{1}{2}}\bigr]\in K_{\mathsf T}\bigl(\mathsf X\times\mathsf X^{\mathsf A}\bigr).
\end{equation}
For every connected component $\mathcal F\subset \msf X^{\msf A}$,
\begin{align*}
\msf{Stab}^{\mathsf{s}}_{\mathfrak C}([\mathcal F])|_{\mathcal F}&=(-1)^{\operatorname{rk}\mathsf{ind}}\left(\frac{\det N^-_{\mathsf X/{\mathcal F}}}{\det \mathrm{Pol}_{\mathsf X}|_{\mathcal F,\mathrm{moving}}}\right)^{\frac{1}{2}}\bigwedge(N^-_{\mathsf X/{\mathcal F}})^\vee\\
&=\bigl(-\sqrt{\hbar}\bigr)^{\operatorname{rk} \mathsf{ind}}(\det \mathrm{Pol}_{\mathsf X}|_{\mathcal F,>0})^{-1} (\det \mscr G_{\mathrm{odd}}|_{\mathcal F,<0})^{-\frac{1}{2}} \bigwedge(N^-_{\mathsf X/{\mathcal F}})^\vee.
\end{align*}
Moreover, $\msf{Stab}^{\mathsf{s}}_{\mathfrak C}$ is supported on $\msf{Attr}^f_{\mathfrak C}$, i.e., if $\mathcal F_2\npreceq \mathcal F_1$ then $\msf{Stab}^{\mathsf{s}}_{\mathfrak C}([\mathcal F_1])|_{\mathcal F_2}=0$. We expect that~$\msf{Stab}^{\mathsf{s}}_{\mathfrak C}$ is the K-theoretic stable envelope with slope $\mathsf s$ for the chamber $\mfk C$. In other words, $\msf{Stab}^{\mathsf{s}}_{\mathfrak C}$ also satisfies the degree condition, which is checked in the following.

 For a connected component $\mathcal F$ of $\mathsf X^{\mathsf A}$, define
\[
 \Delta_{\mathcal F}:=\text{Convex hull}\left(\mathrm{Supp}_{\mathsf A}\msf{Stab}^{\mathsf{s}}_{\mfk{C}}([\mathcal F])|_{\mathcal F}\right),
\]
where $\mathrm{Supp}_{\mathsf A}$ of a $K_{\mathsf T}(\mathcal F)$ class is the set of $\mathsf A$-weights that appears in that class.
\begin{prop}\label{prop: K-theory reduction}
Let $\mathsf X$ be a quotient of Hamiltonian reduction {\rm\eqref{subsubsec: Example: Quotient of Hamiltonian Reduction}}, assume moreover that either
\begin{itemize}\itemsep=0pt
 \item the gauge group is abelian, i.e., the situation in Example~{\rm\ref{ex: abelian gauge theory}}, or
 \item $\mathsf X$ is constructed from the quiver representation, i.e., the situation in Section~{\rm\ref{subsubsec: Example: The Doubled Quivers}},
\end{itemize}
then for every pair of connected components $\mathcal F_2\preceq \mathcal F_1$ in $\mathsf X^{\mathsf A}$ with respect to the chamber $\mfk{C}$,
\begin{align}\label{eqn: newton polytope}
 \mathrm{Supp}_{\mathsf A}\msf{Stab}^{\mathsf{s}}_{\mfk{C}}([\mathcal F_1])|_{\mathcal F_2}\subset \Delta_{\mathcal F_2}+\mathsf A\text{-weight of }\mathsf s|_{\mathcal F_2}- \mathsf A\text{-weight of }\mathsf s|_{\mathcal F_1}.
\end{align}
\end{prop}

\begin{proof}
Suppose that $\mathsf X=\mu_{\mathrm{ev}}^{-1}(0)^{\zeta-ss}/\!\!/G$, where $G=G_{\mathrm{ev}}\times G_{\mathrm{odd}}$ acts on the representation~$\mcal R$ with moment map $\mu\colon\mcal R\oplus\mcal R^\vee\to \mfk g^\vee$ and $\mu_{\mathrm{ev}}$ is the composition of $\mu$ with projection~${\mathrm{pr}_{\mathrm{ev}}\colon\mfk g^\vee\to \mfk g_{\mathrm{ev}}^\vee}$, and $\zeta$ is a generic stability such that $\mu_{\mathrm{ev}}^{-1}(0)^{\zeta-ss}=\mu_{\mathrm{ev}}^{-1}(0)^{\zeta-s}$. As we have assumed in the setup of Section~\ref{subsubsec: Example: Quotient of Hamiltonian Reduction}, $\mu_{\mathrm{ev}}^{-1}(0)^{\zeta-ss}$ is smooth with free $G$ action.

 \textbf{Step 1.} Assume that $G$ is abelian, then according to our previous discussions in Example~\ref{ex: abelian stable envelope}, the equivariant elliptic cohomology of $\mathsf X$ is isomorphic to that of the central fiber ${\mathsf X_0=\mu^{-1}(0)^{\zeta-ss}/\!\!/G}$: $\mathrm{Ell}_{\mathsf T}(\mathsf X)\cong \mathrm{Ell}_{\mathsf T}(\mathsf X_0)$, and elliptic stable envelopes of $\mathsf X$ and $\mathsf X_0$ are identified via this isomorphism. $\mathsf X_0$ is a hypertoric variety, thus the result follows from \cite[Proposition~4.2]{AganagicOkounkov201604}.

 \textbf{Step 2.} If $G$ is nonabelian, then we shall utilize the abelianization \eqref{eqn: abelianization}. It suffices to prove in the special case when $\mathsf A=\mathbb C^{\times}$, by the same argument of \cite[Proposition~4.2]{AganagicOkounkov201604}. By step 1, we have
\begin{align*}
 \mathrm{Supp}_{\mathsf A}\msf{Stab}^{\mathsf{s}}_{S}\bigl(\bigl[\mathcal F'_1\bigr]\bigr)|_{\mathcal F'_2}\subset \Delta_{\mathcal F'_2}+\mathsf A\text{-weight of }\mathsf s|_{\mathcal F'_2}- \mathsf A\text{-weight of }\mathsf s|_{\mathcal F'_1}.
\end{align*}
The abelianization gives
\begin{align*}
 \msf{Stab}^{\mathsf{s}}([\mathcal F_1])|_{\mathcal F_2}&= \frac{\bigl(\det\mscr N_f^{\vee}|_{\mathcal F'_2}\bigr)^{-\frac{1}{2}}\msf{Stab}^{\mathsf{s}}_{S}\bigl(\bigl[\mathcal F'_1\bigr]\bigr)|_{\mathcal F'_2}}{\bigwedge((N_{\mathsf j_-}|_{\mathcal F'_2})_{\mathrm{moving}})^\vee\cdot \bigwedge((T_\pi|_{\mathcal F'_2})_{\mathrm{moving}})^\vee} \bigl(\det\mscr N_f^{\vee}|_{\mathcal F'_1,\mathrm{fixed}}\bigr)^{\frac{1}{2}},\\
 \msf{Stab}^{\mathsf{s}}([\mathcal F_2])|_{\mathcal F_2}&= \frac{(\det\mscr N_f^{\vee}|_{\mathcal F'_2,\mathrm{moving}})^{-\frac{1}{2}}\msf{Stab}^{\mathsf{s}}_{S}\bigl(\bigl[\mathcal F'_2\bigr]\bigr)|_{\mathcal F'_2}}{\bigwedge((N_{\mathsf j_-}|_{\mathcal F'_2})_{\mathrm{moving}})^\vee\cdot \bigwedge((T_\pi|_{\mathcal F'_2})_{\mathrm{moving}})^\vee}.
\end{align*}
The first line in the above equation implies that
\begin{gather*}
 \mathrm{Supp}_{\mathsf A}\msf{Stab}^{\mathsf{s}}([\mathcal F_1])|_{\mathcal F_2}\cdot \bigwedge((N_{\mathsf j_-}|_{\mathcal F'_2})_{\mathrm{moving}})^\vee\cdot \bigwedge((T_\pi|_{\mathcal F'_2})_{\mathrm{moving}})^\vee \\
 \qquad \subset \Delta_{\mathcal F'_2}-\frac{1}{2}\mathsf A\text{-weight of }\det\mscr N_f^{\vee}|_{\mathcal F'_2}+\mathsf A\text{-weight of }\mathsf s|_{\mathcal F'_2}- \mathsf A\text{-weight of }\mathsf s|_{\mathcal F'_1},
\end{gather*}
and the second line implies that $$\Delta_{\mathcal F'_2}=\Delta_{\mathcal F_2}\cdot\bigwedge((N_{\mathsf j_-}|_{\mathcal F'_2})_{\mathrm{moving}})^\vee\cdot \bigwedge((T_\pi|_{\mathcal F'_2})_{\mathrm{moving}})^\vee+\frac{1}{2}\mathsf A\text{-weight of }\det\mscr N_f^{\vee}|_{\mathcal F'_2},$$
thus we have
\begin{align*}
 \mathrm{Supp}_{\mathsf A}\msf{Stab}^{\mathsf{s}}([\mathcal F_1])|_{\mathcal F_2}\subset\Delta_{\mathcal F_2}+\mathsf A\text{-weight of }\mathsf s|_{\mathcal F_2}- \mathsf A\text{-weight of }\mathsf s|_{\mathcal F_1}.
\end{align*}
This finishes the proof.
\end{proof}

A further reduction from the K-theory to cohomology can be defined by
\begin{align}\label{eqn: reduction to cohomology}
 \mathrm{Stab}_{\mathfrak C}:=\text{lowest cohomological degree term in }\mathrm{ch}\bigl(\msf{Stab}^{\mathsf{s}}_{\mfk C}\bigr),
\end{align}
where $\mathrm{ch}\colon K_{\mathsf T}\bigl(\mathsf X\times \mathsf X^{\mathsf A}\bigr)\to H_{\mathsf T}\bigl(\mathsf X\times \mathsf X^{\mathsf A}\bigr)$ is the $\mathsf T$-equivariant Chern character map. The composition of reductions $\mathbf{Stab}_{\mfk C}\to \msf{Stab}^{\mathsf{s}}_{\mfk C}\to \mathrm{Stab}_{\mfk C}$ amounts to replacing $\vartheta(x)$ by $x$.

 For every connected component $\mathcal F\subset \msf X^{\msf A}$,
\smash{$ \mathrm{Stab}_{\mathfrak C}([\mathcal F])|_{\mathcal F}=(-1)^{\operatorname{rk} \mathsf{ind}}e(N^-_{\mathsf X/{\mathcal F}})$},
where $e(\cdot)$ is the $\mathsf T$-equivariant Euler class. Moreover, $\mathrm{Stab}_{\mathfrak C}$ is supported on \smash{$\msf{Attr}^f_{\mathfrak C}$}, i.e., if~${\mathcal F_2\npreceq \mathcal F_1}$ then $\mathrm{Stab}_{\mathfrak C}([\mathcal F_1])|_{\mathcal F_2}=0$. We expect that $\mathrm{Stab}_{\mathfrak C}$ is the cohomological stable envelope for the chamber~$\mfk C$. In other words, $\mathrm{Stab}_{\mathfrak C}$ also satisfies the degree condition, which is checked in the following.

\begin{prop}\label{prop: cohomology reduction}
Under the assumption of Proposition~{\rm\ref{prop: K-theory reduction}}, then for every pair of connected components $\mathcal F_2\preceq \mathcal F_1$ in $\mathsf X^{\mathsf A}$ with respect to the chamber $\mfk{C}$,
\smash{$ \deg_{\mathsf A}\mathrm{Stab}_{\mfk{C}}([\mathcal F_1])|_{\mathcal F_2}<\dim N^-_{\mathsf X/{\mathcal F}_2}$},
where $\deg_{\mathsf A}$ of an element in $H_{\mathsf T}\bigl(\mathsf X^{\mathsf A}\bigr)$ is the degree in $\mathbb C[\mfk a]$ under the non-canonical splitting $H_{\mathsf T}\bigl(\mathsf X^{\mathsf A}\bigr)\cong H_{\mathsf T/\mathsf A}\bigl(\mathsf X^{\mathsf A}\bigr)\otimes\mathbb C[\mfk a]$.\footnote{Different choice of splitting does not change the degree in $\mathbb C[\mfk a]$.}
\end{prop}

\begin{proof}
The idea is essentially the same as Proposition~\ref{prop: K-theory reduction}. Namely, the abelian case is the same as hypertoric varieties, which is shown by using the explicit formula \cite[equation~(56)]{AganagicOkounkov201604}. For the nonabelian case, we use the abelianization and get
\begin{align*}
 \mathrm{Stab}([\mathcal F_1])|_{\mathcal F_2}&= \frac{\mathrm{Stab}_{S}\bigl(\bigl[\mathcal F'_1\bigr]\bigr)|_{\mathcal F'_2}}{e((N_{\mathsf j_-}|_{\mathcal F'_2})_{\mathrm{moving}})\cdot e((T_\pi|_{\mathcal F'_2})_{\mathrm{moving}})}.
\end{align*}
Since the abelian stable envelope satisfies
\smash{$\deg_{\msf A}\bigl(\mathrm{Stab}_{S}\bigl(\bigl[\mathcal F'_1\bigr]\bigr)|_{\mathcal F'_2}\bigr)<\dim N^-_{\mathsf X_S/{\mathcal F}_{2,S}}$},
we conclude that
\begin{align*}
 \deg_{\msf A}\mathrm{Stab}([\mathcal F_1])|_{\mathcal F_2}&=\deg_{\msf A}(\mathrm{Stab}_{S}\bigl(\bigl[\mathcal F'_1\bigr]\bigr)|_{\mathcal F'_2})-\dim (N_{\mathsf j_-}|_{\mathcal F'_2})_{\mathrm{moving}}-\dim (T_\pi|_{\mathcal F'_2})_{\mathrm{moving}}\\
 &<\dim N^-_{\mathsf X_S/{\mathcal F}_{2,S}}-\dim (N_{\mathsf j_-}|_{\mathcal F'_2})_{\mathrm{moving}}-\dim (T_\pi|_{\mathcal F'_2})_{\mathrm{moving}}\\
 &=\dim N^-_{\mathsf X/{\mathcal F}_2}.
\end{align*}
This finishes the proof.
\end{proof}

Propositions~\ref{prop: K-theory reduction} and \ref{prop: cohomology reduction} are summarized in the following
\begin{cor}\label{cor:k-theoretic stable envelope from the elliptic one}
Under the assumption of Proposition~{\rm\ref{prop: K-theory reduction}}, $\msf{Stab}^{\mathsf{s}}_{\mfk{C}}$ is the K-theoretic stable envelope with slope $\msf s$ for the chamber $\mathfrak C$, and $\mathrm{Stab}_{\mfk{C}}$ is the cohomological stable envelope for the chamber~$\mathfrak C$.
\end{cor}

\begin{ex}[Resolved Determinantal Varieties]\label{sec:k-theory and cohomology limit for resolved determinantal variety} We apply \eqref{eq:k-theoretic stable envelope as the limit of elliptic one} to \eqref{eqn: stable envelope for M(N,L)} and obtain the K-theoretic stable envelope for the resolved determinantal variety\footnote{Here we use the identities $\lim_{q\to 0}\vartheta(w)=w^{\frac{1}{2}}-w^{-\frac{1}{2}}$ and $\lim_{q\to 0}\frac{\vartheta(wz)}{\vartheta(z)}=w^{\lfloor\msf{s}\rfloor+\frac{1}{2}}$.}
\begin{align}\label{eqn: K-theoretic stable envelope for M(N,L)}
 \msf{Stab}^{\mathsf{s}}_{\mathfrak C}([\mathcal F_p])=\hbar^{\frac{\#(i>p(a))}{2}}\tenofo{Sym}_{S_N}\left[\left(\prod_{a=1}^N\tfunction_{p(a)}(s_a,\mbs{x},\hbar,\mbs{s})\right)
 \cdot\Biggl(\prod_{\substack{a>b}}\hat{\mathbf{a}}\bigl(s_as_b^{-1}\bigr)\Biggr)\right],
\end{align}
where $\tfunction_m(s,\mbs{x},\hbar,\msf{s})$ is the following function
\[
\tfunction_m(s,\mbs{x},\hbar,\msf{s}):=(sx_{m})^{\lfloor{\msf{s}}\rfloor}\prod_{i<m}\bigl(1-s^{-1}x_i^{-1}\bigr)\prod_{j>m}\bigl(1-\hbar^{-1} s^{-1}x_j^{-1}\bigr),
\]
and $\hat{\mathbf a}$ is a version of $\hat{A}$-genus
\smash{$\hat{\mathbf a}(w)=\frac{1}{w^{\frac{1}{2}}-w^{-\frac{1}{2}}}$}.
A further reduction to cohomology using \eqref{eqn: reduction to cohomology} gives the cohomological stable envelope
\begin{align}\label{eqn: cohomological stable envelope for M(N,L)}
 \mathrm{Stab}_{\mathfrak C}([\mathcal F_p])=\tenofo{Sym}_{S_N}\left[\left(\prod_{a=1}^N\rfunction_{p(a)}(s_a,\mbs{x},\hbar)\right)\cdot\Biggl(\prod_{\substack{a>b}}\frac{1}{s_a-s_b}\Biggr)\right],
\end{align}
where $\rfunction_m(s,\mbs{x},\hbar)$ is the following function
\smash{$
\rfunction_m(s,\mbs{x},\hbar):=\prod_{i<m}(s+x_i)\prod_{j>m}(s+x_j+\hbar)$}.
Note that \eqref{eqn: cohomological stable envelope for M(N,L)} equals to \smash{$(-1)^{\#(i>p(a))}W^{(10)}_{w_0,I}(-\mbs{s},\mbs{x},-\hbar)$}, where \smash{$W^{(10)}_{w_0,I}$} is the super weight function in \cite{RimanyiRozansky202105} for the longest element $w_0\in S_N$ and $I=w_0(\mathrm{image}(p))$. The sign $(-1)^{\#(i>p(a))}$, which equals to $(-1)^{\operatorname{rk} \mathsf{ind}}$, reflects the choice of normalization of the stable envelope in the \emph{loc.\ cit}.\ being~${ \mathrm{Stab}_{\mathfrak C}([\mathcal F])|_{\mathcal F}=e(N^-_{\mathsf X/{\mathcal F}})}$ as opposed to $(-1)^{\operatorname{rk} \mathsf{ind}}e(N^-_{\mathsf X/{\mathcal F}})$, see Axiom~A1 in Definition~4.1 of loc.\ cit.
\end{ex}

\part[The solution to dYBE for $\mfk{sl}(1|1)$ from 3d $\mcal{N}=2$ SQCD]{The solution to dYBE for $\boldsymbol{\mfk{sl}(1|1)}$\\ from 3d $\boldsymbol{\mcal{N}=2}$ SQCD} \label{part:physics}

\section[Stable envelopes and $R$-matrix from gauge theory: the setup]{Stable envelopes and $\boldsymbol{R}$-matrix from gauge theory: the setup}
\label{sec:the setup}

In the previous sections, we have defined the elliptic stable envelope for certain non symplectic varieties. Examples of such varieties include Higgs branches of 3d $\cN=2$ theories. We have characterized the stable envelopes by their formal properties and using these stable envelopes, we defined the $R$-matrix satisfying the dynamical Yang--Baxter equations for superspin chains. In practice, it is rather difficult to construct stable envelopes from their definitions alone. To write down explicit formulas for these elliptic functions, we turn to the 3d $\cN=2$ theories. In this section, we explain the motivation behind computing certain partition functions of these theories that we argue should give us stable envelopes. Finally, in Section~\ref{sec:stable envelopes and the $R$-matrix from gauge theory} when we actually compute these partition functions, we shall give an a posteriori justification for the computation by showing that the partition functions satisfy the criteria for elliptic stable envelopes set up earlier and that the $R$-matrix satisfies the dynamical Yang--Baxter equations. In fact, results from these computations were already used to define the functional forms of the elliptic $\sl(1|1)$ stable envelopes \rf{eqn: stable envelope for M(N,L)} and the $R$-matrix \rf{eqn: quiver $R$-matrix for sl(1|1)} and we checked that they indeed fulfill all the necessary criteria. Let us now go through the details of how we actually computed them.

\subsection{The Bethe/Gauge correspondence}\label{sec:bethe-gauge correspondence}
The correspondence relates integrable spin chains to supersymmetric quiver gauge theories, for both purely bosonic \cite{NekrasovShatashvili200901c,NekrasovShatashvili200901b} and superspin \cite{Nekrasov201811} chains. One way to introduce this is from a~correspondence between Kac--Dynkin diagrams and quivers.

 An integrable spin chain is based on a Lie algebra, say $\sl(m|n)$, which can be characterized by its Kac--Dynkin diagram. For example, the following is a Kac--Dynkin diagram for $\sl(1|3)$
\[
 \begin{tikzpicture}[x={(1.5cm,0cm)}, y={(0cm,1.5cm)}]
 \node[draw, circle] at (2,0) (n3) {};
 \node[draw, circle] at (1,0) (n2) {};
 \node[draw, circle, cross] at (0,0) (n1) {};

 \draw[-] (n1) -- (n2);
 \draw[-] (n2) -- (n3);
 \end{tikzpicture} .
\label{dynkin31}
\]
A hollow circle represents an even simple root and a crossed circle represents an odd simple root. In order to fully characterize a spin chain, in addition to the symmetry, we also need to specify the representations or the spins that appear in the chain and the anisotropy. For simplicity, let us only consider spin chains with highest-weight representations, with the highest-weight being the fundamental weight, then we just need to specify the number of sites in the spin chain. For the anisotropy, we have three choices (see Table~\ref{table:spinchains}), fully isotropic (XXX/rational), partially anisotropic (XXZ/trigonometric/K-theoretic), and fully anisotropic (XYZ/elliptic).

The Bethe/Gauge correspondence relates this data to a quiver such as the following:\footnote{In the hierarchy of quivers, the first one is just a Kac--Dynkin diagram (e.g., \rf{an example of Kac--Dynkin diagram}), then there's the framed and doubled quiver (e.g., \rf{an example of type A quiver}), and finally the one in \rf{quiver31} is a tripled quiver that is traditionally used in physics to define gauge theories with 4 supercharges.}
\beq
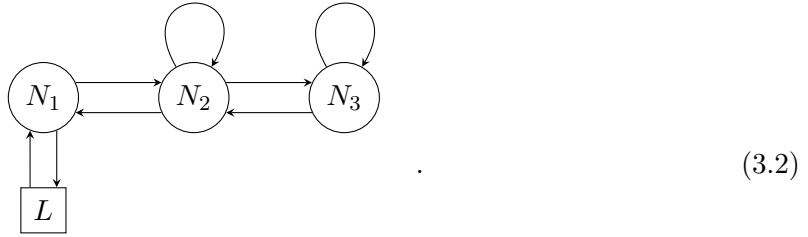
\begin{tikzpicture}[x={(2cm,0cm)}, y={(0cm,2cm)}, baseline=-1cm]
\node[draw, circle] at (0,0) (n1) {$N_1$};
\node[draw, circle] at (1,0) (n2) {$N_2$};
\node[draw, circle] at (2,0) (n3) {$N_3$};

\draw[-stealth] (n1.25) to (n2.155);
\draw[-stealth] (n2.25) to (n3.155);
\draw[-stealth] (n3.205) to (n2.335);
\draw[-stealth] (n2.205) to (n1.335);

\draw[-stealth] (n2.120) to [out=120, in=60, looseness=7] (n2.60);
\draw[-stealth] (n3.120) to [out=120, in=60, looseness=7] (n3.60);

\node[draw, rectangle, minimum size=.6cm] at (0,-.75) (f1) {$L$};

\draw[-stealth] (n1.292) to (f1.60);
\draw[-stealth] (f1.120) to (n1.248);
\end{tikzpicture}
\label{quiver31}
\eeq
Each circular/gauge node corresponds to a node in the Kac--Dynkin diagram~-- even nodes are distinguished from the odd ones by the presence of a self-loop. The presence of a single framing node attached to the left-most gauge node reflects our choice regarding the spin chain only having the fundamental highest weights. The value $L$ of the framing node is the size of the spin chain. The gauge dimensions $N_1$, $N_2$ and $N_3$ are not part of the data defining the spin chain. Rather, they correspond to a sector of the spin chain with definite magnon numbers. To find the gauge theory dual of the full spin chain we need to consider multiple quivers, varying the gauge dimensions over all possible magnon numbers. The choice of anisotropy is not part of the quiver itself, but it will be reflected in the choice of gauge theory that we shall assign to the quiver. The quiver \rf{quiver31} is a special case of Example~\ref{ex: quiver with one f node 2} with $n=3$.

 To quivers like \rf{quiver31}, we can assign supersymmetric gauge theories with 4 supercharges in various dimensions. In the version of the Bethe/Gauge correspondence in which we are interested in this paper, the relation between anisotropy and gauge theory is as follows:
\beq
 \begin{tabular}{rl}
 Rational & 1d $\cN=4$ (quantum mechanics), \\
 Trigonometric & 2d $\cN=(2,2)$, \\
 Elliptic & 3d $\cN=2$.
\end{tabular} \label{B/G}
\eeq
In the absence of any odd node, the amount of supersymmetry is doubled in all cases. All these theories possess the same classical Higgs branch $\higgs$, which is the quiver variety $\cM^\ze(\mbf v, \mbf w)$~\rf{HiggsSyQ} with $\mbf v = (N_1, N_2, N_3)$ and $\mbf w = (L, 0, 0)$. The stability condition $\ze$ corresponds to the real FI parameter in gauge theories. Supersymmetric ground states of these theories correspond to different cohomologies of the Higgs branch. For quantum mechanics, Witten's Morse theory arguments \cite{Witten198204} describe the space of ground states as the ordinary equivariant cohomology~$H_{\msf A}(\higgs)$ where $\msf A$ is the maximal torus of the $\SU(L)$ flavor symmetry of the quiver. In the 2d and 3d cases the cohomology is replaced by K-theory $K_{\msf A}(\higgs)$ and elliptic cohomology $\text{Ell}_{\msf A}(\higgs)$ respectively \cite{Witten198604, Witten1988}. On the spin chain side, for the rational, trigonometric, and elliptic cases, we have the Yangian $\msf Y_\hbar(\mfr g)$, the quantum affine algebra $\mcr U_\hbar(\what{ \mfr g})$, and the elliptic dynamical quantum group $E_{\tau, \hbar}(\mfr g)$ as the spectrum generating algebras. The Bethe/Gauge correspondence identifies the Bethe eigenstates of the spin chain with the space of supersymmetric vacua of the corresponding gauge theories. Consequently, it conjectures an action of the spectrum generating algebra on the corresponding cohomology \cite{FelderRimanyiTarasov201702,MaulikOkounkov201211, Okounkov201512}.\footnote{Since a single gauge theory corresponds to a specific magnon sector, we need to take direct sum of the cohomologies over gauge dimensions to get the full spin chain Hilbert space and the algebra action (cf.\ the sums over $\mbf v$ in Conjecture \ref{conj: algebra action}).}

There are more choices of gauge theories that we can assign to quivers such as \rf{quiver31} and find dual gauge theoretic interpretations of spin chain quantities. For example, in the original formulation of the Bethe/Gauge correspondence \cite{Nekrasov201811, NekrasovShatashvili200901c,NekrasovShatashvili200901b}, for both bosonic and supersymmetric spin chains, the gauge dual of the rational spin chain was presented as a 2d $\cN=(4,4)$ (bosonic spin chain) or an $\cN=(2,2)$ (superspin chain) theory. The quantum mechanics from \rf{B/G} is simply the dimensional reduction of these 2d theories. The supersymmetric vacua of these 2d theories correspond to the quantum deformation of the equivariant cohomology. In principle, there should be 3d and 4d theories with 4 supercharges that quantize K-theory \cite{Jockers:2018sfl, Jockers:2021omw,Kapustin:2013hpk, Okounkov201512} and elliptic cohomology, and whose dimensional reductions correspond to the 2d and 3d theories in \rf{B/G}. In this paper we shall not be concerned with quantum cohomology, regardless, it is of some interest to look at the 2d dual of the rational spin chains. We consider the example of a~$\sl(1|1)$ spin chain in the following.

\begin{ex}[$\sl(1|1)$ spin chains and 2d $\cN=(2,2)$ Gauge theories] \label{ex:2dgl11}
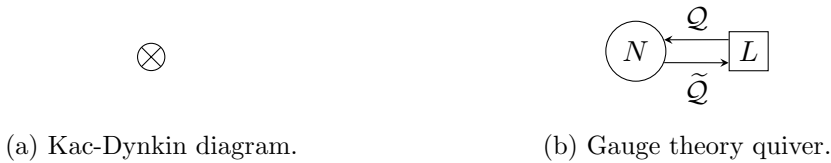
\begin{figure}[H]
\centering
\begin{subfigure}[b]{.45\textwidth}
\centering
\begin{tikzpicture}[x={(1.5cm,0cm)}, y={(0cm,1.5cm)}, baseline=-.75cm]
 \node[draw, circle, cross] at (0,0) {};
 \end{tikzpicture}
 \caption{Kac--Dynkin diagram.}\label{fig:gl11Dynkin}
\end{subfigure}
\begin{subfigure}[b]{.45\textwidth}
\centering
	\begin{tikzpicture}[x={(1.5cm,0cm)}, y={(0cm,1.5cm)}]
 	\node[draw, circle] at (0,0) (N) {$N$};
 	\node[draw, rectangle] at (1,0) (L) {$L$};
 	\draw[-stealth] (L.150) to node[midway, above] {$\cQ$} (N.22);
 	\draw[-stealth] (N.-22) to node[midway, below] {$\wtd \cQ$} (L.210);
 	\end{tikzpicture}
\caption{Gauge theory quiver.}\label{fig:gl11Quiver}
\end{subfigure}
\caption{$\sl(1|1)$ Kac--Dynkin diagram and the corresponding quiver.}
\end{figure}

Consider a length $L$ spin chain whose sites are labeled by $i=1,\dots,L$. At the $i$th site we assign the $\sl(1|1)$ representation with the highest-weight \smash{$\bigl(\la^{(1)}_i, \la^{(2)}_i\bigr)$}. We are using the $\gl(1|1)$ notation to write down the weights, i.e., \smash{$\la^{(1)}_i + \la^{(1)}_i$} is the $\sl(1|1)$ weight of the $i$th site. We also consider complex inhomogeneities~$ v_i$. Excitations containing $N$ magnons with rapidities $\si^\bC_a$ for $a=1,\dots,N$ correspond to Bethe eigenstates if the rapidities satisfy the Bethe ansatz equations (BAE) \cite{RagoucySatta200706} (also see \cite{Volin201012})
\begin{equation}\label{eq:Bethe equations for gl(1|1)}
 \prod_{i=1}^L\frac{+\sigma^\mbb{C}_a+\frac{\mfk{i}}{2}\ep - v_i+\mfk{i} \ep \lambda_i^{(1)}}{-\sigma_a^\mbb{C}-\frac{\mfk{i}}{2}\ep + v_i+\mfk{i} \ep \lambda_i^{(2)}}={\rm e}^{\mfk{i}\varphi}, \qquad a=1,\dots,N,
\end{equation}
The inhomogeneities and the phase on the right-hand side are arbitrary auxiliary parameters of the spin chain. $\ep$ is the quantization parameter of the spin chain.\footnote{Usually the quantization parameter in a quantum mechanical system is written as $\hbar$, however in this paper, we are using $\hbar$ for the elliptic deformation parameter which is roughly related to $\ep$ by $\ep \sim \ln \hbar$.} Note that the value of~\smash{$\la_i^{(1)} - \la_i^{(2)}$} is somewhat ambiguous since it can be changed by shifting the (arbitrary) inhomogeneity $ v_i$. The $\sl(1|1)$ charge \smash{$\la^{(1)}_i + \la^{(2)}_i$} is invariant under such shifts.

The BAE is the condition of extremizing a potential function called the Yang--Yang function
 \begin{align}
 Y\bigl(\mbs{\sigma}^\mbb{C}\bigr)={}&\sum_{a=1}^N\sum_{i=1}^L\bigl(\sigma^\mbb{C}_a-m^\mbb{C}_i\bigr)\bigl(\ln\bigl(\sigma^\mbb{C}_a-m^\mbb{C}_i\bigr)-1\bigr)\nonumber
 \\
 &+\sum_{a=1}^N\sum_{i=1}^L\bigl(-\sigma^\mbb{C}_a-\wt{m}^\mbb{C}_i\bigr)\bigl(\ln\bigl(-\sigma^\mbb{C}_a-\wt{m}^\mbb{C}_i\bigr)-1\bigr)-\mfk{i}\varphi\sum_{a=1}^N\sigma^\mbb{C}_a ,\label{eq:the Yang--Yang function for gl(1|1)}
 \end{align}
where we have defined the new parameters
\beq
	m_i^\mbb{C}:= v_i-\mfk{i} \ep \lambda_i^{(1)}-\frac{\mfk{i}}{2} \ep, \qquad
	 \wt{m}_i^\mbb{C}:= - v_i-\mfk{i} \ep \lambda_i^{(2)}+\frac{\mfk{i}}{2} \ep. \label{2dtwistedmass}
\eeq
In terms of this function, the BAE becomes:
\smash{$\exp \bigl(2\pi \mfr i \frac{\partial Y(\bm\si^\bC)}{\partial \si^\bC_a}\bigr)\! =\! 1$}, $ a \!=\! 1,\dots, N$.
The Bethe/Gauge correspondence identifies the Yang--Yang function \rf{eq:Bethe equations for gl(1|1)} as the effective twisted superpotential of the 2d $\cN=(2,2)$ theory defined by the quiver Figure \ref{fig:gl11Quiver} in the presence of twisted masses~\cite{Nekrasov201811}. We identify the rapidities $\si^\bC_a$ with the adjoint scalars in the Cartan part of the $\cN=(2,2)$ $\U(N)$ vector multiplet. Then we recognize the 1-loop contributions of the fundamental and the anti-fundamental chirals $\cQ^i_a$, $\wtd \cQ^a_i$ in the first and the second line of \rf{eq:the Yang--Yang function for gl(1|1)} \cite{HoriKatzKlemmPandharipandeThomasVafaVakilZaslow2003,Witten199301}. Masses of these chirals are precisely the parameters defined in \rf{2dtwistedmass}. We also notice the microscopic linear twisted superpotential as the last term in the second line.

 For positive real FI parameter, the 2d theory flows in the infrared to a phase where it is described by a sigma model into its Higgs branch \cite{Witten199301}. The Higgs branch is acted on by $\SU(L)$ and the fixed points of the $\U(1)^{L-1}$ action correspond to the vacua of the theory. When put on an interval, such as $I \times S^1$ for some interval $I$, the theory has brane-type boundary conditions supported on the attracting and repelling sets of these fixed points. All this Higgs branch structure can just as well be studied by looking at the dimensionally reduced $\cN=4$ quantum mechanics (from \rf{B/G}) on $I$ using Morse theory \cite{Witten198204}.
\end{ex}

\subsection[3d $\cN=2$ SQCD and its parameters]{3d $\boldsymbol{\cN=2}$ SQCD and its parameters}\label{sec:3d N=2 SQCD and its parameters}

The 3d theory assigned to the quiver Figure~\ref{fig:gl11Quiver} is the $\cN=2$ SQCD. The theory has a gauge group
$G := \U(N)$ with Lie algebra $\mfr g := \mfr u(N)$.
The flavor symmetry visible in the quiver is~$F_{\mathsf A} := \SU(L)$ with Lie algebra $\mfr f_{\msf A} := \su(L)$.
There is an additional $\U(1)_\hbar$ symmetry with Lie algebra $\bR_\hbar$. The total flavor symmetry we consider is
$F := \SU(L) \times \U(1)_\hbar$,
whose Lie algebra we denote by $\mfr f$.

 Let $\tr A$ be the component of the gauge field valued in the center of $\mfr u(N)$. From this abelian gauge field, we can form a current $\star \tr F$ in 3d which is conserved by Bianchi identity. This gives rise to an abelian topological symmetry $\U(1)_\text{top}$. It is topological in the sense that only the monopoles and none of the fields in the Lagrangian are coupled to this current. The Lie algebra of $\U(1)_\text{top}$ is written as $\bR_\text{top}$.
This global symmetry can be weakly gauged by introducing a~background vector multiplet for it. Let $A_\text{top}$, $\si_\text{top}$ and $\sfD_\text{top}$ be the gauge field, the scalar, and the auxiliary scalar in this vector multiplet. The purely bosonic part of the Lagrangian coupling this multiplet to the topological current is
\beq
 A_\text{top} \wedge \tr F + \si_\text{top} \tr \sfD + \sfD_\text{top} \tr \si . \label{topGauge}
\eeq
The theory has a $\mfr g$-valued vector multiplet $\si$ and chiral multiplets\footnote{We refer to a multiplet by its scalar component. Their full field content and our conventions about supersymmetry are presented in Section~\ref{sec:3dN=2susy}.}
\begin{gather}
	\bigl(\cQ, \wtd \cQ\bigr) \in \cR \oplus \cR^\vee ,\qquad
	\text{where} \qquad \cR = \Hom\bigl(\bC^L, \bC^N\bigr) , \qquad \cR^\vee = \Hom\bigl(\bC^N, \bC^L\bigr) .\label{chiralRep}
\end{gather}
We also define $\mbf M := \cR \oplus \cR^\vee$ to refer to the space of all the chirals. The chiral multiplets are accompanied by their conjugate anti-chiral multiplets \smash{$\bigl(\ov \cQ, \ov{\wtd \cQ}\bigr) \in \cR^\vee \oplus \cR$}. Charges of these multiplets under the groups defined above are as follows
\begin{table}[H]\renewcommand{\arraystretch}{1.2}
$$
\begin{array}{|c|cccc|}
\hline
& \U(N) & \SU(L) & \U(1)_\hbar & \U(1)_\text{top} \\
\hline
\si & \mbf{ad} & \mbf 1 & 0 & 0 \\
\cQ & \mbf N & \mbf L^\vee & \frac{1}{2} & 0 \\
\wtd \cQ & \mbf N^\vee & \mbf L & \frac{1}{2} & 0 \\
\hline
\end{array}
\nn
$$
\caption{Charges of 3d $\cN=2$ multiplets in the quiver gauge theory of Figure~\ref{fig:gl11Quiver}. $\mbf{N}$ and $\mbf L$ refer to the fundamental representations of $\U(N)$ and $\SU(L)$ respectively.
}\label{tab:charges of supersymmetric multiplets of the non-Abelian gauge theories}
\end{table}
There is also a $\U(1)_R$ $R$-symmetry but charges of various multiplets under this symmetry are a bit ambiguous since the symmetry can be mixed with any other $\U(1)$ global symmetry. Also, we do not turn on any background for the $R$-symmetry.

We put the 3d theory on a Euclidean space-time manifold $[y_-, y_+] \times \bE_\tau $.
Here $\bE_\tau$ is the elliptic curve with complex structure $\tau$, defined as the quotient
$\bE_\tau = \bC^\times/q^\bZ $, $ q := {\rm e}^{2\pi \ii \tau}$ .
For the global symmetries $\SU(L) \times \U(1)_\hbar \times \U(1)_\text{top}$ we turn on background vector multiplets. The multiplets contain, among other fields, flat connections.\footnote{The connections need to be flat to preserve supersymmetry (cf.\ \rf{Fzzbar}). The full constraints on the background fields follow from the BPS equations of Section~\ref{sec:BPSeq}.} These flat connections are parameterized by their holonomies around the non-contractible cycles, $S^1_\A$ and $S^1_\B$, of $\bE_\tau$. For example, if $A_{\mfr a}$ is the background connection for the $F_{\mathsf A} =\SU(L)$ flavor symmetry, then we can compute the~${\mfr a_\bC = \su(L) \otimes \bC}$ valued periods
\smash{$a_{\mfr a} := \oint_{S^1_\B} A_{\mfr a} - \tau \oint_{S^1_\A} A_{\mfr a}$}.
We can always take the background fields to be Cartan valued
\[
a_{\mfr a} = \text{diag}((a_{\mfr a})_1, \dots, (a_{\mfr a})_L) , \qquad (a_{\mfr a})_1 + \cdots + (a_{\mfr a})_L = 0.
\]
Once we quotient by equivalences due to gauge transformations, these periods are really valued in the torus~\smash{$a_{\mfr a} \in \frac{\mfr a_\bC}{\La^\vee_{\mfr a} \oplus \tau \La^\vee_{\mfr a}}$},
where $\La^\vee_{\mfr a}$ is the coroot lattice of $\mfr a$. Consequently, the holonomies of this connection around the cycles of the elliptic curve, also called the elliptic equivariant parameters, are valued in the elliptic curve as well
\[
	(x_i, \dots, x_L) := \bigl({\rm e}^{2\pi \ii \tau (a_{\mfr a})_1}, \dots, {\rm e}^{2\pi \ii \tau (a_{\mfr a})_L} \bigr) \in \bE^{L}_\tau , \qquad i = 1, \dots, L .
\]
The tracelessness condition translates to $x_1 \cdots x_L = 1$ for the equivariant parameters. Another common term for these parameters is fugacity, which we shall also use.

 In a completely analogous fashion, we turn on elliptic fugacities for the flavor symmetries~$\U(1)_\hbar$,~$\U(1)_\text{top}$, as well as for the gauge symmetry $\U(N)$
$$\renewcommand{\arraystretch}{1.2}
\begin{array}{|c|c|}
\hline
\text{Symmetry} & \text{Fugacities} \\
\hline
\U(1)_\hbar & \hbar \in \bE_\tau \\
\U(1)_\text{top} & z \in \bE_\tau \\
\U(N) & (s_1, \dots, s_N) \in \bE^N_\tau \\
\hline
\end{array} .
$$
The values of the flavor fugacities are constant along the interval $[y_-, y_+]$ due to BPS equations~\rf{flowAzbar}. The values of the gauge fugacities along the interval are unconstrained and only their values at the boundaries with fixed holonomy are parameters of the theory. These fugacities are the same equivariant parameters as in Example~\ref{ex: stab for M(N,L)}.

 In addition to the holonomies, we also turn on $\SU(L)$ twisted masses, by giving nonzero VEV to the adjoint scalar in the background $\SU(L)$ vector multiplet
$\lan \si_{\mfr a} \ran = m = \text{diag}(m_1, \dots, m_L)$,
as well as $\U(1)_\text{top}$ twisted mass
$\lan \si_\text{top} \ran = \ze $.
Note that, due to the coupling \rf{topGauge} between the topological and the dynamical gauge multiplets, the $\ze$ is nothing but the real FI parameter. The absolute values of these mass parameters do not play any role in this paper. However, the relative ordering of their values matter. We fix the sign of the FI parameter once and for all localization computations
$\ze > 0$.
For the $\SU(L)$-masses, we shall refer to different chambers at different times. A chamber simply refers to a particular ordering of the masses. For any permutation $\varsigma\colon \{1,\dots, L\} \hookrightarrow \{1, \dots, L\}$ we have a chamber
\beq
	\mfr C_\varsigma = \{(m_1, \dots, m_L) \mid m_{\varsigma(1)} < m_{\varsigma(2)} < \cdots < m_{\varsigma(L)}\} \subseteq \mfr a .
\label{chamberDef}\eeq
These are the same chambers as in \rf{chamberDefM}. Two arbitrary but specific chambers defined by $i \mapsto i$ and $i \mapsto L+1-i$, reintroduced in \rf{C01}, will be used for explicit computations of stable envelopes and the $R$-matrix in Section~\ref{sec:stable envelopes and the $R$-matrix from gauge theory}.

{\bf Reduction to 2d.} In Section~\ref{sec:2d and 1d avatars of elliptic stable envelopes}, we shall consider dimensional reduction of the 3d partition functions down to 2d and 1d. The 3d $\rightarrow$ 2d reduction is done by compactifying the $S^1_\B$ cycle of the elliptic curve. Then the period of the connections along the $S^1_\B$ cycle becomes a real adjoint scalar field in the 2d theory. Combined with the real adjoint scalar already in the 3d vector multiplet, this becomes the complex adjoint scalar of the 2d $\cN=(2,2)$ vector multiplet. For example, for the $F_{\msf A}$ vector multiplet we define the following complex scalar
\begin{equation}
 \si_{\mfr a}^\bC := \si_{\mfr a} + \ii \bigointsss_{S^1_\B} A_{\mfr a} . \label{complexScalar}
\end{equation}
We use $(\si_{\mfr g})^\bC$, $(\si_\hbar)^\bC$, and $\ze^\bC$ for the similarly defined complex adjoint scalars in the dynamical vector multiplet, in the $\U(1)_\hbar$ vector multiplet, and the complex FI parameter in 2d respectively. A further reduction on $S^1_\A$ gives an $\mcal{N}=4$ quantum mechanics.

\subsection{Branes and the Bethe/Gauge correspondence}
In this section, we are going to look at the brane construction of the 3d quiver theory and relations between certain string dualities and the Bethe/Gauge correspondence. The goal is to find hints regarding which gauge theoretic quantities may correspond to stable envelopes and the $R$-matrix. The brane constructions will serve merely as a motivation to set up the purely gauge theoretic computations of the future sections, we will be rather schematic in the present section.\looseness=-1

 Gauge theories with 4 supercharges have standard brane constructions in terms of D-branes suspended between rotated NS5s \cite{deBoerHoriOzYin, Elitzur:1997fh, Elitzur:1997hc}. We consider a 10d space-time $\bR_y \times \bR_1 \times \bE_\tau \times \bR_2 \times \bR_3 \times \bR^2_{+\ep} \times \bR^2_{-\ep}$. Here $\bR_y$ is a real direction parameterized by the coordinate $y$ and $\bR_1$, $\bR_2$, $\bR_3$ are three real directions parameterized by some unspecified coordinates. $\bR^2_{\pm \ep}$ are two $\Om$-deformed planes, $\ep$ being the deformation parameter. The 3d $\cN=2$ SQCD can be seen as the world-volume theory of a stack of $N$ D3 branes in this background. The complete brane configuration is in Table \ref{tab:brane configuation of the three-dimensional theory}, and we draw a cartoon of the branes in Figure~\ref{fig:SQCDbranes}.

\begin{table}[h]\centering\renewcommand{\arraystretch}{1.2}
$$
\begin{array}{|cc|ccccccc|}
\hline
\text{No.} & & \bR_y & \bR_1 & \bE_\tau & \bR_2 & \bR_3 & \bR^2_{+\ep} & \bR^2_{-\ep} \\
\hline
1 & \text{NS5} & \times & \times & \times &&& \times & \\
1 & \text{NS5}' & \times & \times & \times &&&& \times \\
N & \text{D3} & \times && \times & \times &&& \\
L & \text{D5} & \times && \times && \times & \times & \\
\hline
\end{array}
$$
 \caption{Brane configuration for 3d $\cN=2$ SQCD with $\U(N)$ gauge group and $\SU(L)$ flavor (see Figure~\ref{fig:gl11Quiver}). Simultaneous rotation of the $\bR^2_{+\ep}$ and $\bR^2_{-\ep}$ planes is also a symmetry of this configuration. This corresponds to the $\U(1)_\hbar$ symmetry of the 3d theory (see Table \ref{tab:charges of supersymmetric multiplets of the non-Abelian gauge theories}), the two parameters being roughly related by $\ep \sim \ln \hbar$.}
 \label{tab:brane configuation of the three-dimensional theory}
\end{table}

\begin{figure}[h]
\centering
\begin{tikzpicture}[x={(1.25cm,0cm)}, y={(0cm,1.25cm)}]
 \draw[line width=1] (0,0) -- (0,3);

 \draw[dashed, line width=1] (3,0) -- (3,3);

 \foreach \y in {2, 2.1, 2.6, 2.7} {
 \draw (0,\y) -- (3,\y);
 }
 \node at (1.5, 2.4) {$\vdots$};
 \coordinate (midpoint) at (1.5,2.7);
 \node[above] at (midpoint) {};

 \node[cross, draw, circle, minimum size=5mm] at (1.5, 1.3) {};

 \node at (1.5, .9) (5dots) {$\vdots$};

 \node[cross, draw, circle, minimum size=5mm] at (1.5, .3) {};

 \node[above left=0cm] at (0, 0) {NS5};
 \node[above right=0cm] at (3,0) {NS5$'$};

 \node at (4, 1.25) (D5) {$L$ D5s};
 \node at (-1,3) (D3) {$N$ D3s};

 \draw[-stealth] (D5.180) to [out=180, in=0] ($(5dots.-30) + (.3cm,0)$);

 \draw[-stealth] (D3.20) to [out=10, in=120] ($(midpoint.120) + (-1cm, .2cm)$);

 \node at (6,1) (origin) {};
 \node at ($(origin)+(0,1)$) (up) {};
 \node at ($(origin)+(1,0)$) (right) {};
 \draw[-stealth] (6,1) -- (up);
 \draw[-stealth] (6,1) -- (right);
 \node[left=0] at (up) {$\bR_1$};
 \node[below=0] at (right) {$\bR_2$};
\end{tikzpicture}
\caption{A cross-section of the brane configuration from Table \ref{tab:brane configuation of the three-dimensional theory}.}
\label{fig:SQCDbranes}
\end{figure}
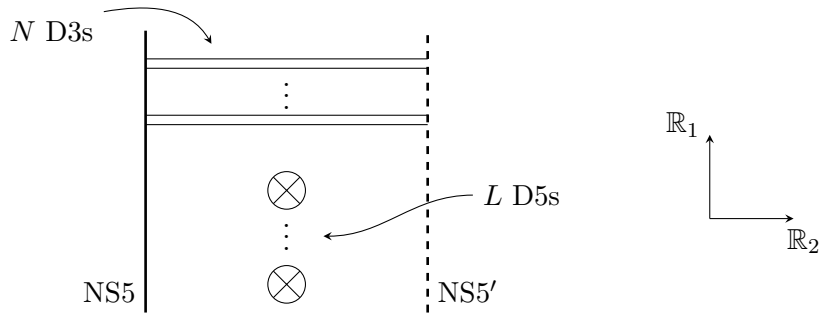

T-dualizing the elliptic curve and the $\bR_1$ direction takes us to a type IIA configuration with~$N$~D2 branes suspended between rotated NS5 branes, with $L$ D4 branes about. This is precisely the configuration studied in \cite{IshtiaqueMoosavianRaghavendranYagi202110}. In this duality frame, the world-volume theory of the~D2 branes is the 2d $\cN=(2,2)$ theory from Example~\ref{ex:2dgl11}. If we do not T-dualize $\bR_1$, then we find~D1 branes suspended between NS5s and additional D3 branes. The D1 branes are described the quantum mechanics from \rf{B/G} for which the D3 branes provide flavor symmetry. The main result of \cite{IshtiaqueMoosavianRaghavendranYagi202110} was to relate this setup to 4d Chern--Simons theory.\footnote{Also see \cite{CostelloYagi201810} for relations between $\Om$-deformed D5s and bosonic 4d Chern--Simons theory.} By a further S-duality we can convert the D1s, NS5s, and D3s to F1s, D5s, and D3s respectively. The $\Om$-deformed D5s give rise to a $\gl(1|1)$ 4d holomorphic-topological Chern--Simons theory on $\bR_y \times \bR_1 \times \bE_\tau$ where~${\bR_y \times \bR_1}$ is the topological plane and the elliptic curve provides the holomorphic direction. The D3 branes create line operators in this theory \cite{IshtiaqueMoosavianRaghavendranYagi202110, Ishtiaque:2022bck}. From \cite{CostelloWittenYamazaki201709} we know that crossing of line operators in the topological plane of 4d Chern--Simons creates $R$-matrix for certain spin chains. In the starting configuration of Table \ref{tab:brane configuation of the three-dimensional theory} this crossing corresponds to changing the $\bR_1$-positions of two~D5 branes as we move along in the $\bR_y$ direction, ultimately swapping the two (see Figure~\ref{fig:RD5}).
\begin{figure}[h]
\centering
\begin{tikzpicture}[x={(1cm,0cm)}, y={(0cm,1cm)}]
 \draw[line width=1, -stealth] (0,0) .. controls (0,2) and (2,1) .. (2,3);
 \draw[line width=1, -stealth] (2,0) .. controls (2,2) and (0,1) .. (0,3);
 \node[below=0] at (0,0) {$\ket{i}$};
 \node[below=0] at (2,0) {$\ket{j}$};
 \node[above=0] at (0,3) {$\bra{k}$};
 \node[above=0] at (2,3) {$\bra{l}$};

 \node[align=left] at (5.5,2) (text) {Line operators in 4d CS \\ or D5 branes from Table \ref{tab:brane configuation of the three-dimensional theory}};

 \draw[-stealth] (text) to (2.1,2.5);
 \draw[-stealth] (text) to (2.1,.5);

 \node at (-3.5,1) (origin) {};
 \node at ($(origin)+(0,1.25)$) (up) {};
 \node at ($(origin)+(1.25,0)$) (right) {};
 \draw[-stealth] (-3.5,1) -- (up);
 \draw[-stealth] (-3.5,1) -- (right);
 \node[left=0] at (up) {$\bR_y$};
 \node[below=0] at (right) {$\bR_1$};
\end{tikzpicture}
\caption{Line operators crossing in the topological plane of 4d Chern--Simons, computing the matrix element $\bra{k} \otimes \bra{l} R \ket{i} \otimes \ket{j}$ of an $R$-matrix. This translates to bending D5 branes (from Table \ref{tab:brane configuation of the three-dimensional theory}) in the $\bR_y \times \bR_1$ plane.} \label{fig:RD5}
\end{figure}
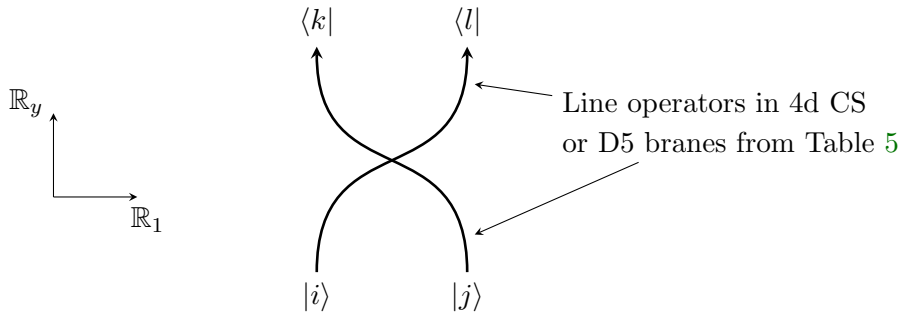
The incoming and outgoing asymptotic states correspond to states in the spin chain, or according to the Bethe/Gauge philosophy, vacua of the D3 brane theory. Without further analysis of the string background, it is unclear how exactly to bend branes while preserving supersymmetry or how to specify incoming and outgoing states. However, there is a natural supersymmetric process in the field theory language that accomplishes something analogous. The locations of the D5 branes in the $\bR_1$ direction are real twisted masses for the chiral multiplets of the SQCD that lives on the D3 branes. Swapping the $\bR_1$-locations, $m_1$ and $m_2$, of two D5 branes then correspond to giving $y$-dependent masses $m_1(y)$ and $m_2(y)$ such that
\beq
	\lim_{y \to -\infty} (m_1(y), m_2(y)) = (m_1, m_2) , \qquad \lim_{y \to +\infty} (m_1(y), m_2(y)) = (m_2, m_1) . \label{masschange}
\eeq
Such position-dependent masses can be made supersymmetric, creating a mass Janus interface \cite{Arav:2020obl, Bobev:2013yra, DHoker:2009lky, Gran:2008vx}. The interface interpolates between two distinct chambers, as defined in \rf{chamberDef}. In Section~\ref{sec:3d N=2 SQCD and half-bps boundary conditions}, we shall be more specific and identify suitable boundary conditions in the 3d theory, in any chamber, that mimics supersymmetric vacua, and define mass Janus interpolating between them.

\subsection{Gauge-theoretic definition of elliptic stable envelope}\label{sec:gauge-theoretic construction of elliptic stable envelope}

The gauge-theoretic definition of elliptic stable envelopes as Janus partition functions first appeared in \cite{BullimoreZhang202109,DedushenkoNekrasov202109}. They define the Janus interfaces for 3d $\cN=2$ supersymmetry but provide boundary conditions for 3d $\cN=4$ vacua, which give elliptic stable envelopes of 3d $\cN=4$ Higgs branches. We take our previous discussion on the brane construction and in particular the position dependent masses \rf{masschange} as motivations to study Janus interfaces. In this section, we briefly review the definition of the elliptic stable envelopes as Janus partition functions, and in Section~\ref{sec:3d N=2 SQCD and half-bps boundary conditions}, we generalize the boundary conditions to correspond to the vacua of 3d $\cN=2$ theories computing the elliptic stable envelopes for the 3d $\cN=2$ Higgs branches.

Our starting point here is the 3d $\cN=2$ SQCD on $[y_-, y_+] \times \bE_\tau$. For the moment, replace the interval $[y_-, y_+]$ with the infinite line $(-\infty, +\infty)$. Suppose, for any chamber $\mfr C$, we have a set of asymptotic boundary conditions $\{p\}$ for the 3d $\cN=2$ SQCD corresponding to its Higgs branch vacua.\footnote{In the presence of twisted masses, the classical Higgs branch of 3d $\cN=2$ SQCD retracts down to the fixed point set of the flavor torus action. This set is a disjoint union of connected components. Here $p$ labels a connected component, and by the boundary condition labeled by $p$ we roughly refer to a brane-type boundary supported on the attracting/repelling set of $p$. We shall describe in detail the fixed points, their attracting sets and the corresponding boundary conditions in Sections~\ref{sec:fixedPoints}, \ref{sec:MorseAndHamilton} and~\ref{sec:bc}, respectively.} We can construct a Janus interface interpolating between a vacuum $p_-$ in a chamber $\mfr C_-$ at infinite past and a vacuum $p_+$ in a chamber $\mfr C_+$ at infinite future. Then, our previous discussion on branes suggests that the expectation value of this interface with $p_-$ and $p_+$ boundary conditions at infinite past and future should give us a matrix element of an $R$-matrix
\[
 \mbf R_{\mfr C_+ \leftarrow \mfr C_-}(p_+, p_-) = \bra{p_+} \cJ(m_{\mfr C_-}, m_{\mfr C_+}) \ket{p_-} ,
\]
where $m_{\mfr C_-}$ and $m_{\mfr C_+}$ are two generic masses from the chambers $\mfr C_-$ and $\mfr C_+$ respectively, and~$\cJ$ is the Janus interface interpolating between these two masses. A similar expression for the $R$-matrix in terms of Janus interfaces was also used in \cite{BullimoreKimLukowski201708} in the context of bosonic rational spin chains.

 As already mentioned in the introduction (equations~\eqref{HXAtoHX}, \eqref{R=StabInvStab} and \eqref{RR...R}), a more fundamental quantity than the $R$-matrix is the stable envelope. While the $R$-matrix corresponds to changing masses across chambers, the stable envelope corresponds to changing masses from some nonzero value to zero. The full $R$-matrix is then a composition of a stable envelope and an inverse stable envelope \rf{eqn: dfn of R-matirx}, the first one taking the nonzero masses from one chamber to zero, the second one taking the zero masses to a nonzero value in a different chamber. In the gauge theory setup, a stable envelope then corresponds to a Janus interface interpolating between nonzero and zero masses. We expect that due to the presence of nonzero equivariant elliptic parameters, even in the massless side of the interface we can choose a basis of boundary conditions labeled by the same $p$ that labels the massive boundaries.\footnote{Note that, if we go down in dimension then the periods of the background connections become part of the twisted masses \rf{complexScalar}.} Then the stable envelope corresponds to the Janus partition function
\[
 \mbf{Stab}_{\mfr C}(p_+, p_-) = \bra{p_+} \cJ(m_{\mfr C}, 0) \ket{p_-} .
\]
We want to compute these partition functions using supersymmetric localization. In practice, localization on non-compact manifolds are technically subtle and it is more convenient to replace the infinite line with the finite interval $[y_-, y_+]$. Then we need to find boundary conditions for the boundaries at finite distances that are cohomologous to the asymptotic vacuum boundaries with respect to our choice of localizing supercharge. We describe our choice of supercharge in Appendix~\ref{sec:Qloc} and in Section~\ref{sec:3d N=2 SQCD and half-bps boundary conditions}, we shall identify the vacua and various boundary conditions equivalent to these vacua. For the massive boundary, we shall construct two types of boundary conditions, the so-called thimbles $\mcr D_{\mfr C}(p)$ (see Section~\ref{sec:thimble}) and the enriched Neumann $\mcr N_{\mfr C}(p)$ (see Section~\ref{sec:neumann}). For the massless boundary, we shall construct one type of boundary conditions depending on a choice of polarization $\bL(p)$ of the space of chirals $\mcr B_{\bL(p)}(p)$ (see Section~\ref{sec:masslessBC}). All these boundary conditions are labeled by the connected components of the (flavor) torus fixed-point set of the Higgs branch. We shall compute the matrix elements of the stable envelopes as the partition function
\beq
 \mbf{Stab}_{\mfr C}(p_+, p_-) = \bra{\mcr B_{\bL(p_+)}(p_+)} \cJ(m_{\mfr C}, 0) \ket{\mcr N_{\mfr C}(p_-)} .
 \label{eq:the gauge-theoretic definition of elliptic stable envelope}
\eeq
Here we have the boundary conditions $\mcr N_{\mfr C}(p_-)$ and $\mcr B_{\bL(p_+)}(p_+)$ at the two endpoints of the finite interval $[y_-, y_+]$ and a Janus interface interpolating between them. By keeping the masses constant throughout most of the interval and varying them rapidly near $y=0$ we can think of the Janus interface as being located at $y=0$. The setup of this computation is depicted in Figure~\ref{fig:the physical configuration for the computation of elliptic stable envelopes}.

\begin{figure}[h]
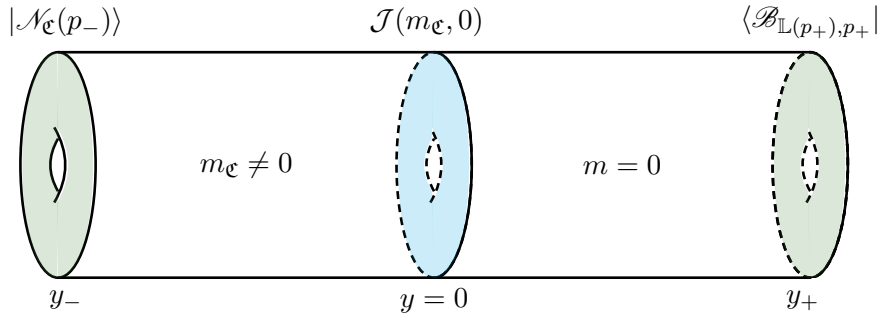

 \centering
 \SetupForComputationOfStableEnvelopes
 \caption{The physical setup for the computation of elliptic stable envelopes. We start with an enriched Neumann boundary condition $\mscr{N}_{\mfr C}(p_-)$ at the past boundary $y_-$. This boundary imitates the choice of the vacuum $p_-$ at $-\infty$. We insert a mass-Janus interface $\mcal{J}(m_\mfk{C},0)$ at~${y=0}$ whose role is to change the real masses from nonzero to zero as we cross over to $y>0$. We then choose an exceptional Dirichlet boundary condition $\mscr{B}_{\mbb{L}(p_+),p_+}$ at $y_+$. This future state is associated with the choice of a vacuum $p_+$ at $+\infty$ and a~polarization $\mbb{L}(p_+)$ of the space of the chirals. The interval partition function of this configuration gives the matrix elements \rf{eq:the gauge-theoretic definition of elliptic stable envelope} of elliptic stable envelopes.}
 \label{fig:the physical configuration for the computation of elliptic stable envelopes}
\end{figure}

Our choice of localizing supercharge makes the 3d theory topological in the $y$-direction \rf{y-exact}, and therefore the length of the interval $[y_-, y_+]$ will not be a parameter of the theory. In fact, we shall find an effective description of the theory as a 2d $\cN=(0,2)$ theory on $\bE_\tau$ -- the amount of supersymmetry reduced to half by the half BPS boundaries and interfaces. The 3d interval partition function will be given by the elliptic genus of this 2d theory.

Due to the presence of a mass-Janus interface along which real masses change, we have two different effective descriptions of the theory on the two sides of the interface. The theory living on $[0,y_+]$, where $m=0$, is the full gauge theory with the gauge Lie algebra $\mfr g = \mfr u(N)$. On the other hand, the theory living on $[y_-,0]$, where $m \ne 0$, is the gauge theory describing the effective theory of low energy (relative to $m$) excitations in the vacuum $p_-$. This theory depends on the choice of a chamber $\mfk{C}$ for the mass $m$, which in turn determines the phase of the theory. For generic masses, the adjoint scalar $\si$ from the vector multiplet takes a generic VEV at this boundary (according to \rf{siP}), and as a result gauge symmetry is broken down to its commutant, the maximal abelian subgroup $\mfr h = \bR^N$. %

\section[Boundaries and interfaces in 3d $\cN=2$ SQCD]{Boundaries and interfaces in 3d $\boldsymbol{\cN=2}$ SQCD}
\label{sec:3d N=2 SQCD and half-bps boundary conditions}

In this section, we look at the classical Higgs branch vacua of 3d $\cN=2$ SQCDs and $\cN=(0,2)$ boundary conditions that are cohomologous to these vacua. Our goal is to compute Janus partition functions \rf{eq:the gauge-theoretic definition of elliptic stable envelope} of the 3d theory on space-time manifolds of the form $[y_-, y_+] \times \bE_\tau$. We do so using supersymmetric localization using the supercharge \rf{Qloc}. The supercharge belongs to an $\cN=(0,2)$ subalgebra of the 3d $\cN=2$ supersymmetry algebra, which is why we shall look at boundary conditions preserving this subalgebra. We follow \cite{DedushenkoNekrasov202109} to find such boundary conditions that mimic the Higgs branch vacua of the theory. The main difference in our case is that, unlike in the 3d $\cN=4$ cases, we do not have a discrete set of completely massive vacua even in the presence of generic twisted masses. This requires us to slightly modify the notion of the so-called thimble boundaries so that they can be attached to extended subspaces of the Higgs branch.

\subsection[Vacua and BPS equations in 3d $\cN=2$ Gauge theories]{Vacua and BPS equations in 3d $\boldsymbol{\cN=2}$ Gauge theories}

\subsubsection{Vacuum equations} \label{sec:vacEq}
Consider a 3d $\cN=2$ Lagrangian gauge theory with gauge group $G$ and flavor symmetry group~$F$. We minimize its potential (consisting of bosonic fields only) to find the classical moduli of vacua.

 All terms of the action containing dynamical and background vector multiplet fields can be packaged together in terms of a single vector multiplet for the group $G' := G \times F$. We denote the fields of the $G'$-multiplet as $\bigl(A, \si, \la, \ov\la, \sfD\bigr)$. Part of the Lagrangian containing $A$, $\si$ and $\sfD$ only is
\beq
	\frac{1}{2} \tr\left(\frac{1}{2} F^{\mu\nu} F_{\mu\nu} + D^\mu \si D_\mu \si + \sfD^2 \right) - \ii \ze \tr\sfD . \label{Lvec}
\eeq
The first trace is from the Yang--Mills action and the second is the FI term. We are assuming that~$G$ has a $\U(1)$ center and so there is a single real FI parameter $\ze$. Note that $\ze$ corresponds to the background value of the adjoint scalar in the vector multiplet for the $\U(1)_\text{top}$ global symmetry~\rf{topGauge}.

 We consider chiral multiplets in the representation $\mbf M$ of the gauge group $G$ and anti-chiral multiplets in the contragradient representation $\mbf M^\vee$. Let us denote the gauge invariant pairing between $\mbf M$ and $\mbf M^\vee$ as $\lan\,,\,\ran$. Then, the part of the chiral multiplet Lagrangian containing only bosonic fields from the chiral and vector multiplets is
\beq
	-\lan \ov\phi, D^\mu D_\mu \phi \ran + \lan \si \cdot \ov\phi, \si \cdot \phi\ran + \ii \lan \ov\phi, \sfD \phi \ran + \lan \ov \sfF, \sfF \ran , \label{Lchi}
\eeq
The equations of motion for the auxiliary fields $\sfD$ and $\sfF$ are
\beq
	\sfD = \ii\ze - \ii \mu_{\mfr g} , \qquad \sfF = \ov\sfF = 0 , \label{auxEOM}
\eeq
where $\mu_{\mfr g} = \phi \ov\phi$ is the real moment map of the $G$-action on $\mbf M$. Eliminating the auxiliary fields from the Lagrangians \rf{Lvec} and \rf{Lchi} using the equations of motion we find the on-shell Lagrangian containing bosonic fields only
\[
	 \tr\left(\frac{1}{4} F^{\mu\nu} F_{\mu\nu} + \frac{1}{2} D^\mu \si D_\mu \si \right) - \lan \ov\phi , D^\mu D_\mu \phi \ran
	+ \tr(-\ze + \mu_{\mfr g} )^2 + \lan \si \cdot \ov\phi, \si \cdot \phi \ran .
\]
The derivative terms are the kinetic terms and the non-derivative terms constitute the potential that we want to minimize. The potential, being a sum of squares, is minimized when the individual terms are zero
\begin{subequations}\begin{align}
	&\mu_{\mfr g} - \ze = 0 , \label{momentEq} \\
	&(\si_{\mfr g} + \si_{\mfr f}) \cdot \phi = 0 . \label{fixedEq}
\end{align}\label{vacEq}\end{subequations}
Here $\si_{\mfr g} + \si_{\mfr f} = \si$ is the decomposition of $\si$ into its dynamical and background components.

 The classical Higgs branch of the theory is defined as the real symplectic quotient $\mu^{-1}_{\mfr g}(\ze)/G$. The first vacuum equation \rf{momentEq} is precisely this real moment map equation, and after implementing the $G$-quotient, the second equation \rf{fixedEq} becomes the criteria for fixed points under the infinitesimal action generated by $\si_{\mfr f}$. $\si_{\mfr f}$ will be some element from the Cartan subalgebra $\mfr t$ of the flavor symmetry group. Nonzero $\si=\si_{\mfr g} + \si_{\mfr f}$ generates mass terms for various chiral multiplets and W-bosons. Thus we come to the following characterization of the (classical) massive vacua of a 3d $\cN=2$ gauge theory: they correspond to the fixed point set in the Higgs branch under the action of the maximal torus of the flavor symmetry group.

\subsubsection{BPS equations} \label{sec:BPSeq}
We look for field configurations left invariant by the localizing supercharge $\thdQ$ \rf{Qloc}. We consider the 3d $\cN=2$ theory on the space-time $I \times \bE_\tau$, where $I$ is an interval with real coordinate $y$ and $\bE_\tau$ is an elliptic curve with holomorphic coordinate $z$. A supercharge of the theory is parameterized by 2 Dirac spinors $\ep$ and $\ov\ep$, as in \rf{Qee}, and our choice of $\mcr Q$ is determined by
\[
	\thdQ = \thdQ_{\ep, \ov\ep}, \qquad \text{where} \quad \ep = \ov\ep = \frac{1}{\sqrt{2}} \begin{pmatrix}
	 1 \\ 1
	\end{pmatrix} .
\]
We thus find the BPS equations by setting the supersymmetry variations for the fermions from the vector \rf{susyVec} and the chiral \rf{susyChiral} multiplets to zero for this specific choice of parameters. Doing so, we find the BPS equations~-- the ones containing derivatives along the elliptic curve~$\bE_\tau$%
\begin{subequations}\label{DzbarEq}\samepage
\begin{align}
	&F_{z \ov z} = 0 , \label{Fzzbar} \\
	&D_{\ov z} \si_{\mfr g} = D_{\ov z} \si_{\mfr f} = 0 , \label{Dzbarsi} \\
	& D_{\ov z} \phi = D_{\ov z} \ov\phi =0 , \label{Dzbarphi}
\end{align}
\end{subequations}
and the ones containing derivatives along $I$
\begin{subequations}\label{flowEq}
\begin{align}
	&D_y A_{\ov z} = 0 , \label{flowAzbar} \\
	&D_y \si_{\mfr g} + \ze + \mu_{\mfr g} = 0 , \label{flowsiG}\\
	&D_y \si_{\mfr f} + \ii \sfD_F = 0 , \label{flowsiF}\\
	&D_y \phi + (\si_{\mfr g} + \si_{\mfr f}) \cdot \phi = 0 , \label{flowphi}
\end{align}
\end{subequations}
along with the conjugate equation for $\ov\phi$. %
We have used the equations of motion \rf{auxEOM} to eliminate the auxiliary fields $\sfD_{\mfr g}$, $\sfF$ and $\ov\sfF$. The only constraint on the background auxiliary field $\sfD_F$ is~\rf{flowsiF} and we use this equation as the definition of this particular background. The equations~\rf{DzbarEq} determine the fields along the elliptic curve $\bE_\tau$ and the equations \rf{flowEq} determine their flow in the $y$-direction.

 The connections in the above equations all contain both dynamical and background fields for the gauge and the flavor symmetry groups respectively. According to \rf{Fzzbar} all these connections are flat and will be characterized by their holonomies around the cycles $S^1_A$ and $S^1_B$ of $\bE_\tau$. Equation \rf{flowAzbar} further implies that these holonomies are constant along the flow in the $y$-direction. $\si_{\mfr g}$ and $\si_{\mfr f}$ are holomorphic according to \rf{Dzbarsi} but since they are elements of real Lie algebras, they have to be constant on $\bE_\tau$. Similarly, $\phi$ and $\ov\phi$ are holomorphic according to \rf{Dzbarphi}. However, in the path integral we would like to impose the reality condition that the anti-chiral field $\ov\phi$ is the complex conjugate of the chiral field $\phi$. Then both $\phi$ and $\ov\phi$ have to be constant on $\bE_\tau$ as well. Having nonzero holonomies for the flat connections is only compatible with constant VEVs for $\phi$ if the holonomies act trivially on $\phi$
\beq
 s x^{-1} \hbar^{\frac{1}{2}} z \phi = \phi . \label{screenGeneric}
\eeq
Here $s$, $x$, $\hbar$ and $z$ are the group valued holonomies for the Cartan valued connections for the gauge group $G$, flavor groups $F_{\msf A}$ and $\U(1)_\hbar$, and the topological symmetry group $\U(1)_\text{top}$. None of the chirals are charged under $\U(1)_\text{top}$ so $z$ acts as identity on all chirals. The equation \rf{screenGeneric} is saying that the combined holonomies are invisible to the chirals, which is also referred to as a screening condition.

 There is no equation for $A_y$, which is gauge equivalent to the zero connection since the $y$-direction is contractible. We are now left with the flow equations for $\si_{\mfr g}$ \rf{flowsiG} and $\phi$ \rf{flowphi}, which we shall address in Section~\ref{sec:MorseAndHamilton}.

\subsection{Classical Higgs branch}\label{sec:higgs branch and fixed points for nonabelian gauge theory}
We consider the classical Higgs branch, where the chiral multiplet scalars have VEVs in the absence of twisted masses. The vacuum equation \rf{fixedEq} is then automatically satisfied because~${\si_{\mfr g} = \si_{\mfr f} = 0}$ and we find the classical Higgs branch to be (cf.\ \rf{HiggsSyQ})
\beq
	\higgs = \mu_{\mfr g}^{-1}(\ze)/G , \label{MH}
\eeq
where $G$ is the gauge group, $\mu_{\mfr g}$ is the real moment map of the $G$-action on the representation space $\mbf M$ of the chirals and $\ze$ is a positive central element in the gauge Lie algebra. The construction \rf{MH} is the real symplectic reduction of $\mbf M$ by $G$
$\higgs = \mbf M/\!\!/G $.
In 3d $\cN=2$, the space of chirals is always K\"ahler and so is its symplectic quotient by a real Lie group. If we are only interested in the complex structure of the Higgs branch, we can equivalently define it as a~GIT quotient by the complexified gauge group instead\footnote{This is how the Higgs branch was defined in \rf{MHgitDef}.}
\beq
	\higgs = \mbf M^s/G_\bC . \label{MHgit}
\eeq
Here $\mbf M^s$ is the stable locus of $\mbf M$, which consists of all points of $\mbf M$ whose $G_\bC$ orbits intersect~$\mu_{\mfr g}^{-1}(\ze)$.

Twisted masses are turned on by giving $\mfr t$-valued nonzero VEV to the adjoint scalar $\si_{\mfr f}$ in the flavor vector multiplet. We then have the additional vacuum equation $(\si_{\mfr g} + \si_{\mfr f}) \cdot \phi = 0$ \rf{fixedEq}. This tells us that most of the Higgs branch is now lifted and only the fixed points of the $\mathsf T$-action in the Higgs branch remain as vacua of the theory.\footnote{Later in computations, we only turn on $\msf A$-twisted masses and look for $\msf A$-fixed points where $\msf T = \msf A \times \U(1)_\hbar$.}

 We work out some details of the construction of the classical Higgs branches and their flavor fixed points for the specific theories defined by the quiver in Figure~\ref{fig:gl11Quiver}.

\subsubsection[Classical Higgs branches of 3d $\cN=2$ SQCDs]{Classical Higgs branches of 3d $\boldsymbol{\cN=2}$ SQCDs} The real moment map of the $G$-action on the space of chirals is dictated by the D-term in the chiral multiplet Lagrangian \rf{Lchi}. For the theory defined by Figure~\ref{fig:gl11Quiver} with the representations of the chirals as in \rf{chiralRep}, the moment map is\footnote{$\cQ$ and $\ov \cQ$ are a priori independent fields in Euclidean signature, but we have to choose some reality condition on them to do our path integrals and we choose $\ov \cQ = \cQ^\dagger$.}
$\mu = \cQ \cQ^\dagger - \wtd \cQ^\dagger \wtd \cQ$.
And so the real moment map equation becomes:
\beq
	\cQ \cQ^\dagger - \wtd \cQ^\dagger \wtd \cQ - \ze = 0 , \label{momentEqSQCD}
\eeq
where $\ze \in \bR_+$ is the positive FI parameter. The Higgs branch is the quotient
$\higgs(N,L) := \mu^{-1}(\ze)/U(N)$,
From this, we get the dimension of the Higgs branch
\beq
	\dim_\bC \higgs(N, L) = 2NL - N^2 . \label{MHdim}
\eeq
Furthermore, since \smash{$\cQ, \ov{\wtd \cQ} \in \Hom\bigl(\bC^L, \bC^N\bigr)$}, we have
$\operatorname{rk} \mu \le L$.
On the other hand, the moment map equation sets $\mu_{\mfr g}$ to $\ze$ which has rank $N$. We thus get the constraint
\beq
	N \le L . \label{N<=L}
\eeq
Notice from \eqref{MHdim} that for odd values of $N$ the complex dimension is odd and the resulting space can not have a holomorphic symplectic form. This is in sharp contrast with the case of 3d~${\cN=4}$ supersymmetry where the hyper-K\"ahler structure of Higgs branches equips them with a~holomorphic symplectic form.

For our purposes, we mainly need the complex structure of the Higgs branch instead of the full K\"ahler structure. Let us look at an example.

\begin{ex}[$N=1$, $L=2$] \label{ex:conifold}
For gauge group $G = \U(1)$ and flavor group $F_{\msf A}=\SU(2)$, the Higgs branch is a $3$-dimensional complex manifold. The space of the chirals is, $\mbf M = \Hom\bigl(\bC^2, \bC\bigr) \oplus \Hom\bigl(\bC, \bC^2\bigr)$ with holmorphic coordinates $\cQ^i$ and \smash{$\wtd \cQ_i$} for $i=1,2$. The moment-map equation is given by
\begin{equation}\label{eq:moment-map equation of Abelian gauge theories with L=2}
 \big|\cQ^1\big|^2+\big|\cQ^2\big|^2-\big|\wt{\cQ}_1\big|^2-\big|\wt{\cQ}_2\big|^2 - \ze = 0.
\end{equation}
Under a complexified gauge transformation by $\la \in \bC^\times$ the coordinates transform as
\[
	\bigl(\cQ, \wtd \cQ\bigr) \mapsto \bigl(\la \cQ, \la^{-1} \wtd \cQ\bigr) .
\]
So a set of gauge invariant holomorphic functions on the Higgs branch is given by the matrix elements of \smash{$\wtd \cQ \cQ$}
\[
 \wtd \cQ \cQ = \begin{pmatrix} \wtd \cQ_1 \cQ^1 & \wtd \cQ_1 \cQ^2 \\ \wtd \cQ_2 \cQ^1 & \wtd \cQ_2 \cQ^2 \end{pmatrix} .
\]
Any gauge invariant polynomial function on the Higgs branch can be written in terms of these. Since \smash{$\wtd \cQ \cQ\colon \bC^2 \to \bC^2$} factors through $\bC$, it has rank at most $1$, consequently
\begin{equation}\label{eq:conifold equation for L=2}
 \det \wtd \cQ \cQ=0.
\end{equation}
The equation \rf{eq:conifold equation for L=2} defines the singular conifold as a subspace of $\Hom\bigl(\bC^2, \bC^2\bigr)$ with singularity at \smash{$\cQ = \wtd \cQ = 0$} where the rank of \smash{$\wtd \cQ \cQ$} is strictly less than 1. The Higgs branch $\higgs(1,2)$ is a~resolution of this singularity by the moment map equation \rf{eq:moment-map equation of Abelian gauge theories with L=2}. The moment map equation requires a nonzero \smash{$\big|\cQ^1\big|^2 + \big|\cQ^2\big|^2$}, which means that with $\ze > 0$, at $\det \wtd \cQ \cQ = 0$ we have $\wtd \cQ = 0$. Then the moment map equation becomes \smash{$\big|\cQ^1\big|^2+\big|\cQ^2\big|^2=\zeta$}. This defines an $S^3$ with radius $\sqrt{\ze}$, which together with the $\U(1)$-gauge identification implies that the singularity of the conifold is replaced with $S^3/\U(1) \simeq \mbb{P}^1$ of size determined by $\ze$.

 The resolved complex geometry can be described as follows. For $\ze >0$, the requirement $\bigl(\cQ^1, \cQ^2\bigr) \ne (0,0)$ coming from the moment map equation says that the map $\cQ\colon\bC^2 \to \bC$ is surjective. Due to the quotient by complex gauge transformations, only the kernel of this map is physical. Different choices of $\ker \cQ \simeq \bC \subseteq \bC^2$ are parameterized by a $\bP^1$. $\cQ^1$ and $\cQ^2$ are the homogeneous coordinates on this $\bP^1$. We can define the two local coordinates
\smash{$\msf z := \frac{\cQ^1}{\cQ^2}$}, \smash{$ \msf z' := \frac{\cQ^2}{\cQ^1}$},
to cover this $\bP^1$. There is no constraint on $\wtd \cQ\colon \bC \to \bC^2$ and since we have already used up the gauge transformations to classify $\cQ$, all choices of $\wtd \cQ\colon \bC \to \bC^2$ are physically significant. A map $\bC \to \bC^2$ is characterized by the image of $1$ which is unconstrained. Over the chart of $\bP^1$ with local coordinate $\msf z$ we can parameterize these images of $1$ by \smash{$\bigl(\frac{\wtd \cQ_1}{\cQ^2}, \frac{\wtd \cQ_2}{\cQ^2}\bigr)$} and over the chart with coordinate $\msf z'$, we can parameterize the same images by
\[
\left(\frac{\wtd \cQ_1}{\cQ^1}, \frac{\wtd \cQ_2}{\cQ^1}\right) = \msf z' \left(\frac{\wtd \cQ_1}{\cQ^2}, \frac{\wtd \cQ_2}{\cQ^2}\right).
\] Overall, the maps $\bigl(\cQ, \wtd \cQ\bigr)$ modulo complex gauge transformations form a rank 2 vector bundle over $\bP^1$ with transition function $\msf z'$ as we change coordinate $\msf z \mapsto \msf z'$. Dualizing this transition function we get the underlying space as the Higgs branch
\beq
	\higgs(1,2) = \mcr O(-1) \oplus \mcr O(-1) \to \bP^1 . \label{MH12}
\eeq
\end{ex}
It is easy to generalize this example to arbitrary $N$ and $L \ge N$. First, as a consequence of the moment map equation we have
\begin{prop} \label{prop:rankQ=N}
 For $\ze > 0$, on the Higgs branch $\operatorname{rk} \cQ = N$.
\end{prop}
\begin{proof}
 Suppose $\cQ$ has rank less than $N$. Then $\cQ \cQ^\dagger$ has rank less than $N$ and we can choose a nonzero $v \in \ker \cQ\cQ^\dagger$. The moment map equation \rf{momentEqSQCD} then implies
$- \wtd \cQ^\dagger \wtd \cQ v = \ze v $.
 However, the operator on the left-hand side is negative semi-definite and can not have a positive eigenvalue~-- we have a contradiction.
\end{proof}

Now let us consider the Higgs branch. Once again, a generating set of holomorphic functions on the Higgs branch are given by the matrix elements of $\wtd \cQ \cQ$. Since $\wtd \cQ \cQ$ factors through $\bC^N$ we have the constraint
$ \operatorname{rk} \wtd \cQ \cQ \le N$.
This defines the determinantal variety, denoted by $\msf{Det}(N, L)$, which has singularities where the rank of $\wtd \cQ\cQ$ is strictly less than $N$. These singularities are resolved by the real moment map equation.

 Since $\cQ\colon\bC^L \to \bC^N$ is surjective (see Proposition~\ref{prop:rankQ=N}), modulo complexified gauge transformation by $\GL_N$, $\cQ$ is completely fixed by its kernel. Once $\cQ$ is fixed, $\wtd \cQ\colon \bC^N \to \bC^L$ is unconstrained, though the choices of $\wtd \cQ$ are fibered nontrivially over the choices of $\cQ$ due to the~$\GL_N$ quotient that acts on both the choice of $\cQ$ and that of $\wtd \cQ$. The quotient of $\bC^L$ by $\ker \cQ$ is parameterized by the Grassmannian $\text{Gr}(N, L)$. We can remember the $\GL_N$ action on $\cQ$ by the tautological principal $\GL_N$ bundle $\cK_P$ over $\text{Gr}(N, L)$. Then the choices of $\bigl(\cQ, \wtd \cQ\bigr)$ such that~$\cQ$ is surjective are parameterized by $\cK_P \times \Hom\bigl(\bC^N, \bC^L\bigr)$. Taking the gauge quotient, we find the Higgs branch
\[
 \higgs(N, L) = \cK_P \times^{\GL_N} \Hom\bigl(\bC^N, \bC^L\bigr) = \Hom\bigl(\cK, \bC^L\bigr) ,
\]
where $\cK$ and $\bC^L$ refer to the tautological vector bundle and a trivial bundle over $\text{Gr}(N, L)$. By construction, this is a vector bundle over $\text{Gr}(N, L)$.\footnote{Had we used the opposite sign $\ze < 0$ of the real FI parameter in the moment map equation \rf{momentEqSQCD}, the Higgs branch would be the dual bundle $\Hom\bigl(\bC^L, \cK\bigr)$.} The above space is also called the desingularization of the determinantal variety $\msf{Det}(N, L)$ \cite{Lascoux1978}, and can be denoted by \smash{$\what{\msf{Det}}(N, L)$}.

 The singular conifold \rf{eq:conifold equation for L=2} is the special case of a determinantal variety for $N=1$ and~${L=2}$, and the resolved conifold \rf{MH12} is the aforementioned resolution thereof. In all abelian cases, we get
\[
 \higgs(1, L) = \what{\msf{Det}}(1, L) = \mcr O(-1)^{\oplus L}\to\mbb{P}^{L-1} .
\]

\subsubsection{Fixed points} \label{sec:fixedPoints}
Flavor fixed points in the Higgs branch are points of $\mbf M$ satisfying the real moment map equation where the flavor torus action can be undone by a gauge transformation. We are considering the gauge and flavor symmetry group $G = \U(N)$ and $\msf A = \U(1)^{L-1} \subseteq \SU(L)$. Under an infinitesimal~${G \times \msf A}$ action, the chiral multiplets transform as
$ \cQ\to \cQ+ \si \cQ-\cQ m$,
$
 \wt{\cQ}\to \wt{\cQ} - \wt{\cQ} \si + m \wt{\cQ}$
where $ \si \in\mfk{u}(N)$, and \beq
 m = \begin{pmatrix} m_1 &&& \\ & m_2 && \\ && \ddots & \\ &&& m_L \end{pmatrix} \in \mfr a = \bR^{L} ,
\label{twistedm}\eeq
with the constraint $m_1 + \cdots + m_L = 0$. Therefore, $\bigl(\cQ,\wt{\cQ}\bigr)$ is a fixed point if the following equations admit solutions for $\si$ \begin{equation}\label{eq:conditions for fixed points in the non-Abelian gauge theories}
 \begin{aligned}
 \si \cQ-\cQ m &=0 ,
 \qquad
 \wt{\cQ} \si - m \wt{\cQ} &=0 .
 \end{aligned}
\end{equation}
And $\bigl(\cQ, \wtd \cQ\bigr)$ must satisfy the moment map equations \rf{momentEqSQCD} to be in the Higgs branch.

 Let us denote by $\cQ^i$ and $\wtd \cQ_i$ for $i=1,\dots,L$ the columns and rows of $\cQ$ and $\wtd \cQ$, respectively. Then we can write \rf{eq:conditions for fixed points in the non-Abelian gauge theories} as
\beq
	\si \cQ^i = m_i \cQ^i , \qquad \wtd \cQ_i \si = m_i \wtd \cQ_i .\label{eigenEq}
\eeq
These are eigenvalue equations for the $2L$ vectors $\cQ^i$, $\wtd \cQ_i$. We assume that the masses $m_i$ are all distinct. There are then $L$ different eigenvalue equations since $\cQ^i$ and $\wtd \cQ_i$ share an eigenvalue. Since $\si$ is an $N \times N$ matrix with $N \le L$, at most $N$ out of these $L$ eigenvalue equations can produce nonzero eigenvectors. On the other hand, for $\ze > 0$, the moment map equation~\rf{momentEqSQCD} implies that $\operatorname{rk} \cQ = N$ (see Proposition~\ref{prop:rankQ=N}). Consequently, a solution to the eigenvalue equations~\rf{eigenEq} contains exactly $N$ nonzero eigenvectors. The nonzero eigenvectors are labeled by an inclusion
\beq
	p \colon\ \{1,\dots,N\} \hookrightarrow \{1, \dots, L\} , \label{LchooseN}
\eeq
such that the nonzero eigenvectors are \smash{$\cQ^{p(a)}$}. Eigenvectors of hermitian operators with distinct eigenvalues are mutually orthogonal. So, as an orthonormal basis for the $N$-dimensional vector space spanned by $\cQ^{p(a)}$ we can choose \smash{$\bigl\{\frac{\cQ^{p(1)}}{\lVert \cQ^{p(1)} \rVert}, \dots, \frac{\cQ^{p(N)}}{\lVert \cQ^{p(N)} \rVert} \bigr\}$}. In this basis, the matrix $\cQ$ satisfies
\begin{gather}
	\cQ^i_a = \de^i_{p(a)} \big\lVert \cQ^{p(a)} \big\rVert  \qquad \text{(no sum over $a$)} . \label{Qx}
\end{gather}
$\wtd \cQ_i$ must be a multiple of $\cQ^i$ since they have the same eigenvalue
\beq
 \wtd \cQ^a_i = \de^{p(a)}_i {\rm e}^{\ii \varphi} \big\lVert \wtd \cQ_{p(a)} \big\rVert . \label{Qtildex}
\eeq
The phase is not fixed. Substituting \rf{Qx} and \rf{Qtildex} in the moment map equation \rf{momentEqSQCD}, we get
\[
	\cQ \cQ^\dagger - \wtd \cQ^\dagger \wtd \cQ = \text{diag}\bigl(\big\lVert \cQ^{p(1)} \big\rVert^2 - \big\lVert \wtd \cQ_{p(1)} \big\rVert^2, \dots, \big\lVert \cQ^{p(N)} \big\rVert^2 - \big\lVert \wtd \cQ_{p(N)} \big\rVert^2\bigr) = \ze \id_N .
\]
Using gauge symmetry, we can fix all components of $\cQ$ to be real and positive, then there is a~one-parameter family of solutions to the above moment map equation:
\beq
 \cQ^i_a = \de^i_{p(a)} \sqrt{\ze + |\la|^2} , \qquad \wtd \cQ^a_i = \de^{p(a)}_i \la , \qquad \la \in \bC . \label{Qvev}
\eeq
To summarize, we have the following characterization of $\msf A$-fixed points of the Higgs branch $\higgs(N,L)$. There are $\binom{L}{N}$ mutually disjoint connected components of the fixed point set, labeled by a choice $p$ of $N$ numbers out of $L$ \rf{LchooseN}.\footnote{Only maps \rf{LchooseN} modulo permutations of $\{1,\dots, N\}$ lead to physically inequivalent vacua since the permutations coincide with the action of the Weyl subgroup of the gauge group $\U(N)$.} For such a choice $p$, the corresponding component of the fixed point set in the Higgs branch is given by the gauge orbit of a reference $N$-dimensional subspace $\what \cF_p \subseteq \mbf M$ which we take to be (no sum over $a$, $b$)
\begin{gather}
 \what \cF_p := \bigl\{\bigl(\cQ, \wtd \cQ\bigr) \in \mbf M \mid \bigl(\cQ^i_a, \wtd \cQ^b_j\bigr) = \bigl(\de^i_{p(a)} \sqrt{\ze + |\la_a|^2}, \de^{p(b)}_j \la_b\bigr), (\la_1, \dots, \la_N) \in \bC^N \bigr\} .
\label{fixedP}
\end{gather}
We denote the image of $\what \cF_p $ in the complex Higgs branch \rf{MHgit} under the bundle projection
\beq
 \mbf M^s \xrightarrow{\pi} \mbf M^s/\GL_N = \higgs \label{higgsProj}
\eeq
by
$ \cF_p := \pi\bigl(\what \cF_p\bigr) $.
The projection $\pi$ is one-to-one when restricted to $\what \cF_p $.

 Observe that the fixed points also satisfy the vacuum equations \rf{vacEq} in the case $\si_{\mfr f} = \si_{\msf A}$,
in other words, when we turn on twisted masses only with respect to the flavor symmetry $\SU(L)$ (the framing node of the quiver Figure~\ref{fig:gl11Quiver}).
At these vacua, the value of the vector multiplet scalar $\si$ is determined by \rf{fixedEq} (or equivalently \rf{eigenEq}), which is
\beq\begin{aligned}
 \si_p(m) := \begin{pmatrix} m_{p(1)} &&& \\ & m_{p(2)} && \\ && \ddots & \\ &&& m_{p(N)} \end{pmatrix} .
\end{aligned}\label{siP}\eeq
Thus, the $\msf A$-fixed point set corresponds to the residual classical Higgs branch in the presence of the $\msf A$-twisted masses
\beq
 \higgs^{\msf A} = \bigsqcup_p \cF_p = \text{classical Higgs branch with $\msf A$-twisted masses.} \label{MHA}
\eeq
At the vacuum $p$, we also get constraints on the holonomies for the flat connections around the elliptic curve from the screening conditions \rf{screenGeneric}. The nontrivial screening conditions come from the chirals \rf{Qvev} that take nonzero VEVs, the conditions being
\beq
 s_a x_{p(a)}^{-1} \hbar^{1/2} = 1 . \label{screening}
\eeq
We shall refer to the gauge holonomies satisfying this condition by $s^{(p)}_a$.

\subsubsection{Attracting sets and Morse flow} \label{sec:MorseAndHamilton}

Consider the definition of the attracting set of a fixed point $x_0 \in \cF_p $ \rf{attrDef}
\[
 \msf{Attr}_{\mfr C}(x_0) = \bigl\{ x \in \higgs \mid \lim_{t \to 0} f(t) \cdot x = x_0 \text{ for all } f \in \mfr C \bigr\} .
\]
Choose some generic $m = \text{diag}(m_1, \dots, m_L) \in \mfr a$ such that $t \mapsto t^m$ is a cocharacter in the chamber $\mfr C$. Now, introduce the variable $y = -\ln t$ such that the limit $t \to 0$ corresponds to the limit $y \to \infty$. We take the limit $t \to 0$ along the positive real axis so that $y$ becomes a~non-negative real number. Then we can rewrite the definition of the above attracting set as
\[
 \msf{Attr}_{\mfr C}(x_0) = \bigl\{ x \in \higgs \mid \lim_{y \to \infty} {\rm e}^{-ym} \cdot x = x_0 \bigr\}.
\]
Define the trajectory of $x$ in the Higgs branch, as we take the limit $y \to \infty$, as a path
\[
 x\colon\ [0,\infty) \to \higgs , \qquad x(y) := {\rm e}^{-ym} \cdot x_0 .
\]
This path satisfies the flow equation
\beq
 \pa{y} x + m \cdot x = 0 , \label{xiMorse}
\eeq
with the asymptotic boundary condition
$\lim_{y \to \infty} x = x_0$.
The definition of the above attracting set can now be rewritten as
\begin{align}
 \msf{Attr}_{\mfr C}(x_0) = \bigl\{ x(0) \in \higgs \mid \pa{y} x + m \cdot x = 0,\, \lim_{y \to \infty} x = x_0 \bigr\} . \label{attrx2}
\end{align}
Any smooth path $ x(y)$ in the Higgs branch can be lifted to a smooth path in the stable locus~$\mbf M^s$ by undoing the projection \rf{higgsProj} and non-uniquely choosing some point from each fiber over~$ x(y)$. Let $\what x(y)$ be such a lift of $ x(y)$. The flow equation \rf{xiMorse} for $ x$ lifts to the following flow equation for $\what x$
$D_{y} \what x + (\si + m) \cdot \what x = 0$.
Here $D_y$ involves some $\gl_N$-connection and $\si(y)$ is some $\gl_N$-valued function of $y$. These are the extra ingredients that will be lost once we take the projection by the $\GL_N$-quotient. Let $\what x_0 \in \what \cF_p $ be a lift of $x_0$. We can constraint the lift~$\what x$ by requiring that its asymptotic value, instead of being something arbitrary from the $\GL_N$-orbit of $\what x_0$, be exactly equal to $\what x_0$
$ \lim_{y \to \infty} \what x(y) = \what x_0$.
This completely breaks gauge symmetry at $y=\infty$ and therefore the value of $\si$ must also be constrained asymptotically to take the value~\rf{siP}~${\lim_{y \to \infty} \si(y) = \si_p(m)}$.
We can rewrite the definition of the attracting set \rf{attrx2} once more by projecting $\what x$ as follows
\begin{align}
 \msf{Attr}_{\mfr C}(x_0) =& \pi\bigl(\text{Flow}^+_{\mfr C}\bigl(\what x_0\bigr)\bigr) , \label{attrphi}
\end{align}
where we have defined
\beq
 \text{Flow}^\pm_{\mfr C}(\what x_0) := \bigl\{\what x(0) \in \mbf M^s \mid D_y \what x + (\si + m) \cdot \what x = 0,\,\lim_{y \to \pm\infty} (\what x, \si) = \bigl(\what x_0, \si_p(m)\bigr) \bigr\} .
\label{flow}
\eeq
This relates Morse flow in $\mbf M$ to the attracting sets in the Higgs branch. The flows in $\mbf M$ are not completely fixed because both the connection $D_y$ and $\si$ can be changed by $\GL_N$ gauge transformations as long as we satisfy the limit condition. We can choose some convenient representatives of these flows as follows.

 Let us zoom in on the relation \rf{attrphi} at $x_0$, i.e., we look at the tangent space to the attracting set at $x_0$
\begin{align}
 T_{x_0}\msf{Attr}_{\mfr C}(x_0) =& \pi_*\bigl(\bigl\{\what x(0)\! \in \!\mbf M^s \mid D_y \what x + (\si_p(m) + y D_y \si|_{x_0}\!\! + m)\! \cdot\! \what x = 0, \,\lim_{y \to \infty} \what x = \what x_0 \bigr\}\bigr) .\!\!\! \label{attrphi*}
\end{align}
We further use our gauge freedom to set $D_y = \pa{y}$ and \smash{$D_y \si\big|_{x_0} = 0$}. Small displacements around~${x_0 \in \mbf M^s}$ can be parameterized by the chiral scalars $\cQ^i_a$ and \smash{$\wtd \cQ^a_i$}. These displacements are scaled by the action of $(\si_p(m)+m)$ as
\begin{gather}\label{si+mAction}
 (\si_p(m)+m) \cdot \cQ^i_a= (m_{p(a)} - m_i) \cQ^i_a ,
\qquad
 (\si_p(m)+m) \cdot \wtd \cQ_i^a= -(m_{p(a)} - m_i) \wtd \cQ_i^a .
\end{gather}
Solutions to the flow equation $\pa{y}\what x + (\si_p(m) + m)\cdot \what x = 0$ are therefore given by
\[\cQ^i_a(y)= {\rm e}^{-y(m_{p(a)} - m_i)} \cQ^i_a(0) ,
\qquad
 \wtd \cQ_i^a(y)= {\rm e}^{y(m_{p(a)} - m_i)} \wtd \cQ_i^a(0) .
\]
The linear combination \smash{$\what x(y) = \what x_0 + \ka^a_i \cQ^i_a(y) + \wtd \ka^i_a \wtd \cQ^a_i(y)$} then satisfies the same flow equation as well as the limit condition $\lim_{y \to \infty} \what x = \what x_0$ if and only if the coefficients $\ka$ and $\wtd \ka$ vanish for all chirals with non-positive weights. Let us decompose the space of chirals into subspaces according to the sign of their weights under the action \rf{si+mAction}
\begin{subequations}
\beq
 \mbf M = \mbf M_{\mfr C;+}(p) \oplus \mbf M_{0}(p) \oplus \mbf M_{\mfr C;-}(p) , \label{XweightDecomp}
\eeq
where
\begin{gather*}
 \mbf M_{\mfr C;+}(p) := \text{Span}_\bC \bigl\{\cQ^i_a, \wtd \cQ^b_j \mid m_{p(a)} - m_i > 0,\, m_{p(b)} - m_j < 0 \bigr\} , \\
 \mbf M_{0}(p) := \text{Span}_\bC \bigl\{\cQ^i_a, \wtd \cQ^b_j \mid m_{p(a)} - m_i = m_{p(b)} - m_j = 0 \bigr\} , \\
 \mbf M_{\mfr C;-}(p) := \text{Span}_\bC \bigl\{\cQ^i_a, \wtd \cQ^b_j \mid m_{p(a)} - m_i < 0,\, m_{p(b)} - m_j > 0 \bigr\} .
\end{gather*}
\end{subequations}
We can now rewrite \rf{attrphi*} as
\begin{align}
 T_{x_0}\msf{Attr}_{\mfr C}(x_0) \simeq& \pi_*\left(\mbf M_{\mfr C;+}(p) \right) . \label{attrphi*RR}
\end{align}
Here the vector space inside the projection on the right-hand side has its origin at $\what x_0$. The attracting set $\msf{Attr}_{\mfr C}(x_0)$ is generated by the action of $f(t) = t^m$ for some generic cocharacter $f$ in~the chamber $\mfr C$. This action can be lifted to an action $\what f$ on $\mbf M$ via an embedding $\mathsf A_\bC \hookrightarrow \mathsf H_\bC \times \mathsf A_\bC$, such that for $\what x \in \mbf M$ we have $\what f(t) \cdot \what x = t^{\si_p(m)+m} \cdot \what x$. We can then define the attracting set of $\what x_0$ in $\mbf M^s$ with respect to this action:
\[ \what{\msf{Attr}}_{\mfr C}(\what x_0) := \bigl\{ \what x \in \mbf M \mid \lim_{t \to 0} t^{\si_p(m)+m} \cdot \what x = \what x_0 \bigr\} .
\]
The tangent space to this affine space at $\what x_0$ is $\mbf M_{\mfr C;+}(p)$. The projection $\pi$ is $\mathsf A_\bC$ equivariant and it is a surjection \rf{attrphi*RR} on tangent spaces. Therefore, it is a surjection globally, i.e., \smash{$\what{\msf{Attr}}_{\mfr C}(\what x_0)$} projects onto the full attracting set $\msf{Attr}_{\mfr C}(x_0)$. We can also immediately see that the attracting set \smash{$\what{\msf{Attr}}_{\mfr C}(\what x_0)$} is nothing but $\{\what x_0\} \times \mbf M_{\mfr C;+}(p)$. Therefore,
\[ \msf{Attr}_{\mfr C}(x_0) = \pi\bigl(\bigl\{\what x_0\bigr\} \times \mbf M_{\mfr C;+}(p) \bigr) .
\]
Comparing with \rf{attrphi} we find a gauge equivalence between positive weight spaces in $\mbf M$ and Morse flow in the stable locus $\mbf M^s$ that asymptotes some fixed point in $\what \cF_p $. More specifically, by varying $x_0$ over $ \cF_p $, we find
\beq
\pi\bigl( \text{Flow}^+_{\mfr C}\bigl(\what \cF_p \bigr) \bigr)
= \msf{Attr}_{\mfr C}( \cF_p )
= \pi\bigl(\what \cF_p \times \mbf M_{\mfr C;+}(p)\bigr) .
\label{morse+}
\eeq
Replacing the limit $y \to \infty$ with $y \to -\infty$, and following the same arguments we find
\beq
\pi\bigl( \text{Flow}^-_{\mfr C}\bigl(\what \cF_p \bigr) \bigr)
= \msf{Rep}_{\mfr C}( \cF_p )
= \pi\bigl(\what \cF_p \times \mbf M_{\mfr C;-}(p) \bigr) .
\label{morse-}\eeq
Here $\msf{Rep}$ refers to the repelling set defined similarly to the attracting set \rf{attrDef} but with the opposite limit $t \to \infty$.

\subsection{Boundaries and interfaces} \label{sec:bc}
We discuss half-BPS boundary conditions that are equivalent, in localization computations, to forcing the fields to reach any particular connected component of the solutions to the vacuum equation \rf{vacEq}. The choice of localizing supercharge is discussed in Appendix~\ref{sec:Qloc}, \rf{Qloc} in particular. We discuss half-BPS boundaries preserving $\mcr Q_-$, $\ov{\mcr Q}_-$. So these are $\cN=(0,2)$ boundary conditions in the 3d $\cN=2$ theory. We also discuss the mass Janus interface preserving the same supersymmetry. Our discussion mostly follows similar boundary conditions for the 3d $\cN=4$ case considered in \cite{DedushenkoNekrasov202109}. Our description of the thimble boundaries (see Section~\ref{sec:thimble}) differs slightly from this reference due to the non-isolated nature of 3d $\cN=2$ Higgs branch vacua. %

\subsubsection[Thimble (exceptional Dirichlet), $\mcr D_{\mfr C}(p)$]{Thimble (exceptional Dirichlet), $\boldsymbol{\mcr D_{\mfr C}(p)}$} \label{sec:thimble}
To find a boundary condition at $y=y_-$ that is cohomologous to a choice of vacuum, we extend our space-time to negative infinity and consider the theory on $(-\infty \times y_-] \times \bE_\tau$. We impose the constraint that the theory must reach the vacuum $p$ at infinity. Then we observe the supersymmetric evolution of the fields along the interval $(-\infty, y_-]$ and see what values the fields take at $y=y_-$. More specifically, we solve the BPS equations \rf{flowEq} with the boundary condition that at $y=-\infty$ the fields satisfy the vacuum equation \rf{vacEq}. Then the values of these solutions at $y=y_-$ provide a past boundary condition for the theory at $y>=y_-$ that mimics the chosen vacuum in BPS computations. Similar boundary conditions were constructed and called thimble or exceptional Dirichlet boundary conditions in \cite{Bullimore:2016nji}.

 In this process, we observe that the BPS equation \rf{flowphi} with the condition that at past infinity the chiral and vector multiplet scalars satisfy the vacuum equations \rf{vacEq} is precisely what defines the set $\text{Flow}^-_{\mfr C}\bigl(\what \cF_p \bigr)$ \rf{flow}. The choice of chamber and vacuum determine the Morse function and the connected component of the vacuum configurations defining the boundary condition. The BPS equation \rf{flowsiG} puts constraints on the evolution of $\si_{\mfr g}$ but $\text{Flow}^-_{\mfr C}\bigl(\what \cF_p \bigr)$ is agnostic to such constraints. According to \rf{morse-}, the set \smash{$\text{Flow}^-_{\mfr C}\bigl(\what \cF_p \bigr)$} is gauge equivalent to the subspace~${\what \cF_p \times \mbf M_{\mfr C;-}(p)}$ of $\mbf M$. So, we formulate a $(0,2)$ boundary condition of brane type with support for the chiral multiplet scalars on \smash{$\what \cF_p \times \mbf M_{\mfr C;-}(p)$}.

\begin{rmk}[thimbles in $\cN=4$ vs $\cN=2$]
Support of this brane (at the past boundary) on the Higgs branch is the projection of $\what \cF_p \times \mbf M_{\mfr C;-}(p)$, which is the repelling set \rf{morse-}. In \cite{Bullimore:2016nji}, the thimble boundaries are supported on holomorphic Lagrangians in the hyperk\"ahler Higgs branches of 3d $\cN=4$ theories. For 3d $\cN=2$ theories, the Higgs branches are generally only K\"ahler and the attracting/repelling sets \eqref{morse+}, \eqref{morse-} are not supported on holomorphic Lagrangians, they are just holomorphic.
\end{rmk}

Let us decompose the space of chirals further from \rf{XweightDecomp} based on where the $\cQ$s and the~$\wtd \cQ$s belong
\[
 \mbf M_{\mfr C; \pm}(p) = \cR_{\mfr C; \pm}(p) \oplus \wtd\cR_{\mfr C; \pm}(p) , \qquad
 \mbf M_{0}(p) = \cR_{0}(p) \oplus \wtd\cR_{0}(p) .
\]
Since $\cQ^i_a$ and $\wtd \cQ^a_i$ have weights with opposite signs (see \rf{si+mAction}), we have the following relations among these subspaces
\[
 \cR_{\mfr C;\pm}(p) = \wtd \cR_{\mfr C;\mp}^\vee(p) , \qquad
 \cR_0(p) = \wtd \cR_0^\vee(p) .
\]
We thus have a Lagrangian splitting of the chirals
\begin{gather}
 \mbf M = \bL_{\mfr C}(p) \oplus \bL^\perp_{\mfr C}(p) , \qquad
 \bL_{\mfr C}(p) := \mbf M_{\mfr C; +}(p) \oplus \cR_0(p) , \nonumber\\
 \bL^\perp_{\mfr C}(p) := \mbf M_{\mfr C; -}(p) \oplus \wtd\cR_0(p) .\label{LMdecomp}
\end{gather}
$\what \cF_p $ is parameterized by arbitrary VEVs of fields from $\wtd \cR_0(p)$ (see \rf{fixedP}), therefore, chirals from~$\bL^\perp_{\mfr C}(p)$ parameterize the brane \smash{$\what \cF_p \times \mbf M_{\mfr C;-}(p)$}. So we shall put Dirichlet and Neumann boundary conditions on the scalars from the chirals in $\bL_{\mfr C}(p)$ and $\bL^\perp_{\mfr C}(p)$ respectively. The boundary conditions will be extended to the rest of the fields of these multiplets by the requirement of preserving $(0,2)$ supersymmetry.

 Gauge symmetry is broken at the vacuum by the VEVs of the chirals and as such the gauge field receives Dirichlet boundary condition, the rest of the vector multiplet receiving boundary conditions accordingly. The holonomy of the dynamical gauge field at the vacuum is given in terms of the holonomy of the background gauge fields via the screening condition \rf{screening}.

 Altogether, the thimble boundary condition $\mcr D_{\mfr C}(p)$ involves:
\beq\begin{aligned}
 &\text{Chirals in } \bL_{\mfr C}(p)\colon \ \text{Dirichlet with VEVs according to vacuum } p, \\
& \text{Chirals in } \bL^\perp_{\mfr C}(p)\colon \ \text{Neumann}, \\
 &\text{Vector multiplet} \colon \ \text{Dirichlet with screening \rf{screening} of gauge holonomy.}
\end{aligned}\label{thimbleBC}\eeq

{\bf Boundary 't Hooft anomaly.} Let ${\mbf f} \in \mfr h$, ${\mbf f}_x \in \mfr a$ and ${\mbf f}_\hbar \in \bR_\hbar$ be the curvatures of the gauge, $\SU(L)$ flavor, and $\U(1)_\hbar$ flavor symmetry at the boundary. The gauge symmetry is completely broken at the boundary and ${\mbf f}$ is a function of the flavor curvatures. Their dependence is fixed by demanding that the nonzero VEVs at the boundary do not observe any flux at the boundary ${\mbf f} = f^{(p)}({\mbf f}_x, {\mbf f}_\hbar)$ such that,
$w({\mbf f} + {\mbf f}_x + {\mbf f}_\hbar) = 0$ for all $w \in \cR_0(p)$ .
Contributions to the anomaly polynomial come from two sources: 1)~Chern--Simons terms created by integrating out massive fermions and the Chern--Simons coupling to background multiplet associated with~$\U(1)_\text{top}$ symmetry, 2)~purely from boundary conditions. In the present case, these contributions are \cite{DedushenkoNekrasov202109,DimofteGaiottoPaquette201712}\footnote{The trace without explicit representation is taken over the fundamental representation.}
\[
 P^\text{CS}_-(\mcr D_{\mfr C}(p)) = 2\tr_{\mbf M_{\mfr C;+}(p)} ({\mbf f}_\hbar ({\mbf f} + {\mbf f}_x)) - 2 \tr({\mbf f} {\mbf f}_\text{top}) , \qquad
 P^\partial(\mcr D_{\mfr C}(p)) = 2\tr_{\cR_0(p)}({\mbf f}_\hbar ({\mbf f} + {\mbf f}_x)) .
\]
The subscript in $P^\text{CS}_-$ refers to the past boundary. The Chern--Simons anomaly depends on the orientation of the boundary and therefore $P^\text{CS}_- = -P^\text{CS}_+$. $P^\partial$ is independent of the choice of past or future. If we define
\[
 P(\mcr D_{\mfr C}(p)) := P^\text{CS}_-(\mcr D_{\mfr C}(p)) + P^\partial(\mcr D_{\mfr C}(p)) = 2\tr_{\bL_{\mfr C}(p)} ({\mbf f}_\hbar ({\mbf f} + {\mbf f}_x)) - 2 \tr({\mbf f} {\mbf f}_\text{top}) ,
\]
then the total anomaly polynomial associated to the thimble boundary is $P(\mcr D_{\mfr C}(p))$ evaluated at ${\mbf f} = {\mbf f}^{(p)}({\mbf f}_x, {\mbf f}_\hbar)$.

\subsubsection[Enriched Neumann, $\scrN_{\mfr C}(p)$]{Enriched Neumann, $\boldsymbol{\scrN_{\mfr C}(p)}$} \label{sec:neumann}

Swap the boundary conditions on the chiral multiplets from \rf{thimbleBC}. Meaning, the chirals in~$\bL_{\mfr C}(p)$ now receive Neumann, and those in $\bL^\perp_{\mfr C}(p)$ receive Dirichlet boundary conditions with zero VEV. The gauge field in the vector multiplet receives the Neumann boundary condition. This implies that the adjoint scalar $\si_{\mfr g}$ from the vector multiplet gets Dirichlet \cite[Section~2.5.2]{DimofteGaiottoPaquette201712}. We set the value of $\si_{\mfr g}$ equal to the VEV $\si_p(m)$ \rf{siP} corresponding to the vacuum $p$. This breaks the gauge symmetry down to the maximal torus $\msf H = \U(1)^N$. Additionally, we couple a~2d~$\cN=(0,2)$ theory $\Upsilon_{\mfr C}(p)$ at the boundary by gauging some of its global symmetries with the boundary gauge symmetry $\msf H$. We shall characterize the theory $\Upsilon_{\mfr C}(p)$ momentarily. Let us first summarize this enriched Neumann boundary condition $\mcr N_{\mfr C}(p)$
\begin{gather*}\label{TCBC}
 \text{Chirals in } \bL_{\mfr C}(p) \colon\ \text{Neumann}, \nonumber\\
 \text{Chirals in } \bL^\perp_{\mfr C}(p) \colon\ \text{Dirichlet with zero VEVs}, \nonumber\\
 \text{Vector multiplet} \colon\ \text{Neumann}, \nonumber\\
 \text{Coupled to the boundary} \colon\ \text{2d } \cN=(0,2) \text{ theory } \Upsilon_{\mfr C}(p).
\end{gather*}

The relation between the thimble boundary conditions \rf{thimbleBC} and vacua are relatively straightforward. However, in computation, we shall use this new boundary condition $\mcr N_{\mfr C}(p)$ to represent the vacuum $p$. In \cite[Section~6.3]{DedushenkoNekrasov202109}, Dedushenko--Nekrasov argue that both $\ket{\mcr D_{\mfr C}(p)} \bra{\mcr N_{\mfr C}(p)}$ and~${\ket{\mcr N_{\mfr C}(p)} \bra{\mcr D_{\mfr C}(p)}}$ are projections onto the vacuum $p$ in cohomology, where $\bra{\cdot}$ and $\ket{\cdot}$ refer to future and past boundary condition, respectively. A necessary condition for this to be true is that these projectors are free of 't Hooft anomaly
\beq
 P_\pm(\mcr D_{\mfr C}(p))\big|_{{\mbf f} = {\mbf f}^{(p)}({\mbf f}_x, {\mbf f}_\hbar)} + P_\mp(\mcr N_{\mfr C}(p)) = 0 . \label{anomalymatching}
\eeq
Here $P_+$ and $P_-$ refer to anomaly polynomials of the future and past boundary conditions, respectively. In this paper, we only describe the $\mcr N_{\mfr C}(p)$ boundary condition and show that this anomaly condition is satisfied, beyond that, we take their equivalence as an ansatz. For more details, see \cite{DedushenkoNekrasov202109}.

{\bf Boundary theory $\boldsymbol{\Upsilon_{\mfr C}(p)}$.} At the boundary preserving $\cN=(0,2)$ supersymmetry, we have~$\U(1)^N$ gauge symmetry. So we can couple a 2d $\cN=(0,2)$ theory with $\U(1)^N$ global symmetry by gauging it as long as the anomaly matching condition \rf{anomalymatching} is satisfied. It is straightforward to write down a boundary theory for any chamber, but for the sake of explicit formulas we shall focus here on the following two chambers
\beq
 \mfr C_1 := m_1 < m_2 < \cdots < m_L , \qquad
 \mfr C_2 := m_1 > m_2 > \cdots > m_L .
 \label{C01}
\eeq
For $\mfr C_1$, we propose the following field content and charges for the theory under various symmetry groups applying the procedure of field content of $\Upsilon_{\mfr C_1}(p)$:
$$
\begin{array}{|c|ccc|}
\hline
& \U(1)_a & \U(1)_\hbar & \U(1)_\text{top} \\
\hline
\text{Chiral}_{1,a} & -1 & 1 & 0 \\
\text{Chiral}_{2,a} & 0 & \frac{L}{2} - p(a) & 1 \\[1mm]
\text{Fermi}_{1,a} & 1 & \frac{L}{2} - p(a) & 1 \\
\text{Fermi}_{2,a} & 0 & 1 & 0 \\
\hline
\end{array} .
$$
We have one set of these fields for each $a=1,\dots, N$, totaling in $2N$ Fermi and $2N$ chiral multiplets. The $\U(1)_a$ global symmetry is identified with a $\U(1)$ subgroup of the bulk ${\U(1)^N \times \U(1)^{L-1} \times \U(1)_\hbar}$ symmetry such that the 2d $\U(1)_a$-fugacity is identified with the bulk fugacity $s_a x^{-1}_{p(a)} \hbar^{1/2}$. The remaining symmetries are identified straightforwardly. The elliptic genus of this theory is \cite{BeniniEagerHoriTachikawa201305, BeniniEagerHoriTachikawa201308}
\beq
 \bW(\Upsilon_{\mfr C_1}(p)) = \prod_{a=1}^N \frac{\vth\bigl( s_a x_{p(a)}^{-1} \hbar^{\frac{1}{2} + \frac{L}{2} - p(a)} z \bigr) \vth( \hbar )}
 {\vth\bigl( s^{-1}_a x_{p(a)} \hbar^{\frac{1}{2}} \bigr) \vth\bigl( z \hbar^{\frac{L}{2} - p(a)} \bigr)} . \label{EGUp}
\eeq
Here we are using the theta functions defined in the notation of \cite{AganagicOkounkov201604}
\begin{equation}\label{eq:elliptic theta function in the aganagic-okounkov convention}
 \vartheta(t;q):= \bigl(t^{\frac{1}{2}}-t^{-\frac{1}{2}}\bigr)\mbs{\prod}_{n=1}^\infty\bigl(1-tq^n\bigr)\bigl(1-t^{-1}q^n\bigr) .
\end{equation}
Some of its identities are
 \begin{gather*}
 \vartheta(1;q)=0 , \qquad \vartheta\bigl(t^{-1};q\bigr)=-\vartheta(t;q) , \qquad \vartheta\bigl({\rm e}^{2\pi \mfk{i}}t;q\bigr)=-\vartheta(t;q) ,
 \\
 \vartheta\bigl(q^kt;q\bigr)=q^{-\frac{k^2}{2}}(-t)^{-k}\vartheta(t;q),\qquad
 k\in\mbb{Z} .
 \end{gather*}
This theta function can be related to Jacobi's theta function and the Dedekind eta function as follows:
\[
 q^{\frac{1}{12}} \vth(t; q) = \frac{\ii \tht_1\bigl(\frac{\ln t}{2\pi\ii} ; \frac{\ln q}{2\pi \ii} \bigr)}{\eta\bigl(\frac{\ln q}{2\pi\ii}\bigr)} .
\]
We omit the argument $q$ in the theta function when it is clear from the context.

 For the chamber $\mfr C_2$, we simply write down the elliptic genus
\beq
 \bW(\Upsilon_{\mfr C_2}(p)) = \prod_{a=1}^N \frac{\vth\bigl( s_a x_{p(a)}^{-1} \hbar^{-\frac{1}{2} - \frac{L}{2} + p(a)} z \bigr) \vth( \hbar )}
 {\vth\bigl( s^{-1}_a x_{p(a)} \hbar^{\frac{1}{2}} \bigr) \vth\bigl( z \hbar^{-\frac{L}{2} + p(a) - 1} \bigr)} . \label{EGUpOp}
\eeq
The field content of the theory and their charges can be read off from this if needed. Notice that when the screening conditions \rf{screening} are imposed the elliptic genera reduce to 1.

 The choice of the boundary theory is defined only up to the addition of anomaly-free 2d $\cN=(0,2)$ theories. This leads to some ambiguities regarding the overall partition function, and consequently regarding the stable envelope. We can fix some of this ambiguity by imposing certain normalization conditions on the stable envelopes (as was done in their definition in Theorem~\ref{thm: stable envelope}). The normalization condition restricts our choice of boundary theory to those without anomaly and whose elliptic genus reduces to 1 when the screening conditions \rf{screening} are imposed. Even then the choice is not unique. Our choice differs minimally from the closely related example of the $T^*\text{Gr}(N,L)$ Grassmannian theory with 3d $\cN=4$ symmetry from \cite[Section~6.5.3]{DedushenkoNekrasov202109}.

{\bf Boundary 't Hooft anomaly.}
The enriched Neumann boundary condition \rf{TCBC} receives contributions to its anomaly polynomial from Chern--Simons terms, boundary conditions on chirals and vectors, and the boundary theory
\[
 P(\mcr N_{\mfr C}(p)) := P^\text{CS}_-(\mcr N_{\mfr C}(p)) + P^\partial(\mcr N_{\mfr C}(p)) + P(\Upsilon_{\mfr C}(p)) .
\]
The Chern--Simons anomalies only depend on the orientation of the boundary, $P^\text{CS}_-(\mcr N_{\mfr C}(p)) = P^\text{CS}_-(\mcr D_{\mfr C}(p))$. Whereas, the boundary conditions on the chirals and the vector multiplet fields are swapped between the thimble \rf{thimbleBC} and the enriched Neumann \rf{TCBC}. So they contribute with opposite signs: $P^\partial_-(\mcr N_{\mfr C}(p)) = -P^\partial_-(\mcr D_{\mfr C}(p))$. Since we have a gauge theory at the boundary, the gauge anomalies from all the different sources must neatly cancel. In fact, the boundary theory is defined to satisfy the constraint that only 't Hooft anomalies survive and they match the negative of the future thimble boundary anomaly \rf{anomalymatching} define $\Upsilon_{\mfr C}(p)$ such that
\[
P_-(\mcr N_{\mfr C}(p)) = -P_+(\mcr D_{\mfr C}(p))\big|_{{\mbf f} = {\mbf f}^{(p)}({\mbf f}_x, {\mbf f}_\hbar)} .
\]
The anomaly of the elliptic genera \rf{EGUp} and \rf{EGUpOp} satisfy this condition for the chambers $\mfr C_1$ and $\mfr C_2$ \rf{C01} respectively. For reference, we note that in any chamber $\mfr C$, the anomaly of the boundary theory must be
\[
 P(\Upsilon_{\mfr C}(p)) = - 2 (\tr_{\bL_{\mfr C}(p)} - 2\tr_{\cR_0(p)} ) \bigl(\bigl({\mbf f} - {\mbf f}^{(p)}\bigr) {\mbf f}_\hbar \bigr) + 2 \tr\bigl(\bigl({\mbf f} - {\mbf f}^{(p)}\bigr) {\mbf f}_\text{top}\bigr) .
\]
It can be checked that the elliptic genera \rf{EGUp} and \rf{EGUpOp} contribute the appropriate anomalies for their respective chambers.

\subsubsection[Massless boundary condition, $\mcr B_{\bL(p)}(p)$]{Massless boundary condition, $\boldsymbol{\mcr B_{\bL(p)}(p)}$} \label{sec:masslessBC}
Broadly speaking, stable envelope is a map from the cohomology of the flavor fixed points of the Higgs branch to the cohomology of the Higgs branch itself \rf{HXAtoHX}. To compute this map, we compute the partition function of the 3d theory on $[y_-, y_+] \times \bE_\tau$ such that at the past boundary we have twisted masses which eventually vanish at the future boundary \rf{eq:the gauge-theoretic definition of elliptic stable envelope}. At the past boundary, we can, in principle, use either the thimble (see Section~\ref{sec:thimble}) or the enriched Neumann (see Section~\ref{sec:neumann}) boundary conditions both of which are labeled by connected components of the fixed point set and provide a basis for the cohomology of the fixed point set. As vector spaces, the cohomology of the Higgs branch is isomorphic to that of its fixed points. We, therefore, look for a set of boundary conditions to impose at the future massless boundary parameterized by connected components of the fixed point set as well, such that we can write down the matrix elements of the stable envelope in terms of pairs of connected components. Such a set of boundary conditions was proposed by Dedushenko--Nekrasov in \cite{DedushenkoNekrasov202109}. These boundary conditions are very similar to the thimble boundary conditions. However, since there are no masses at this boundary, there is no choice of chamber and we can not choose a polarization, as we did for the thimbles in~\rf{LMdecomp}, based on positive and negative masses. More specifically, there are no obvious candidates for $\mbf M_{\mfr C;\pm}(p)$. The massless parts of the polarization, namely $\cR_0(p)$ and~\smash{$\wtd\cR_0(p)$} still make sense as they are chamber independent. The proposal of Dedushenko--Nekrasov in this case is to arbitrarily choose a polarization of the chirals
$\mbf M = \bL(p) \oplus \bL^\perp(p)$,
as long as they satisfy (cf.\ \rf{LMdecomp})
\beq
 \cR_0(p) \subseteq \bL(p) , \qquad \wtd\cR_0(p) \subseteq \bL^\perp(p) . \label{RinL}
\eeq
In all our computations, we make the following choice
$ \bL(p) := \cR$, \smash{$ \bL^\perp(p) := \wtd\cR$},
which obviously satisfies \rf{RinL}.

 Once the choice of polarization is made, this boundary condition, which we call $\mcr B_{\bL(p)}(p)$, can be described just as in \rf{thimbleBC}:
\[\begin{aligned}
 &\text{Chirals in} \bL(p) \colon\ \text{Dirichlet with VEVs according to vacuum}\ p, \\
 &\text{Chirals in} \bL^\perp(p) \colon\ \text{Neumann}, \\
 &\text{Vector multiplet} \colon\ \text{Dirichlet with screening \rf{screening} of gauge holonomy.}
\end{aligned}\label{masslessBC}\]

\subsubsection[Janus interface, $\cJ(m_{\mfr C}, 0)$]{Janus interface, $\boldsymbol{\cJ(m_{\mfr C}, 0)}$}

The transition from a massive boundary in the past labeled by a fixed point $p$ to a massless boundary in the future is achieved via a mass Janus interface. We assume that the masses in the past belong to the chamber $\mfr C$. Take $y_-<0$ and $y_+>0$ and put the Janus interface at $y=0$ wrapping the elliptic curve $\bE_\tau$.

 We only vary the $F_{\msf A}=\SU(L)$ twisted masses in the theory across the Janus interface. We do not turn on any twisted masses for the $\U(1)_\hbar$ flavor symmetry. In principle, the mass Jauns is implemented by giving a $y$-dependent background value to the adjoint scalar $\si_{\mfr f_{\msf A}}$ in the flavor vector multiplet. The background value interpolates between some generic mass~$m$~\rf{twistedm} in the chamber $\mfr C$ and zero. This is supersymmetric as long as we also turn on the auxiliary field in the multiplet satisfying the BPS equation \rf{flowsiF} (with $\si_\hbar = \sfD_\hbar = 0$ such that ${\si_{\mfr f} = \si_{\mfr f_{\msf A}}}$ and~${\sfD_F = \sfD_{F_{\msf A}}}$). In particular, we can choose a background profile for the mass that is constant on the intervals $[y_-,0)$ and $(0, y+]$ and changes rapidly near $y=0$. Similar Janus interfaces in~3d~${\cN=4}$ theories were previously used to compute stable envelopes \cite{DedushenkoNekrasov202109} and $R$-matrices \cite{BullimoreKimLukowski201708, BullimoreZhang202109} for bosonic spin chains.

 The supersymmetric quantities we compute are supposed to be renormalization group invariants and capture the infrared dynamics. We, therefore, do not lose anything by considering the nonzero masses to be arbitrarily large. In this limit, we can integrate out the massive fields from the interval $[y_-, 0)$. But these fields are all present on the other side $(0, y_+]$. Thus, the fields that acquire masses in the massive vacuum $p$, including chirals and W-bosons, terminate at the interface when approaching from the massless side. This is implemented by a set of half-BPS boundary conditions at $y=0$ that were worked out in \cite{DedushenkoNekrasov202109}. We are simply going to use their prescription.

 The chirals with positive and negative masses in the vacuum $p$ belong to $\mbf M_{\mfr C;+}$ and $\mbf M_{\mfr C;-}$ respectively, according to our earlier definition of these spaces \rf{XweightDecomp}. To characterize the massive W-bosons, let us decompose the set $\De$ of roots of the gauge Lie algebra into subsets corresponding to the sign of their masses at the vacuum $p$. At $p$, the adjoint scalar in the vector multiplet takes the VEV $\si_p(m)$ given by \rf{siP}. Therefore, the vector field associated to the root $\al \in \De$ picks up a mass $\al(\si_p(m))$. The decomposition of the roots can then be written as~$ \De= \De_+(p) \sqcup \De_0 \sqcup \De_-(p)$,
where
\begin{gather*}
 \De_+(p) := \{\al \in \De | \al(\si_p(m)) > 0\} ,
 \qquad
 \De_0 := \{\al \in \De | \al(\si_p(m)) = 0\} ,
 \\
 \De_-(p) := \{\al \in \De | \al(\si_p(m)) < 0\} .
\end{gather*}
For generic masses, always only the Cartan valued fields remain massless, so $\De_0 = \{0\}$ in any vacuum.

 The interface is then implemented by the following boundary conditions on fields \cite{DedushenkoNekrasov202109}:
\beq\begin{aligned}
 &\text{Chirals in } \mbf M_{\mfr C;+}(p)\colon\ \text{Dirichlet}, \\
 & \text{Chirals in } \mbf M_{\mfr C;-}(p)\colon\ \text{Neumann}, \\
 & \text{Vectors in } \De_+(p)\colon\ \text{Dirichlet}, \\
 & \text{Vectors in } \De_-(p) \colon\ \text{Neumann}.
\end{aligned}\label{janusBC}
\eeq
From the definition of the Janus interface as a supersymmetric position-dependent background profile for the masses, it is clear that it does not introduce any anomaly. However, in practice, we implement the Janus interface effectively using the above-mentioned boundary conditions and it is not obvious that they are anomaly free. After all, both the boundary conditions and the Chern--Simons terms created by integrating out massive fermions contribute to anomaly at the boundary. It is however a simple exercise to check that these two anomalies precisely cancel each other, as was checked explicitly in \cite{DedushenkoNekrasov202109}.

\section[Stable envelopes and $R$-matrix from gauge theory: the computation]{Stable envelopes and $\boldsymbol{R}$-matrix from gauge theory:\\ the computation}
\label{sec:stable envelopes and the $R$-matrix from gauge theory}

\subsection{Elliptic stable envelope}\label{sec:explicit construction of elliptic stable envelope}
We are now in a position to compute the matrix elements of the stable envelope of $\sl(1|1)$ elliptic spin chains. For the purpose of BPS computations, we can shrink the interval $[y_-, y_+]$ to a point, resulting in an effective 2d $\cN=(0,2)$ theory on $\bE_\tau$. Then the 3d $\cN=2$ partition function on~${[y_-, y_+]\times \bE_\tau}$ is given by the elliptic genus of this effective 2d theory. To derive the field content of this effective 2d theory, we need to look at the intervals $[y_-, 0]$ and $[0, y_+]$, in addition to their union, and determine which 3d multiplets have compatible boundary conditions at both ends of all intervals where it is defined. If any field living on $[y_\text{start}, y_\text{end}]$ has opposite boundary conditions at $y_\text{start}$ and $y_\text{end}$ then it vanishes in the effective theory. A 3d multiplet that has the same boundary condition at both ends of an interval survives the collision of boundaries and decomposes into 2d $\cN=(0,2)$ multiplets. Of these $(0,2)$ multiplets, the ones with Dirichlet boundary conditions are simply frozen to their boundary values. The remaining~$(0,2)$ multiplets with Neumann boundary conditions, along with the theory $\Upsilon_{\mfr C_1}(p_1)$ at the past massive boundary, comprise the field content of the effective 2d theory. Note that throughout the computation of the stable envelope, we use the specific chamber \rf{C01} to write down formulas.

The vector fields that remain massless in the vacuum $p_1$ are part of the theory along the entire interval $[y_-, y_+]$ and they have opposite boundary conditions at $\ket{\mcr B_{\bL(p_2)}(p_2)}$ (Dirichlet) and $\bra{\mcr N_{\mfr C_1}(p_1)}$ (Neumann). These vector fields, therefore, do not contribute to the elliptic genus. All the vector fields receive Dirichlet boundary condition at $\ket{\mcr B_{\bL(p_2)}(p_2)}$ and all the vector fields with positive mass in the vacuum $p_1$ receive Dirichlet at the Janus interface. In the IR, these vector multiplets with positive masses in $p_1$ are only part of the massless theory along $y>0$ and therefore they survive. They decompose into 2d $\cN=(0,2)$ vector and adjoint chiral multiplets. Dirichlet for 3d $\cN=2$ vector multiplet imposes Dirichlet and Neumann on its constituent 2d vector and 2d adjoint chiral multiplets respectively \cite{DimofteGaiottoPaquette201712}. Contributions from these adjoint chirals to the elliptic genus are
\beq
 \bV(p_1) := \prod_{\substack{a, b=1 \\ p_1(a) > p_1(b)}}^N \frac{q^{-1/12}}{\vth\bigl(s_a s_b^{-1}\bigr)} . \label{V(p)}
\eeq
These 2d chirals are not charged under any global symmetry group.

Next, let us look at all the 3d $\cN=2$ chiral multiplets that survive with Dirichlet boundary conditions. Multiplets receiving Dirichlet at $\ket{\mcr N_{\mfr C_1}(p_1)}$ are in $\bL^\perp_{\mfr C_1}(p_1) = \mbf M_{\mfr C_1;-}(p_1) \oplus \wtd \cR_0(p_1)$ (\rf{TCBC} and \rf{LMdecomp}). But everything in $\mbf M_{\mfr C_1;-}(p_1)$ receives Neumann at the Janus \rf{janusBC} and therefore vanishes in the effective 2d theory. Multiplets in $\wtd \cR_0(p_1)$ extends unconstrained across the Janus and intersects with the multiplets $\bL(p_2)$ receving Dirichlet at $\bra{\mcr B_{\bL(p_2)}(p_2)}$. Similarly, everything in $\mbf M_{\mfr C_1;+}(p_1)$ receives Dirichlet at the Janus and intersects with multiplets receiving Dirichlet at $\bra{\mcr B_{\bL(p_2)}(p_2)}$. Therefore, 3d chiral multiplets surviving the collision of boundaries with Dirichlet are from
\[
 \bL(p_2) \cap \bigl(\wtd \cR_0(p_1) \oplus \mbf M_{\mfr C_1;+}(p_1)\bigr) = \cR \cap \bigl(\wtd \cR_0(p_1) \oplus \cR_{\mfr C_1;+}(p_1) \oplus \wtd \cR_{\mfr C_1;+}(p_1)\bigr) = \cR_{\mfr C_1;+}(p_1) .
\]
Dirichlet on a 3d chiral multiplet means Dirichlet and Neumann on its constituent 2d $\cN=(0,2)$ chiral and fermi multiplets respectively. Fields with Dirichlet are frozen at their boundary values in the effective 2d theory and the remaining fermi multiplets contribute to the elliptic genus
\[
 \bM_\text{fe}(p_1) := \prod_{a=1}^N \prod_{\substack{i=1 \\ i < p_1(a)}}^L q^{1/12} \vth\bigl(s_a x_i^{-1} \hbar^{1/2} \bigr) .
\]
Now we look for the 3d $\cN=2$ chiral multiplets with Neumann boundary condition contributing to the effective 2d theory. At $\ket{\mcr N_{\mfr C_1}(p_1)}$ all fields in $\bL_{\mfr C_1}(p_1) = \mbf M_{\mfr C_1;+} \oplus \cR_{0}(p_1)$ get Neumann. Only $\cR_0(p_1)$ passes through the Janus unconstrained. On the interval $[0, y_+]$ all fields with negative masses in the vacuum $p_1$ receive Neumann boundary condition at the Janus interface. Taking the intersection of all these fields with the fields receiving Neumann at $\bra{\mcr B_\bL(p_2)}$, we find
\[
 \bL^\perp(p_2) \cap ( \cR_0(p_1) \oplus \mbf M_{\mfr C_1;-}(p_1) ) = \wtd \cR \cap \bigl( \cR_0(p_1) \oplus \cR_{\mfr C_1;-}(p_1) \oplus \wtd \cR_{\mfr C_1;-}(p_1) \bigr) = \wtd\cR_{\mfr C_1;-}(p_1) .
\]
Neumann on a 3d chiral multiplet means Neumann and Dirichlet on its constituent 2d $\cN=(0,2)$ chiral and fermi multiplets respectively \cite{DimofteGaiottoPaquette201712}. These 2d chirals contribute to the elliptic genus
\[
 \bM_\text{ch}(p_1) := \prod_{a=1}^N \prod_{\substack{i=1 \\ i < p_1(a)}}^L \frac{q^{-1/12}}{\vth\bigl(s_a^{-1} x_i \hbar^{1/2} \bigr)} .
\]
Lastly, we have the contribution \rf{EGUp} of the boundary theory $\Up_{\mfr C_1}(p_1)$ to the elliptic genus. Putting all these together, we find the (unnormalized) elliptic stable envelope%
\beq
 \wtd{\estab }_{\mfr C_1}(p) = \text{Sym}_{S_N} \left(\bV(p) \bM_\text{fe}(p) \bM_\text{ch}(p) \bW(\Upsilon_{\mfr C_1}(p)) \right) . \label{Stab1}
\eeq
The symmetrization is over the gauge fugacities. In addition to the maximal torus $\U(1)^N$, the Weyl group $S_N$ of $\U(N)$ also remains as gauge symmetry at the massive vacuum $p_1$ and it acts on the gauge fugacities by permuting them.

Let us define
\beq
 \efunction_{\mfr C_1, m}(s, \bm x, \hbar, z) := \Biggl(\prod_{i < m} \vth(s x_i ) \Biggr)\frac{\vth\bigl(s x_m z \hbar^{m - L}\bigr)}{\vth\bigl( z \hbar^{m - L} \bigr)} \Biggl( \prod_{i > m} \vth(s x_i \hbar ) \Biggr) . \label{fC0m}
\eeq
Then the stable envelope \rf{Stab1} becomes
\[
 \wtd{\estab }_{\mfr C_1}(p) = \frac{q^{- \binom{N}{2}/12} \vth(\tilde \hbar )^N}{\prod_{a,i} \vth\bigl( s_a \tilde x_i \tilde \hbar \bigr)} \estab _{\mfr C_1}(p) ,
\]
where
\begin{align}
 \estab _{\mfr C_1}(p) :={}& (-1)^{\sum_a\#(i<p(a))}\nonumber\\
 &\times\text{Sym}_{S_N} \Biggl[ \Biggl(\prod_a \efunction_{\mfr C_1, p(a)}(s_a, \tilde{\bm x}, \tilde \hbar, \tilde z)\Biggr) \Biggl( \prod_{p(a) > p(b)} \frac{1}{\vth\bigl(s_a s_b^{-1}\bigr)} \Biggr) \Biggr] ,
\label{Stab2}
\end{align}
and we have also redefined the fugacities as
\beq
 \tilde x_i := x^{-1}_i \hbar^{\frac{1}{2}} , \qquad \tilde \hbar := \hbar^{-1} , \qquad \tilde z := z \hbar^{-\frac{L}{2}} . \label{tildeFug}
\eeq
These redefinitions and the specific form of the stable envelope in terms of the function \rf{fC0m} are merely to facilitate comparison with analogous known formulas coming from BPS computations in 3d $\cN=4$ SQCDs \cite{DedushenkoNekrasov202109} or from the theory of elliptic stable envelopes on complex symplectic varieties \cite{AganagicOkounkov201604} that correspond to bosonic XYZ spin chains.

\begin{rmk}[normalization of the stable envelopes] \label{rmk:normStab}
The overall factor in \smash{$\wtd{\estab }_{\mfr C_1}(p)$} is independent of the choice of the chamber $\mfr C$ and the vacuum $p$. This makes it ambiguous in the context of defining stable envelopes since we can change it by normalizing all the stable envelopes simultaneously. Only the relative normalization of the stable envelopes associated to different vacua and different chambers is meaningful. We have therefore defined the normalized stable envelope $\estab _{\mfr C_1}(p)$ without any leading factor. This normalized version is used in Section~\ref{sec:ellStab} to check that it satisfies all formal properties of elliptic stable envelopes for $\sl(1|1)$ spin chains. Note however that, the overall sign from \rf{Stab2} is missing in \rf{eqn: stable envelope for M(N,L)} due to a slightly different choice of normalization for the stable envelopes, see footnote \ref{foot:normStab}.
\end{rmk}

In the expression \rf{Stab2}, we have not yet fully imposed the boundary condition $\mcr B_{\bL(p_2)}(p_2)$. We still need to impose the screening condition \rf{screening} associated to the fixed point $p_2$. Doing so completes the computation of the 3d $\cN=2$ interval partition function, providing us with the matrix elements of the stable envelope
\[
 \estab _{\mfr C_1}(p_2, p_1) = \bra{ \mcr B_{\bL(p_2)}(p_2) } \cJ(m, 0) \ket{\mcr N_{\mfr C_1}(p_1)} = \estab _{\mfr C_1}(p_1)\big|_{\bm s = \bm s^{(p_2)}} .
\]
The specific values of the mass $m$ is not important as long as it belongs to the chamber $\mfr C_1$.

\begin{rmk}[SQED and comparison with 3d $\cN=4$/bosonic stable envelopes]\label{rmk:abelianN=2vsN=4}
The difference between the quiver \rf{fig:gl11Quiver} we used to compute elliptic $\sl(1|1)$ stable envelope and the quiver used in \cite{DedushenkoNekrasov202109} to compute the elliptic $\sl(2)$ stable envelope is the lack of a self-loop at the gauge node in the $\sl(1|1)$-quiver. In field theory, this simply means the absence of the 3d $\cN=2$ adjoint chiral multiplet, which is necessarily a part of the 3d $\cN=4$ vector multiplet. In the computation of the last section, the 3d $\cN=4$ vector multiplet contribution to the elliptic genus would involve~$\bV(p)$ from \rf{V(p)} along with the contributions from this adjoint chiral. This is the only difference between the $\sl(1|1)$ and $\sl(2)$ cases. Notice then, in the 1-magnon sector, i.e., for SQEDs, there is no distinction between the two cases. Because the vector multiplet, as well as any adjoint chiral multiplet, do not contribute anything to the elliptic genus in abelian theories. Thus, the~$\sl(1|1)$ stable envelopes and the $R$-matrix in the 1-magnon sector must coincide with those of the $\sl(2)$ spin chains. More generally, the following two quivers will always lead to the same stable envelopes, assuming that the hidden parts of the quivers coincide:
\[
\begin{tikzpicture}[x={(1cm,0cm)}, y={(0cm,1cm)}, baseline=0]
 	\node[draw, circle] at (0,0) (N) {$1$};
 	\node[circle] at (0,-1.25) (D) {$\vdots$};
 	\draw[-stealth] (D.70) to (N.296);
 	\draw[-stealth] (N.244) to (D.110);
 \node[circle] at (1.25,0) (R) {$\cdots$};
 \node[circle] at (-1.25,0) (L) {$\cdots$};
 \draw[-stealth] (N.148) to (L.24);
 \draw[-stealth] (L.-24) to (N.212);
 \draw[-stealth] (N.-32) to (R.204);
 \draw[-stealth] (N.32) to (R.156);
 	\end{tikzpicture}
 \simeq
 \begin{tikzpicture}[x={(1cm,0cm)}, y={(0cm,1cm)}, baseline=0]
 	\node[draw, circle] at (0,0) (N) {$1$};
 	\node[circle] at (0,-1.25) (D) {$\vdots$};
 	\draw[-stealth] (D.70) to (N.296);
 	\draw[-stealth] (N.244) to (D.110);
 \node[circle] at (1.25,0) (R) {$\cdots$};
 \node[circle] at (-1.25,0) (L) {$\cdots$};
 \draw[-stealth] (N.148) to (L.24);
 \draw[-stealth] (L.-24) to (N.212);
 \draw[-stealth] (N.-32) to (R.204);
 \draw[-stealth] (N.32) to (R.156);

 \draw[-stealth] (N.60) to [out=60, in=120, looseness=6] (N.120);
 	\end{tikzpicture}
\]
Another way of saying this is that in a magnon sector with magnon number 1 assigned to a~simple root, the stable envelopes (and consequently the $R$-matrix) will remain invariant if we change the parity of the simple root. The equivalence between even and odd gauge nodes in the case of abelian gauge groups was also pointed out in Example~\ref{ex: abelian stable envelope} from a more mathematical perspective.
\end{rmk}

\subsection[$R$-matrix for elliptic $\sl(1|1)$ spin chains]{$\boldsymbol{R}$-matrix for elliptic $\boldsymbol{\sl(1|1)}$ spin chains}\label{sec:elliptic r-matrix for sl(1|1)}
From the perspective of integrable systems, we are interested in computing the $R$-matrix for supersymmetric XYZ spin chains, not just the stable envelope. As stated in \rf{eqn: dfn of R-matirx}, $R$-matrix is a composition of a stable envelope and an inverse stable envelope.

The action of an $R$-matrix on a spin chain preserves magnon numbers. As a result, an $R$-matrix takes a block diagonal form, one block for each sector with fixed magnon numbers. We focus on the $\sl(1|1)$ spin chain with two fundamental spins.\footnote{The spin chain can be arbitrarily large, but the $R$-matrix always acts on the tensor product of two representations. So for the purpose of computing the $R$-matrix it suffices to consider just two sites. While testing Yang--Baxter we will have to extend the $R$-matrix to act trivially on a third site as well.} The fundamental representation of~$\sl(1|1)$ is $\bC^{1|1}$, where the lowest weight state is bosonic, and the only excited state is fermionic. Thus, the tensor product of two such representations is four complex dimensional, consisting of a~bosonic one-dimensional subspace of magnon number 0, a fermionic two-dimensional subspace of magnon number 1, and a bosonic one-dimensional subspace of magnon number 2. The $R$-matrix also decomposes as
\beq
 \ermatrix = \begin{pmatrix}
 \ermatrix^{(0)} && \\ & \ermatrix^{(1)} & \\ && \ermatrix^{(2)}
 \end{pmatrix} , \label{Rschematic}
\eeq
where the superscript refers to the magnon number.

The gauge theory dual of a spin chain with two fundamental spins for $\sl(1|1)$ corresponds to a 3d $\cN=2$ SQCD with $\SU(2)$ flavor symmetry. Each magnon sector corresponds to a different unitary gauge group, the rank of the gauge group is equal to the magnon number. The zero magnon sector thus corresponds to a trivial gauge theory and accordingly, the $R$-matrix in the 0-magnon sector is always trivial
\beq
 \ermatrix^{(0)}_{\mfr C_2 \leftarrow \mfr C_1} = 1 . \label{R0}
\eeq
In the above formula, $\mfr C_1$ and $\mfr C_2$ are the two chambers for the Cartan of $\su(2)$. To compute the $R$-matrix, we always need stable envelopes in two different chambers. We use two real masses~$m_1$ and $m_2$ to parameterize the Cartan subalgebra of $\su(2)$ with the tracelessness condition~${m_1 + m_2 = 0}$. There are then two chambers for the masses
$\mfr C_1 = m_1 < m_2$, $ \mfr C_2 = m_1 > m_2$.
We computed the stable envelope for $\mfr C_1$ in \rf{Stab2}. The stable envelope in the other chamber can be computed following identical steps as in the previous section. Here we simply write down the answer
\begin{align}
 \estab _{\mfr C_2}(p) := {}&(-1)^{\sum_a \#(i>p(a))}\nonumber \\
 &\times\text{Sym}_{S_N} \Biggl[ \Biggl(\prod_a \efunction_{\mfr C_2, p(a)}(s_a, \bm x, \hbar, z)\Biggr) \Biggl( \prod_{p(a) < p(b)} \frac{1}{\vth\bigl(s_a s_b^{-1}\bigr)} \Biggr) \Biggr] ,
\label{Stab3}\end{align}
where we have defined
\[
 \efunction_{\mfr C_2, m}(s, \bm x, \hbar, z) := \Biggl(\prod_{i > m} \vth(s x_i ) \Biggr)\frac{\vth\bigl(s x_m z \hbar^{1-m}\bigr)}{\vth\bigl( z \hbar^{1-m} \bigr)} \Biggl( \prod_{i < m} \vth(s x_i \hbar ) \Biggr) .
\]
Note that, in this section, all fugacities are the newly defined ones from \rf{tildeFug}, but to avoid clutter we omit the tildes.

Let us now specialize these formulas to the cases of gauge group $\U(1)$ and $\U(2)$ to compute the stable envelopes and the $R$-matrix in the 1 and 2-magnon sectors.

{\bf 1-magnon.}
The Higgs branch (see Example~\ref{ex:conifold}) has two fixed points $p_1$ and $p_2$, corresponding to the two possible maps $\{1\} \hookrightarrow \{1, 2\}$ mapping 1 to either 1 or 2 (cf.\ \rf{LchooseN}). There is a~single gauge fugacity and two screening conditions \rf{screening} for the two fixed points. We can now compute the four matrix elements of each of the two stable envelopes \rf{Stab2} and \rf{Stab3} in the two chambers, we find
\[\begin{aligned}
 \estab ^{(1)}_{\mfr C_1} =& \begin{pmatrix}
 \vth\left(x_1^{-1} x_2 \hbar \right) & 0 \\
 \frac{\vth(x_1 x^{-1}_2 z \hbar^{-1} ) \vth(\hbar)}{\vth(z h^{-1} )} & -\vth\left(x_1 x^{-1}_2\right)
 \end{pmatrix} ,
 \\
 \estab ^{(1)}_{\mfr C_2} =& \begin{pmatrix}
 -\vth\left(x_1^{-1} x_2 \right) & \frac{\vth\left(x^{-1}_1 x_2 z \hbar^{-1} \right) \vth(\hbar)}{\vth(z h^{-1} )} \\
 0 & \vth\left(x_1 x^{-1}_2 \hbar\right)
 \end{pmatrix} .
\end{aligned}\]
The $R$-matrix is a simple matrix product
\begin{align}
 \ermatrix^{(1)}_{\mfr C_2 \ot \mfr C_1} ={}& \estab ^{-1}_{\mfr C_2} \cdot \estab _{\mfr C_1} \nn\\
 ={}& \frac{-1}{\vth\left(x^{-1}_1 x_2 \hbar^{-1}\right)} \begin{pmatrix}
 \frac{\vth(x^{-1}_1 x_2) \vth(z) \vth(z\hbar^{-2})}{\vth(z \hbar^{-1})^2} & \frac{\vth(\hbar) \vth(x^{-1}_1 x_2 z \hbar^{-1})}{\vth(z \hbar^{-1})} \\
 \frac{\vth(\hbar) \vth(x_1 x^{-1}_2 z \hbar^{-1})}{\vth(z \hbar^{-1})} & -\vth(x_1 x^{-1}_2)
 \end{pmatrix} .\label{R1}
\end{align}
To simplify the upper left component, we have used the identity
\[
 \sum_\text{cyclic} \vth(AB) \vth\bigl(AB^{-1}\bigr) \vth(C)^2 = 0 .
\]

{\bf 2-magnon.}
The gauge theory dual of the 2-magnon sector has gauge group $\U(2)$. There are two gauge fugacities and a unique vacuum corresponding to the unique (modulo permutation) map $\{1,2\} \hookrightarrow \{1,2\}$. The stable envelopes are just functions
\[
 \estab ^{(2)}_{\mfr C_1} = -\vth\bigl(x^{-1}_1 x_2 \hbar\bigr) , \qquad
 \estab ^{(2)}_{\mfr C_2} = -\vth\bigl(x_1 x^{-1}_2 \hbar\bigr) .
\]
As is the $R$-matrix,
\beq
 \ermatrix^{(2)}_{\mfr C_2 \ot \mfr C_1} = \estab ^{-1}_{\mfr C_2} \estab _{\mfr C_1} = \frac{\vth\bigl(x^{-1}_1 x_2 \hbar\bigr)}{\vth\bigl(x_1 x^{-1}_2 \hbar\bigr)} . \label{R2}
\eeq

Putting together \rf{R0}, \rf{R1}, and \rf{R2} as in \rf{Rschematic}, we find the complete $R$-matrix for the elliptic $\sl(1|1)$ spin chain
\begin{equation}\label{eq:gl(1|1) elliptic r-matrix}
 \def\arraystretch{1.5}
	\setlength{\arrayrulewidth}{.6pt}
 \setlength{\arraycolsep}{6pt}
 \ermatrix_{\mfk{C}_2\ot\mfk{C}_1}=
 \left(\begin{array}{c|cc|c}
 1 & 0 & 0 & 0 \\
 \hline
 0 & -\frac{\vth (x^{-1}_1 x_2 ) \vth(z) \vth(z\hbar^{-2})}{\vth (x^{-1}_1 x_2 \hbar^{-1} ) \vth(z \hbar^{-1})^2} & \frac{\vth(\hbar) \vth (x^{-1}_1 x_2 z \hbar^{-1} )}{\vth (x_1 x^{-1}_2 \hbar ) \vth (z \hbar^{-1} )} & 0
 \\[6pt]
 0 & \frac{\vth(\hbar) \vth (x_1 x^{-1}_2 z \hbar^{-1} )}{\vth (x_1 x^{-1}_2 \hbar ) \vth (z \hbar^{-1} )} & - \frac{\vth (x^{-1}_1 x_2 )}{\vth (x^{-1}_1 x_2 \hbar^{-1} )} & 0
 \\[4pt]
 \hline
 0 & 0 & 0 & \frac{\vth (x^{-1}_1 x_2 \hbar )}{\vth (x_1 x^{-1}_2 \hbar )}
 \end{array}\right) .
\end{equation}
Here $x_1$ and $x_2$ are referred to as the spectral parameters associated to the two sites, $z$ as the dynamical parameter, and $\hbar$ as the quantization parameter. This $R$-matrix solves the dynamical Yang--Baxter equation for $\sl(1|1)$ spin chains \rf{eqn: DYBE for sl(1|1)}. Note that, relative to the $R$-matrix presented in \rf{eqn: quiver $R$-matrix for sl(1|1)}, the signs in front of the two diagonal entries in the 1-magnon sector are different in the above formula. This is due to the slightly different choice of normalization of the stable envelopes, as pointed out in Remark~\ref{rmk:normStab}.

\section[2d and 1d avatars of elliptic stable envelopes and the $R$-matrix]{2d and 1d avatars of elliptic stable envelopes and the $\boldsymbol{R}$-matrix}
\label{sec:2d and 1d avatars of elliptic stable envelopes}

This section is devoted to taking the 3d $\to$ 2d and 2d $\to$ 1d reduction of elliptic stable envelopes. In Section~\ref{sec:the K-theory limit}, we take the 3d $\to$ 2d reduction, which produces the K-theoretic stable envelopes and the trigonometric $R$-matrix for $\mfk{sl}(1|1)$. The 2d $\to$ 1d reduction, which produces cohomological stable envelopes and the rational $R$-matrix, is the subject of Section~\ref{sec:the cohomology limit}. In the case of cohomological stable envelopes and the rational $R$-matrix for $\mfk{sl}(1|1)$, we compare our result with that of Rim\'anyi and Rozansky in \cite{RimanyiRozansky202105} and find perfect agreement. The discussion in this section is based on Section~\ref{sec:k-theory and cohomology limit} and also the physical description in \cite[Section~2.1]{DedushenkoNekrasov202109}.

\subsection{The K-theory limit}\label{sec:the K-theory limit}

In this section, we reduce the formulas for elliptic stable envelopes to two dimensions and give a prediction for K-theoretic stable envelopes associated with Lie superalgebra $\mfk{sl}(1|1)$.

{\bf Procedure for 3d $\boldsymbol{\to}$ 2d reduction.} In Section~\ref{sec:3d N=2 SQCD and its parameters}, we explained the parameters in three-dimensional theory on $I\times\mbb{E}_\tau$. One could shrink one of the cycles of $\mbb{E}_\tau$, corresponding to $\tau\to 0$, to get a two-dimensional theory $S^1_\A\times I$. The procedure of going from three to two dimensions is that we shrink the B-cycles of $\mbb{E}_\tau$. This limit could be better handled if we take our parameters to take value in the $S$-transformed elliptic curve \smash{$\mbb{E}_{-\frac{1}{\tau}}$}. Then, the two-dimensional limit is given by
$\tau\to 0$, \smash{$ q^{-\frac{2\pi\mfk{i}}{\tau}}\to 0$}.
In this limit, the elliptic curve becomes nodal and the smooth locus is isomorphic to $\mbb{C}^\times$. It turns out that the K-theoretic stable envelope depends on the choice of a slope. This makes the behavior of elliptic K\"ahler parameters in the $q\to 0$ limit more subtle. Defining the slope parameter as \eqref{eq:the k-theory slope parameter} such that $\msf{s}$ is generic (i.e., it belongs to an alcove). The precise procedure to take the K-theory limit is explained in Section~\ref{sec:k-theory and cohomology limit}.

In the case of $\mcal{N}=2$ SQCD, we need two results for deriving the K-theoretic stable envelopes:
\begin{itemize}\itemsep=0pt
\item[(1)] the behavior of this function in the limit $q\to 0$, that can be deduced from its very definition \eqref{eq:elliptic theta function in the aganagic-okounkov convention}
\smash{$\lim_{q\to 0}\vartheta(t;q)=t^{\frac{1}{2}}-t^{-\frac{1}{2}}$}.
\item[(2)] and that \cite[equation~(2)]{AganagicOkounkov201604}
\begin{equation}\label{eq:q->0 limit of ratio of theta functions}
 \lim_{q\to 0}\left.\left(\frac{\vartheta(t_1z;q)}{\vartheta(t_2z;q)}\right)\right|_{z=q^{\msf{s}}}=\left(\frac{t_1}{t_2}\right)^{\lfloor\msf{s}\rfloor+\frac{1}{2}}.
\end{equation}
\end{itemize}

Finally, we note that all the holonomies introduced in Section~\ref{sec:3d N=2 SQCD and its parameters} are still multiplicative but are now valued in $\mbb{C}^\times$.

{\bf K-theoretic stable envelopes for $\boldsymbol{\mfk{sl}(1|1)}$.}
We now perform a 3d $\to$ 2d reduction of the elliptic stable envelope, which gives us the K-theoretic stable envelopes.

From the definition of the function $\efunction_{\mfk{C}_1,m}$ \eqref{fC0m} and \eqref{eq:q->0 limit of ratio of theta functions}, we find that
\[
 \lim_{q\to 0}\Biggl[\frac{\vartheta\bigl(s_ax_{p(a)}\hbar^{-(L-p(a))}z\bigr)}{\vartheta\bigl(\hbar^{-(L-p(a))}z\bigr)}\Bigg|_{z=q^{\msf{s}}}\Biggr]=(s_ax_{p(a)})^{\lfloor{\msf{s}}\rfloor+\frac{1}{2}}.
\]
Hence,
\begin{align*}
 \lim_{q\to 0}\efunction_{\mfk{C}_1,m}(s,\mbs{x},\hbar,z)&=\biggl(\hbar^{\frac{\#(i>p(a))}{2}}\prod_{i}(sx_i)^{\frac{1}{2}}\bigg)(sx_m)^{\lfloor\msf s\rfloor}\prod_{i<m}\bigl(1-(sx_i)^{-1}\bigr)\prod_{i>m}\bigl(1-(sx_i\hbar)^{-1}\bigr)
 \\
 &=\biggl(\hbar^{\frac{\#(i>p(a))}{2}}\prod_{i}(sx_i)^{\frac{1}{2}}\bigg)\tfunction_{\mfk{C}_1,\msf{s},m}(s,\mbs{x},\hbar),
\end{align*}
where we have defined
\[
 \tfunction_{\mfk{C}_1,\msf{s},m}(s,\mbs{x},\hbar):= (sx_{m})^{\lfloor{\msf{s}}\rfloor}\prod_{i<m}\bigl(1-(sx_i)^{-1}\bigr)\prod_{i>m}\bigl(1-(sx_i\hbar)^{-1}\bigr).
\]
Furthermore,
\[
 \lim_{q\to 0}\left(\prod_{\substack{a>b}}\frac{1}{\vartheta\bigl(s_as_b^{-1}\bigr)}\right)=\prod_{\substack{a>b}}\frac{1}{\bigl(s_as_b^{-1}\bigr)^{\frac{1}{2}}-\bigl(s_as_b^{-1}\bigr)^{-\frac{1}{2}}}.
\]
Finally, we know from \eqref{eq:k-theoretic stable envelope as the limit of elliptic one} that we need to multiply the result by two factors: (1) the inverse of the square root of the determinant of the partial polarization \smash{$(\det\text{Pol}_{\msf{X}})^{-\frac{1}{2}}$} for $\msf{X}=\mcal{M}_H(N,L)$, which is given by the product of fugacities of the quarks $Q$
\[
(\det\text{Pol}_{\msf{X}})^{-\frac{1}{2}}=\prod_{a,i}\frac{1}{\bigl(s_ax_i^{-1}\hbar^{1/2}\bigr)^{\frac{1}{2}}}.
\]
However, we need to implement the redefinition of fugacities according to \eqref{tildeFug}. Therefore,
\[
 (\det\text{Pol}_{\msf{X}})^{-\frac{1}{2}}=\prod_{a,i}\frac{1}{(s_ax_i)^{\frac{1}{2}}},
\]
(2) the square root of the determinant of the polarization of the fixed loci $(\det\text{Pol}_{\msf{X}^{\msf{A}}})^{\frac{1}{2}}$, which in our case is trivial
\smash{$(\det\text{Pol}_{\msf{X}^{\msf{A}}})^{+\frac{1}{2}}=1$}.
Putting these results together, and using \eqref{eq:k-theoretic stable envelope as the limit of elliptic one}, we see that\footnote{The K-theory limit of elliptic stable envelope for the resolved determinantal variety $\mcal{M}_H(N,L)$ has also been computed in \eqref{eqn: K-theoretic stable envelope for M(N,L)} using a different polarization as explained there. Compared to \eqref{eqn: K-theoretic stable envelope for M(N,L)}, the expression \eqref{eq:k-theoretic stable envelope coming from reduction of bps computation} contains the extra factor $(-1)^{\sum_a\#(i<p(a))|}$, which is the result of a choice of different polarization in the BPS computation of elliptic stable envelope.}
\begin{align}
 \tstab_{\mfk{C}_1,\msf{s}}(p)={}&(-1)^{\sum_a\#(i<p(a))}\hbar^{\frac{1}{2}\sum_a\#(i>p(a))}\nonumber
 \\
 &\times \tenofo{Sym}_{S_N}\left[\left(\prod_{a=1}^N \tfunction_{\mfk{C}_1,\msf{s},p(a)}(s_a,\mbs{x},\hbar)\right)\cdot
 \Biggl(\prod_{\substack{a>b}}\frac{1}{\bigl(s_as_b^{-1}\bigr)^{\frac{1}{2}}-\bigl(s_as_b^{-1}\bigr)^{-\frac{1}{2}}}\Biggr)\right].\label{eq:k-theoretic stable envelope coming from reduction of bps computation}
\end{align}
By Lemma \ref{cor:k-theoretic stable envelope from the elliptic one}, this is the K-theoretic stable envelope for the Lie superalgebra $\mfk{sl}(1|1)$. Notice that it is manifestly a locally-constant function of $\msf{s}$, which is a characteristic feature of K-theoretic stable envelopes as a function of $\msf{s}$.

{\bf Trigonometric $\boldsymbol{R}$-matrix for $\boldsymbol{\mfk{sl}(1|1)}$.}
We can now perform the reduction of the $\mfk{sl}(1|1)$ elliptic $R$-matrix \eqref{eq:gl(1|1) elliptic r-matrix} to construct the~$\mfk{sl}(1|1)$ trigonometric $R$-matrix. Let us define the function~\smash{$\msf{G}(x):= x^{\frac{1}{2}}-x^{-\frac{1}{2}}$}.
The result of the reduction is
\begin{equation}\label{eq:gl(1|1) trigonometric r-matrix}
	\def\arraystretch{1.5}
	\setlength{\arrayrulewidth}{.8pt}
 \setlength{\arraycolsep}{6pt}
	\trmatrix_{\mfk{C}_2\ot\mfk{C}_1}(u)=
	\left(
	\begin{array}{c|cc|c}
		1 & 0 & 0 & 0
		\\
		\hline
		0 & \frac{\msf{G}(u)}{\msf{G} (u^{-1}\hbar )} & u^{\lfloor \msf{s}\rfloor+\frac{1}{2}}\frac{\msf{G}(\hbar)}{\msf{G} (u^{-1}\hbar )} & 0
 \\[6pt]
 0 & u^{-\lfloor \msf{s}\rfloor-\frac{1}{2}}\frac{\msf{G}(\hbar)}{\msf{G} (u^{-1}\hbar )} & \frac{\msf{G}(u)}{\msf{G} (u^{-1}\hbar )} & 0
		\\[4pt]
		\hline
		0 & 0 & 0 & \frac{\msf{G}(u\hbar)}{\msf{G} (u^{-1}\hbar )}
	\end{array}\right).
\end{equation}
One can check that it satisfies the trigonometric version of Yang--Baxter equations given by
\[
 \trmatrix_{12}(u/v) \trmatrix_{13}(u/w) \trmatrix_{23}(v/w)= \trmatrix_{23}(v/w) \trmatrix_{13}(u/w) \trmatrix_{12}(u/v).
\]

\subsection{The cohomology limit}\label{sec:the cohomology limit}

In the previous section, we have constructed the K-theoretic stable envelopes and the trigonometric $R$-matrix for Lie superalgebra $\mfk{sl}(1|1)$ by performing a 3d $\to$ 2d reduction. It is then natural to consider a further 2d $\to$ 1d reduction to construct the cohomological stable envelopes and the rational $R$-matrix for $\mfk{sl}(1|1)$. Fortunately, the mathematical construction of cohomological stable envelope for Lie superalgebra $\mfk{sl}(1|1)$ is available in the literature \cite{RimanyiRozansky202105}, and we can compare our result. In this section, we perform this 2d $\to$ 1d reduction, and we recover the cohomological stable envelopes and the geometrical $R$-matrix of Rim\'anyi and Rozansky.

{\bf Procedure for 2d $\boldsymbol{\to}$ 1d reduction.}
We need to know what happens by going from 2d to~1d. In the 3d $\to$ 2d, reduction, one considers the reduction of the theory on one of the circles of the elliptic curve. Similarly, the 2d $\to$~1d result is obtained by the reduction of the theory on the remaining circle. This amounts to the following substitution of gauge and flavor holonomies\footnote{The precise statement in going from K-theoretic stable envelope to the cohomological one is given in \eqref{eqn: reduction to cohomology}.} (see Section~\ref{sec:k-theory and cohomology limit} and \cite[Section~2.1]{DedushenkoNekrasov202109})
\begin{equation}\label{eq:multiplicative to additive chern roots}
 (s_a,x_i,\hbar)\mapsto \lim_{\epsilon\to 0}\bigl({\rm e}^{\epsilon s_a},{\rm e}^{\epsilon x_i},{\rm e}^{\epsilon\hbar}\bigr),
\end{equation}
On the other hand, the holonomy associated with the topological symmetry defines K\"aher parameter in three dimensions. Upon 3d $\to$ 2d reduction, this becomes a theta angle, i.e., the slope parameter, as explained in Section~\ref{sec:the K-theory limit}. By 2d $\to$ 1d reduction, it turns out that this parameter completely disappears \cite[p.~16]{DedushenkoNekrasov202109}. This is consistent with the fact that cohomological stable envelopes are not dependent on extra choices like a slope. Therefore, one should (1) use \eqref{eq:multiplicative to additive chern roots} to replace multiplicative holonomies with additive parameters in the $\epsilon\to 0$ limit by throwing away terms of order $\mcal{O}\bigl(\epsilon^2\bigr)$, and (2) discard all $\msf{s}$-dependent factors. In this process, all gauge and flavor holonomies become additive and $\mbb{C}$-valued.

{\bf Cohomological stable envelope for $\boldsymbol{\mfk{sl}(1|1)}$.}
We can now perform the reduction explicitly. Different factors of \eqref{eq:k-theoretic stable envelope coming from reduction of bps computation} are reduced as follows
\[
 \tfunction_{\mfk{C}_1,\msf{s},m}\quad\to\quad \rfunction_{\mfk{C}_1,m}(s,\mbs{x},\hbar):= \prod_{i<m}(s+x_i)\prod_{i>m}(s+x_i+\hbar),
\]
and
\[
 \prod_{\substack{a>b}}\frac{1}{\bigl(s_as_b^{-1}\bigr)^{\frac{1}{2}}-\bigl(s_as_b^{-1}\bigr)^{-\frac{1}{2}}}\quad \to\quad \prod_{\substack{a>b}}\frac{1}{s_a-s_b}.
\]
Therefore, according to Lemma \ref{cor:k-theoretic stable envelope from the elliptic one}, we end up with the cohomological stable envelope\footnote{The cohomology limit of K-theoretic stable envelope for the resolved determinantal variety $\mcal{M}_H(N,L)$ has also been computed in \eqref{eqn: cohomological stable envelope for M(N,L)} using a different polarization as explained there. Compared to \eqref{eqn: cohomological stable envelope for M(N,L)}, the expression~\eqref{eq:cohomological stable envelope coming from reduction of bps computation} contains the extra factor $(-1)^{\sum_a\#(i<p(a))|}$, which is the result of a choice of different polarization in the BPS computation of elliptic stable envelope.}
\begin{equation}\label{eq:cohomological stable envelope coming from reduction of bps computation}
 \rstab_{\mfk{C}_1}(p)=(-1)^{\sum_a\#(i<p(a))}\tenofo{Sym}_{S_N}\left[\left(\prod_{a=1}^N\rfunction_{\mfk{C}_1,m}(s,\mbs{x},\hbar)\right)\cdot
 \Biggl(\prod_{\substack{a>b}}\frac{1}{s_a-s_b}\Biggr)\right].
\end{equation}

{\bf Rational $\boldsymbol{R}$-matrix for $\boldsymbol{\mfk{sl}(1|1)$}.}
We can now perform the reduction of the $\mfk{sl}(1|1)$ trigonometric $R$-matrix \eqref{eq:gl(1|1) trigonometric r-matrix} to construct the $\mfk{sl}(1|1)$ rational $R$-matrix. The result is
\begin{equation}\label{eq:gl(1|1) rational r-matrix}
	\def\arraystretch{1.5}
	\setlength{\arrayrulewidth}{.6pt}
 \setlength{\arraycolsep}{6pt}
	\rrmatrix_{\mfk{C}_2\ot\mfk{C}_1}(u)=
	\left(
	\begin{array}{c|cc|c}
		1 & 0 & 0 & 0
		\\
		\hline
		0 & \frac{u}{\hbar-u} & \frac{\hbar}{\hbar-u} & 0
 \\[4pt]
 0 & \frac{\hbar}{\hbar-u} & \frac{u}{\hbar-u} & 0
		\\[4pt]
		\hline
		0 & 0 & 0 & \frac{\hbar+u}{\hbar-u}
	\end{array}\right).
\end{equation}
This satisfies the rational Yang--Baxter equations given by
\[
 \rrmatrix_{12}(u-v) \rrmatrix_{13}(u-w) \rrmatrix_{23}(v-w)= \rrmatrix_{23}(v-w) \rrmatrix_{13}(u-w) \rrmatrix_{12}(u-v).
\]
Let us finally compare our result with the construction in the mathematical literature \cite{RimanyiRozansky202105}. For this matter, we first give a synopsis of \cite{RimanyiRozansky202105}.

{\bf Brief recap of the work of Rim\'anyi and Rozanski.}
Consider an index with values $(00), (01), (10),$ or $(11)$. Depending on the choice of $r$, we have
	\begin{gather*}
		r=00\colon\ \mbb{C}_{\tenofo{even}}\oplus \mbb{C}_{\tenofo{even}}, \qquad
		r=10\colon\ \mbb{C}_{\tenofo{even}}\oplus \mbb{C}_{\tenofo{odd}},\qquad
		r=01\colon\ \mbb{C}_{\tenofo{odd}}\oplus \mbb{C}_{\tenofo{even}}, \\
		r=11\colon\ \mbb{C}_{\tenofo{odd}}\oplus \mbb{C}_{\tenofo{odd}},
	\end{gather*}
which corresponds to $\mfk{sl}(2,0)$, $\mfk{sl}(1,1)$, $\mfk{sl}(1|1)$, and $\mfk{sl}(0|2)$, respectively. We are only interested in~${r=(10)}$ (or equivalently $r=(01)$).\footnote{The case $r=(00)$ corresponds to $\mfk{sl}(2)$ spin chains, which is explored in \cite{BullimoreKimLukowski201708, MaulikOkounkov201211}.}
Analog to the classical case, there is a super-stable map~\smash{$\tenofo{Stab}^{(r)}_\pi$}, which depends on a permutation $\pi\in S_L$. The definition of this map involves the triple~${(N,L,\pi)}$ and are defined if certain classes \smash{$\kappa^{(r)}_{p,\pi}:= \tenofo{Stab}^{(r)}_\pi(1_{p})$}\footnote{In \cite{RimanyiRozansky202105}, these classes are denoted as \smash{$\kappa^{(r)}_{\mbs{I},\pi}$}, where $\mbs{I}=\{I_1,\ldots,I_N\}\subseteq\{1,\ldots,L\}$. $\mbs{I}$ is determined by the injective map $p\colon\{1,\ldots,N\}\to\{1,\ldots,L\}$. We prefer to use the map $p$ to have a uniform notation throughout the paper.} exist and satisfy certain axioms with $p$ denoting the injective map \eqref{LchooseN}, that uniquely determine \smash{$\tenofo{Stab}^{(r)}_\pi$} if it exists~\cite[Definition 4.1]{RimanyiRozansky202105}.
The explicit forms of the stable envelopes are then given in terms of certain rational functions \smash{$W^{(r)}_{l;N,L}(\mbs{s},\mbs{x},\hbar)$} that are called super-weight functions and are the super-analog of weight functions defined by Tarasov and Varchenko \cite{RimanyiTarasovVarchenko201212,TarasovVarchenko199311}. Here, $\mbs{x}:= \{s_1,\ldots,s_N\}$ and~${\mbs{x}:= \{x_1,\ldots,x_L\}}$ are Chern roots and are equivariant parameters, respectively.
The superweight functions are defined as \cite[Section~5.2]{RimanyiRozansky202105}
\begin{equation}\label{eq:definiton of functions W}
	W^{(10)}_{p}(\mbs{s},\mbs{x},\hbar):= \tenofo{Sym}_{S_N}\bigl(U^{(10)}_{p}(\mbs{s},\mbs{x};\hbar)\bigr),
\end{equation}
where\footnote{The explicit form of $U^{(r)}_{l}(\mbs{s};\mbs{x};\hbar)$ for other values of $r$ can be found in \cite[Section~5]{RimanyiRozansky202105}.}
\begin{equation}\label{eq:explicit form of U}
 U^{(10)}_{p}(\mbs{s},\mbs{x},\hbar)=\prod_{a=1}^N\left\{\prod_{i=1}^{p(a)-1}\left(s_a-x_i+\hbar\right)\prod_{i=p(a)+1}^{L}\left(-s_a+x_i\right)\right\}\cdot\left(\prod_{a>b}\frac{1}{(s_a-s_b)}\right).
\end{equation}
More generally, one can define the superweight functions for any permutation $\pi\in S_L$ through
\begin{equation}\label{eq:function W for a permutation}
	W^{(10)}_{p,\pi}(\mbs{s},\mbs{x};\hbar):= W^{(10)}_{\pi^{-1}(p)}(\mbs{s},\pi(\mbs{x}),\hbar),
\end{equation}
where $\pi(\mbs{x})=\{x_{\pi(1)},\ldots,x_{\pi(L)}\}$. Consider the set of all subsets of $\{1,\dots,L\}$ consisting of $N$ elements. Each such subset is determined by an injective map $p\colon\{1,\ldots,N\}\to\{1,\ldots,L\}$. We denote the set of all such injective maps as $\mbs{p}_N$. For any $p,p'\in \mbs{p}_N$, we denote the polynomial function\footnote{The fact that $W_{p,\pi}(\mbs{x}_{p'};\mbs{x};\hbar)$ is a polynomial function is proven in \cite[Proposition~6.1]{RimanyiRozansky202105}.} obtained by substituting $s_a=x_{p'(a)}$\footnote{This is the additive analog of the screening condition \smash{$s_a\to s^{(p)}_a$} that we used to construct matrix elements of elliptic stable envelopes.}
by \smash{$W^{(10)}_{p,\pi}(\mbs{x}_{\mbs{J}},\mbs{x},\hbar)$}. Then, the tuple~\smash{$\bigl(W^{(10)}_{p,\pi}(\mbs{x}_{p'},\mbs{x},\hbar)\bigr)_{p'\in\mbs{p}_N}$},
is a class in the equivariant cohomology $H_{T}(\tenofo{Gr}(N,L))$, which we denote as \smash{$\bigl[W^{(10)}_{p,\pi}\bigr]$} \cite[Proposition~7.1]{RimanyiRozansky202105}, and satisfies the axioms for the class \smash{$\kappa^{(r)}_{p,\pi}$} \cite[Theorem~7.3]{RimanyiRozansky202105}. Hence by \cite[Definition~4.1]{RimanyiRozansky202105},
\smash{$\kappa^{(10)}_{p,\pi}=\tenofo{Stab}^{(10)}_\pi(1_{p})=\bigl[W^{(10)}_{p,\pi}\bigr]$}.
Therefore, the stable envelope maps the identity class to the class \smash{$\bigl[W^{(10)}_{p,\pi}\bigr]$}, and as a result is given by the superweight functions \smash{$W^{(10)}_{p,\pi}(\mbs{s},\mbs{m},\hbar)$}. Furthermore, the geometric $R$-matrix for~${r=(10)}$ is given by \cite[Section~8.2]{RimanyiRozansky202105}
\begin{equation}\label{eq:rimanyi-rozansky gl(1|1) rational r-matrix}
	\def\arraystretch{1.5}
	\setlength{\arrayrulewidth}{.6pt}
 \setlength{\arraycolsep}{6pt}
	\rrmatrix^{\tenofo{RR}}(u)=
	\left(
	\begin{array}{c|cc|c}
		1 & 0 & 0 & 0
		\\
		\hline
		0 & \frac{u}{\hbar-u} & \frac{\hbar}{\hbar-u} & 0
 \\[4pt]
 0 & \frac{\hbar}{\hbar-u} & \frac{u}{\hbar-u} & 0
		\\[4pt]
		\hline
		0 & 0 & 0 & \frac{\hbar+u}{\hbar-u}
	\end{array}\right).
\end{equation}

{\bf Recovering the result of Rim\'anyi and Rozansky.}
We can now compare our results to that of \cite{RimanyiRozansky202105}. First consider a permutation $\pi$ such that $\{x_1,\ldots,x_{p(a)},\ldots, x_L\}\to\{x_L,\ldots,x_{p(a)},\ldots,x_1\}$. Then,
\begin{align*}
 W^{(10)}_{p,\pi}(\mbs{s},-\mbs{x},\hbar)={}&(-1)^{\sum_a\#(i<p(a))}
 \\
 &\times \tenofo{Sym}_{S_N}\left[\prod_{a=1}^N\Biggl(\prod_{i<p(a)}\!(s_a+x_i)\!\prod_{i>p(a)}\!(s_a+x_i+\hbar)\Biggr)\!\cdot\!\Biggl(\prod_{a>b}\frac{1}{(s_a-s_b)}\Biggr)\right].
\end{align*}
Comparing this expression to \eqref{eq:cohomological stable envelope coming from reduction of bps computation}, we find that
$\tenofo{Stab}_{\mfk{C}_1}(p)=W^{(10)}_{p,\pi}(\mbs{s},-\mbs{x},\hbar)$.
Furthermore, by comparing \eqref{eq:rimanyi-rozansky gl(1|1) rational r-matrix} to \eqref{eq:gl(1|1) rational r-matrix}, we find that
$ \rrmatrix_{\mfk{C}_2\ot\mfk{C}_1}(u)= \rrmatrix^{\tenofo{RR}}(u)$.
This verifies our construction of cohomological stable envelopes for $\mfk{sl}(1|1)$.

\appendix

\section{Equivariant elliptic cohomology}
\label{subsec: Equivariant Elliptic Cohomology}

In this appendix, we give a brief overview of equivariant elliptic cohomology, we only list the essential ingredients that are used in this paper. We basically follow the presentation in \cite[Section~1.4, Appendices~A and B]{okounkov2021inductive}. The general theory on equivariant elliptic cohomology can be found in \cite{ganter2014elliptic,Gepner2006,GinzburgKapranovVasserot199502,Grojnowski200703,lurie2009survey,rosu2001equivariant},

Let $q$ be a nonzero complex number such that $|q|<1$, then take the elliptic curve $\mathbb E=\mathbb C^{\times}/q^{\mathbb Z}$. For simplicity, we choose the elliptic curve $\bE$ to be general enough such that it has no complex multiplication, i.e., $\mathrm{End}(\mathbb E)=\mathbb Z$.

Fix a compact Lie group $G$ and denote its complexification by $G_{\mathbb C}$, then the $G$-equivariant elliptic cohomology is a functor
\[
\{\text{pairs of $G$-spaces}\}\to \{\text{graded super schemes over }\mscr E_G\}\qquad
(X,\partial X)\mapsto \mathrm{Ell}_G(X,\partial X),
\]
where $\mscr E_G$ is the moduli scheme of semistable principal $G_{\mathbb C}$-bundles of trivial topological type on the dual elliptic curve $\mathbb E^\vee$. For a torus $T$,
$\mscr E_T=\mathbb E\otimes_{\mathbb Z}\mathrm{Cochar}(T)$,
and for a compact Lie group~$G$ with maximal torus $T$, it is known that \cite{friedman1997principal}
$ \mscr E_G\cong \mscr E_T/W$,
where $W$ is the Weyl group of $G$ which acts on $\mathrm{Cochar}(T)$ naturally. When $G$ is simple and simply-connected, it is known that~\cite{looijenga1976root} $\mscr E_G$ is isomorphic to the weighted projective space $\mathbb P(1,g_1,\dots,g_r)$, where $g_i$ are coefficients in the decomposition
$\theta^\vee=\sum_{i}g_i\alpha_i^\vee$,
of the dual of highest root into simple coroots. The map $\mathrm{Ell}_G$ is functorial with respect to the change of groups, i.e., for a group homomorphism $G\to G'$, there is a natural transformation $\mathrm{Ell}_G\to \mathrm{Ell}_{G'}$. In the following, we focus on the case when $\partial X=\varnothing$ and $X$ is a complex algebraic variety.

The graded super scheme means the structure sheaf \smash{$\mscr O^{\bullet}_{\mathrm{Ell}_G(X)}=\bigoplus_{d\in \mathbb Z}\mscr O^{d}_{\mathrm{Ell}_G(X)}$} is graded and graded-commutative. It is known that each homogeneous piece $\mscr O^{d}_{\mathrm{Ell}_G(X)}$ is a coherent sheaf on~$\mscr E_G$. Moreover, $\mathrm{Ell}_G$ is two-periodic
\smash{$\mscr O^{d}_{\mathrm{Ell}_G(X)}=\mscr O^{d+2}_{\mathrm{Ell}_G(X)}\otimes \omega$}, $\omega=T^*_{0}E$.
Since the stable~enve\-lope is in degree zero, we will focus on \smash{$\mscr O^{0}_{\mathrm{Ell}_G(X)}$}, which is the structure sheaf of a~scheme, and it is finite over $\mscr E_G$ \cite{GinzburgKapranovVasserot199502}. In the later discussions, we use the same notation $\mathrm{Ell}_G(X)$ to denote the degree zero piece.

\subsection{Chern class} If $H$ is another compact Lie group, and $P\to X$ is a $G$-equivariant principal $H$-bundle, then $P$ induces a map
$ c\colon\mathrm{Ell}_G(X)\to \mathrm{Ell}_H(\mathrm{pt})=\mscr E_H$,
called the \textit{Chern class} map. For a vector bundle~$V$ of rank $r$, which is equivalent to a principal $\mathrm{GL}_r$-bundle, the \textit{Thom line bundle} associated to~$V$ is defined as
\[
 \Theta(V):=c^*\mscr O(D_{\Theta}), \qquad D_{\Theta}=\{0\}+S^{r-1}\mathbb E\subset S^r\mathbb E=\mscr E_{\mathrm{GL}_r}.
\]
$\Theta(V)$ inherits a canonical section $\vartheta(V)$ from the effective divisor $D_{\Theta}$. $\Theta\colon V\to \Theta(V)$ descends to a group homomorphism $K_G(X)\to \mathrm{Pic}(\mathrm{Ell}_G(X))$, this follows from
\[
 \Theta(V)=\Theta(V_1)\otimes \Theta(V_1)\qquad \text{for short exact sequence } 0\to V_1\to V\to V_2\to 0.
\]
Note that the canonical section simply multiplies: $\vartheta(V)=\vartheta(V_1)\vartheta(V_2)$. We also have
$\Theta\bigl(V^\vee\bigr)\cong \Theta(V)$,
such that the canonical section picks up a sign $\vartheta\bigl(V^\vee\bigr)=(-1)^{\operatorname{rk} V}\vartheta(V)$.

For vector bundles $V_1$, $V_2$ of ranks $r_1$, $r_2$, $\Theta(V_1\otimes V_2)$ is the pullback of $\mscr O(D_{\Theta})$ on $S^{r_1r_2}E$ via the composition of maps
\[
 \mathrm{Ell}_G(X)\overset{c_1\times c_2}{\longrightarrow} S^{r_1}\mathbb E\times S^{r_2}\mathbb E\overset{m}{\longrightarrow} S^{r_1r_2}\mathbb E,
\]
where $m$ is induced from the multiplication of elliptic curve
\[
 m\colon\ (\{x_1,\dots,x_{r_1}\},\{y_1,\dots,y_{r_2}\})\mapsto \{x_i+y_j\}.
\]
In particular, for $G=1$ and considering the tensor product of line bundles, the above gives rise to the formal group law of non-equivariant elliptic cohomology.

\subsection{Gysin map}
For a proper $G$-equivariant map $f\colon X\to Y$, assume that $f$ factors as a regular embedding ${i\colon X\hookrightarrow Z}$ and a smooth projection $p\colon Z\to Y$, and that both $i$ and $p$ are $G$-equivariant, then there exists a distinguished element
\[
 f_{\circledast}\in \mathrm{Hom}_{\mscr O_{\mathrm{Ell}_G(Y)}}(f_*\Theta(T_f),\mscr O_{\mathrm{Ell}_G(Y)}),
 \]
where $f_*\colon\mathrm{Ell}_G(X)\to \mathrm{Ell}_G(Y)$ is the induced map between elliptic cohomologies, and $T_f$ is the relative tangent bundle. We call $f_{\circledast}$ the Gysin map.

If $T_f$ equals to $f^*V$ in $K_G(X)$ for some $V\in K_G(Y)$, then we denote by $[X]$ the section~${\Gamma(\mathrm{Ell}_G(Y),\Theta(-V))}$ induced by the Gysin map $f_{\circledast}\colon f_*\mscr O_{\mathrm{Ell}_G(X)}\to \Theta(-V)$ precomposed with the canonical map $\mscr O_{\mathrm{Ell}_G(Y)}\to f_*\mscr O_{\mathrm{Ell}_G(X)}$.

For example, let $N\to X$ be a $G$-equivariant vector bundle and let $i\colon X\hookrightarrow N$ be the zero section, then $f_*\colon\mathrm{Ell}_G(X)\to \mathrm{Ell}_G(N)$ is an isomorphism of schemes due to the homotopy invariance of elliptic cohomology, and $i_{\circledast}\in \Gamma(\mathrm{Ell}_G(X),\Theta(N))$ is the section $\vartheta(N)$.

\subsection{Supports}
For a section $\alpha$ of a coherent sheaf $\mscr F$ on $\mathrm{Ell}_G(X)$, and a $G$-invariant open subset $j\colon U\hookrightarrow X$, we say that $\alpha$ is supported on $X\setminus U$ if $j^*(\alpha)=0$ in $\mathrm{Ell}_G(U)$. Define the support $\mathrm{supp}(\alpha)$ to be the intersection of $G$-invariant closed subset that $\alpha$ is supported on.

The Gysin map can be defined for compactly supported sections. Namely for a $G$-equivariant map $f\colon X\to Y$ such that it factors as a regular embedding followed by a smooth projection, then there exists a distinguished element
\[
 f_{\circledast}\in \mathrm{Hom}_{\mscr O_{\mathrm{Ell}_G(Y)}}(f_*\Theta(T_f)_c,\mscr O_{\mathrm{Ell}_G(Y)}),
 \]
where ${\Theta(T_f)_c\subset \Theta(T_f)}$ is the subsheaf of sections $\alpha$ such that $f|_{\mathrm{supp}(\alpha)}$ is proper.

\subsection{Correspondences}
Consider the diagram
\[
\begin{tikzcd}
& X_2\times X_1 \arrow[dl,"p_2" ']\arrow[dr,"p_1"] & \\
X_2 \arrow[dr,"q_2" '] & & X_1. \arrow[dl,"q_1"]\\
&\mathrm{pt} &
\end{tikzcd}
\]
Assume that $X_1$ is smooth, then for a pair of line bundles $\mscr L_i\in \mathrm{Pic}(\mathrm{Ell}_G(X_i))$, and for any section
$ \alpha\in \Gamma\bigl(\mathrm{Ell}_G(X_2\times X_1),\mscr L_2\boxtimes\bigl(\mscr L_1^\vee\otimes\Theta(T_{X_1})\bigr)\bigr)$,
such that $\mathrm{supp}(\alpha)$ is proper over $X_2$, $\alpha$ induces a map
\[\begin{tikzcd}
 q_{1*}\mscr L_1\arrow[rr,"p_{2\circledast}(\alpha p_1^*(\cdot))"] & & q_{2*}\mscr L_2,
\end{tikzcd}\]
in $\mathrm{Coh}(\mscr E_G)$.

\subsection{Degree of a line bundle}
For a line bundle $\mscr L$ on an abelian variety $\mcal A$, we define its degree $\deg\mscr L$ to be its image in the N\'eron--Severi group $\mathrm{NS}(\mcal A)=\mathrm{Pic}(\mcal A)/\mathrm{Pic}^0(\mcal A)$. It is known that $\mathrm{NS}(\mcal A)$ is isomorphic to the subgroup of homomorphisms $f\in \mathrm{Hom}\bigl(\mcal A,\mcal A^\vee\bigr)$ such that $f=f^\vee$, and the isomorphism is given by\footnote{This fact can be derived from Theorem~2 together with the remark after that in \cite[Section~20]{mumford1974abelian}.}
$ \mscr L\mapsto \bigl(\phi_{\mscr L}\colon x\mapsto x^*\mscr L\otimes\mscr L^{-1}\bigr)$.
Let $T$ be a torus, then $\mathrm{Hom}\bigl(\mscr E_T,\mscr E^\vee_T\bigr)$ is isomorphic to~${\mathrm{Char}(T)^{\otimes 2}\otimes_{\mathbb Z}\mathrm{End}(\mathbb E)}$. As we have assumed in the beginning, $\mathbb E$ is chosen such that~${\mathrm{End}(\mathbb E)=\mathbb Z}$, so the N\'eron--Severi group $\mathrm{NS}(\mscr E_T)$ is isomorphic to $S^2\mathrm{Char}(T)$. Explicitly, any $\mu\in \mathrm{Char}(T)$ gives rise to a map $\phi_{\mu}\colon\mscr E_T\to \mathbb E$, then $\deg \phi_{\mu}^*\mscr O(D_{\Theta})=\mu\otimes\mu\in S^2\mathrm{Char}(T)$. So for $V=\sum_{\mu}V_{\mu}\cdot\mu\in K_{T}(\mathrm{pt})$, we have
$\deg\Theta(V)=\sum_{\mu}(\dim V_{\mu})\mu\otimes\mu$.

\section{Notations and conventions for supersymmetry}
\label{sec:details of three-dimensional computations}
In this appendix, we collect the notations and conventions regarding supersymmetry used in the computation of elliptic stable envelopes in Section~\ref{sec:stable envelopes and the $R$-matrix from gauge theory}.

\subsection{Conventions for spinors} \label{sec:spinors}

For a two component $\SU(2)$ spinor $\ep$, we write its components as $\ep^\al$ where $\al=1,2$. We use~${\al,\be,\ldots}$ as spinor indices. Our spinors will be anti-commuting. Let $\vep^{\al\be}$ and $\vep_{\al\be}$ be the totally antisymmetric tensors satisfying
$\vep^{12} = \vep_{21} = 1$.
These tensors are used to raise and lower spinor indices
$\ep_\al = \vep_{\al\be} \ep^\be$, $ \ep^\al = \vep^{\al\be} \ep_\be$.
Let $\ga^\mu$ for $\mu=1,2,3$ be the Pauli matrices
\[
	\bigl[\tensor{\bigl(\ga^1\bigr)}{_\al^\be}\bigr] = \begin{pmatrix} & 1 \\ 1 & \end{pmatrix} , \qquad
	\bigl[\tensor{\bigl(\ga^2\bigr)}{_\al^\be}\bigr] = \begin{pmatrix} & -\ii \\ \ii & \end{pmatrix} , \qquad
	\bigl[\tensor{\bigl(\ga^3\bigr)}{_\al^\be}\bigr] = \begin{pmatrix} 1 & \\ & -1 \end{pmatrix} .
\]
We also define the charge conjugation matrix
\[
	C = \bigl[C^{\al\be}\bigr] := -\ii \ga^2 = \begin{pmatrix} & -1 \\ 1 & \end{pmatrix} .
\]
In other words, $C^{\al\be} = -\vep^{\al\be}$. Some properties of the charge conjugation matrix are
\beq
	C^2 = -1, \qquad C^T = C^{-1} = -C , \qquad C\ga^\mu = -(\ga^\mu)^T C . \label{Cproperties}
\eeq
Contraction of spinors $\ep$ and $\la$ are defined as
\beq
	\ep\la := \ep^\al \la_\al , \label{epla}
\eeq
i.e., we use the NE-SW convention for contracting indices on anti-commuting spinors. We can think of two-component spinors as column vectors, with lowered indices, e.g.,
$\ep = \left(\begin{smallmatrix} \ep_1 \\ \ep_2 \end{smallmatrix}\right)$.
Then we can write the spinor bilinear \rf{epla} in matrix notation as\footnote{There should not be any confusion between when the matrix notation is being used in spinor bilinears, since in the matrix notation there will always be the transpose of a spinor to the left (as in the right-hand side of \rf{eplaMatrix}), otherwise there will not be any transpose (as in the left-hand side of \rf{eplaMatrix}).}
\beq
	\ep\la = \ep^T C \la . \label{eplaMatrix}
\eeq
Using the properties \rf{Cproperties} of the charge conjugation matrix, we can derive some symmetry properties of various spinor bilinears:
\[\ep \ga^{\mu_1} \cdots \ga^{\mu_n} \la = (-1)^{\frac{n(n+1)}{2}} \la \ga^{\mu_1} \cdots \ga^{\mu_n} \ep .
\]

\subsection[3d $\mcal{N}=2$ supersymmetry]{3d $\boldsymbol{\mcal{N}=2}$ supersymmetry} \label{sec:3dN=2susy}
The supersymmetry algebra has four supercharges which we can denote as $\thdQ_\al$, $ \ov\thdQ_\al$ with $\al$ being a spinor index. A generic supercharge is a linear combination of these
\beq
	\thdQ_{\ep,\ov\ep} := \ep^\al \thdQ_\al + \ov\ep^\al \ov\thdQ_\al , \label{Qee}
\eeq
where the coefficients $\ep$, $\ov\ep$ are two components spinors. The variation of a field in the theory caused by the supercharge $\thdQ_{\ep,\ov\ep}$ will be denoted $\de_{\ep, \ov\ep}$.

For a Lie group $G$, the corresponding 3d $\cN=2$ vector multiplet consists of a vector field $A$, two scalars $\si$ and $\sfD$ and two spinors $\la$ and $\ov\la$. All these fields are valued in the Lie algebra $\mfr g$ of $G$. We adopt the convention that elements of real Lie algebras are hermitian. Their supersymmetry variations are
\begin{gather}
	\de_{\ep, \ov\ep} A_\mu = \frac{\ii}{2} \bigl(\ov\ep \ga_\mu \la - \ov\la \ga_\mu \ep\bigr) , \qquad
	\de_{\ep, \ov\ep} \si = \frac{1}{2} \bigl(\ov \ep \la - \ov\la \ep\bigr) , \nonumber\\
	\de_{\ep, \ov\ep} \la = - \frac{1}{2} \ga^{\mu\nu} F_{\mu\nu} \ep - \sfD \ep + \ii \ga^\mu D_\mu \si \ep ,\qquad
	\de_{\ep, \ov\ep} \ov\la = - \frac{1}{2} \ga^{\mu\nu} F_{\mu\nu} \ov\ep + \sfD \ov\ep - \ii \ga^\mu D_\mu \si \ov\ep , \nonumber\\
	\de_{\ep, \ov\ep} \sfD = - \frac{\ii}{2} \ov\ep \ga^\mu D_\mu \la - \frac{\ii}{2} D_\mu \ov\la \ga^\mu \ep + \frac{\ii}{2} [\ov\ep \la, \si] + \frac{\ii}{2} \bigl[\ov \la \ep, \si\bigr] .\label{susyVec}
\end{gather}
Given a representation $R$ of $G$ we can add an $R$-valued chiral multiplet containing two scalars~$\phi$ and $\sfF$, and a spinor~$\psi$. We also have the anti-chiral multiplet containing scalars $\ov\phi$, $\ov\sfF$, and spinor~$\ov\psi$ valued in the dual representation $R^\vee$. These fields transform under supersymmetry as
\begin{gather}
	\de_{\ep,\ov\ep} \phi = \ov\ep \psi , \qquad
	\de_{\ep,\ov\ep} \ov\phi = \ep \ov\psi , \qquad
	\de_{\ep,\ov\ep} \psi = \ii \ga^\mu \ep D_\mu \phi + \ii \ep \si \cdot \phi + \ov\ep \sfF , \nonumber\\
	\de_{\ep,\ov\ep} \ov\psi = \ii \ga^\mu \ov \ep D_\mu \ov\phi - \ii \ov\ep \si \cdot \ov\phi + \ep \ov \sfF ,\qquad
	\de_{\ep,\ov\ep} \sfF = \ep\bigl(\ii \ga^\mu D_\mu \psi - \ii \si \cdot \psi - \ii \la \cdot \phi\bigr) , \nonumber\\
	\de_{\ep,\ov\ep} \ov \sfF = \ov\ep\bigl(\ii \ga^\mu D_\mu \ov\psi + \ii \si \cdot \ov\psi - \ii \ov\la \cdot \ov\phi\bigr) .\label{susyChiral}
\end{gather}
Here $\si\cdot$ is the action of $\si$ on whatever follows in the corresponding representation. %

By computing the commutator $[\de_{\ep,\ov\ep}, \de_{\eta, \ov\eta}]$ on fields, we find the (anti)commutation relations between the supercharges
\beq
	\bigl\{\thdQ_\al, \ov\thdQ_\be\bigr\} = -\ii \ga^\mu_{\al\be} D_\mu - \ii \vep_{\al\be} \si \cdot . \label{3dsusy}
\eeq
In a gauge theory, we always have a dynamical vector multiplet for the gauge group. So the last term above is an infinitesimal gauge transformation by the adjoint scalar $\si$ of this vector multiplet. But we can also introduce background vector multiplets for flavor symmetry groups. If $G$ is the gauge group and $F$ the flavor group with Lie algebras $\mfr g$ and $\mfr f$, respectively, then we can treat the vector multiplet from \rf{susyVec} as $(\mfr g \oplus \mfr f)$-valued. Thus all the fields of the multiplet become sums of $\mfr g$ and $\mfr f$-valued fields. In particular, we write $\si = \si_G + \si_F$ according to this decomposition and the supersymmetry algebra \rf{3dsusy} contains a gauge and a flavor symmetry transformation by the dynamical $\si_G$ and the background $\si_F$ respectively. We only turn on background values for scalar fields, namely $\si_F$ and $\sfD_F$. Furthermore, these values are constrained by the requirement to preserve some amount of supersymmetry, i.e., to solve the BPS equations. More on BPS equations in Section~\ref{sec:BPSeq}.

\subsection{Localizing supercharge} \label{sec:Qloc}
Suppose the elliptic curve $\bE_\tau$ in our 3d space-time $I \times \bE_\tau$ has a holomorphic volume form $\dd z$ and two 1-cycles, $S^1_A$ and $S^1_B$ defined by:
\smash{$\oint_{S^1_A} \dd z = 1$}, \smash{$ \oint_{S^1_B} \dd z = \tau$}.
Let $y$, $\tht_A$ and $\tht_B$ be the real coordinates on $I$, $S^1_A$ and $S^1_B$, respectively. We take the orientation defining volume form on $M$ to be $\dd y \wedge \dd \tht_A \wedge \dd \tht_B$. This induces the volume forms $\dd y \wedge \dd \tht_A$ on $\Si := \bR \times S^1_A$ and $\dd \tht_A \wedge \dd \tht_B$ on $\bE_\tau$. We use holomorphic coordinates $w$ and $z$ on $\Si$ and $\bE_\tau$ respectively such that the volume forms in real and complex coordinates are related by
$\dd y \wedge \dd \tht_A = \frac{\ii}{2} \dd w \wedge \dd \ov w$, $
	\dd \tht_A \wedge \dd \tht_B = \frac{\ii}{2} \dd z \wedge \dd \ov z $.
The real and complex coordinates are related by
$w = y + \ii \tht_A$, $ z = \tht_A + \ii \tht_B$.
We denote the derivatives with respect to $y$, $\tht_A$ and $\tht_B$ by $\pa{y}$, $\pa{A}$ and $\pa{B}$, respectively. They are related to the derivatives with respect to the complex coordinates as
$\pa{w} = \frac{1}{2} (\pa{y} - \ii \pa{A})$, $ \pa{z} = \frac{1}{2} (\pa{A} - \ii \pa{B})$.
Similarly for the covariant derivatives $D_w$, $D_z$, $D_y$, $D_A$ and $D_B$.

We can relabel our 3d $\cN=2$ supercharges $\thdQ_\al$, $\ov\thdQ_\al$ to adjust to the complex structures of either $\Si$ or $\bE_\tau$. If we define
\begin{gather*}
	\thdQ_+ := \frac{1}{\sqrt{2}} (\thdQ_1 + \thdQ_2) , \qquad \thdQ_- := \frac{1}{\sqrt{2}} (\thdQ_1 - \thdQ_2) ,\qquad
	\ov\thdQ_+ := \frac{1}{\sqrt{2}} \bigl(\ov\thdQ_1 + \ov\thdQ_2\bigr) , \\
 \ov\thdQ_- := \frac{1}{\sqrt{2}} \bigl(\ov\thdQ_1 - \ov\thdQ_2\bigr) ,
\end{gather*}
then the subscripts $\pm$ refer to the chirality on $\bE_\tau$ and the algebra \rf{3dsusy} reflects the complex structure of $\bE_\tau$
\begin{gather*}
	\bigl\{\thdQ_+, \ov\thdQ_+\bigr\} = 2 D_z , \qquad \bigl\{\thdQ_-, \ov\thdQ_-\bigr\} = 2 D_{\ov z} , \qquad
	\bigl\{\thdQ_+, \ov\thdQ_-\bigr\} = \ii ( D_y - \si) , \\ \bigl\{\thdQ_-, \ov\thdQ_+\bigr\} = \ii (D_y + \si) .
\end{gather*}
Alternatively, if we define
\begin{gather*}
	q_+ :=\frac{\thdQ_1}{\sqrt{2}} , \qquad q_- := -\frac{\thdQ_2}{\sqrt{2}} , \qquad
	\ov q_+ := \frac{\ov \thdQ_1}{\sqrt{2}} , \qquad \ov q_- := -\frac{\ov \thdQ_2}{\sqrt{2}} ,
\end{gather*}
then the subscripts $\pm$ refer to the chirality on $\Si$ and the algebra \rf{3dsusy} reflects the complex structure of $\Si$
\begin{gather*}
	\{q_+, \ov q_+\} = \ii D_w , \qquad \{q_-, \ov q_-\} = - \ii D_{\ov w} , \qquad
	\{q_+, \ov q_-\} = \frac{\ii}{2}(D_B - \si) , \\ \{q_-, \ov q_+\} = \frac{\ii}{2}(D_B + \si) .
\end{gather*}
The localization results used in this note are with respect to the supercharge
\beq
	\thdQ := \thdQ_- + \ov\thdQ_- = q_+ + q_- + \ov q_+ + \ov q_- , \label{Qloc}
\eeq
which satisfies
\beq
	\thdQ^2 = 2 D_{\ov z} . \label{Q^2}
\eeq
Furthermore, translation in the $y$-direction is $\mcr Q$-exact
\beq
 \bigl\{\mcr Q, \mcr Q_+ + \ov{\mcr Q}_+\bigr\} = 2 \ii D_y . \label{y-exact}
\eeq
From the point of view of both $\Si$ and $\bE_\tau$, the 3d $\cN=2$ algebra appears as the 2d $\cN=(2,2)$ algebra. On the elliptic curve, the supercharge $\thdQ$ belongs to the $\cN=(0,2)$ subalgebra. On~$\Si$, the supercharge $\thdQ$ does not belong to any proper subalgebra of $(2,2)$ but rather it can be related to certain deformations of both the A-model and the B-model BRST operators, as we find below.

{\bf Dimensional reduction.}
By compactifying the $S^1_B$ direction of the 3d space-time, we land on a 2d $\cN=(2,2)$ theory on $\Si$ with four supercharges $\sfQ_\pm$, $\ov \sfQ_\pm$ that descend from the 3d supercharges $q_\pm$, $\ov q_\pm$, respectively. The 2d $\cN=(2,2)$ supersymmetry has the usual A-model and B-model supercharges
$\sfQ_A := \ov \sfQ_+ + \sfQ_-$, $ \sfQ_B := \ov\sfQ_+ + \ov\sfQ_-$,
that are nilpotent up to gauge transformations.\footnote{After compactifying $S^1_B$ the component of a gauge field in this direction becomes an adjoint scalar field in the~2d theory and the derivative $D_B$ becomes a gauge transformation generated by this field.} The remaining supersymmetry in the A and B-model can be labeled as
\[
	\sfG_+^A := -\ii \sfQ_+ , \qquad \sfG_-^A := \ii \ov\sfQ_- , \qquad
	\sfG_+^B := -\ii \sfQ_+ , \qquad \sfG_-^B := \ii \sfQ_- .
\]
For any vector field $V = V_+ \pa{w} + V_- \pa{\ov w}$, we define the linear combination $\iota_V \sfG := V_+ \sfG_+ + V_- \sfG_-$ and we have the anti-commutation relations
\[
	\bigl\{\sfQ_A, \iota_V \sfG^A\bigr\} =\bigl\{\sfQ_B, \iota_V \sfG^B \bigr\} = V_+ D_w + V_- D_{\ov w} = \cL_V + \text{gauge transformation} .
\]
$\Om$-deforming the A and B-model supercharges with respect to a vector field $V$ means changing the BRST operator from $\sfQ_A$ and $\sfQ_B$ to $\sfQ_A + \iota_V \sfG^A$ and $\sfQ_B + \iota_V \sfG^B$ such that the new BRST operator squares to the space-time transformation generated by $V$
\[
	\bigl(\sfQ_A + \iota_V \sfG^A\bigr)^2 = \bigl(\sfQ_B + \iota_V \sfG^B\bigr)^2 = \cL_V + \text{gauge transformation} .
\]
We observe that the localizing supercharge \rf{Qloc}, upon reduction to 2d, descends to an $\Om$-deformed supercharge both in the A and in the B-model
\[
	\thdQ \xrightarrow{\text{3d $\to$ 2d}} \sfQ_+ + \sfQ_- + \ov\sfQ_+ + \ov \sfQ_- =
	 \sfQ_A + \ii \sfG^A_+ - \ii \sfG^A_- = \sfQ_B + \ii \sfG^B_+ - \ii \sfG^B_- .
\]
These $\Om$-deformations are defined with respect to the vector field
$\ii \pa{w} - \ii \pa{\ov w} = \pa{A}$,
which rotates the $S^1_A$ circle. So, the 3d localizing supercharge can be thought of as a lift of the $S^1_A$-rotating $\Om$-deformed supercharge from either the A or the B-model on $\Si$.

If we further compactify the $S^1_A$ direction and reduce to a 1d theory on $I$, we get an ${\cN=4}$ supersymmetric gauged quantum mechanics. The localizing supercharge $\thdQ$ descends to a~nilpotent (up to gauge transformation) supercharge. Nilpotency in 1d follows from the fact that there is no $D_y$ on the right-hand side of \rf{Q^2}, which is the only remaining translation in 1d.

\section[$\mfk{sl}(1|1)$ rational $R$-matrix from the A-model localization]{$\boldsymbol{\mfk{sl}(1|1)}$ rational $\boldsymbol{R}$-matrix from the A-model localization}
\label{sec:geometric r-matrix from the a-model localization}

The aim of this appendix is to compute the geometric rational $R$-matrix for $\mfk{sl}(1|1)$ superspin chain from its natural theory from the perspective of the Bethe/Gauge correspondence, i.e., a~2d gauge theory corresponding to the $\mfk{sl}(1|1)$ spin chains with fundamental representations, and using the A-model localization formula of Closset, Cremonesi and Park \cite{ClossetCremonesiPark201504}. The computation here is inspired by insights and observations in \cite{BullimoreKimLukowski201708,Nekrasov2013,Nekrasov2014}, and its $\mfk{sl}(2)$ counterpart has been performed in \cite{BullimoreKimLukowski201708}. Unlike the construction in the main body of this work, the computation in this appendix should be considered an interesting observation rather than being based on a~systematic framework. Furthermore, in this section, we assume that the explicit form of cohomological stable envelopes is known. Of course, by knowing the cohomological stable envelopes, the computation of $R$-matrix is a simple matrix manipulation. However, viewing the $R$-matrix elements as the two-point functions of certain (normalized) A-model observables would provide an alternative route and may lead to a better understanding of the Bethe/Gauge correspondence.\looseness=-1

\subsection[Construction of geometric rational $R$-matrix]{Construction of geometric rational $\boldsymbol{R}$-matrix}\label{sec:construction of geometric r-matrix}

In Section~\ref{sec:2d and 1d avatars of elliptic stable envelopes}, we have successfully constructed cohomological stable envelopes and the geometric $R$-matrix associated with $\mfk{sl}(1|1)$ by starting from their 3d counterparts and performing a~3d$\to$~2d$\to$ 1d reduction, corresponding to reducing the three-dimensional gauge theory, living on $I\times\mbb{E}_\tau$, on the two cycles of the elliptic curve $\mbb{E}_\tau$. This is expected since the interval \mbox{partition} function of the 3d theory reduces to the interval partition function of the corresponding~1d quantum mechanical system after discarding all instanton corrections coming from particles wrapping~$S^1_\A$ and $S^1_\B$. Before presenting the setup and the computation of the $R$-matrix, let us make a~series of remarks.

\begin{rmk}[2d Gauge theory] The first remark is the theory we use for localization computation. This is the gauge theory that Bethe/Gauge correspondence naturally associates to an $\mfk{sl}(1|1)$ spin chain in a fundamental representation. This gauge theory is constructed in Example~\ref{ex:2dgl11}. For the $N$-magnon sector of a spin chain with $L$ sites, the 2d gauge theory is a $\U(N)$ gauge theory with $\mcal{N}=(2,2)$ supersymmetry coupled to a pair of $\SU(L)\times\U(1)_\hbar$ fundamental/anti-fundamental chiral multiplets.
\end{rmk}

\begin{rmk}[relation to the dimensional reduction of 3d Gauge theory of Section~\ref{sec:3d N=2 SQCD and its parameters}] A~natural question is the relation between the gauge theory we use for the localization computation of this section and the 3d$\to$2d dimensional reduction of the 3d theory of Section~\ref{sec:3d N=2 SQCD and its parameters} used for the computation of elliptic stable envelopes. The dimensional reduction of the 3d theory of Section~\ref{sec:3d N=2 SQCD and its parameters} is again a 2d gauge theory with an $\mcal{N}=(2,2)$ supersymmetry. However, this is a {\it different} 2d theory, although with the same gauge and matter content. To see the relation, it is expected~\cite{DedushenkoNekrasov202109} that we should start from a 4d theory\footnote{The reason we need to start in 4d is that such a theory is a natural arena for the quantum version of equivariant elliptic cohomology. The BRST cohomology of a 4d analog of the A-model supercharge is expected to be equivalent to quantum equivariant elliptic cohomology of the Higgs branch of the 4d theory.} with $\mcal{N}=1$ supersymmetry on~${I\times S^1_\A\times S^1_\B\times S^1_\C}$,\footnote{Alternatively, we can think of the 4d theory on $S^2\times S^1_\A\times S^1_\B$. We can then think of $S^2$ as $I\times S^1_\C$ to which two cigars are capped. Reduction of the theory on $S^1_\A\times S^1_\B$ gives rise to a theory on $S^2$.} where $S^1_\C$ is an extra circle.
 The dimensional reduction of the 4d theory on $S^1_C$ would be the 3d theory of Section~\ref{sec:3d N=2 SQCD and its parameters}.\footnote{The reason that this 3d theory is enough for the purpose of the computation of the $R$-matrix is that the latter is insensitive to the deformed non-perturbative corrections coming from the extra circle $S^1_\C$ or equivalently, the deformed ring structure of the quantum equivariant elliptic cohomology.} On the other hand, the dimensional reduction of the 4d theory on $S^1_\A\times S^1_\B$ gives a 2d theory on $I\times S^1_\C$, which is the gauge theory that the 2d gauge theory we use for the localization of this section. This is precisely the gauge theory that the Bethe/Gauge correspondence assigns to an $\mfk{sl}(1|1)$ spin chain with fundamental representations. Furthermore, the 2d gauge theory on $I\times S^1_\C$ can be thought of as a theory living on the sphere from which two cigars have been capped off as shown in Figure~\ref{fig:from sphere to cylinder}.
\begin{figure}[H]
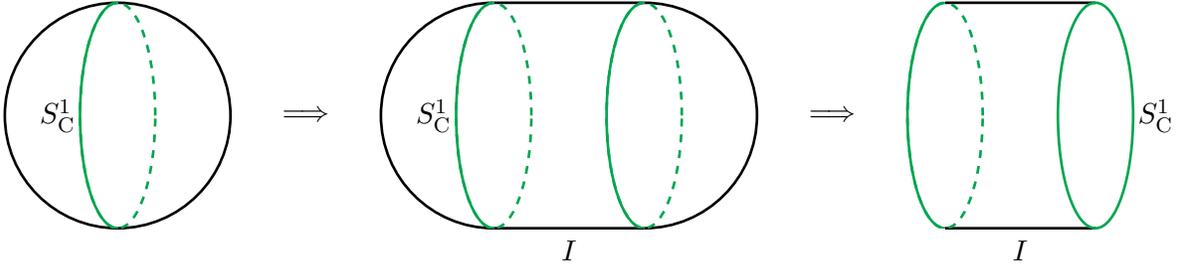

 \centering
 \FromSphereToCylinder
 \caption{The cylinder $I\times S^1_\C$ can be thought as starting from the sphere (left), stretching (middle), and capping off two cigars (right).}
 \label{fig:from sphere to cylinder}
\end{figure}
The hierarchy of theories we have considered in this work is shown in Figure~\ref{fig:from four to lower dimensions}.
\end{rmk}

\begin{figure}[t]
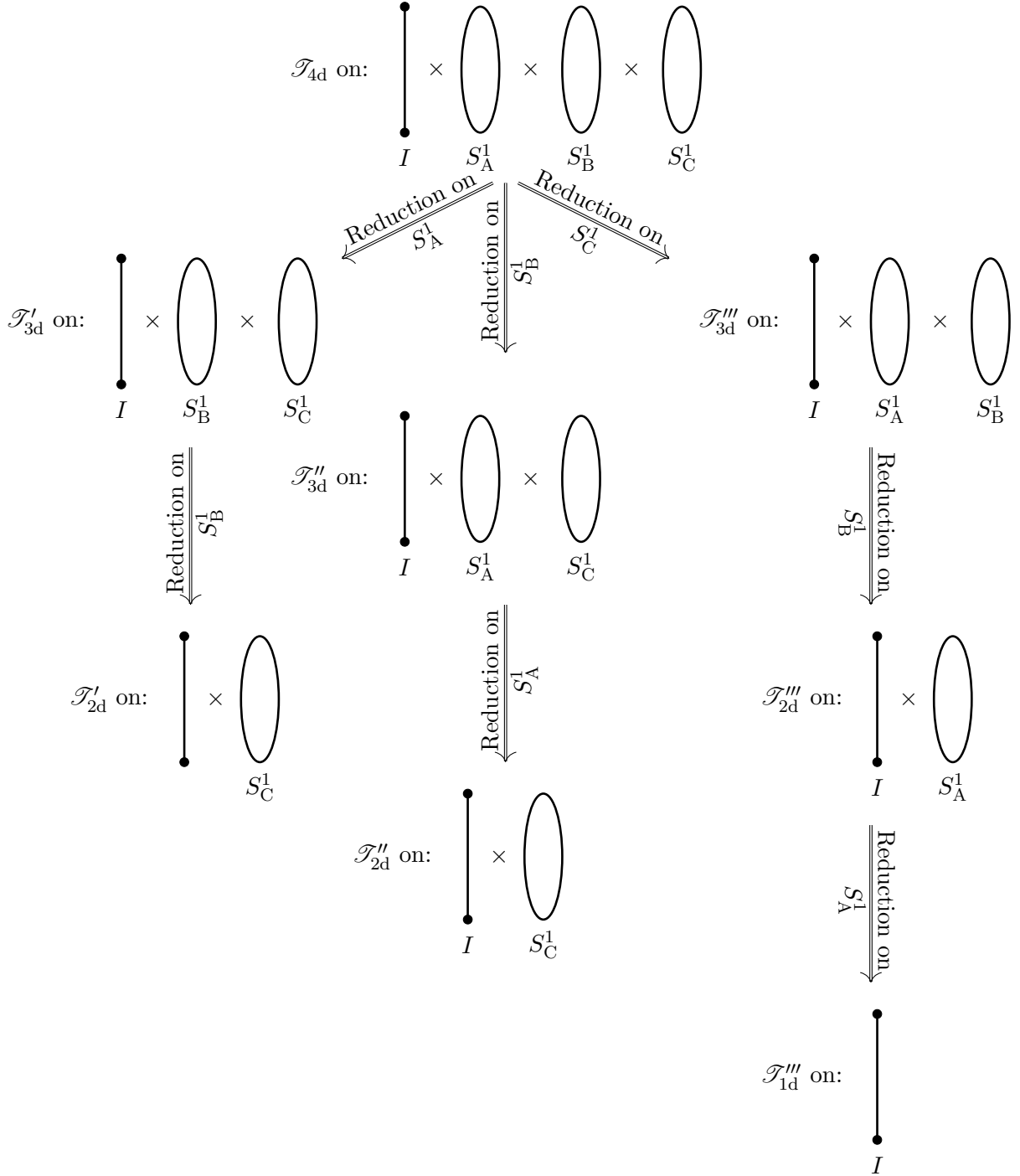

 \centering
 \FromFourToLowerDimensions

 \vspace{-1mm}

 \caption{Since the order of dimensional reduction should not matter, one can identify $\mscr{T}'_{\text{2d}}$ and~$\mscr{T}''_{\text{2d}}$ and it is the theory we use in this appendix. $\mscr{T}'''_{\text{3d}}$ is the theory we used for the computation of elliptic stable envelopes in Section~\ref{sec:explicit construction of elliptic stable envelope}.}
 \label{fig:from four to lower dimensions}\vspace{-1mm}
\end{figure}

\begin{rmk}[identification of parameters] As we have seen in Section~\ref{sec:explicit construction of elliptic stable envelope}, the elliptic stable envelopes depend on the holonomies for the gauge $\mbs{s}$, flavor $(\mbs{x},\hbar)$ and the topological symmetry $z$ along the cycle of the elliptic curve. Furthermore, we constructed cohomological stable envelopes by reducing the elliptic ones to 1d upon which the holonomy $z$ disappears and all other holonomies become complex twisted masses of the 2d theory on $I\times S^1_\C$. We thus need to compare the quantities in the theory $\mscr{T}'''_{\text{1d}}$ (the theory used in deriving \eqref{eq:cohomological stable envelope coming from reduction of bps computation}) and $\mscr{T}'_{\text{2d}}$ (the theory we use in this appendix) to be able to write the formulas in terms of parameters of the latter (see Figure~\ref{fig:from four to lower dimensions}). First note that in the theory $\mscr{T}'''_{1d}$,
\[
 \mbs{s}:=\bigointsss_\A A^{\mfk{h}}_{\text{4d}}+\mfk{i}\bigointsss_\B A^{\mfk{h}}_{\text{4d}},
\]
where $A^{\mfk{h}}_{\text{4d}}$ is the four-dimensional gauge field for the Cartan of the gauge Lie algebra $\mfk{h}=\mbb{R}^{\oplus N}$. Hence, they will be identified with $\mbs{\sigma}^\mbb{C}:=\text{diag}(\sigma^\mbb{C}_1,\ldots,\sigma^\mbb{C}_N)$, the complex scalar in the $\mfk{u}(N)$ vector multiplet, in the theory $\mscr{T}'_{\text{2d}}$. Furthermore, in the theory $\mscr{T}'''_{1d}$,
\begin{gather*}
\mbs{x}:=\bigointsss_\A A^{\mfk{a}}_{\text{4d}}+\mfk{i}\bigointsss_\B A^{\mfk{a}}_{\text{4d}},
 \qquad
 \hbar:=\bigointsss_\A A^{\hbar}_{\text{4d}}+\mfk{i}\bigointsss_\B A^{\hbar}_{\text{4d}},
\end{gather*}
where $A^{\mfk{a}}_{\text{4d}}$ is the four-dimensional gauge field for the Cartan of the flavor symmetry $\mfk{a}=\mfk{u}(1)^{\oplus L-1}$,\footnote{Note that $x_1+\dots+x_L=0$.} and $A^{\hbar}_{\text{4d}}$ is the four-dimensional gauge field for the $\mbb{R}_\hbar$ flavor symmetry. Therefore, they are identified with $\mbs{m}^\mbb{C}:=\bigl(m^\mbb{C}_1,\ldots,m^\mbb{C}_L\bigr)$ and $\hbar^{\mbb{C}}$, the complex twisted masses, in the theory $\mscr{T}'_{\text{2d}}$. Hence, we identify the parameters as follows:
\begin{gather}
 s_a \longleftrightarrow \sigma^\mbb{C}_a,\qquad a=1,\ldots,N,
 \qquad
 x_i \longleftrightarrow m_i^\mbb{C}, \qquad i=1,\ldots,L,\qquad
 \hbar \longleftrightarrow \hbar^{\mbb{C}}.\label{eq:identification of parameters of reduction and in quantum mechanics}
\end{gather}
Let us emphasize that say $m^\mbb{C}_i$ has nothing to do with the (expectation value of the) quantity defined in \eqref{complexScalar}, which is the complex scalar in the vector multiplet of the theory $\mscr{T}'''_{2d}$ (see the right portion of Figure~\ref{fig:from four to lower dimensions}), while $\bigl(\sigma^{\mbb{C}}_a,m^\mbb{C}_i,\hbar^{\mbb{C}}\bigr)$ in \eqref{eq:identification of parameters of reduction and in quantum mechanics} are parameters of the theory $\mscr{T}'_{\text{2d}}$ (or equivalently $\mscr{T}''_{\text{2d}}$). On the other hand, there are real masses in the theory $\mscr{T}'''_{1d}$ whose order determines a chamber $\mfk{C}$ in the theory $\mscr{T}'''_{1d}$,
\begin{equation}\label{eq:real masses in the quantum mechanics}
m_i=\bigointsss_\C A^{\mfk{a}}_{\text{4d}}, \qquad i=1,\ldots,L.
\end{equation}
From the perspective of $\mscr{T}'_{2d}$, these are holonomies of the 2d gauge field\footnote{Note that the gauge field in the theory $\mscr{T}'_{\text{2d}}$ comes from the components of the gauge field in 4d along $I\times S^1_\C$.} along $S^1_\C$. This completes the identification of parameters. Furthermore, for convenience in the computation, we consider a shift in the parameters\footnote{These shifts are not essential and as we mentioned are done just for the convenience in the computations and coping with existing literature.}
\begin{gather}\label{eq:the shift in the 1d holonomies}
s_a\to s_a, \qquad a=1,\ldots,N,
 \qquad
 x_i\to x_i+\frac{\hbar}{2}, \qquad i=1,\ldots,L.
\end{gather}
\end{rmk}

\begin{rmk}[using the A-model localization]\label{rmk:philosophy of using a-model supercharge}
 Note that the stable envelope we have constructed is written in the stable basis, which is a basis for the equivariant quantum cohomology of the Higgs branch of the gauge theory. The quantum equivariant cohomology consists of the vector space of the underlying classical equivariant cohomology together with a deformed ring structure \cite{Vafa199111,Witten199112}. On the other hand, for an A-twisted theory, the equivariant quantum cohomology of the target space is the cohomology of the ordinary A-model supercharge, which by definition is the twisted chiral ring of the gauge theory \cite{HoriKatzKlemmPandharipandeThomasVafaVakilZaslow2003,Vafa199111,Witten199112}. Therefore, the full Bethe--Gauge correspondence would have to invoke the cohomology of A-model supercharge, as has been argued in \cite{DedushenkoNekrasov202109,NekrasovShatashvili200901c,NekrasovShatashvili200901b}. Since the stable envelopes provide a map between quantum equivariant cohomology\footnote{Stable envelope map preserves the ring structure of the quantum equivariant cohomology in some weak sense~\cite[Theorem~7.2.1]{MaulikOkounkov201211}.} of the fixed-point loci\footnote{This is the fixed-point loci of the flavor-symmetry action on the corresponding Higgs branch.} and the target, one would like to construct a~basis for the twisted chiral ring which coincides with the stable basis for the equivariant quantum cohomology. Therefore, cohomological stable envelopes \eqref{eq:cohomological stable envelope coming from reduction of bps computation} are the corresponding A-model observables. As it is known, these generators are functions of the complex scalar in the vector multiplet of the 2d theory {\normalfont\cite{ClossetCremonesiPark201504}}.
 \end{rmk}

 \begin{rmk}[identifying the A-model observables] The A-model observables $\mcal{O}\bigl(\sigma^{\mbb{C}}\bigr)$ are functions of complex scalars $\sigma^\mbb{C}$ in the vector multiplet associated with the gauge invariance~\cite{ClossetCremonesiPark201504}. In the computation of the $R$-matrix, we use particular A-model observables obtained from the stable envelopes using the identification \eqref{eq:identification of parameters of reduction and in quantum mechanics}. Considering the identification \eqref{eq:identification of parameters of reduction and in quantum mechanics} and the shifts~\eqref{eq:the shift in the 1d holonomies}, and using \eqref{eq:definiton of functions W} and \eqref{eq:explicit form of U}, the cohomological stable envelopes take the following form:
\begin{align}
 \rstab_{\mfk{C}}(p)={}&\tenofo{Sym}_{S_N}\left[\prod_{a=1}^N\Biggl(\prod_{i<p(a)}\left(+\sigma^\mbb{C}_a-m^\mbb{C}_i+\frac{\hbar^\mbb{C}}{2}\right)\prod_{i>p(a)}
 \left(-\sigma^\mbb{C}_a+m^\mbb{C}_i+\frac{\hbar^\mbb{C}}{2}\right)\Biggr)\nonumber\right.\\
 &\left.\times \Biggl(\prod_{a>b}\frac{1}{\sigma^{\mbb C}_a-\sigma^{\mbb C}_b}\Biggr)\right].\label{eq:cohomological stable envelope after identification of parameters}
\end{align}
This is the form of the cohomological stable envelopes for $\mfk{sl}(1|1)$ that we use in the following. Note that these observables can be inserted anywhere on $S^2$ and the final result is independent of the location of insertions \cite[Section~3.3]{ClossetCremonesiPark201504}.
\end{rmk}

\begin{rmk}[Coulomb vs.\ Higgs branch localization schemes] For GLSMs, there are two different localization schemes depending on the choice of the localizing action: the Higgs-branch localization and the Coulomb-branch localization \cite{BeniniCremonesi201206,ClossetCremonesiPark201504}. These localization schemes are expected to give rise to the same final results but the explicit details and the localization procedure could be very different. The A-localization computation of this section is an example of the Coulomb-branch localization.
 \end{rmk}

 \begin{rmk}[$R$-matrix from the A-model localization] From the mathematical definition, it is expected that the $R$-matrix is related to the two-point function of stable envelopes~\cite{BullimoreKimLukowski201708}. As we briefly recall in Section~\ref{sec:two-point function of GLSM observables}, the A-model localization computations involve a summation over all instanton sectors. However, $R$-matrix is insensitive to these non-perturbative corrections. As such, the computation of $R$-matrix reduces to the computation of a two-point function of stable envelopes in the zero-flux sector. We will further elaborate on the construction below.
 \end{rmk}

 \begin{rmk}[interpretation in terms of $\mcal{N}=4$ quantum mechanics]
 Since the computation of $R$-matrix is the restriction of certain two-point functions to zero-flux section, we can interpret the result in the corresponding quantum mechanics. This theory is obtained by the dimensional reduction of our 2d theory on $I\times S^1_\C$. The insertion of A-model observables on $S^2$ leads to certain boundary conditions in quantum mechanics, as illustrated in Figure~\ref{fig:from 2d gauge theory to 1d QM}. Therefore, the two-point function of A-model observables in the zero-flux sector is the interval partition function of the corresponding $\mcal{N}=4$ quantum mechanics in the presence of these boundary conditions.
 \end{rmk}

Having summarized all the necessary ingredients, we can now compute the $R$-matrix. In this construction, we use the cohomological stable envelopes \eqref{eq:cohomological stable envelope after identification of parameters}.

 \begin{figure}[t]
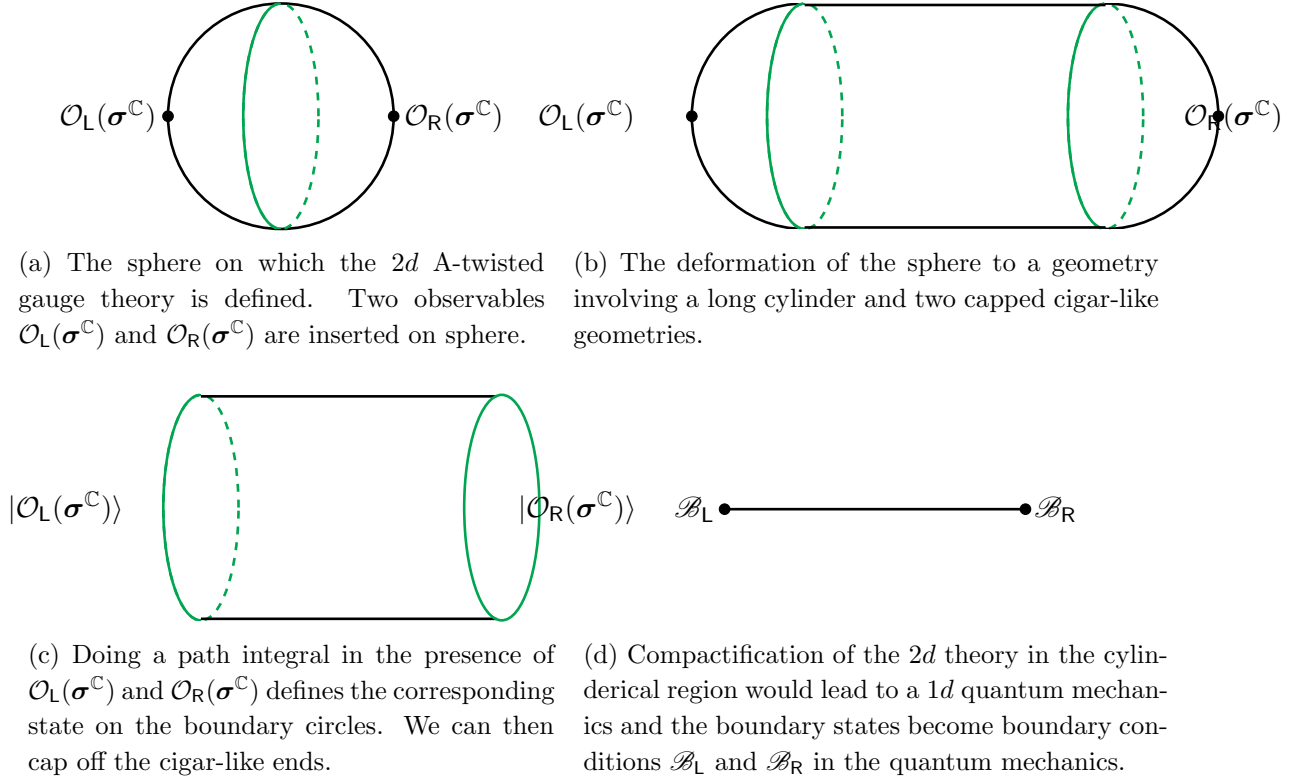

 \centering
 \FromTwoDimensionToOneDimension
 \caption{The 2d $\to$ 1d reduction of the 2d $\mcal{N}=(2,2)$ theory to a quantum mechanical system with boundary conditions prescribed by the observables $\mcal{O}_{\msf{L}}(\sigma)$ and $\mcal{O}_{\msf{R}}(\sigma)$. The subscripts just mean left and right. Restricting to the zero-flux sector, the two-point function on the sphere becomes the partition function of the quantum mechanics with certain boundary conditions.}
 \label{fig:from 2d gauge theory to 1d QM}
\end{figure}

{\bf Rational $\boldsymbol{R}$-matrix from Gauge theory.}
The geometric definition of $R$-matrix \eqref{eqn: dfn of R-matirx} (and its trigonometric and rational versions) depends on two pieces of data: (1) the cohomological stable envelope and its inverse, and (2) a choice of the chamber for the cohomological stable envelope and another (in general different) chamber for its inverse. Therefore, it is natural to consider the configuration of Figure~\ref{fig:two janus interfaces in two dimensions that compute the r-matrix} of the 2d gauge theories on $S^2$.

\begin{figure}[t]
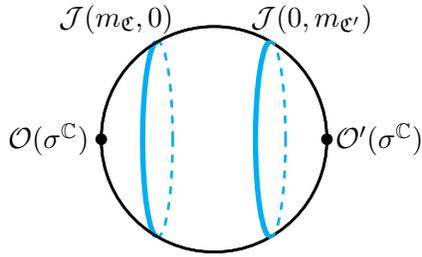
 \centering
\vspace{-1mm}

 \TwoJanusInterfacesInTwoDimensions
 \caption{The configuration of two Janus interfaces for computing the rational $R$-matrix. $\mcal{O}\bigl(\sigma^\mbb{C}\bigr)$ and $\mcal{O}'\bigl(\sigma^\mbb{C}\bigr)$ are certain A-model observables.}
 \label{fig:two janus interfaces in two dimensions that compute the r-matrix}
\end{figure}
$\mcal{J}(m_{\mfk{C}},0)$ and $\mcal{J}(0,m_{\mfk{C}'})$ are certain supersymmetric Janus interfaces.\footnote{Through \eqref{eq:real masses in the quantum mechanics} and going to 1d by compactification on $S^1_\C$, we can think of these interfaces as real-mass Janus interfaces in the corresponding 4d $\mcal{N}=4$ quantum mechanics. The presence of these types of interfaces leads to an exact deformation of the theory \cite[Section~4.2]{DedushenkoNekrasov202109}. Therefore, the precise shape of the profile of the parameters does not matter and only their asymptotic values affect the computations. This is the only fact we need.} Colliding these interfaces gives rise to a single interface: $\mcal{J}(m_{\mfk{C}},0)\mcal{J}(0,m'_{\mfk{C}'})=\mcal{J}(m_{\mfk{C}},m'_{\mfk{C}'})$. Setting $m'_{\mfk{C}'}=m_{\mfk{C}}$ and noting that $\mcal{J}(m_{\mfk{C}},m_{\mfk{C}})$ is equivalent to no interface, we see that $\mcal{J}^{-1}(m_\mfk{C},0)=\mcal{J}(0,m_{\mfk{C}})$. Therefore, we consider the configuration of Figure~\ref{fig:configuration for the computation of r-matrix} for the computation of $R$-matrix.

\begin{figure}[t]
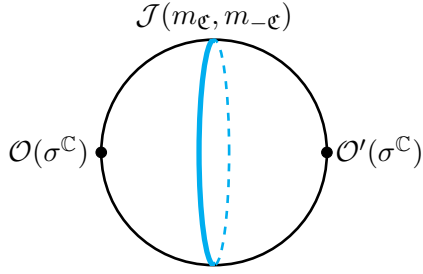
 \centering
\vspace{-1mm}

 \OneJanusInterfaceInTwoDimensions
 \caption{The configuration for the computation of $R$-matrix obtained by colliding the two interfaces in Figure~\ref{fig:two janus interfaces in two dimensions that compute the r-matrix}.}
 \label{fig:configuration for the computation of r-matrix}
\end{figure}
The effect of the presence of the Janus interface $\mcal{J}(m_{\mfk{C}},m_{-\mfk{C}})$ is that the parameters belong to the two opposite chambers $\mfk{C}$ and$-\mfk{C}$ on the two sides of the interface.

The next ingredient is the formula for the computation of the two-point functions of GLSM observables on sphere $S^2$ in the zero-flux sector. In the case of theory $\mscr{T}'_{\text{2d}}$, it is given by \cite{ClossetCremonesiPark201504} (see Appendix~\ref{sec:explicit construction of two-point functions} for the derivation)
\begin{equation}\label{eq:generic expression of correlation functions of GLSM observables}
 \langle\mcal{O}\mcal{O}'\rangle_{S^2,0}=\frac{1}{N!}\bigointsss_{\tenofo{JK}} \sd\mbs{\sigma}^\mbb{C} \mcal{Z}[\mcal{O},\mcal{O}']\bigl(\mbs{\sigma}^\mbb{C}\bigr),
\end{equation}
where the subscript $0$ denotes the zero-flux sector, $\mbs{\sigma}^\mbb{C}=\bigl(\sigma^\mbb{C}_1,\ldots,\mbs{\sigma}^\mbb{C}_N\bigr)$,
\[
 \sd\mbs{\sigma}^\mbb{C}=\sd\sigma^\mbb{C}_1\wedge \dots\wedge \sd \sigma^\mbb{C}_N,
 \qquad
 \mcal{Z}[\mcal{O},\mcal{O}']\bigl(\mbs{\sigma}^\mbb{C}\bigr)=\mcal{Z}\bigl(\mbs{\sigma}^\mbb{C}\bigr)\mcal{O}\bigl(\mbs{\sigma}^\mbb{C}\bigr)\mcal{O}'\bigl(\mbs{\sigma}^\mbb{C}\bigr),
\]
and
\begin{equation}\label{eq:partition function of GLSM on the coulomb branch}
 \mcal{Z}\bigl(\mbs{\sigma}^\mbb{C}\bigr)\equiv \prod_{a=1}^{N}\prod_{i=1}^L\frac{1}{\bigl(\sigma^\mbb{C}_a-m^\mbb{C}_i+\frac{1}{2}\hbar^\mbb{C}\bigr)\bigl(-\sigma^\mbb{C}_a+m^\mbb{C}_i+\frac{1}{2}\hbar^\mbb{C}\bigr)}\cdot\prod_{\substack{a\ne b}}^N\bigl(\si^\mbb{C}_{a}-\si^\mbb{C}_b\bigr).
\end{equation}
The contribution to the integral is determined by the JK-pole prescription and comes from the contribution of chiral multiplets. The location of the poles is given by \eqref{eq:the JK poles}.

The A-model observables relevant for the computation of $R$-matrix are constructed from~\eqref{eq:cohomological stable envelope after identification of parameters} as follows. First of all these observables depend on the pair $(N,L)$ and also the injective map~$p$. Furthermore, we can consider a permutation $\pi\colon\{1,\ldots,L\}\to\{1,\ldots,L\}$ which leads to a~permutation of complex masses and the real parameters \eqref{eq:real masses in the quantum mechanics}, and hence determines a choice of chamber. We denote these observables as $\mcal{O}_{p,\pi,N,L}$. It turns out that the relevant two-point functions are \cite{BullimoreKimLukowski201708}\footnote{In \cite[Section~6.2]{BullimoreKimLukowski201708}, this type of two-point functions has been interpreted as an inner product in the corresponding $\mcal{N}=4$ quantum mechanics, where the $\star$ operation provides a notion of ``complex conjugation".}
\begin{eqgathered}\label{eq:relevant correlation functions of a-model observables}
 \bigl\langle \mcal{O}_{p,\pi;N,L}\bigl(\sigma^\mbb{C}\bigr)\mcal{O}^\star_{p',\pi';N,L}\bigl(\sigma^\mbb{C}\bigr)\bigr\rangle_{S^2,0},
\end{eqgathered}
where
\begin{equation}\label{eq:relevant a-model observables}
 \mcal{O}_{p,\pi;N,L}\propto\rstab_{\mfk{C}}(p),
\end{equation}
where the chamber $\mfk{C}$ is determined through the permutation $\pi$, and $\star$ operation corresponds to replacing $\sigma^\mbb{C}_a\to-\sigma^\mbb{C}_a$ and $m^\mbb{C}_i\to-m^\mbb{C}_i$. This operation corresponds to the longest permutation~${\pi\colon \{1,\ldots,L\}\to\{L,\ldots,1\}}$ in \eqref{eq:function W for a permutation}. We comment below (see Remark~\ref{rmk:normalization of observables}) on how to determine the proportionality constant in \eqref{eq:relevant a-model observables}.

We are now in a position to compute the rational $R$-matrix for $\mfk{sl}(1|1)$. For this purpose, it suffices to restrict to $L=2$, which through \eqref{N<=L}, implies $N=0,1,2$, corresponding to the vacuum, one-magnon, and two-magnon sectors of the spin chain. The $R$-matrix of the vacuum sector is trivially one. We can compute the $R$-matrix in the one- and two-magnon sectors as follows.

\begin{rmk}[normalization of observables]\label{rmk:normalization of observables} Let us make a curious observation before proceeding to details. In the following computations, we demand that the observables $\mcal{O}_{p,\pi;N,L}$ are orthonormal, meaning $\langle\mcal{O}_{p,\pi;N,L}\mcal{O}^\star_{p,\pi;N,L}\rangle=1$ and $\langle\mcal{O}_{p,\pi;N,L}\mcal{O}^\star_{p',\pi;N,L}\rangle=0$. If we do not stress on orthonormality of these observables, then we end up with a $4\times 4$ matrix that all the components of the middle $2\times 2$ diagonal block\footnote{For an $R$-matrix, this $2\times 2$ block is the one-magnon sector $R$-matrix.} are multiplied by $\bigl(\hbar^{\mbb{C}}\bigr)^{-1}$ and the bottom $1\times 1$ diagonal block\footnote{For an $R$-matrix, this $1\times 1$ block is the two-magnon sector $R$-matrix.} is multiplied by \smash{$\bigl(\hbar^{\mbb{C}}\bigr)^{-2}$}. Such a matrix {\it does not} satisfy the Yang--Baxter equation and as such is not an $R$-matrix at all. In the case of $\mfk{sl}(2)$ spin chains considered in \cite[Section~6.1]{BullimoreKimLukowski201708}, the analogous observables are already normalized and there is no need for a rescaling. The reason is a certain contribution of the adjoint chiral multiplet that takes care of the normalization. In the present case, there is no adjoint chiral multiplet and as such we need this rescaling.
\end{rmk}

{\bf $\boldsymbol{R}$-matrix in the one-magnon sector.} The one-magnon sector is described by setting~${N=1}$, hence there are two possible choices of the injective map, which we call $p$ and $p'$: $p(1)=1$ or $p'(1)=2$.
There are also two possible choices of permutation, which we denote by~$\pi$ and $\pi'$: $\pi(\{1,2\})=\{1,2\}$ and $\pi'(\{1,2\})=\{2,1\}$. There are two complex twisted masses: $m^\mbb{C}=\bigl(m_1^\mbb{C},m_2^\mbb{C}\bigr)$ with $m_1^\mbb{C}+m_2^\mbb{C}=0$, and there is only one effective complex twisted mass~${m^\mbb{C}\equiv m_1^\mbb{C}=-m_2^{\mbb{C}}}$. Similarly, there are two real parameters, defined in \eqref{eq:real masses in the quantum mechanics}, with~${m_1+m_2=0}$, and a single effective parameter $m\equiv m_1=-m_2$. Correspondingly, there are two chambers $\mfk{C}=\{m>0\}$ and $-\mfk{C}=\{m<0\}$. These chambers correspond to the permutations $\pi$ and $\pi'$. We then have
\begin{gather*}
 \mcal{O}_{p,\pi;1,2}\bigl(\si^\mbb{C}\bigr)=\bigl(\hbar^{\mbb{C}}\bigr)^{\frac{1}{2}}\left(-\si^\mbb{C}+m^{\mbb{C}}_2+\frac{\hbar^\mbb{C}}{2}\right),
 \\
 \mcal{O}_{p,\pi';1,2}\bigl(\si^\mbb{C}\bigr)=\bigl(\hbar^{\mbb{C}}\bigr)^{\frac{1}{2}}\left(+\si^\mbb{C}-m_2^\mbb{C}+\frac{\hbar^\mbb{C}}{2}\right)=\mcal{O}^\star_{p,\pi;1,2}\bigl(\si^\mbb{C}\bigr),
 \\
 \mcal{O}_{p',\pi;1,2}\bigl(\si^\mbb{C}\bigr)=\bigl(\hbar^{\mbb{C}}\bigr)^{\frac{1}{2}}\left(+\si^\mbb{C}-m^{\mbb{C}}_1+\frac{\hbar^\mbb{C}}{2}\right),
 \\
 \mcal{O}_{p',\pi';1,2}\bigl(\si^\mbb{C}\bigr)=\bigl(\hbar^{\mbb{C}}\bigr)^{\frac{1}{2}}\left(-\si^\mbb{C}+m_1^\mbb{C}+\frac{\hbar^\mbb{C}}{2}\right)=\mcal{O}^\star_{p',\pi;1,2}\bigl(\si^\mbb{C}\bigr).
\end{gather*}
The prefactor $\bigl(\hbar^\mbb{C}\bigr)^{\frac{1}{2}}$ is chosen to guarantee the orthonormality of the observables $\mcal{O}_{p,\pi;N,L}$, which can be seen using \eqref{eq:generic expression of correlation functions of GLSM observables} as follows
\begin{alignat*}{3}
& \bigl\langle \mcal{O}_{p,\pi;1,2}\bigl(\si^\mbb{C}\bigr)\mcal{O}^\star_{p,\pi;1,2}\bigl(\si^\mbb{C}\bigr)\bigr\rangle_{S^2,0}=1,
 \qquad&&
 \bigl\langle \mcal{O}_{p,\pi;1,2}\bigl(\si^\mbb{C}\bigr)\mcal{O}^\star_{p',\pi;1,2}\bigl(\si^\mbb{C}\bigr)\bigr\rangle_{S^2,0}=0,&
 \\
& \bigl\langle \mcal{O}_{p',\pi;1,2}\bigl(\si^\mbb{C}\bigr)\mcal{O}^\star_{p',\pi;1,2}\bigl(\si^\mbb{C}\bigr)\bigr\rangle_{S^2,0}=1,
 \qquad&&
 \bigl\langle \mcal{O}_{p',\pi;1,2}\bigl(\si^\mbb{C}\bigr)\mcal{O}^\star_{p,\pi;1,2}\bigl(\si^\mbb{C}\bigr)\bigr\rangle_{S^2,0}=0.&
\end{alignat*}
All these observables only depend on a single permutation $\pi$, which means that real masses are chosen in a single chamber, say $\mfk{C}$. Physically this means that there is no Janus interface, and these are just orthonormality conditions on the observables.
Furthermore,
\begin{gather*}
 \bigl\langle \mcal{O}_{p,\pi';1,2}\bigl(\si^\mbb{C}\bigr)\mcal{O}^\star_{p,\pi;1,2}\bigl(\si^\mbb{C}\bigr)\bigr\rangle_{S^2,0}=-\frac{m_2^\mbb{C}-m_1^\mbb{C}}{m_2^\mbb{C}-m_1^\mbb{C}+\hbar^\mbb{C}},
 \\
 \bigl\langle \mcal{O}_{p,\pi';1,2}\bigl(\si^\mbb{C}\bigr)\mcal{O}^\star_{p',\pi;1,2}\bigl(\si^\mbb{C}\bigr)\bigr\rangle_{S^2,0}=\frac{\hbar^\mbb{C}}{m_2^\mbb{C}-m_1^\mbb{C}+\hbar^\mbb{C}},
 \\
 \bigl\langle \mcal{O}_{p',\pi';1,2}\bigl(\si^\mbb{C}\bigr)\mcal{O}^\star_{p,\pi;1,2}\bigl(\si^\mbb{C}\bigr)\bigr\rangle_{S^2,0}=\frac{\hbar^\mbb{C}}{m_2^\mbb{C}-m_1^\mbb{C}+\hbar^\mbb{C}},
 \\
 \bigl\langle \mcal{O}_{p',\pi';1,2}\bigl(\si^\mbb{C}\bigr)\mcal{O}^\star_{p',\pi;1,2}\bigl(\si^\mbb{C}\bigr)\bigr\rangle_{S^2,0}=-\frac{m_2^\mbb{C}-m_1^\mbb{C}}{m_2^\mbb{C}-m_1^\mbb{C}+\hbar^\mbb{C}}.
\end{gather*}
These correlation functions involve two observables and in the presence of a Janus interface $\mscr{J}(m_\mfk{C},m_{-\mfk{C}})$, the observables in the two sides of the interval depend on two different permutations $\pi$ and $\pi'$.
Hence, we find the one-magnon sector $R$-matrix constructed from the quantum mechanics
\begin{gather}\label{eq:r-matrix in one-magnon sector from gauge theory}
 \rrmatrix^{\tenofo{gt},(1)}\equiv
 \begin{pmatrix}
 -\frac{m_2^\mbb{C}-m_1^\mbb{C}}{m_2^\mbb{C}-m_1^\mbb{C}+\hbar^\mbb{C}} & \frac{\hbar^\mbb{C}}{m_2^\mbb{C}-m_1^\mbb{C}+\hbar^\mbb{C}}
 \vspace{1mm}\\
 \frac{\hbar^\mbb{C}}{m_2^\mbb{C}-m_1^\mbb{C}+\hbar^\mbb{C}} & -\frac{m_2^\mbb{C}-m_1^\mbb{C}}{m_2^\mbb{C}-m_1^\mbb{C}+\hbar^\mbb{C}}
 \end{pmatrix}.
\end{gather}

\begin{rmk}
 By comparing this with \cite[equation~(6.7)]{BullimoreKimLukowski201708}, we see that the one-magnon sector $R$-matrix is the same for $\mfk{sl}(2)$ and $\mfk{sl}(1|1)$. This is expected since the statistics of the excitation do not play any role in the one-magnon. We have seen this above but in a different disguise: the stable envelopes have been the same as we have explained in Remark~\ref{rmk:abelianN=2vsN=4}.
\end{rmk}
Finally,
\begin{gather*}
 \bigl\langle \mcal{O}_{p,\pi;1,2}\bigl(\si^\mbb{C}\bigr)\mcal{O}^\star_{p,\pi';1,2}\bigl(\si^\mbb{C}\bigr)\bigr\rangle_{S^2,0}=-\frac{m_2^\mbb{C}-m_1^\mbb{C}}{m_2^\mbb{C}-m_1^\mbb{C}-\hbar^\mbb{C}},
 \\
 \bigl\langle \mcal{O}_{p,\pi;1,2}\bigl(\si^\mbb{C}\bigr)\mcal{O}^\star_{p',\pi';1,2}\bigl(\si^\mbb{C}\bigr)\bigr\rangle_{S^2,0}=-\frac{\hbar^\mbb{C}}{m_2^\mbb{C}-m_1^\mbb{C}-\hbar^\mbb{C}},
 \\
 \bigl\langle \mcal{O}_{p',\pi;1,2}\bigl(\si^\mbb{C}\bigr)\mcal{O}^\star_{p,\pi';1,2}\bigl(\si^\mbb{C}\bigr)\bigr\rangle_{S^2,0}=-\frac{\hbar^\mbb{C}}{m_2^\mbb{C}-m_1^\mbb{C}-\hbar^\mbb{C}},
 \\
 \bigl\langle \mcal{O}_{p',\pi;1,2}\bigl(\si^\mbb{C}\bigr)\mcal{O}^\star_{p',\pi';1,2}\bigl(\si^\mbb{C}\bigr)\bigr\rangle_{S^2,0}=-\frac{m_2^\mbb{C}-m_1^\mbb{C}}{m_2^\mbb{C}-m_1^\mbb{C}-\hbar^\mbb{C}},
\end{gather*}
that regarding \eqref{eq:r-matrix in one-magnon sector from gauge theory} is just the inverse \smash{$\bigl(\rrmatrix^{\tenofo{gt},(1)}\bigr)^{-1}$} of $\rrmatrix^{\tenofo{gt},(1)}$, as it is expected
\[
 \bigl(\rrmatrix^{\tenofo{gt},(1)}\bigr)^{-1}\bigl(m_1^\mbb{C},m_2^\mbb{C}\bigr)=
 \begin{pmatrix}
 -\frac{m_2^\mbb{C}-m_1^\mbb{C}}{m_2^\mbb{C}-m_1^\mbb{C}-\hbar^\mbb{C}} & -\frac{\hbar^\mbb{C}}{m_2^\mbb{C}-m_1^\mbb{C}-\hbar^\mbb{C}}
 \vspace{1mm}\\
 -\frac{\hbar^\mbb{C}}{m_2^\mbb{C}-m_1^\mbb{C}-\hbar^\mbb{C}} & -\frac{m_2^\mbb{C}-m_1^\mbb{C}}{m_2^\mbb{C}-m_1^\mbb{C}-\hbar^\mbb{C}}
 \end{pmatrix}.
\]
Let us now compute the correlation functions \eqref{eq:relevant correlation functions of a-model observables} for the gauge theory corresponding to the two-magnon sector.

{\bf $\boldsymbol{R}$-matrix in the two-magnon sector.} The two-magnon sector is described by setting~${N=2}$. There is only one possible choice of the injective map $p(\{1,2\})=\{1,2\}$. As the Abelian case, there are two possible choices of permutation, which we denote by $\pi$ and $\pi'$: $\pi(\{1,2\})=\{1,2\}$ and $\pi'(\{1,2\})=\{2,1\}$. The complex masses and the real parameters are as in the Abelian case. We have $\bigl(\sigma^\mbb{C}=\bigl(\si^\mbb{C}_1,\si^\mbb{C}_2\bigr)$)
\[
 \mcal{O}_{p,\pi;2,2}\bigl(\sigma^\mbb{C}\bigr)=\hbar^{\mbb{C}}\left(\frac{\bigl(-\si^\mbb{C}_1+m_2^\mbb{C}+\frac{\hbar^\mbb{C}}{2}\bigr)
 \bigl(\si^\mbb{C}_2-m_1^\mbb{C}+\frac{\hbar^\mbb{C}}{2}\bigr)}{\sigma_2^\mbb{C}-\sigma_1^\mbb{C}}\right)+\bigl(\sigma_1^\mbb{C}\longleftrightarrow\sigma_2^\mbb{C}\bigr),
\]
and $\mcal{O}_{p,\pi';2,2}\bigl(\sigma^\mbb{C}\bigr)=\mcal{O}^\star_{p,\pi;2,2}\bigl(\sigma^\mbb{C}\bigr)$. Again, the prefactor $\hbar^\mbb{C}$ is determined by demanding orthogonality. Using \eqref{eq:generic expression of correlation functions of GLSM observables}, we see that
\[
 \bigl\langle \mcal{O}_{p,\pi;2,2}\bigl(\sigma^\mbb{C}\bigr)\mcal{O}^\star_{p,\pi;2,2}\bigl(\sigma^\mbb{C}\bigr)\bigr\rangle_{S^2,0}=1,
\]
which means that $\mcal{O}_{p,\pi;2,2}\bigl(\sigma^\mbb{C}\bigr)$ is normalized. Furthermore,
\[
 \bigl\langle \mcal{O}_{p,\pi';2,2}\bigl(\sigma^\mbb{C}\bigr)\mcal{O}^\star_{p,\pi;2,2}\bigl(\sigma^\mbb{C}\bigr)\bigr\rangle_{S^2,0}=-\frac{m_2^\mbb{C}-m_1^\mbb{C}-\hbar^\mbb{C}}{m_2^\mbb{C}-m_1^\mbb{C}+\hbar^\mbb{C}}.
\]
This gives the geometric $R$-matrix in the two-magnon sector
\begin{equation}\label{eq:r-matrix in two-magnon sector from gauge theory}
 \rrmatrix^{\tenofo{gt},(2)}=-\frac{m_2^\mbb{C}-m_1^\mbb{C}-\hbar^\mbb{C}}{m_2^\mbb{C}-m_1^\mbb{C}+\hbar^\mbb{C}}.
\end{equation}
Finally,{\samepage
\[
 \bigl\langle \mcal{O}_{p,\pi;2,2}\bigl(\sigma^\mbb{C}\bigr)\mcal{O}^\star_{p,\pi';2,2}\bigl(\sigma^\mbb{C}\bigr)\bigr\rangle_{S^2,0}=-\frac{m_2^\mbb{C}-m_1^\mbb{C}+\hbar^\mbb{C}}{m_2^\mbb{C}-m_1^\mbb{C}-\hbar^\mbb{C}},
\]
which is just the inverse of \eqref{eq:r-matrix in two-magnon sector from gauge theory}.}

Putting \eqref{eq:r-matrix in one-magnon sector from gauge theory} and \eqref{eq:r-matrix in two-magnon sector from gauge theory} together, the $R$-matrix of $\mfk{sl}(1|1)$ superspin chain constructed through the Bethe--Gauge correspondence is
\[
	\def\arraystretch{1.5}
	\setlength{\arrayrulewidth}{.6pt}
 \setlength{\arraycolsep}{6pt}
	\rrmatrix^{\tenofo{gt}}=
	\left(
	\begin{array}{c|cc|c}
		1 & 0 & 0 & 0
		\\
		\hline
		0 & -\frac{m_2^\mbb{C}-m_1^\mbb{C}}{m_2^\mbb{C}-m_1^\mbb{C}+\hbar^\mbb{C}} & \frac{\hbar^\mbb{C}}{m_2^\mbb{C}-m_1^\mbb{C}+\hbar^\mbb{C}} & 0
 \\[6pt]
 0 & \frac{\hbar^\mbb{C}}{m_2^\mbb{C}-m_1^\mbb{C}+\hbar^\mbb{C}} & -\frac{m_2^\mbb{C}-m_1^\mbb{C}}{m_2^\mbb{C}-m_1^\mbb{C}+\hbar^\mbb{C}} & 0
		\\ [6pt]
		\hline
		0 & 0 & 0 & -\frac{m_2^\mbb{C}-m_1^\mbb{C}-\hbar^\mbb{C}}{m_2^\mbb{C}-m_1^\mbb{C}+\hbar^\mbb{C}}
	\end{array}\right).
\]
If we identify $u\to m^\mbb{C}_1-m^\mbb{C}_2$ and $\hbar\to\hbar^{\mbb{C}}$, we see from \eqref{eq:rimanyi-rozansky gl(1|1) rational r-matrix} that
$\rrmatrix^{\tenofo{gt}}(u)=\rrmatrix^{\tenofo{RR}}(u)$.
Therefore, our A-model localization successfully reproduces the geometric $R$-matrix of $\mfk{sl}(1|1)$ superspin chain, as it is expected from the Bethe/Gauge correspondence.

\subsection{Further details on two-point functions of 2d GLSM observables}
\label{sec:two-point function of GLSM observables}
Here, we will collect concepts and formulas for the computation of two-point function of $\Omega$-deformed A-model observables. We follow \cite{ClossetCremonesiPark201504} to which we refer for all the details.

\subsubsection{Brief recap of Closset--Cremonesi--Park localization formula}

Consider a 2d $\mcal{N}(2,2)$ GLSM with $\tenofo{rank}(\mfk{g})=r$ gauge Lie algebra $\mfk{g}$, whose Cartan subalgebra we denote by $\mfk{h}$, in the $\Omega$-deformed sphere background $S^2_\Omega$.\footnote{We do not need the details of this background, which can be found in \cite[Section~2.1]{ClossetCremonesiPark201504}.} We denote the corresponding groups by $G$ and $H$ and the complexified Lie algebras by $\mfk{g}_\mbb{C}$ and $\mfk{h}_\mbb{C}$, respectively. Different topological sectors are labeled by magnetic flux $\mbs{n}=(n_1,\dots,n_r)\in\Gamma_{\mfk{g}^\vee}\simeq\mbb{Z}^{r}\subset \mfk{ih}$, where
\[
	\Gamma_{\mfk{g}^\vee}\equiv \{\mbs{n}\mid  \rho(\mbs{n}),\, \forall\rho\in\Gamma_{\mfk{g}}\},
\]
and $\rho(\mbs{n})$ is given by the canonical pairing of dual vector spaces and electric charges take value in $\Gamma_{\mfk{g}}$, the weight lattice. Supersymmetric observables belong to the cohomology of left- and right-moving supercharges $\twdQ$ and $\wt{\twdQ}$. In the presence of $\Omega$-deformation, we have $\delta_{\twdQ}\sigma=\delta_{\wt{\twdQ}}\sigma=0$ when the vector field $V$, that generates the rotation isometry, vanishes. Hence such operators are functions of $\sigma$, the complex scalar field in the vector multiplet of the theory. Furthermore, they can only be inserted at the north and south poles of $S^2_\Omega$. When we turn off the $\Omega$ background, they can be inserted anywhere on the sphere. The two-point function of such operators, which we denote by $\mcal{O}^N(\sigma_N)$ and $\mcal{O}^S(\sigma_S)$ for $N$ and $S$ denoting the north and south poles, are given by \cite{ClossetCremonesiPark201504}
\begin{gather}
	\left\langle\mcal{O}^N(\sigma_N)\mcal{O}^R(\sigma_R)\right\rangle_{S^2_\Omega}\nonumber\\
\qquad=\frac{(-1)^{\#}}{|W|}\sum_{\mbs{n}\in\Gamma_{\mfk{g}^\vee}}\msf{q}^{\mbs{n}}\bigointsss_{\tenofo{JK}}\mcal{Z}_{\mbs{n}}^{\tenofo{$1$-loop}}(\sigma,\epsilon)\mcal{O}^N\left(\sigma-\frac{\epsilon\mbs{ n}}{2}\right)\mcal{O}^S\left(\sigma+\frac{\epsilon \mbs{n}}{2}\right),\label{eq:the correlation functions of GLSMs observables}
\end{gather}
where $|W|$ denotes the order of the Weyl group $W$ of $\mfk{g}$, $\#$ denotes the number of chiral multiplets with $R$-charge $2$,\footnote{The prefactor $(-1)^{\#}$ is a sign ambiguity which is irrelevant for the rest of our discussion since we do not have a chiral multiplet of $R$-charge $+2$. The prescription to fix this sign ambiguity, which we have stated here, is given in \cite[Section~4.5]{ClossetCremonesiPark201504}.} $\sigma$ is the coordinate on $\wt{\mscr{M}}\equiv \mfk{h}_{\mbb{C}}\simeq\mbb{C}^{r}$, which is the cover of $\mscr{M}\equiv \wt{\mscr{M}}/W$, the Coulomb-branch moduli space, and $\epsilon$ is the $\Omega$-deformation parameter. $\mcal{Z}_{\mbs{n}}^{\tenofo{$1$-loop}}(\sigma)$ is the 1-loop contribution to the partition function from the $\mbs{n}$\textsuperscript{th} topological sector, which is give by
\[
	\mcal{Z}_{\mbs{n}}^{\tenofo{$1$-loop}}(\sigma)=\mcal{Z}_{\mbs{n}}^{\tenofo{vector}}(\sigma;\epsilon)\prod_i\mcal{Z}_{\mbs{n}}^{\tenofo{chiral}_i}(\sigma;\epsilon),
\]
where $\tenofo{chiral}_i$ denotes the $i$\textsuperscript{th} chiral multiplet. These contributions could be written explicitly using
\begin{equation}\label{eq:the function Zp}
	Z_{p}(x,\epsilon)\equiv
	\begin{cases}
		\displaystyle\prod_{q=-\frac{p}{2}+1}^{\frac{p}{2}-1}(x+\epsilon q), &\quad p>1,
		\\
1, &\quad p=1,
		\\
		\displaystyle\prod_{q=-\frac{|p|}{2}}^{\frac{|p|}{2}} (x+\epsilon q)^{-1}, &\quad p<1.
	\end{cases}
\end{equation}
For our purpose, we only need to know $Z_0(x,\epsilon=0)$ and $Z_2(x,\epsilon=0)$ for $\mbs{n}=\mbs{0}$. From \eqref{eq:the function Zp}, we have
\begin{equation}\label{eq:functions Z0 and Z2}
	Z_{2}(x,\epsilon=0)=x, \qquad Z_{0}(x,\epsilon=0)=\frac{1}{x}.
\end{equation}
We then have
\begin{equation}\label{eq:the contribtion of vector and chiral fields to the correlation functions of GLSMs}
	\begin{aligned}
		\mcal{Z}_{\mbs{n}}^{\tenofo{vector}}(\sigma,\epsilon)&=\prod_{\alpha\in\Delta}Z_{2-\alpha(\mbs{n})}(\alpha(\sigma),\epsilon),
		\\ \mcal{Z}_{\mbs{n}}^{\tenofo{chiral}_i}(\sigma,\epsilon)&=\prod_{\rho_i\in\mcal{R}_i}Z_{q^R_i-\rho_i(\mbs{n})}(\rho_i(\sigma)+m_i,\epsilon).
	\end{aligned}
\end{equation}
In this formula, $\Delta$ denotes the set of roots, $\rho_i$ labels the weights of representation $\mcal{R}_i$ of $G$ in which the $i$\textsuperscript{th} chiral multiplet transforms, and $m_i$ and $q^R_i$ denote the twisted mass and $U(1)_V$ $R$-symmetry charge of the $i$\textsuperscript{th} chiral multiplet. Defining the notation
\begin{equation}\label{eq:the collective notation for multiplets}
	\bigl(\mscr{W}_\mcal{I},m_\mcal{I},q^R_\mcal{I}\bigr)=
	\begin{cases}
		\bigl(\rho_i,m_i,q^R_i\bigr), &  \mcal{I}=(i,\rho_i),
		\\
		(\alpha,0,2), & \mcal{I}=\alpha.
	\end{cases}
\end{equation}
We can write the 1-loop partition function as follows:\footnote{Considering Figure~\ref{fig:from 2d gauge theory to 1d QM}, we only need the zero-flux sector of the GLSM observables in the computation of interval partition function of quantum mechanics, as we will see momentarily.}
\begin{equation}\label{eq:the compact form of the 1-loop partition function}
	\mcal{Z}_{\mbs{n}}^{\tenofo{1-loop}}(\sigma,\epsilon)=\prod_{\mcal{I}}\mcal{Z}_{\mscr{W}_\mcal{I}(\mbs{n})-q^R_\mcal{I}}(\mscr{W}_\mcal{I}(\sigma)+m_\mcal{I},\epsilon).
\end{equation}
The parameter $\msf{q}$ in \eqref{eq:the correlation functions of GLSMs observables} is related to the complexified FI parameters, defined as follows:
\[
	\eta^I\equiv \frac{\theta^I}{2\pi}+\mfk{i}\zeta^I\in\mbb{C}, \qquad 1\le I\le \dim(\mfk{c}^*_{\mbb{C}}),
\]
where $\theta^I$ and $\zeta^I$ are the theta angle and real FI parameters. They belong to the dual of the complexified center $\mfk{c}^*_{\mbb{C}}\subset\mfk{g}^*_{\mbb{C}}$ and could be thought of as elements of $\mfk{h}^*_{\mbb{C}}$ through the embedding~${\mfk{c}^*_{\mbb{C}}\hookrightarrow\mfk{h}^*_{\mbb{C}}}$. We denote the image of the physical FI parameters $\zeta^I$ under this embedding by~${\zeta^a, a=1,\dots,r}$. Then,
\smash{$
\msf{q}^{\mbs{n}}\equiv \exp\bigl(2\pi\mfk{i}\sum_{I}\mbs{\eta}\cdot\mbs{n}\bigr)$},
where $n_I$ denotes the flux in the free $I$\textsuperscript{th} $\U(1)$ part of the center of $G$. Note that $\mbs{n}\in\mfk{i\mfk{h}}$ and hence $\mbs{\eta}\cdot\mbs{n}$ is given by the canonical pairing between $\mfk{h}_\mbb{C}$ and $\mfk{h}^*_{\mbb{C}}$.

Finally, we need to specify the contour of integration, which is given by the JK contour \cite{BrionVergne199903,JeffreyKirwan199307,SzenesVergne200306}. The poles of $\mcal{Z}_{\mbs{n}}^{\tenofo{1-loop}}(\sigma,\epsilon)$ are the loci of intersections of $s\ge r$ hyperplanes $\bigl\{\mcal{H}^{k_1}_{\mcal{I}_1},\ldots,\mcal{H}^{k_s}_{\mcal{I}_s}\bigr\}$, where
\begin{equation}\label{eq:the singular hyperplanes}
	\mcal{H}^{k}_\mcal{I}\equiv \left\{\sigma\,\biggl|\, \mscr{W}_\mcal{I}(\sigma)=-m_\mcal{I}-\left(k+\frac{q^R_\mcal{I}-\mscr{W}_\mcal{I}(\mbs{n})}{2}\right)\epsilon \right\}, \qquad k\in\bigl[0,\mscr{W}_\mcal{I}(\mbs{n})-q^R_\mcal{I}\bigr]_{\tenofo{int}},
\end{equation}
and $[a,b]_{\tenofo{int}}$ denotes the set of integers between $a$ and $b$. There are two types of singularities: the non-degenerate singularities corresponding to the intersection of exactly $r$ hyperplanes and the degenerate ones corresponding to the intersection of $s>r$ hyperplanes. This complex-codimension-$r$ singular loci is denoted as $\wt{\mscr{M}}^{\tenofo{sing}}_{\mbs{n}}\subset\wt{\mscr{M}}$. Note that $\mcal{H}^k_\mcal{I}$ is a singular locus of~\smash{$\mcal{Z}_{\mbs{n}}^{\tenofo{1-loop}}$} only when $\mcal{I}$ labels a chiral multiplet field. $\mcal{H}^k_\alpha$ labels the codimension-1 poles with zero residue. The definition of JK residues at \smash{$\sigma_*\in\wt{\mscr{M}}_{\mbs{n}}^{\tenofo{sing}}$} depends on an additional parameter $\eta\in \mfk{i}\mfk{h}^*$, which is defined as follows. The charge $\mscr{W}_{\mcal{I}}$, defined in \eqref{eq:the collective notation for multiplets}, belongs to $\mfk{i}\mfk{h}^*$ and hence any set of charges $\{\mscr{W}_{\mcal{I}_1},\dots,\mscr{W}_{\mcal{I}_r}\}$ generates a cone
\begin{equation}\label{eq:the generic cones}
	\tenofo{Con}[\mscr{W}_{\mcal{I}_1},\dots,\mscr{W}_{\mcal{I}_r}]\subset\mfk{i}\mfk{h}^*.
\end{equation}
The set of all such cones can be divided into dimension-$r$ chambers that are separated by codimension-$1$ walls. The chamber to FI parameters $\xi^a\in\mfk{i}\mfk{h}^*, a=1,\dots,r$ belong to determines the phase of the GLSM \cite{MorrisonPlesser199412,Witten199301}. The effective FI parameter $\zeta^{\tenofo{UV},I}_{\tenofo{eff}}$ at the infinity (i.e., when~${|\sigma|\to\infty}$) of Coulomb branch $\mscr{M}$ is defined to be
\[
	\zeta^{\tenofo{UV},a}_{\tenofo{eff}}\equiv \zeta^a+\frac{b^a}{2\pi}\lim_{\msf{R}\to \infty}\ln \msf{R},\qquad b^a\equiv \sum_{i}\tenofo{Tr}_{\mcal{R}_i}(t^a)=\sum_{\mcal{I}}\mscr{W}_\mcal{I}.
\]
$t^a$ denotes the generator of the $a$\textsuperscript{th} $\mfk{u}(1)\subset \mfk{g}$,\footnote{$t^a$ is the image of the $I$\textsuperscript{th} $\mfk{u}(1)$ generator of the center of $\mfk{g}^*$ under the embedding $\mfk{c}^*\hookrightarrow\mfk{h}^*$.} and the last equality follows from embedding~${\mfk{c}^*\subset \mfk{i}\mfk{h}^*}$. Hence, phases of the theory on the Coulomb branch depend on the chamber~${\mfk{C}\subset\mfk{i}\mfk{h}^*}$ that the parameter $\zeta^{\tenofo{UV}}_{\tenofo{eff}}$ belong to. This can be stated as follows: $\zeta^{\tenofo{UV}}_{\tenofo{eff}}$ belongs to a chamber $\mcal{C}\subset\mfk{i}\mfk{h}^*$ if
$\zeta+\frac{b}{2\pi}\ln \msf{R}\in\mcal{C}$, $ \forall \msf{R}\ge \msf{R}_*$,
for some $\msf{R}_*$, which always exists for practical purposes \cite{ClossetCremonesiPark201504}. When $b\ne 0$, there could be two cases: 1) either $b$ belongs to a definite chamber~$\mcal{C}$; In this case, $\zeta$ is irrelevant and the theory has only one phase determined by the chamber~$\mcal{C}$, or 2) $b$ belongs to a wall where chambers~${\mcal{C}_1,\dots,\mcal{C}_p}$ meet; In this case, one can turn on $\zeta$ and depending on its value, $\zeta^{\tenofo{UV}}_{\tenofo{eff}}$ could belong to any of the chambers $\mfk{C}_i,i=1,\dots,p$, which in turn determines the phase of the theory. When $b=0$, and hence the IR theory is a CFT, then~${\zeta^{\tenofo{UV}}_{\tenofo{eff}}=\zeta}$ could be taken to be in any chamber $\mcal{C}\subset\mfk{i}\mfk{h}^*$. We then set
\begin{equation}\label{eq:the definition of effective FI parameter}
	\xi\equiv \zeta^{\tenofo{FI}}_{\tenofo{eff}}.
\end{equation}
The JK residues can now be defined as follows. If a singular point $\sigma_*$ comes from the intersection of $s\ge r$ hyperplanes $\bigl\{\mcal{H}^{k_1}_{\mcal{I}_1},\dots,\mcal{H}^{k_s}_{\mcal{I}_s}\bigr\}$, we denote the set of $\mcal{H}$-charges as $\mbs{\mscr{W}}(\sigma_*)\equiv \{\mscr{W}_{\mcal{I}_1},\dots,\mscr{W}_{\mcal{I}_s}\}$. The JK residues at $\sigma_*=\{\sigma_{*,1},\dots,\sigma_{*,r}\}$ is then give by
\begin{gather}
	\left.\tenofo{JK-Res}[\mbs{\mscr{W}}(\sigma),\xi]\prod_{\mscr{W}_j\in \mbs{\mscr{W}}_r}\frac{\sd\sigma_1\wedge\cdots\wedge\sd\sigma_{r}}{\mscr{W}_{1}(\sigma)\cdots \mscr{W}_{r}(\sigma)}\right|_{\sigma=\sigma_*}\nonumber\\
\qquad=
	\begin{cases}
		\dfrac{1}{|{\det}(\mscr{W}_i)|}, & \xi\in\tenofo{Con}[\mscr{W}_1,\dots,\mscr{W}_r],
		\\
0, & \xi\notin\tenofo{Con}[\mscr{W}_1,\dots,\mscr{W}_r],
	\end{cases}\label{eq:the generic prescription for JK residues}
\end{gather}
where $\mbs{\mscr{W}}_r$ is any set among $s \choose r$ sets of $r$ $\mfk{h}$-charges in $\mbs{\mscr{W}}(\sigma_*)$. This concludes our discussion of the necessary results from \cite{ClossetCremonesiPark201504}.

\subsubsection{Explicit computation of two-point functions}\label{sec:explicit construction of two-point functions}
We would like to compute the contribution of vector and chiral multiplets to the correlation functions \eqref{eq:relevant correlation functions of a-model observables}, given in \eqref{eq:the correlation functions of GLSMs observables} or more compactly in \eqref{eq:the compact form of the 1-loop partition function}, explicit.
For the fields in the vector multiplet, $\mcal{I}=(i,\rho_i)\equiv (ab,\rho_{ab})$, as \eqref{eq:the collective notation for multiplets}
$\rho_{ab}=\alpha_{ab}=e_a-e_b$, $ a\ne b$,
where $\alpha_{ab}$ denotes the $ab$\textsuperscript{th} root of $\mfk{u}(N)$, associated to the generator $E_{ab}$, and $\{e_1,\dots,e_N\}$ is the standard basis for $\mbb{R}^N$.\footnote{For $\alpha_{ab}=e_a-e_b$ for $a\ne b$, we consider the one-dimensional vector space generated by $E_{ab}$. This space is an eigenspace of the Cartan subalgebra with eigenvalue $\alpha_{ab}$. These eigenspaces together with the Cartan subalgebra generate $\mfk{u}(N)$.} Therefore,
\begin{gather}
		\mscr{W}_\mcal{I}\bigl(\sigma^\mbb{C}\bigr)=\rho_{ab}\bigl(\sigma^\mbb{C}\bigr)=\si_a^\mbb{C}-\si_b^\mbb{C}, \qquad a\ne b,\nonumber
		\\
		\mscr{W}_\mcal{I}(\mbs{n})=\rho_{ab}(\mbs{n})=n_a-n_b,\qquad a\ne b.\label{eq:quantities W(si) and W(n)}
\end{gather}
Similarly, for chiral fields in the fundamental $(+)$ and antifundamental $(-)$ representations, the weights are ($\mcal{I}=(i,\rho_i)\equiv (a,\rho_{a,\pm})$ as in \eqref{eq:the collective notation for multiplets})
\begin{equation}\label{eq:the weights for the fundamental and antifundamental representations}
	\rho_{a,\pm}=\pm e_a, \qquad a=1,\dots,N,
\end{equation}
and hence for the chiral multiplets\footnote{We use the same notation for chiral multiplets of the theory $\mscr{T}''_{\text{2d}}$ in Figure~\ref{fig:from four to lower dimensions} as the ones for $\mscr{T}'''_{\text{3d}}$ and its dimensional reductions to 2d and 1d.} $\bigl(\mcal{Q}\indices{_a^i}, \wt{\mcal{Q}}\indices{_i^a}\bigr)$, we have
\begin{gather*}
 \mcal{Q}\indices{_a^i}\colon\ 	\mscr{W}_\mcal{I}\bigl(\sigma^\mbb{C}\bigr)+m_\mcal{I}=\rho_{a,+}\bigl(\sigma^\mbb{C}\bigr)+m_{\mcal{Q}\indices{_a^i}}=\si_a^\mbb{C}-m^\mbb{C}_i+\hbar^\mbb{C},
		\\
	\wt{\mcal{Q}}\indices{_i^a}\colon\ 	\mscr{W}_\mcal{I}\bigl(\sigma^\mbb{C}\bigr)+m_\mcal{I}=\rho_{a,-}\bigl(\sigma^\mbb{C}\bigr)+m_{\wt{\mcal{Q}}\indices{_i^a}}=-\si^\mbb{C}_a+m_i^\mbb{C}+\hbar^\mbb{C}.
\end{gather*}
Using \eqref{eq:functions Z0 and Z2} and \eqref{eq:the contribtion of vector and chiral fields to the correlation functions of GLSMs},
the contribution of vector multiplet is\footnote{Note that the contribution of vector multiplet is the same as an adjoint chiral with $R$-charge $+2$ \cite{ClossetCremonesiPark201504}. Hence, we have taken $q_\tenofo{vector}^R=2$.}
\begin{equation}\label{eq:the contribution of vector multiplet to partition function}
	\mcal{Z}^{\tenofo{vector}}_2(\sigma,\epsilon=0)=\prod_{\substack{a,b=1 \\ a\ne b}}^N\bigl(\si_{a}^\mbb{C}-\si_b^\mbb{C}\bigr).
\end{equation}
Similarly, setting $q_{\mcal{Q}}^R=q_{\wt{\mcal{Q}}}^R=0$, we obtain the contribution of chiral multiplets
\begin{gather}
		\mcal{Z}_0^\mcal{Q}(\sigma^\mbb{C};\epsilon=0)=\prod_{a=1}^{N}\prod_{i=1}^L\frac{1}{\bigl(\sigma_a^\mbb{C}-m_i^\mbb{C}+\frac{\hbar^\mbb{C}}{2}\bigr)},\nonumber
		\\
		\mcal{Z}_0^{\wt{\mcal{Q}}}(\sigma^\mbb{C};\epsilon=0)=\prod_{a=1}^{N}\prod_{i=1}^L\frac{1}{\bigl(-\sigma_a^\mbb{C}+m_i^\mbb{C}+\frac{\hbar^\mbb{C}}{2}\bigr)}.\label{eq:the contribution of chiral multiplets to partition function}
\end{gather}
Putting together \eqref{eq:the contribution of vector multiplet to partition function} and \eqref{eq:the contribution of chiral multiplets to partition function} and using \eqref{eq:the correlation functions of GLSMs observables}, we thus see that the zero-flux sector contribution to the sphere partition function can be written explicitly as\footnote{As we mentioned above, there is no overall $(-1)^{\#}$ sign since there is no chiral multiplet of $R$-charge $+2$ in our construction.} (${\mbs{n}=0}$ in \eqref{eq:the correlation functions of GLSMs observables})
\[
	\mcal{Z}\equiv \langle 1\rangle_{S^2,0}=\frac{1}{N!}\bigointsss_{\tenofo{JK}}\sd\sigma_1^\mbb{C}\wedge \dots\wedge \sd \sigma_N^\mbb{C} \mcal{Z}\bigl(\sigma^\mbb{C}_1,\dots,\si_N^\mbb{C}\bigr),
\]
where
\begin{equation}\label{eq:partition function on the coulomb branch, derived}
 \mcal{Z}\bigl(\sigma^\mbb{C}_1,\dots,\si^\mbb{C}_N\bigr)\equiv \prod_{a=1}^{N}\prod_{i=1}^L\frac{1}{\bigl(\sigma^\mbb{C}_a-m^\mbb{C}_i+\frac{1}{2}\hbar^\mbb{C}\bigr)\bigl(-\sigma^\mbb{C}_a+m^\mbb{C}_i+\frac{1}{2}\hbar^\mbb{C}\bigr)}\cdot\prod_{\substack{a\ne b}}^N(\si^\mbb{C}_{a}-\si^\mbb{C}_b).
\end{equation}
The factor $1/N!$ comes from the order of the Weyl-group factor in \eqref{eq:the correlation functions of GLSMs observables}. This proves \eqref{eq:partition function of GLSM on the coulomb branch}. Substituting \eqref{eq:partition function on the coulomb branch, derived} in \eqref{eq:the correlation functions of GLSMs observables} gives \eqref{eq:generic expression of correlation functions of GLSM observables}. The poles of JK prescription for the positive (effective) FI parameter ($\xi>0$, where $\xi$ is defined in \eqref{eq:the definition of effective FI parameter}) are given by
\begin{gather}\label{eq:the JK poles}
	\tenofo{JK Poles}\colon\ \sigma_a^\mbb{C}=m_i^\mbb{C}-\frac{\hbar^\mbb{C}}{2}, \qquad a=1,\dots,N,\qquad i=1,\dots,L,
\end{gather}
coming from the intersection of hyperplanes $\mcal{H}^k_{\mcal{I}}$ in \eqref{eq:the singular hyperplanes} associated with fundamental chiral multiplets $\mcal{Q}\indices{_a^i}$. Note that we need to take the contribution from all poles, i.e., all pairs $(a,i)$. In the following, we provide further details on this choice of JK poles in the computation of $R$-matrix.

{\bf Jeffrey--Kirwan poles for the computation of $\boldsymbol{R}$-matrix.}
Let us give the details of the choice \eqref{eq:the JK poles} of JK poles necessary for the computation of $R$-matrix in \ref{sec:construction of geometric r-matrix}. Although we have stated the poles in \eqref{eq:the JK poles}, we give more details on this choice using the formalism of \cite{ClossetCremonesiPark201504}, as reviewed in Section~\ref{sec:two-point function of GLSM observables}. As we have explained there, the choice of JK poles depends on the choice of the (effective) FI parameter $\xi$, defined in \eqref{eq:the definition of effective FI parameter}.

For the $1$-magnon sector, we have a $\U(1)$ gauge theory with $\SU(2)$ fundamental $\mcal{Q}^i$ and antifundamental $\wt{\mcal{Q}}_i$ flavors. Using \eqref{eq:the weights for the fundamental and antifundamental representations} and \eqref{eq:the singular hyperplanes}, the equations of hyperplanes \eqref{eq:the singular hyperplanes} are given by
\begin{gather*}
		\mcal{H}^{k^\mcal{Q}_1}_{\mcal{I}_{\mcal{Q}^1}}=\left\{\si^\mbb{C}\,\biggl|\, \si^\mbb{C}=m^\mbb{C}_1-\frac{\hbar^\mbb{C}}{2}-\left(k_{1}^\mcal{Q}-\frac{n}{2}\right)\epsilon\right\} , \qquad k_{1}^\mcal{Q}\in[0,n]_{\tenofo{int}},
		\\
		\mcal{H}^{k^\mcal{Q}_2}_{\mcal{I}_{\mcal{Q}^2}}=\left\{\si^\mbb{C}\,\biggl|\, \si^\mbb{C}=m^\mbb{C}_2-\frac{\hbar^\mbb{C}}{2}-\left(k_{2}^\mcal{Q}-\frac{n}{2}\right)\epsilon\right\} , \qquad k_{2}^\mcal{Q}\in[0,n]_{\tenofo{int}},
		\\
		\mcal{H}^{k^{\wt{\mcal{Q}}}_1}_{\mcal{I}_{\wt{\mcal{Q}}_1}}=\left\{\si^\mbb{C}\,\biggl|\, -\si^\mbb{C}=-m^\mbb{C}_1-\frac{\hbar^\mbb{C}}{2}-\left(k_{1}^{\wt{\mcal{Q}}}+\frac{n}{2}\right)\epsilon\right\} , \qquad k_{1}^{\wt{\mcal{Q}}}\in[0,n]_{\tenofo{int}},
		\\
		\mcal{H}^{k^{\wt{\mcal{Q}}}_2}_{\mcal{I}_{\wt{\mcal{Q}}_2}}=\left\{\si^\mbb{C}\,\biggl|\, -\si^\mbb{C}=-m_2^\mbb{C}-\frac{\hbar^\mbb{C}}{2}-\left(k_{2}^{\wt{\mcal{Q}}}+\frac{n}{2}\right)\epsilon\right\} , \qquad k_{2}^{\wt{\mcal{Q}}}\in[0,n]_{\tenofo{int}},
\end{gather*}
with $m^\mbb{C}_1+m^\mbb{C}_2=0$. The weight diagram is
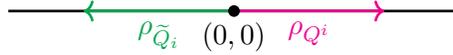
\begin{figure}[H]\centering
	\WeightDiagramOneMagnonSector
\end{figure}
We only need to consider the zero-flux sector $n=0$. From \eqref{eq:the generic cones}, we see that there are two ``cones"
\[
		\tenofo{Con}(\mscr{W}_{\mcal{I}_{\mcal{Q}^1}},\mscr{W}_{\mcal{I}_{\mcal{Q}^2}})=\tenofo{Con}(+1,+1),\qquad \tenofo{Con}(\mscr{W}_{\mcal{I}_{\wt{\mcal{Q}}_1}},\mscr{W}_{\mcal{I}_{\wt{\mcal{Q}}_2}})=\tenofo{Con}(-1,-1).
	\]
Using \eqref{eq:the generic prescription for JK residues} and taking $\xi>0$ and using \eqref{eq:the generic prescription for JK residues}, we see that only the poles coming from \smash{$H^{k^\mcal{Q}_1}_{\mcal{Q}_1}$} and \smash{$H^{k^\mcal{Q}_2}_{\mcal{Q}_2}$}, which is the same as \eqref{eq:the JK poles}.

For the $2$-magnon sector, the gauge group is $\U(2)$ and the $\SU(2)\times\U(1)_\hbar$ matter multiplet consists of $\mcal{Q}\indices{_a^i}$ and $\wt{\mcal{Q}}\indices{^a_i}$ with \smash{$q_\mcal{Q}^R=q_{\wt{\mcal{Q}}}^R=0$}. \eqref{eq:the weights for the fundamental and antifundamental representations}, and \eqref{eq:quantities W(si) and W(n)}, we can write the equations for singular hyperplanes \eqref{eq:the singular hyperplanes} associated to the fundamental chiral multiplets
\begin{gather*}
		\mcal{H}^{k^\mcal{Q}_{11}}_{\mcal{I}_{\mcal{Q}\indices{_1^1}}} =\left\{\bigl(\si_1^\mbb{C},\si_2^\mbb{C}\bigr)\,\biggl|\, \sigma_1^\mbb{C}=m_1^\mbb{C}-\frac{\hbar^\mbb{C}}{2}-\left(k_{11}^\mcal{Q}-\frac{n_1}{2}\right)\epsilon\right\} , \qquad k_{11}^\mcal{Q} \in[0,n_1]_{\tenofo{int}},
		\\
		\mcal{H}^{k^\mcal{Q}_{12}}_{\mcal{I}_{\mcal{Q}\indices{_1^2}}} =\left\{\bigl(\si_1^\mbb{C},\si_2^\mbb{C}\bigr)\,\biggl|\, \sigma_1^\mbb{C}=m_2^\mbb{C}-\frac{\hbar^\mbb{C}}{2}-\left(k_{12}^\mcal{Q}-\frac{n_1}{2}\right)\epsilon\right\} , \qquad k_{12}^\mcal{Q} \in[0,n_1]_{\tenofo{int}},
		\\
		\mcal{H}^{k^\mcal{Q}_{21}}_{\mcal{I}_{\mcal{Q}\indices{_2^1}}} =\left\{\bigl(\si_1^\mbb{C},\si_2^\mbb{C}\bigr)\,\biggl|\, \sigma^\mbb{C}_2=m_1^\mbb{C}-\frac{\hbar^\mbb{C}}{2}-\left(k_{21}^\mcal{Q}-\frac{n_2}{2}\right)\epsilon\right\} , \qquad k_{21}^\mcal{Q} \in[0,n_2]_{\tenofo{int}},
		\\
		\mcal{H}^{k^\mcal{Q}_{22}}_{\mcal{I}_{\mcal{Q}\indices{_2^2}}} =\left\{\bigl(\si_1^\mbb{C},\si_2^\mbb{C}\bigr)\,\biggl|\, \sigma_2^\mbb{C}=m_2^\mbb{C}-\frac{\hbar^\mbb{C}}{2}-\left(k_{22}^\mcal{Q}-\frac{n_2}{2}\right)\epsilon\right\} , \qquad k_{22}^\mcal{Q} \in[0,n_2]_{\tenofo{int}}.
\end{gather*}
Similarly, the hyperplanes associated to the antifundamental chiral multiplets are
\begin{gather*}
		H^{k^{\wt{\mcal{Q}}}_{11}}_{\mcal{I}_{{\wt{\mcal{Q}}}\indices{_1^1}}} =\left\{\bigl(\si_1^\mbb{C},\si_2^\mbb{C}\bigr)\,\biggl|\, -\sigma_1^\mbb{C}=-m_1^\mbb{C}-\frac{\hbar^\mbb{C}}{2}-\left(k_{11}^{\wt{\mcal{Q}}}+\frac{n_1}{2}\right)\epsilon\right\} , \qquad k_{11}^{\wt{\mcal{Q}}} \in[0,n_1]_{\tenofo{int}},
		\\ H^{k^{\wt{\mcal{Q}}}_{12}}_{\mcal{I}_{{\wt{\mcal{Q}}}\indices{_1^2}}} =\left\{\bigl(\si_1^\mbb{C},\si_2^\mbb{C}\bigr)\,\biggl|\, -\sigma_1^\mbb{C}=-m_2^\mbb{C}-\frac{\hbar^\mbb{C}}{2}-\left(k_{12}^{\wt{\mcal{Q}}}+\frac{n_1}{2}\right)\epsilon\right\} , \qquad k_{12}^{\wt{\mcal{Q}}} \in[0,n_1]_{\tenofo{int}},
		\\
		H^{k^{\wt{\mcal{Q}}}_{21}}_{\mcal{I}_{{\wt{\mcal{Q}}}\indices{_2^1}}} =\left\{\bigl(\si_1^\mbb{C},\si_2^\mbb{C}\bigr)\,\biggl|\, -\sigma_2^\mbb{C}=-m_1^\mbb{C}-\frac{\hbar^\mbb{C}}{2}-\left(k_{21}^{\wt{\mcal{Q}}}+\frac{n_2}{2}\right)\epsilon\right\} , \qquad k_{21}^{\wt{\mcal{Q}}} \in[0,n_2]_{\tenofo{int}},
		\\
		H^{k^{\wt{\mcal{Q}}}_{22}}_{\mcal{I}_{{\wt{\mcal{Q}}}\indices{_2^2}}} =\left\{\bigl(\si_1^\mbb{C},\si_2^\mbb{C}\bigr)\,\biggl|\, -\sigma_2^\mbb{C}=-m_2^\mbb{C}-\frac{\hbar^\mbb{C}}{2}-\left(k_{22}^{\wt{\mcal{Q}}}+\frac{n_2}{2}\right)\epsilon\right\} , \qquad k_{22}^{\wt{\mcal{Q}}} \in[0,n_2]_{\tenofo{int}}.
\end{gather*}
We only need the zero-flux sector $n_1=n_2=0$. To specify the JK poles, we take %
$\xi$ to be in the first quadrant, and using \eqref{eq:the generic cones} and \eqref{eq:the generic prescription for JK residues}, we see that we need to take the contribution from the hyperplanes forming the following cones associated to $\mcal{Q}\indices{_a^i}$
\smash{$\tenofo{Con}(\mscr{W}_{\mcal{I}_{\mcal{Q}\indices{_1^i}}},\mscr{W}_{\mcal{I}_{\mcal{Q}\indices{_2^i}}})$}, $ i=1,2$,
which contain $\xi$. These poles are as given in \eqref{eq:the JK poles}.

\subsection*{Acknowledgements}

We thank Kevin Costello, Yuan Miao, Hiraku Nakajima, Tadashi Okazaki, Surya Raghavendran, Junya Yagi, Masahito Yamazaki, and Zijun Zhou for useful discussions. We thank Davide Gaiotto, and Andrei Okounkov for reading the draft and sending us valuable feedback. Special thanks go to Mykola Dedushenko for explaining some parts of his work with Nikita Nekrasov and for providing extensive commentary on an earlier version of this work. We thank the Banff International Research Station for generously hosting us during the last stage of this project. N.I. is supported by the Huawei Young Talents Program Fellowship at IHES. The work of S.F.M. is funded by the Natural Sciences and Engineering Research Council of Canada (NSERC) funding number SAPIN-2022-00028, and also in part by the Alfred P. Sloan Foundation, grant FG-2020-13768. S.F.M would like to thank Davide Gaiotto for his support during the visit to the Perimeter Institute. Kavli IPMU is supported by World Premier International Research Center Initiative (WPI), MEXT, Japan.

\LastPageEnding

\end{document}

%% file: Figures_and_Tables.tex
\newcommand{\SetupForComputationOfStableEnvelopes}{
\begin{tikzpicture}

\draw[line width=1] (-5,0) ellipse (.5cm and 1.5cm);
\node[] at (-4.9,-1.8) {$y_-$};
\node[] at (-4.9,+1.9) {$|\mscr{N}_{\mfr C}(p_-)\rangle$};
\draw[line width=1pt] (-5.05,-.5) to [out=60,in=-60] (-5.05,.5);
\draw[line width=1pt] (-5,-.35) to [out=120,in=240] (-5,.35);

\draw[line width=1] (-5,+1.5) -- (+5,+1.5);
\draw[line width=1] (-5,-1.5) -- (+5,-1.5);

\begin{scope}
 \draw[line width=1pt, densely dashed] (0,0) ellipse (.5cm and 1.5cm);
 \clip[] (0,1.5) rectangle (0.7,-1.5);
 \draw[line width=1] (0,0) ellipse (.5cm and 1.5cm);
\end{scope}
\path[line width=1] (0,0) ellipse (.5cm and 1.5cm);
\node[] at (0,-1.8) {$y=0$};
\node[] at (-.1,+1.9) {$ \mcal{J}(m_{\mfk{C}}, 0)$};
\draw[line width=1pt, densely dashed] (-.05,-.5) to [out=60,in=-60] (-.05,.5);
\draw[line width=1pt, densely dashed] (0,-.35) to [out=120,in=240] (0,.35);

\begin{scope}
 \draw[line width=1pt, densely dashed] (5,0) ellipse (.5cm and 1.5cm);
 \clip[] (5,1.5) rectangle (5.7,-1.5);
 \draw[line width=1] (5,0) ellipse (.5cm and 1.5cm);
\end{scope}
\draw[line width=1pt, densely dashed] (5,0) ellipse (.5cm and 1.5cm);
\node[] at (4.9,-1.8) {$y_+$};
\node[] at (5,+1.9) {$\langle\mscr{B}_{\mbb{L}(p_+),p_+}|$};
\draw[line width=1pt, densely dashed] (4.95,-.5) to [out=60,in=-60] (4.95,.5);
\draw[line width=1pt, densely dashed] (5,-.35) to [out=120,in=240] (5,.35);

\path[line width=2pt,name path=A1] (-5,1.5) to [out=0,in=0,looseness=.53] (-5,-1.5);
\path[line width=2pt,name path=A2] (-5.01,1.5) to [out=180,in=180,looseness=.53] (-5.01,-1.5);
\path[line width=1pt,name path=A3] (-5.01,-.5) to [out=60,in=-60] (-5.01,.5);
\path[line width=1pt,name path=A4] (-5,-.35) to [out=120,in=240] (-5,.35);
\tikzfillbetween[of=A1 and A3]{OliveGreen!90!, opacity=0.2};
\tikzfillbetween[of=A2 and A4]{OliveGreen!90!, opacity=0.2};

\path[line width=2pt,name path=B1] (0,1.5) to [out=0,in=0,looseness=.53] (0,-1.5);
\path[line width=2pt,name path=B2] (0,1.5) to [out=180,in=180,looseness=.55] (0,-1.5);
\path[line width=1pt,name path=B3] (-.03,-.5) to [out=60,in=-60] (-.03,.5);
\path[line width=1pt,name path=B4] (-.03,-.35) to [out=120,in=240] (-.03,.35);
\tikzfillbetween[of=B1 and B3]{Cerulean, opacity=0.2};
\tikzfillbetween[of=B2 and B4]{Cerulean, opacity=0.2};

\path[line width=2pt,name path=C1] (5,1.5) to [out=0,in=0,looseness=.53] (5,-1.5);
\path[line width=2pt,name path=C2] (5,1.5) to [out=180,in=180,looseness=.55] (5,-1.5);
\path[line width=1pt,name path=C3] (4.97,-.5) to [out=60,in=-60] (4.97,.5);
\path[line width=1pt,name path=C4] (4.97,-.35) to [out=120,in=240] (4.97,.35);
\tikzfillbetween[of=C1 and C3]{OliveGreen!90!, opacity=0.2};
\tikzfillbetween[of=C2 and C4]{OliveGreen!90!, opacity=0.2};

\node at (-2.5,0) {$m_{\mfk{C}}\ne 0$};
\node at (+2.5,0) {$m=0$};

\end{tikzpicture}
}

\newcommand{\FromTwoDimensionToOneDimension}{
	\begin{subfigure}{.40\textwidth}\centering
		\begin{tikzpicture}[scale=0.8]
			\begin{scope}
				\draw[line width=1] (0,0) circle (1.5cm);
				\draw[dashed,line width=1pt,Green] (0,0) ellipse (.5cm and 1.5cm);
				\clip (-1.6,-1.5) rectangle (0 cm, 1.5 cm);
				\draw[line width=1pt,Green] (0,0) ellipse (.5cm and 1.5 cm);
			\end{scope}
			\draw[fill] (-1.5,0) circle (2pt);
			\node[xshift=-.4cm] at (-1.9,0) {$\mcal{O}_{\msf{L}}(\mbs{\si}^\mbb{C})$};
			\draw[fill] (1.5,0) circle (2pt);
			\node[xshift=.4cm] at (1.9,0) {$\mcal{O}_{\msf{R}}(\mbs{\si}^\mbb{C})$};
		\end{tikzpicture}
		\subcaption{The sphere on which the $2d$ A-twisted gauge theory is defined. Two observables $\mcal{O}_{\msf{L}}(\mbs{\si}^\mbb{C})$ and $\mcal{O}_{\msf{R}}(\mbs{\si}^\mbb{C})$ are inserted on sphere.}
	\end{subfigure}
	~
	\begin{subfigure}{.55\textwidth}\centering
		\begin{tikzpicture}[scale=0.8]
			\begin{scope}
				\clip (-1.6,-1.5) rectangle (-.001 cm, 1.5 cm);
				\draw[line width=1] (0,0) circle (1.5cm);
			\end{scope}
			\begin{scope}
				\draw[Green, dashed,line width=1pt] (0,0) ellipse (.5cm and 1.5cm);
				\clip (-1.5,-1.5) rectangle (-.001 cm, 1.5 cm);
				\draw[line width=1pt, Green] (0,0) ellipse (.5cm and 1.5 cm);
			\end{scope}
			\begin{scope}
				\draw[Green, dashed,line width=1pt] (4,0) ellipse (.5cm and 1.5cm);
				\clip (3.4,-1.5) rectangle (4 cm, 1.5 cm);
				\draw[line width=1pt, Green] (4,0) ellipse (.5cm and 1.5 cm);
			\end{scope}
			\begin{scope}
				\clip (4,-1.5) rectangle (6.5cm, 1.5cm);
				\draw[line width=1] (4,0) circle (1.5cm);
			\end{scope}
			\draw[line width=1] (0,1.48) -- (4,1.48);
			\draw[line width=1] (0,-1.48) -- (4,-1.48);
			\draw[fill] (-1.5,0) circle (2pt);
			\node[xshift=-.4cm] at (-1.9,0) {$\mcal{O}_{\msf{L}}(\mbs{\si}^\mbb{C})$};
			\draw[fill] (5.5,0) circle (2pt);
			\node[xshift=.4cm] at (6.0,0) {$\mcal{O}_{\msf{R}}(\mbs{\si}^\mbb{C})$};
		\end{tikzpicture}
		\subcaption{The deformation of the sphere to a geometry involving a~long cylinder and two capped cigar-like geometries.\\ \null}
	\end{subfigure}	

\vspace{-7mm}

	\begin{subfigure}{.40\textwidth}\centering\vspace{.5cm}
		\begin{tikzpicture}[scale=0.70]
			\begin{scope}
				\draw[Green, dashed,line width=1pt] (0,0) ellipse (.5cm and 1.5cm);
				\clip (-1.5,-1.5) rectangle (-.001 cm, 1.5 cm);
				\draw[line width=1, Green] (0,0) ellipse (.5cm and 1.5 cm);
			\end{scope}
			\draw[line width=1] (0,1.48) -- (3.93,1.48);
			\draw[line width=1] (0,-1.48) -- (3.93,-1.48);
			\draw[Green,line width=1pt] (4,0) ellipse (.5cm and 1.5cm);
			
			\node[xshift=0cm] at (-1.7,0) {$|\mcal{O}_{\msf{L}}(\mbs{\si}^\mbb{C})\rangle$};
			\node[xshift=0cm] at (5.8,0) {$|\mcal{O}_{\msf{R}}(\mbs{\si}^\mbb{C})\rangle$};
		\end{tikzpicture}
		\subcaption{Doing a path integral in the presence of $\mcal{O}_{\msf{L}}(\mbs{\si}^\mbb{C})$ and $\mcal{O}_{\msf{R}}(\mbs{\si}^\mbb{C})$ defines the corresponding state on the boundary circles. We can then cap off the cigar-like ends.}
	\end{subfigure}
	~
	\begin{subfigure}{.55\textwidth}\centering\vspace{.5cm}
		\begin{tikzpicture}
			\draw[draw=none, line width=1pt] (0,0) ellipse (.5cm and 1.5cm);
			\draw[fill] (-2,0) circle (2pt);
			\draw[line width=1] (-2,0) -- (2,0);
			\draw[fill] (2,0) circle (2pt);
			
			\node at (-2.4,0) {$\mscr{B}_{\msf{L}}$};
			\node at (2.4,0) {$\mscr{B}_{\msf{R}}$};
		\end{tikzpicture}
		\subcaption{Compactification of the $2d$ theory in the cylinderical region would lead to a $1d$ quantum mechanics and the boundary states become boundary conditions $\mscr{B}_{\msf{L}}$ and $\mscr{B}_{\msf{R}}$ in the quantum mechanics.}
	\end{subfigure}
}

\newcommand{\WeightDiagramOneMagnonSector}{
\begin{tikzpicture}
		\draw[line width=1pt] (-3,0) -- (3,0);
		\draw[line width=1, ->, Magenta] (0,0)node[below,black]{$(0,0)$} -- (2,0)node[below,xshift=-1cm]{$\rho_{Q^i}$};
		\draw[line width=1, ->, Green] (0,0) -- (-2,0)node[below,xshift=1cm]{$\rho_{\wt{Q}_i}$};
		\draw[fill] (0,0) circle (2pt);
	\end{tikzpicture}
}

\newcommand{\FromSphereToCylinder}{
\begin{tikzpicture}
    \begin{scope}
    \draw[line width=1] (-5,0) circle (1.5cm);
    \draw[dashed,line width=1pt,Green] (-5,0) ellipse (.5cm and 1.5cm);
    \clip (-6.5,-1.5) rectangle (-5, 1.5);
    \draw[line width=1pt,Green] (-5,0) ellipse (.5cm and 1.5 cm);
    \end{scope}

    \node at (-5.8,0) {$S^1_\C$};

    \node at (-2.5,0) {$\Longrightarrow$};

    \draw[line width=1] (0,1.5) arc [start angle=90,   end angle=270,
                   x radius=1.5cm, y radius=1.5cm];
    \draw[line width=1] (0,1.5) -- (2,1.5);
    \draw[line width=1] (0,-1.5) -- (2,-1.5);
    \draw[line width=1] (2,1.5) arc [start angle=90,   end angle=-90,
                   x radius=1.5cm, y radius=1.5cm];
    \node at (1,-1.8) {$I$};

    \node at (-.8,0) {$S^1_\C$};

    \begin{scope}
    \draw[dashed,line width=1pt,Green] (0,0) ellipse (.5cm and 1.5cm);
    \clip (-1.5,-1.5) rectangle (0, 1.5);
    \draw[line width=1pt,Green] (0,0) ellipse (.5cm and 1.5 cm);
    \end{scope}

    \begin{scope}
    \draw[dashed,line width=1pt,Green] (2,0) ellipse (.5cm and 1.5cm);
    \clip (.5,-1.5) rectangle (2, 1.5);
    \draw[line width=1pt,Green] (2,0) ellipse (.5cm and 1.5 cm);
    \end{scope}

    \node at (4.5,0) {$\Longrightarrow$};

    \draw[line width=1] (6,1.5) -- (8,1.5);
    \draw[line width=1] (6,-1.5) -- (8,-1.5);
    \node at (7,-1.8) {$I$};

    \begin{scope}
    \draw[dashed,line width=1pt,Green] (6,0) ellipse (.5cm and 1.5cm);
    \clip (4.5,-1.5) rectangle (6, 1.5);
    \draw[line width=1pt,Green] (6,0) ellipse (.5cm and 1.5 cm);
    \end{scope}

    \draw[line width=1pt,Green] (8,0) ellipse (.5cm and 1.5cm);
    \node at (8.8,0) {$S^1_\C$};

\end{tikzpicture}
}

\newcommand{\FromFourToLowerDimensions}{
    \begin{tikzpicture}[scale=0.99]

    \node at (-.2,0) {$\mscr{T}_{\text{4d}}$ on:};
    \draw[fill] (1,1) circle (2pt);
    \draw[line width=1pt] (1,1) -- (1,-1);
    \draw[fill] (1,-1) circle (2pt);
    \node at (1,-1.4) {$I$};

    \node at (1.5, 0) {$\times$};
    \node at (3, 0) {$\times$};
    \node at (4.6, 0) {$\times$};

    \draw[line width=1pt] (2.2,0) ellipse (.3cm and 1 cm);
    \node at (2.2,-1.4) {$S^1_\A$};

    \draw[line width=1pt] (3.8,0) ellipse (.3cm and 1 cm);
    \node at (3.8,-1.4) {$S^1_\B$};

    \draw[line width=1pt] (5.4,0) ellipse (.3cm and 1 cm);
    \node at (5.4,-1.4) {$S^1_\C$};

    \draw[double,->] (2.4,-1.8) -- (0,-3);
    \node[rotate=26,yshift=.25cm] at (1.2,-2.4) {Reduction on};
    \node[rotate=26,yshift=-.3cm] at (1.2,-2.4) {$S^1_\A$};

    \begin{scope}[xshift=-4.5cm,yshift=-4cm]

    \node at (-.2,0) {$\mscr{T}'_{\text{3d}}$ on:};
    \draw[fill] (1,1) circle (2pt);
    \draw[line width=1pt] (1,1) -- (1,-1);
    \draw[fill] (1,-1) circle (2pt);
    \node at (1,-1.4) {$I$};

    \node at (1.5, 0) {$\times$};
    \node at (3, 0) {$\times$};

    \draw[line width=1pt] (2.2,0) ellipse (.3cm and 1 cm);
    \node at (2.2,-1.4) {$S^1_\B$};

    \draw[line width=1pt] (3.8,0) ellipse (.3cm and 1 cm);
    \node at (3.8,-1.4) {$S^1_\C$};

    \end{scope}

    \draw[double,->] (2.6,-1.8) -- (2.6,-1.8-2.7);
    \node[rotate=90,yshift=.25cm] at (2.6,-3.2) {Reduction on};
    \node[rotate=90,yshift=-.3cm] at (2.6,-3.2) {$S^1_\B$};

    \begin{scope}[xshift=0cm,yshift=-6.5cm]

    \node at (-.2,0) {$\mscr{T}''_{\text{3d}}$ on:};

    \draw[fill] (1,1) circle (2pt);
    \draw[line width=1pt] (1,1) -- (1,-1);
    \draw[fill] (1,-1) circle (2pt);
    \node at (1,-1.4) {$I$};

    \node at (1.5, 0) {$\times$};
    \node at (3, 0) {$\times$};

    \draw[line width=1pt] (2.2,0) ellipse (.3cm and 1 cm);
    \node at (2.2,-1.4) {$S^1_\A$};

    \draw[line width=1pt] (3.8,0) ellipse (.3cm and 1 cm);
    \node at (3.8,-1.4) {$S^1_\C$};

    \end{scope}

    \draw[double,->] (2.8,-1.8) -- (+5.2,-3);
    \node[rotate=-26,yshift=.25cm] at (4,-2.4) {Reduction on};
    \node[rotate=-26,yshift=-.3cm] at (4,-2.4) {$S^1_\C$};

    \begin{scope}[xshift=+6.5cm,yshift=-4cm]

    \node at (-.2,0) {$\mscr{T}'''_{\text{3d}}$ on:};

    \draw[fill] (1,1) circle (2pt);
    \draw[line width=1pt] (1,1) -- (1,-1);
    \draw[fill] (1,-1) circle (2pt);
    \node at (1,-1.4) {$I$};

    \node at (1.5, 0) {$\times$};
    \node at (3, 0) {$\times$};

    \draw[line width=1pt] (2.2,0) ellipse (.3cm and 1 cm);
    \node at (2.2,-1.4) {$S^1_\A$};

    \draw[line width=1pt] (3.8,0) ellipse (.3cm and 1 cm);
    \node at (3.8,-1.4) {$S^1_\B$};

    \end{scope}

    \draw[double,->] (-2.4,-6) -- (-2.4,-8.5);
    \node[rotate=90,yshift=.25cm] at (-2.4,-7.2) {Reduction on};
    \node[rotate=90,yshift=-.3cm] at (-2.4,-7.2) {$S^1_\B$};

    \begin{scope}[xshift=+1cm,yshift=-6cm]

    \begin{scope}[xshift=-4.5cm,yshift=-4cm]

    \node at (-.2,0) {$\mscr{T}'_{\text{2d}}$ on:};
    \draw[fill] (1,1) circle (2pt);
    \draw[line width=1pt] (1,1) -- (1,-1);
    \draw[fill] (1,-1) circle (2pt);

    \node at (1.5, 0) {$\times$};

    \draw[line width=1pt] (2.2,0) ellipse (.3cm and 1 cm);
    \node at (2.2,-1.4) {$S^1_\C$};

    \end{scope}

    \end{scope}

    \draw[double,->] (2.6,-8.5) -- (2.6,-11);
    \node[rotate=90,yshift=.25cm] at (2.6,-9.7) {Reduction on};
    \node[rotate=90,yshift=-.3cm] at (2.6,-9.7) {$S^1_\A$};

    \begin{scope}[yshift=-6cm,xshift=+1cm]

    \begin{scope}[xshift=0cm,yshift=-6.5cm]

    \node at (-.2,0) {$\mscr{T}''_{\text{2d}}$ on:};

    \draw[fill] (1,1) circle (2pt);
    \draw[line width=1pt] (1,1) -- (1,-1);
    \draw[fill] (1,-1) circle (2pt);
    \node at (1,-1.4) {$I$};

    \node at (1.5, 0) {$\times$};

    \draw[line width=1pt] (2.2,0) ellipse (.3cm and 1 cm);
    \node at (2.2,-1.4) {$S^1_\C$};

    \end{scope}

    \end{scope}

    \draw[double,->] (8.4,-6) -- (8.4,-8.5);
    \node[rotate=-90,yshift=.25cm] at (8.4,-7.2) {Reduction on};
    \node[rotate=-90,yshift=-.3cm] at (8.4,-7.2) {$S^1_\B$};

    \begin{scope}[xshift=1cm, yshift=-6cm]

    \begin{scope}[xshift=+6.5cm,yshift=-4cm]

    \node at (-.2,0) {$\mscr{T}'''_{\text{2d}}$ on:};

    \draw[fill] (1,1) circle (2pt);
    \draw[line width=1pt] (1,1) -- (1,-1);
    \draw[fill] (1,-1) circle (2pt);
    \node at (1,-1.4) {$I$};

    \node at (1.5, 0) {$\times$};

    \draw[line width=1pt] (2.2,0) ellipse (.3cm and 1 cm);
    \node at (2.2,-1.4) {$S^1_\A$};

    \end{scope}

    \end{scope}

    \draw[double,->] (8.4,-12) -- (8.4,-14.5);
    \node[rotate=-90,yshift=.25cm] at (8.4,-13.2) {Reduction on};
    \node[rotate=-90,yshift=-.3cm] at (8.4,-13.2) {$S^1_\A$};

    \begin{scope}[xshift=0cm,yshift=-6cm]

    \begin{scope}[xshift=1cm, yshift=-6cm]

    \begin{scope}[xshift=+6.5cm,yshift=-4cm]

    \node at (-.2,0) {$\mscr{T}'''_{\text{1d}}$ on:};

    \draw[fill] (1,1) circle (2pt);
    \draw[line width=1pt] (1,1) -- (1,-1);
    \draw[fill] (1,-1) circle (2pt);
    \node at (1,-1.4) {$I$};

    \end{scope}

    \end{scope}
    \end{scope}

    \end{tikzpicture}
}

\newcommand{\TwoJanusInterfacesInTwoDimensions}{
\begin{tikzpicture}

\begin{scope}
    \draw[line width=1pt] (0,0) circle (1.5cm);
    \draw[dashed,line width=1pt,Cyan] (-.75,0) ellipse (.2cm and 1.3cm);
    \clip (-1.5,-1.3) rectangle (-.75, 1.3);
    \draw[line width=2pt,Cyan] (-.75,0) ellipse (.2 cm and 1.3 cm);
    \end{scope}
\node at (-1.3,1.6) {$\mcal{J}(m_\mfk{C},0)$};

\begin{scope}
    \draw[line width=1pt] (0,0) circle (1.5cm);
    \draw[dashed,line width=1pt,Cyan] (.75,0) ellipse (.2cm and 1.3cm);
    \clip (0,-1.3) rectangle (.75, 1.3);
    \draw[line width=2pt,Cyan] (.75,0) ellipse (.2 cm and 1.3 cm);
\end{scope}
\node at (+1.3,1.6) {$\mcal{J}(0,m_{\mfk{C}'})$};

    \draw[fill] (-1.5,0) circle (2pt);
    \draw[fill] (+1.5,0) circle (2pt);
    \node at (-2.2,0) {$\mcal{O}(\sigma^\mbb{C})$};
    \node at (+2.2,0) {$\mcal{O}'(\sigma^\mbb{C})$};
\end{tikzpicture}
}

\newcommand{\OneJanusInterfaceInTwoDimensions}{
\begin{tikzpicture}

\begin{scope}
    \draw[line width=1pt] (0,0) circle (1.5cm);
    \draw[dashed,line width=1pt,Cyan] (0,0) ellipse (.2cm and 1.5cm);
    \clip (-.75,-1.5) rectangle (0, 1.5);
    \draw[line width=2pt,Cyan] (0,0) ellipse (.2 cm and 1.5 cm);
    \end{scope}
\node at (0,1.8) {$\mcal{J}(m_\mfk{C},m_{-\mfk{C}})$};

    \draw[fill] (-1.5,0) circle (2pt);
    \draw[fill] (+1.5,0) circle (2pt);
    \node at (-2.2,0) {$\mcal{O}(\sigma^\mbb{C})$};
    \node at (+2.2,0) {$\mcal{O}'(\sigma^\mbb{C})$};
\end{tikzpicture}
}